\documentclass{article}
\usepackage[utf8]{inputenc}
\usepackage[T1]{fontenc}
\usepackage{amssymb,amsmath,amsthm,mathtools,bm}
\usepackage[margin=1.25in]{geometry}
\usepackage[round]{natbib}
\usepackage{float}
\usepackage{graphicx}
\usepackage[dvipsnames]{xcolor}
\usepackage{booktabs}
\usepackage{tabularx,array}
\usepackage{enumitem}
\usepackage{subcaption}
\usepackage{microtype}
\usepackage[section]{placeins}
\usepackage[colorlinks,citecolor=BlueViolet,linkcolor=black]{hyperref}

\newtheorem{assumption}{Assumption}
\newtheorem{theorem}{Theorem}
\newtheorem{proposition}{Proposition}
\newtheorem{corollary}{Corollary}
\newtheorem{remark}{Remark}

\newtheorem{lemma}{Lemma}

\newcommand{\Var}{\mathrm{Var}}
\newcommand{\Cov}{\mathrm{Cov}}
\newcommand{\1}{\mathbf{1}}
\newcommand{\I}{\mathcal{I}}
\newcolumntype{Y}{>{\raggedright\arraybackslash}X}
\title{Design-Based Inference for Time-Series GMM}
\date{June 2026}
\author{Thomas Glinnan\thanks{Department of Economics, London School of Economics and Political Science. Email: \href{mailto:t.m.glinnan@lse.ac.uk}{t.m.glinnan@lse.ac.uk}}\\
LSE}

\begin{document}
\maketitle

\begin{abstract}
This paper studies inference for time-series GMM when uncertainty comes from shock assignment within a realized historical episode. Rather than treating the data as one random draw from a population of hypothetical economies, the framework conditions on the historical environment and considers alternative realizations of shocks and instruments. For locally correctly specified GMM estimators, the centered moment has design long-run variance $\Omega_R$, which determines the sandwich covariance for the finite-history estimand. Conventional HAC estimators instead converge to $\Omega_R^+=\Omega_R+\Omega_\mu$, where $\Omega_\mu\succeq0$ is the long-run variance of the centered mean-moment path. HAC inference is therefore conservative for scalar functions of the finite-history estimand. Projection adjustment using predetermined covariates can reduce this HAC variance limit in Loewner order and, under an additional long-run orthogonality condition, yields a tighter conservative bound on the corresponding asymptotic covariance. Monte Carlo evidence shows when the distinction is quantitatively important. In a monetary-policy application, standard-error reductions from rich macro covariates provide a diagnostic for economically meaningful predictable variation in the mean-moment path.
\end{abstract}

\noindent\textbf{Keywords:} Design-based inference; Time-series GMM; HAC standard errors; Local projections; Finite-population asymptotics; Regression adjustment.

\noindent\textbf{JEL codes:} C12, C22, C32, C36.

\section{Introduction} \label{sec:introduction}
Suppose a researcher uses a local projection or SVAR to study a past historical episode, such as monetary-policy shocks over a particular time period. Often, researchers are not directly interested in how estimates of dynamic treatment effects---such as impulse responses---would vary across hypothetical economies drawn from the same stationary process. Instead, the question of interest may be the following: holding the historical environment fixed, how different would the estimated impulse response have been under alternative realizations of the economic shocks? This question often corresponds more closely to the object of interest in applied work, but it requires a different notion of variance from the one that is most commonly used. Standard HAC standard errors treat the observed moment path as a draw from an unconditional time-series process. By contrast, the design approach below conditions on the realized finite history and treats only the assignment shocks, and possibly instruments, as random. The estimand is therefore the minimizer of the sample-period population GMM criterion based on average design mean moments over the realized dates, with variance given by the design variance under re-randomization of the shocks.

The distinction is the time-series analogue of the finite-population variance calculation in randomized experiments. In cross-sectional settings, design-based inference asks about treatment effects \textit{in the sample}, with uncertainty arising only from treatment assignment. Analogously, this framework targets dynamic treatment effects \textit{for the observed time period} only, with uncertainty generated only by alternative realizations of shocks or instruments. In a simple randomized experiment, studied by \citet{Neyman1923_thesis}, the variance of a difference-in-means estimator contains an unidentified negative term proportional to treatment-effect heterogeneity, so the usual plug-in formula is conservative. Recent finite-population results extend this logic to regression, M-estimation, and overidentified GMM \citep{AbadieAtheyImbensWooldridge2020,Xu2021FinitePopulationMEstimators,KakehiMatsushitaOtsu2026FinitePopulationGMM}. The time-series question is how the same unidentified heterogeneity term changes when a shock at one date can affect outcomes and moments at other dates.

This paper answers that question for a class of time-series GMM estimators. The conditioning environment $\mathcal E_T$ collects the outcome maps, state-transition rules, deterministic sampling windows, and covariate paths that the design holds fixed. Conditional on $\mathcal E_T$, only the assignment shocks are re-randomized. The estimand is denoted $\theta_T^\star$, which may be, for example, the average $h$-period ahead impulse response over the observed time period. While standard methods such as LPs and SVARs can deliver valid point estimation of $\theta_T^\star$ under the design-based framework, statistical inference requires more care. Specifically, write the moment used in estimation as $g_t(W_t,\theta)=\mu_{T,t}(\theta) + e_{T,t}(W_t,\theta)$. Here $\mu_{T,t}(\theta)$ is the design-conditional mean of the date-$t$ moments, so the sequence $(\mu_{T,t}(\theta_T^\star))_{t\le T}$ is the \textit{mean path}, recording predictable date-by-date movement in moments over the fixed historical episode. The key decomposition is that the usual asymptotic variance $\Omega_R^+$ satisfies
\[
\Omega_R^+ = \Omega_R + \Omega_\mu, \qquad \Omega_\mu\succeq 0.
\]
Here $\Omega_R$ is the asymptotic design-based variance of the centered moments. Conventional HAC procedures estimate $\Omega_R^+$ because they are applied to the observed moment sequence and therefore include both design innovation variation and fixed movement in the mean path. The component $\Omega_\mu$ is the long-run variance obtained by applying the same HAC long-run-variance calculation to the mean path after subtracting its sample-period average. This is the dependent-data analogue of the finite-population heterogeneity correction and characterizes when $\Omega_R^+$ and $\Omega_R$ are not equal.

The first contribution of the paper is to show that conventional HAC estimators and HAC multiplier bootstraps estimate a conservative variance limit for scalar design-based inference. If the target is instead an unconditional population of hypothetical economies, as in some forecasting exercises, the usual asymptotic variance remains the relevant object. The results below apply to the finite-history question, where the realized historical environment is held fixed and only the assignment shocks are re-randomized.

The second contribution is a projection-adjusted, or regression-adjusted, refinement based on predetermined covariates. Projecting the moments on such covariates and applying the same HAC long-run-variance calculation to the residualized moments gives an adjusted variance limit $\Omega_R^+(r)$ satisfying $\Omega_R^+(r)\preceq\Omega_R^+$. This positive-semidefinite, or Loewner-order, reduction result does not by itself assert that the adjusted matrix is conservative for the design variance. The stronger interpretation $\Omega_R\preceq\Omega_R^+(r)\preceq\Omega_R^+$ requires long-run orthogonality, meaning zero HAC long-run cross-covariance, between the centered innovation and the part of the date-specific mean path left after projecting on the adjustment covariates. Equality with $\Omega_R$ holds when the residualized centered mean path has zero long-run variance; a simple sufficient case is that the centered mean path lies exactly in the linear span of the adjustment covariates. Without those restrictions, the adjusted matrix is a diagnostic and a feasible reduction of the HAC variance limit rather than an estimator of $\Omega_R$.

The paper studies estimators defined by unconditional moment restrictions, including just-identified LPs, stable VAR normal equations, proxy or event-study IV moments, and locally correctly specified overidentified GMM\@. The main text records primitive sufficient routes for fixed-dimensional LPs and stable VAR normal equations, while Appendix~\ref{app:primitive-verification} gives the corresponding details. The Monte Carlo section evaluates the variance limits directly under fixed conditioning environments. The simulations show that conventional HAC can be materially conservative when movement in the date-specific mean path is large and that aligned predetermined adjustment can remove a substantial part of the gap. They also show that covariate adjustment meaningfully improves standard errors only when the adjustment span tracks the predictable component of that mean path. The empirical application to monthly U.S. monetary-policy shocks follows the same logic. It asks whether a richer set of predetermined macro covariates explains part of the HAC moment variation in standard LP and state-dependent specifications. The exercise diagnoses the variance-limit distinction in a familiar monetary-shock setting rather than supplying new structural evidence about monetary transmission.

The paper proceeds as follows. Section~\ref{sec:motivation} introduces the design-based time-series perspective, Section~\ref{sec:related-literature} relates the argument to design-based econometrics and macro time-series inference, and Section~\ref{sec:setup-and-main-results} defines the finite-history GMM estimand and gives the first-order theory. Sections~\ref{sec:calculating-SEs}--\ref{sec:empirical} then establish the HAC and bootstrap results, give the projection-adjusted refinement, and report the simulations and monetary-shock application. The appendices contain additional supporting theory, proofs, additional simulations and application diagnostics, and examples showing how other macro estimators fit the moment-level setup.

\section{Motivation}\label{sec:motivation}
This section makes the conditioning convention concrete before the GMM notation is introduced. It begins with the direct potential-outcome framework of \citet{BojinovShephard2019_TimeSeriesExperiments} and \citet{RambachanShephard2019_NonparametricDynamicCausalModel}, under which assignment shocks are the source of randomness in the economy. The examples in this section use i.i.d.\ shocks for clarity; the GMM theory below allows weak dependence and states the needed long-run covariance and CLT conditions at the centered-moment level. To keep the observed-period notation one-sided, write $w_{1:t}=(w_1,\ldots,w_t)$ for the realized shock history up to time $t\ge 1$.

Suppose that $(W_t)_{t \in \mathbb{Z}}$ has common law $\mathcal{F}_W$ on $\mathcal{W}\subseteq\mathbb{R}^{d_w}$ and that, for each observed date $t\ge 1$, there is a vector-valued potential-outcome map $Y_t:\mathcal{W}^t\to \mathbb{R}^{d_y}$ with $w_{1:t}\mapsto Y_t(w_{1:t})$. Conditional on these maps, the assignment shocks are the source of design randomness, and the maps are nonanticipating in the sense that the date-$t$ outcome depends only on shocks realized up to date $t$. This construction motivates the moment-level conditioning convention in Assumption~\ref{ass:design-environment}. Fix an observed date $s\ge 1$ and a horizon $h\ge 0$. For any $w,w'\in\mathcal{W}$ and any finite off-date shock values
$w_{-s}^{(h)}:=(w_1,\ldots,w_{s-1},w_{s+1},\ldots,w_{s+h})$, the $h$-step causal effect of replacing $W_s=w'$ by $w$
(holding those off-date shocks fixed) is
\[
\tau_{s,h}(w,w';w_{-s}^{(h)})
:= Y_{s+h}\!\big(w_{1:s-1},\,w,\,w_{s+1:s+h}\big)
   - Y_{s+h}\!\big(w_{1:s-1},\,w',\,w_{s+1:s+h}\big).
\]
In the notation of \citet{RambachanShephard2019_NonparametricDynamicCausalModel}, the usual $h$-period impulse response at date $s$ is $\tau_{s,h}(1,0;W_{-s}^{(h)})$. The observed data are $(Y_t(W_{1:t}), W_t)_{t=1}^T$. This convention is the time-series analogue of the Neyman--Rubin finite-population setup, with shocks playing the role of randomized treatments.

The shock-conditioning convention is already implicit in many time-series models. A stationary SVAR written in its MA\((\infty)\) representation, \(Y_t=\sum_{s=0}^\infty A_sW_{t-s}\), treats the structural shocks as the stochastic input and maps each possible shock history into an outcome path. Conditional on the coefficients and any initial conditions, a shock sequence \(w=(w_1,\ldots,w_T)\) therefore determines a potential outcome path \(Y_t(w)\), while the observed path is the one selected by the realized shocks \(W_t\). The design-based convention takes this potential-outcome map as fixed over the historical episode and re-randomizes only the assignment shocks. The difference from conventional stationary asymptotics is not that SVARs make outcomes random while the design approach does not; outcomes are random in both views. The difference is what is averaged over: conventional asymptotics average over the law generating the economy, whereas the design calculation conditions on the realized finite-history outcome functions and averages only over alternative shock draws.

In order for there to be a difference between the usual and the design-based notions of asymptotic variance, there needs to be heterogeneity in dynamic causal effects. In the cross-sectional Neyman formula, the gap between design-based and conventional variance limits is driven by treatment-effect heterogeneity. In time series, the analogous object is heterogeneity in impulse responses across dates. For example, in the constant-coefficient SVAR(1), \(Y_t=AY_{t-1}+BW_t\), the date-\(t\), horizon-\(h\) response is the same at every date, \(A^hB\). If instead the impact matrix depends on a state variable, \(B=B(Z_t)\), with \(Z_t\) predetermined relative to \(W_t\), then the response is \(A^hB(Z_t)\) and varies with $Z_t$. Holding that fixed, a natural estimand of interest would be the sample average \(\theta_T^\star := (T-h)^{-1}\sum_{t=1}^{T-h}A^hB(Z_t)\). This variation need not imply nonstationarity: the observables and effects may remain strictly stationary, for example if \((Z_t)_{t\in\mathbb Z}\) is strictly stationary. Nor must the heterogeneity itself be identified from a single realized history, as when \(B(z)\) is orthogonal for every \(z\). Design-based and conventional standard errors differ only when the realized economy is treated as fixed and dynamic effects vary over that history. Empirical state dependence is documented, for example, by \citet{AuerbachGorodnichenko2012}, \citet{TenreyroThwaites2016}, and \citet{RameyZubairy2018}, but its relevance must be assessed in each application. The main results do not require the researcher to specify the form of the heterogeneity; the possibility of predictable dynamic-effect variation is enough to make the variance distinction relevant.

The following simple LP with i.i.d.\ Bernoulli shocks gives a Neyman-style setting in which, for each date, the two potential outcomes are treated as fixed. Fix $h$ and set $T_h:=T-h$; that is, the pairs $(Y^{(1)}_{t+h},Y^{(0)}_{t+h})_{t\le T_h}$ are treated as fixed constants. All expectations and variances in this paragraph and Lemma~\ref{lem:neyman-ts} are conditional on that fixed array. The calculation therefore isolates the same-date, or lag-zero, finite-population variance component. If a full dynamic re-randomization lets overlapping histories move together, serial covariance terms can appear; those terms are handled by the later GMM long-run variance.

After this conditioning, write
\[
Y^{(1)}_{t+h}:=Y_{t+h}(W_{1:t-1},1,W_{t+1:t+h}),
\qquad
Y^{(0)}_{t+h}:=Y_{t+h}(W_{1:t-1},0,W_{t+1:t+h}),
\]
with finite-history averages $\bar Y_h^{(w)}:=T_h^{-1}\sum_{t=1}^{T_h}Y^{(w)}_{t+h}$ for $w\in\{0,1\}$. The observed outcome in the horizon-$h$ LP is $Y_{t+h}^{\mathrm{obs}}=W_tY^{(1)}_{t+h}+(1-W_t)Y^{(0)}_{t+h}$. Let $\tau_{t,h}:=Y^{(1)}_{t+h}-Y^{(0)}_{t+h}$ and $\bar\tau_h:=\bar Y_h^{(1)}-\bar Y_h^{(0)}$. The first-order linear term for the LP slope is
\[
\phi_{t,h}:=
\frac{W_t}{p}\big(Y^{(1)}_{t+h}-\bar Y_h^{(1)}\big)
-\frac{1-W_t}{1-p}\big(Y^{(0)}_{t+h}-\bar Y_h^{(0)}\big).
\]

\begin{lemma}\label{lem:neyman-ts}
Conditional on the fixed array described above, suppose $W_t \sim \operatorname{Bernoulli}(p)$ i.i.d.\ with $p\in(0,1)$ fixed, and suppose $T_h^{-1}\sum_{t=1}^{T_h}(|Y^{(1)}_{t+h}|^4+|Y^{(0)}_{t+h}|^4)=O(1)$. Then, as $T_h\to\infty$, the LP slope $\widehat\tau_h$ obtained from the intercept regression $Y_{t+h}^{\mathrm{obs}}=\alpha_h+\tau_h W_t+u_{t,h}$ admits the design expansion
\[
\sqrt{T_h}(\widehat\tau_h-\bar\tau_h)
=
T_h^{-1/2}\sum_{t=1}^{T_h}\phi_{t,h}+r_{T,h},
\qquad
\mathbb E[r_{T,h}^2]=o(1).
\]
Consequently,
\[
\Var(\widehat\tau_h)
=\frac{1}{T_h}\Gamma_h(0)+o(T_h^{-1}),
\]
where
\[
\Gamma_h(0)
:=\frac{1}{T_h}\sum_{t=1}^{T_h}\left[
\frac{\big(Y^{(1)}_{t+h}-\bar Y_h^{(1)}\big)^2}{p}
+
\frac{\big(Y^{(0)}_{t+h}-\bar Y_h^{(0)}\big)^2}{1-p}
-
\big(\tau_{t,h}-\bar\tau_h\big)^2
\right].
\]
Equivalently, $\Gamma_h(0)=S^2_{1,h}/p+S^2_{0,h}/(1-p)-S^2_{\tau,h}$, where $S^2_{1,h}$, $S^2_{0,h}$, and $S^2_{\tau,h}$ are the finite-history variances, computed with denominator $T_h$, of $Y^{(1)}_{t+h}$, $Y^{(0)}_{t+h}$, and $\tau_{t,h}$, respectively.
\end{lemma}

Lemma~\ref{lem:neyman-ts} is a simple special case that connects the time-series setup to the finite-population result of \citet{Neyman1923_thesis}. Once the array of outcomes obtained by setting $W_t=1$ and $W_t=0$ at each date is fixed, the same-date, or lag-zero, term is the Neyman finite-array component. The negative term $-S^2_{\tau,h}$ is the treatment-effect-heterogeneity adjustment: it is not identified from one assignment because $Y^{(1)}_{t+h}$ and $Y^{(0)}_{t+h}$ are not jointly observed at the same date.

The rest of the paper generalizes this kind of result to the GMM setting. The full theory allows centered moment or influence contributions to be serially dependent. For a scalar centered contribution $U_{t,h}$, write $\Gamma_h^U(0):=T_h^{-1}\sum_{t=1}^{T_h}\Var(U_{t,h})$ and $\Gamma_h^U(\ell):=(T_h-\ell)^{-1}\sum_{t=\ell+1}^{T_h}\Cov(U_{t,h},U_{t-\ell,h})$ for $1\le \ell<T_h$. Then
\[
\Var\!\left(T_h^{-1/2}\sum_{t=1}^{T_h}U_{t,h}\right)
=
\Gamma_h^U(0)+2\sum_{\ell=1}^{T_h-1}\left(1-\frac{\ell}{T_h}\right)\Gamma_h^U(\ell).
\]
For vector-valued contributions, the positive-lag term is $\Gamma_h^U(\ell)+\Gamma_h^U(\ell)^\top$ rather than $2\Gamma_h^U(\ell)$. The nonzero lag terms are the time-series component. They vanish in the fixed-array Bernoulli calculation above because $\phi_{t,h}$ is a function only of $W_t$ and fixed constants, but they need not vanish for the serially dependent centered moment arrays used below. In the GMM notation, these lag terms enter the design long-run variance $\Omega_R$ of the centered moment, while movement in the date-specific mean path creates the conservative-HAC gap $\Omega_\mu$.

\section{Related literature}\label{sec:related-literature}

The first connection is to design-based econometrics. The finite-population analysis of \citet{Neyman1923_thesis} conditions on the potential outcomes and treats only the assignment as random. In that setting, the usual conservative variance formula reflects the fact that treatment-effect heterogeneity is not identified from one assignment. \citet{AbadieAtheyImbensWooldridge2020}, \citet{Xu2021FinitePopulationMEstimators}, and \citet{KakehiMatsushitaOtsu2026FinitePopulationGMM} extend related finite-population reasoning to regression, M-estimation, and GMM. Relatedly, \citet{Sancibrian2025RepeatedFinitePopulations} studies estimation uncertainty in repeated finite populations with latent attributes and noisy repeated measurements; that setting is model-based rather than assignment-based, but it shares the distinction between uncertainty about a realized finite population and uncertainty about a superpopulation. The present paper keeps the same conditioning logic but replaces cross-sectional independent assignment with a dependent time-series moment. The resulting correction is not a single finite-population variance of treatment effects; it is the long-run variance $\Omega_\mu$ of the centered mean path, the object estimated by a growing-bandwidth HAC calculation.

The second connection is to regression adjustment in randomized experiments. \citet{Freedman2008RegressionAdjustment} emphasizes that regression adjustment is not automatically justified by random assignment, while \citet{Lin2013AgnosticRegressionAdjustment} gives conditions under which adjustment improves precision. The adjustment studied here has the same model-free character. Predetermined covariates are used to remove predictable components of the mean path. The result is a positive-semidefinite reduction in the Loewner order for the HAC variance limit, while the interpretation as a tighter conservative variance bound requires additional orthogonality between the residualized innovation and the span of adjustment covariates.

A third connection is to time-series experiments and dynamic causal effects. \citet{BojinovShephard2019_TimeSeriesExperiments} and \citet{RambachanShephard2019_NonparametricDynamicCausalModel} develop potential-outcome frameworks in which treatment or assignment histories generate dynamic effects, and \citet{BojinovRambachanShephard2021PanelExperiments} extends this logic to panels. For switchback designs, \citet{BojinovSimchiLeviZhao2023Switchback} study optimal design and randomization-based inference. \citet{LiangRecht2025NEqualsOne} study randomization inference for a single time-series unit using links to system identification, and \citet{LinDing2025TimeSeriesExperiments} study regression-based and design-based inference in time-series experiments. The present paper differs in its estimand and estimators: it studies macroeconomic GMM moments, including smooth transformations, overidentified systems, HAC estimators, dependent multiplier bootstraps calibrated to the HAC long-run variance, and LP or VAR-style impulse-response applications.

Finally, this paper reinterprets standard macroeconometric inference. Classical GMM and extremum theory use long-run covariance matrices to approximate sampling uncertainty \citep{Hansen1982GMM,NeweyMcFadden1994_LargeSampleEstimation}, and HAC estimators estimate those matrices under weak dependence \citep{NeweyWest1987,Andrews1991HAC}. In conventional macro applications, the long-run variance averages over both innovation variation and predictable variation in the mean moment under the maintained stochastic process. In the finite-history interpretation developed here, the environment and the date-specific mean path are conditioned on, so the same HAC calculation estimates $\Omega_R^+=\Omega_R+\Omega_\mu$. This estimand/variance-limit distinction is complementary to work on robust LP and VAR inference, which studies how to approximate the sampling law of impulse-response estimators under alternative persistence, lag-order, and horizon asymptotics. The appendix connects the finite-history decomposition to the LP and VAR impulse-response literature, including the general LP survey of \citet{JordaTaylor2025}, robust LP inference in \citet{MontielOleaPlagborgMoller2021} and \citet{InoueJordaKuersteiner2024}, and LP--VAR equivalence and misspecification comparisons in \citet{PlagborgMollerWolf2021} and \citet{MontielOleaPlagborgMollerQianWolf2024}.

\section{Setup and main results} \label{sec:setup-and-main-results}

\subsection{Estimand and interpretation}
\label{subsec:estimand}

As mentioned above, the estimand is both \emph{finite-history} and \emph{design-based}. For each sample length, the design conditions on a fixed environment $\mathcal{E}_T$, which records the potential-outcome maps, state-transition rules, deterministic windows, and covariate paths held fixed by the design.

\begin{assumption}\label{ass:design-environment}
For each sample length $T$, the design specifies a conditioning environment $\mathcal{E}_T$. This environment contains the potential-outcome maps, state-transition rules, deterministic sampling windows, and any covariate or state paths explicitly held fixed by the design. Conditional on $\mathcal{E}_T$, the moment array used below is generated by the assignment shocks through maps from shock histories to moment contributions. The assignment shocks are the only source of design randomness, and all expectations, probabilities, and variances are taken with respect to this conditional design distribution.
\end{assumption}

Assumption~\ref{ass:design-environment} is a conditioning convention, not a law of large numbers or a central limit theorem. It fixes the probability model before the stochastic regularity conditions are imposed. The later mean moments, covariance limits, and variance comparisons are all evaluated under this conditional design distribution, and the asymptotic arguments consider sequences of such conditioning environments. Thus the choice of $\mathcal{E}_T$ is part of the estimand: changing what is held fixed changes the counterfactual experiment and can change the variance target.

This convention also separates two ideas that are often conflated. A variable may be predetermined, in the sense that it is dated before the shock being adjusted for, without being fixed under the design. It is fixed only if its realized path is included in $\mathcal{E}_T$; otherwise it remains part of the shock-history map and may change when earlier assignment shocks are re-randomized. Hence the same object should not be treated as fixed when defining the estimand and stochastic when computing the variance, unless an explicit two-layer model is introduced. Two versions of this choice recur below. In a fixed-state design, the relevant pre-shock covariate, state, or regime path is included in $\mathcal{E}_T$, so conditional means given $\mathcal{I}_{T,t}$ are fixed date-specific quantities. In a dynamic-state design, those variables are themselves generated by earlier assignment shocks and therefore move under re-randomization; the date-specific mean moment then averages over the moving pre-shock state under the design distribution. The analysis uses whichever convention is built into $\mathcal{E}_T$ and keeps that convention fixed within a variance calculation.

Generated shocks and first-stage instruments raise a separate issue. If a high-frequency shock, proxy, or instrument is estimated rather than directly assigned, its first-stage error is neither just a fixed state nor just a moving predetermined variable. Treating that error as part of the design randomness requires augmenting the moment vector, as in Appendix~\ref{app:two-step}. The main decomposition in Proposition~\ref{lem:fixed-env-decomp} uses the fixed-environment convention. However, the projection-adjustment result can also be read under a dynamic-state convention as a reduction in the conservative HAC variance limit, but its interpretation as a tighter conservative bound then requires the additional orthogonality conditions in Appendix~\ref{app:ra-lower-bound}.

To cover the estimators used in practice, the analysis uses a GMM framework that includes just-identified systems and locally correctly specified overidentified systems. For each date $t$, let $g_t(W_t,\theta)\in\mathbb{R}^k$ denote the moment at parameter value $\theta\in\Theta\subset\mathbb{R}^p$, where $\Theta$ is compact. The argument $W_t$ is shorthand for the relevant shock history under the design distribution, not necessarily only the current shock. Thus an LP moment may contain $y_{t+h}$, lagged controls, or generated shocks, and a VAR moment may contain lagged outcomes. After conditioning on $\mathcal{E}_T$, these objects are components of the map from the shock path to the moment. For brevity, write $g_t(\theta):=g_t(W_t,\theta)$. The main text treats the macro sample as fully observed, so $N=T$; the incomplete-observation case is deferred to the appendix.

Throughout Sections~\ref{sec:setup-and-main-results}--\ref{sec:implementation}, probabilities, expectations, and variances are conditional on $\mathcal{E}_T$. Write $\mathbb{E}_T[\cdot]$ for this design expectation. The date-specific mean moment is $\mu_{T,t}(\theta):=\mathbb{E}_T[g_t(W_t,\theta)]$, and the sequence $(\mu_{T,t}(\theta_T^\star))_{t\le T}$ is the mean path at the estimand. The centered innovation is $e_{T,t}(\theta):=g_t(W_t,\theta)-\mu_{T,t}(\theta)$, and the centered mean path at the estimand is $\tilde\mu_{T,t}:=\mu_{T,t}(\theta_T^\star)-T^{-1}\sum_{s=1}^T\mu_{T,s}(\theta_T^\star)$. Conditional-on-$\mathcal{I}_{T,t}$ expressions in the LP and VAR examples below are primitive sufficient conditions for these fixed date-specific mean moments, not an additional source of randomness in the main decomposition. When the $T$-subscript is not essential, write $\mathcal{I}_t$, $\mu_t$, and $\tilde\mu_t$.

With this convention, the sample moment averages a single design draw of the mapping from shock histories to moments, while the population object averages that same fixed mapping under the conditional design distribution. Define the sample moment and GMM estimator by
\[
 g_N(\theta):=\frac{1}{T}\sum_{t=1}^T g_t(W_t,\theta), \qquad
\widehat\theta_N\in\arg\min_{\theta\in\Theta} g_N(\theta)^\top \widehat A_N g_N(\theta),
\]
where $\widehat A_N$ is a positive definite estimator of a fixed weighting matrix $A\succ0$. Define the sample-period mean moment by
\[
\bar m_T(\theta):=\frac{1}{T}\sum_{t=1}^T \mu_{T,t}(\theta)
=\frac{1}{T}\sum_{t=1}^T \mathbb{E}_T[g_t(W_t,\theta)],
\]
and let
\[
Q_T(\theta):=\bar m_T(\theta)^\top A\bar m_T(\theta), \qquad
\theta_T^\star\in\arg\min_{\theta\in\Theta} Q_T(\theta).
\]
The objects of interest are $\theta_T^\star$ and smooth functions $h(\theta_T^\star)$, such as impulse responses or forecast error variance decompositions.

The theory does not require that $\theta_T^\star$ converges to some $\theta^\star$ as $T \rightarrow \infty$; however, this feature can be generated by supposing that $Q_T$ converges uniformly to a deterministic limit $Q$ with unique minimizer $\theta^\star$. When the mean path drifts with $T$, however, $\theta_T^\star$ remains the relevant econometric object for inference, and $\theta^\star$ is best viewed as a convenient asymptotic approximation rather than the estimand.

The potential-outcome examples of \citet{RambachanShephard2019_NonparametricDynamicCausalModel} often use i.i.d.\ assignments. However, the GMM theory below is stated directly for the centered moment innovations and lag products because those are the objects entering the CLT and HAC estimator. A strictly stationary weakly dependent shock process is a common sufficient route (see, e.g., \citet{Doukhan1994,Rio2000}), but the raw observed moment may still have a nonconstant date-specific mean. Applications often describe exogeneity through an information set. Let $\mathcal F_t:=\sigma(W_s:s\le t)$ be the filtration generated by the shocks, and let $Z_t$ be a vector of predetermined covariates, such as deterministic trends, regime indicators, lags, and pre-announced variables. Throughout, $\mathcal{I}_{T,t}$ denotes the information set with respect to which the assignment mechanism is exogenous. Section~\ref{sec:implementation} discusses how suitable choices of $Z_t$ can be used to tighten conservative variance bounds. By convention the contemporaneous shock $W_t$ is not included in $Z_t$. A variable is predetermined when it is dated before the shock whose moment is being adjusted; this timing condition does not by itself mean that the variable is fixed under every possible design distribution.

The rest of the paper uses this convention to compare three moment-level covariance objects. The centered innovation component has design long-run variance $\Omega_R$; the conventional HAC limit is $\Omega_R^+$; and the difference is the nonnegative mean-path component $\Omega_\mu$, so that $\Omega_R=\Omega_R^+-\Omega_\mu$. Once these are established, the relevant asymptotic variance of GMM estimators follows by the usual theory.

\subsection{Examples and mean-path drift}

Two canonical examples below can be used to show where a mean path drift can come from: LPs and VARs. 

For a given horizon $h$, let the assignment shock be $x_t$ and let $c_t$ denote predetermined controls, such as lags or trends. Define $\psi_t := (1, x_t, c_t^\top)^\top$ and $\theta_h := (\alpha_h,\beta_h,\gamma_h^\top)^\top$. The LP moment condition is $g^{\mathrm{LP}}_{t,h}(\theta_h)=\psi_t(y_{t+h}-\psi_t^\top\theta_h)$, and so stacking over horizons gives $g^{\mathrm{LP}}_t(\theta)=\big(g^{\mathrm{LP}}_{t,h}(\theta_h)\big)_{h\in\mathcal H}$. This one-shock, one-outcome LP is the running example for the design long-run variance in Theorem~\ref{thm:AN}, the HAC conservative-variance result in Section~\ref{sec:calculating-SEs}, and the projection-adjustment construction in Section~\ref{sec:implementation}. If $y_{t+h}$ is vector-valued, let $r_{t,h}(\theta_h):=Y_{t+h}-M_h\psi_t$ with coefficient matrix $M_h$, and write $g^{\mathrm{LP}}_{t,h}(\theta_h)=(I_n\otimes\psi_t)r_{t,h}(\theta_h)$, so that $\operatorname{vec}(M_h)$ is identified by orthogonality between regressors and residuals.

A second example is a stable VAR($p$) with observed shock vector $W_t$, $Y_t=\Phi_0+\Phi_1Y_{t-1}+\cdots+\Phi_pY_{t-p}+BW_t$, where $Y_t\in\mathbb R^n$, $\Phi_i\in\mathbb R^{n\times n}$, and $B\in\mathbb R^{n\times m}$. Let
\[
X_t^\top := (\,1,\, Y_{t-1}^\top,\, \ldots,\, Y_{t-p}^\top,\, W_t^\top\,),
\qquad
u_t(\phi) := Y_t - \Phi_0 - \Phi_1 Y_{t-1} - \cdots - \Phi_p Y_{t-p} - B W_t,
\]
with parameter vector $\phi:=\operatorname{vec}(\Phi_0,\Phi_1,\ldots,\Phi_p,B)$. The VAR moment function is $g^{\mathrm{VAR}}_t(\phi)=\operatorname{vec}\!\big(X_t u_t(\phi)^\top\big)=(I_n\otimes X_t)u_t(\phi)$.

In both LPs and VARs, heterogeneous dynamic causal effects imply that $\mu_{T,t}(\theta_T^\star):=\mathbb{E}_T[g_t(W_t,\theta_T^\star)]$ may not be 0 for all $t$. Instead, orthogonality holds on average at $T$, but not necessarily at dates $t < T$. The resulting mean-path drift is the time-series analogue of the treatment-effect-heterogeneity term in Neyman's finite-population variance formula. Remark~\ref{rem:LP_VAR_mean_drift} records the corresponding conditional mean moments. These formulas are not needed for the abstract GMM theorem, but they explain the economic source of $\Omega_\mu$ in the running examples.

\begin{remark}\label{rem:LP_VAR_mean_drift}
Consider the cases below.
\begin{enumerate}[label=(\roman*)]
\item Suppose, in the local-projection setting, that the shock is centered relative to the assignment information set, $\mathbb{E}[x_t\mid\I_t]=0$, that $\tau_{t,h}$ is $\I_t$-measurable, and that the potential outcome satisfies
\[
\mathbb E\!\left[x_t\{y_{t+h}(0)-\mathbb E[y_{t+h}(0)\mid\I_t]\}\mid\I_t\right]=0.
\]
At the population best linear projection coefficient
$\theta_h^\star=(\alpha_h^\star,\beta_h^\star,\gamma_h^{\star\top})^\top$, the $x_t$-row of the conditional mean moment satisfies
\[
e_x^\top\mathbb{E}\big[g^{\mathrm{LP}}_{t,h}(\theta_h^\star)\,\big|\,\I_t\big]
=
\mathrm{Var}(x_t\mid\I_t)\,\big(\tau_{t,h}-\beta_h^\star\big),
\]
where $e_x$ selects the shock row. If, in addition, the intercept/control block is correctly specified date by date in the sense that
$\mathbb{E}[y_{t+h}(0)\mid\I_t]=\alpha_h^\star+\gamma_h^{\star\top}c_t$, then
\[
\mathbb{E}\big[g^{\mathrm{LP}}_{t,h}(\theta_h^\star)\,\big|\,\I_t\big]
=\begin{bmatrix}
0\\[0.2em]
\mathrm{Var}(x_t\mid\I_t)\,\big(\tau_{t,h}-\beta_h^\star\big)\\[0.2em]
\mathbf 0
\end{bmatrix}.
\]
In that case, the conditional mean is constant across $t$ (and, because its sample-period average is zero, equal to zero) whenever
$\mathrm{Var}(x_t\mid\I_t)\,\big(\tau_{t,h}-\beta_h^\star\big)\equiv 0$.
\item Suppose, in the VAR setting, that the true law is $Y_t=\Phi_0(Z_t)+\sum_{i=1}^p \Phi_i(Z_t)Y_{t-i}+B(Z_t)W_t+\varepsilon_t$, with
$Z_t$ measurable with respect to the pre-shock information set $\I_t$, $\mathbb{E}[W_t\mid\I_t]=0$, and $\mathbb{E}[\varepsilon_t\mid\I_t,W_t]=0$. At the sample-period estimand $\phi^\star$, define $\Delta_t:=(\Phi_0(Z_t)-\Phi_0)+\sum_{i=1}^p(\Phi_i(Z_t)-\Phi_i)Y_{t-i}$, $\Sigma_{W,t}:=\mathbb{E}[W_t W_t^\top\mid\I_t]$, and $Z^{\mathrm{pred}}_t:=\big(1,\,Y_{t-1}^\top,\,\ldots,\,Y_{t-p}^\top\big)^\top$. If the realized path of $Z_t$ is included in $\mathcal E_T$, it is fixed under the design; otherwise it moves with past assignment shocks under re-randomization. Then
\[
\mathbb{E}\big[g^{\mathrm{VAR}}_t(\phi^\star)\,\big|\,\I_t\big]
=\operatorname{vec}\!\begin{pmatrix}
Z^{\mathrm{pred}}_t\,\Delta_t^\top \\[0.3em]
\Sigma_{W,t}\,(B(Z_t)-B)^\top
\end{pmatrix}.
\]
Thus the conditional mean moment vanishes whenever $\Delta_t\equiv 0$ and
$\Sigma_{W,t}(B(Z_t)-B)^\top\equiv 0$; in particular this holds when
$\Phi_0(Z_t)=\Phi_0$, $\Phi_i(Z_t)=\Phi_i$ for all $i$, and $B(Z_t)=B$.
\end{enumerate}
\end{remark}

Since DSGE models can often be rewritten via their state-space form into a VAR model \citep{Giacomini2013_DSGE_VAR_Relationship}, Remark~\ref{rem:LP_VAR_mean_drift}(ii) also characterizes when GMM-estimated DSGE models exhibit mean drift. Many other macro procedures fit into this setting; representative examples include LP \citep{Jorda2005}, external-instrument (proxy) SVARs and high-frequency/event-study IV \citep{StockWatson2012_BPEA,MertensRavn2013,GertlerKaradi2015}, heteroskedasticity-based identification \citep{Rigobon2003,LanneLutkepohl2008}, FAVAR \citep{BernankeBoivinEliasz2005}, and minimum-distance / indirect inference that matches model-implied autocovariances or IRFs \citep{GourierouxMonfortRenault1993,HallInoueNasonRossi2012}. All of these admit per-period moment vectors $g_t(W_t,\theta)$ and a sample-period estimand $\theta_T^\star$. Appendix~\ref{app:macro-estimators} records how these estimators fit into the framework.

These examples show that the notation is not special to LPs or VARs. The next subsection therefore states the moment-level assumptions behind the asymptotic results.

\subsection{Assumptions and asymptotic theory}

The estimator assumptions are stated at the moment level. This keeps the theorem general enough for LPs, VARs, proxy procedures, and minimum-distance estimators, while the surrounding prose records primitive sufficient routes for the canonical cases.

\begin{assumption}\label{ass:dependence}
Let $e_{T,t}:=e_{T,t}(\theta_T^\star)$. Conditional on $\mathcal{E}_T$, the centered innovation array satisfies:
\begin{enumerate}[label=(\roman*)]
\item For each fixed $\ell\ge0$, the limit
\[
\Gamma_e(\ell):=\lim_{T\to\infty}T^{-1}\sum_{t=\ell+1}^T
\mathbb{E}_T[e_{T,t}e_{T,t-\ell}^\top]
\]
exists, with $\Gamma_e(-\ell):=\Gamma_e(\ell)^\top$ for $\ell>0$, and $\sum_{\ell\in\mathbb{Z}}\|\Gamma_e(\ell)\|<\infty$.
\item The normalized innovation sum satisfies
\[
T^{-1/2}\sum_{t=1}^T e_{T,t}\Rightarrow \mathcal{N}(0,\Omega_R),
\qquad
\Omega_R:=\sum_{\ell\in\mathbb{Z}}\Gamma_e(\ell).
\]
\item For some $\delta>0$, $\sup_{T,t}\mathbb{E}_T\|e_{T,t}\|^{2+\delta}<\infty$.
\end{enumerate}
\end{assumption}

Assumption~\ref{ass:dependence} is a condition on the innovation component of the moment, not on the raw moment, the observable series, or the date-specific mean path. It allows the raw moment process to have a changing conditional mean through $\mu_{T,t}$. Strict stationarity of $(e_{T,t})$ is therefore not imposed as part of the baseline theory; it is only one sufficient route.

The primitive sufficient conditions are standard: the assumption follows, for example, if the triangular array $(e_{T,t})$ is uniformly strongly mixing, has the displayed uniform $2+\delta$ moment bound, satisfies a Lindeberg condition, and has fixed-lag Ces\`aro covariance limits with an absolutely summable covariance envelope. In finite-lag LPs with i.i.d.\ or finite-memory assignment shocks, the centered moment is $q_h$-dependent. In stable VARs with short-memory innovations and absolutely summable linear filters, the centered moment is a short-memory linear process. In both cases the assumption reduces to the usual moment and weak-dependence requirements for a time-series CLT. Thus, when a stationary or mixing primitive is invoked below, it is a sufficient condition for the assignment-shock process or for the centered innovation array conditional on $\mathcal E_T$. It is not an assumption that the raw moment $g_t(W_t,\theta_T^\star)$, the observed series, or the fixed mean path $\mu_{T,t}$ is stationary.

The main text specializes to the case where the entire time series of interest is observed. By contrast, the appendix records the modifications under random or deterministic sub-sampling; in other words, the cases where a researcher may be interested in some $T$-length time series, but only observe $N<T$ of those time periods. The asymptotic theory also requires smoothness of the moment map and point identification of the sample-period object. In the canonical LP and VAR cases these requirements reduce to familiar moment, nonsingularity, and stability conditions; the appendix gives the corresponding verification routes.

Let $B_{g,T,t}$, $b_{g,T,t}$, $B_{g,1,T,t}$, and $b_{g,1,T,t}$ denote nonnegative envelope and Lipschitz variables for the moment and its Jacobian.
\begin{assumption}\label{ass:smoothness}
For some $\delta>0$, for each $t$, $g_t(W_t,\theta)$ is measurable and continuously differentiable in $\theta$. For all $\theta,\theta'\in\Theta$, $\|g_t(W_t,\theta)\|\le B_{g,T,t}$, $\|g_t(W_t,\theta)-g_t(W_t,\theta')\|\le b_{g,T,t}\|\theta-\theta'\|$, $\|\nabla_\theta g_t(W_t,\theta)\|\le B_{g,1,T,t}$, and $\|\nabla_\theta g_t(W_t,\theta)-\nabla_\theta g_t(W_t,\theta')\|\le b_{g,1,T,t}\|\theta-\theta'\|$. The envelope moments are uniformly bounded: $\sup_{T,t}\mathbb{E}_T[B_{g,T,t}^{2+\delta}+B_{g,1,T,t}^{2}+b_{g,T,t}^{2}+b_{g,1,T,t}]<\infty$. Finally, $G_T:=T^{-1}\sum_{t=1}^T \mathbb{E}_T[\nabla_\theta g_t(W_t,\theta_T^\star)]\to G$, where $G$ has full column rank.
\end{assumption}

Assumption~\ref{ass:smoothness} is the standard empirical-process/GMM regularity condition. The envelope and Lipschitz bounds allow the proofs to pass from pointwise convergence to uniform convergence over $\Theta$, control plug-in errors when $\widehat\theta_N$ replaces $\theta_T^\star$, and linearize the sample moment map through the Jacobian. In the appendix proofs, the assumption is used mainly for the uniform laws of large numbers for moments and Jacobians and for the plug-in steps in HAC and bootstrap consistency. To keep the appendix notation light, write $B_g(W_t)$, $b_g(W_t)$, $B_{g,1}(W_t)$, and $b_{g,1}(W_t)$ for these triangular-array envelopes when no confusion arises.

The main consistency proof also needs a uniform law of large numbers over $\Theta$ for the moment function and Jacobian. This condition is stated separately because the innovation short-memory condition at $\theta_T^\star$ does not by itself deliver generic-$\theta$ uniform convergence, and the separation avoids making point identification depend on a stronger process-level stationarity condition.

\begin{assumption}\label{ass:uniform-lln}
The moment function and Jacobian satisfy
\[
\sup_{\theta\in\Theta}\left\|\frac{1}{T}\sum_{t=1}^T\Big(g_t(W_t,\theta)-\mathbb{E}_T[g_t(W_t,\theta)]\Big)\right\|\xrightarrow{p}0,
\]
\[
\sup_{\theta\in\Theta}\left\|\frac{1}{T}\sum_{t=1}^T\Big(\nabla_\theta g_t(W_t,\theta)-\mathbb{E}_T[\nabla_\theta g_t(W_t,\theta)]\Big)\right\|\xrightarrow{p}0.
\]
\end{assumption}

Assumption~\ref{ass:uniform-lln} makes explicit the generic-$\theta$ regularity needed
for consistency and linearization. It is conceptually distinct from the innovation CLT imposed at $\theta_T^\star$. A primitive route is compactness of $\Theta$, the Lipschitz envelopes in Assumption~\ref{ass:smoothness}, and a uniform law of large numbers for the underlying finite-dimensional moment function and Jacobian arrays. For the linear LP and VAR moments used below, this reduces to a law of large numbers for finitely many products of shocks, outcomes, instruments, and controls. The main result further requires local correct specification of the moments.

\begin{assumption}\label{ass:local-correct-specification}
The sample-period estimand satisfies $\sqrt T\|\bar m_T(\theta_T^\star)\|=o(1)$.
\end{assumption}

Assumption~\ref{ass:local-correct-specification} is automatic in just-identified linear LPs and VARs because the sample-period normal equations imply $\bar m_T(\theta_T^\star)=0$. In overidentified GMM it rules out a first-order misspecification component in the main theorem. Appendix~\ref{app:misspecified-gmm} records the corresponding expansion for the fixed-weight criterion with $\widehat A_N\equiv A$ when this condition is not imposed; the first-order moment then includes Jacobian fluctuations multiplied by the nonzero pseudo-true mean. If the weight matrix is estimated in that locally misspecified case, the weight-estimation expansion is an additional first-order ingredient rather than part of the main theorem. Note further that this average-moment restriction does not require $\mu_{T,t}(\theta_T^\star)=0$ date by date and therefore does not rule out mean-path drift. The centered path $\tilde\mu_{T,t}=\mu_{T,t}(\theta_T^\star)-\bar m_T(\theta_T^\star)$ can be nonzero even when $\bar m_T(\theta_T^\star)=o(T^{-1/2})$.

Finally, standard \citet{NeweyMcFadden1994_LargeSampleEstimation}-style conditions on the GMM objective function and weight matrix are required to justify estimation. These are the same conditions that would be required in the non-design setting.
\begin{assumption}\label{ass:gmm}
Fix a positive definite weight matrix $A\succ0$. Define $Q_T(\theta):=\bar m_T(\theta)^\top A\bar m_T(\theta)$, and assume that for each $T$, $Q_T(\theta)$ is minimized at some $\theta_T^\star\in\Theta$. A feasible weight matrix satisfies $\widehat A_N\to_p A$.
\end{assumption}

Assumption~\ref{ass:gmm} is the existence-and-weighting part of the GMM setup. It fixes the sample-period population criterion for each sample length and requires the feasible weight matrix to converge to a fixed nonsingular limit. In the proofs, this condition permits comparison of the sample and population criteria and stabilizes the linear GMM expansion. Consistency for a moving estimand also requires $Q_T$ to keep separating $\theta_T^\star$ from the rest of the parameter space as $T$ changes, so the next assumption states that condition explicitly.

\begin{assumption}\label{ass:estimand-separation}
For every $\varepsilon>0$, $\liminf_{T\to\infty}\inf_{\|\theta-\theta_T^\star\|\ge\varepsilon}[Q_T(\theta)-Q_T(\theta_T^\star)]>0$.
\end{assumption}

Assumption~\ref{ass:estimand-separation} is the moving-estimand identification condition.
Because the estimand is $\theta_T^\star$ rather than a fixed $\theta^\star$, the paper
needs a statement saying that the criterion continues to separate $\theta_T^\star$ from
other parameter values as $T$ changes. Its role is narrow but essential: it is the
condition that turns uniform convergence of the sample criterion into consistency of
$\widehat\theta_N$ for the sample-period estimand.

For the linearization in part (ii), the theorem uses the usual interior local solution. In the linear GMM applications considered below, this is the standard closed-form estimator on the event that the sample Gram matrix is nonsingular, so the condition is substantive mainly when boundary solutions or singular sample criteria are plausible.

\begin{assumption}\label{ass:gmm-interior}
There exists $\eta>0$ such that $\{\theta\in\mathbb R^p:\|\theta-\theta_T^\star\|\le \eta\}\subseteq\Theta$ for all sufficiently large $T$. With probability approaching one, $\widehat\theta_N$ lies in this neighborhood and is an interior local minimizer of $J_N(\theta):=g_N(\theta)^\top \widehat A_N g_N(\theta)$.
\end{assumption}

Assumption~\ref{ass:gmm-interior} is only needed for the local first-order expansion in
Theorem~\ref{thm:AN}(ii). The deterministic ball condition makes $\theta_T^\star$ an interior point of $\Theta$ for all large $T$, so differentiability of the population criterion yields the population first-order condition at $\theta_T^\star$. The high-probability sample-local-minimizer condition separately ensures that the sample first-order condition is available and that the line segment used in the mean-value expansion remains inside $\Theta$. The assumption is not needed for the basic consistency conclusion in Theorem~\ref{thm:AN}.

The first main result is that, under the conditions above, the GMM estimator is consistent and asymptotically normal, with asymptotic variance induced by the design variance of the moments.

\begin{theorem}\label{thm:AN}
For the estimator with a fully observed sample in this section, $N=T$. Under Assumptions~\ref{ass:design-environment}, \ref{ass:smoothness}, \ref{ass:uniform-lln}, \ref{ass:gmm}, and \ref{ass:estimand-separation}, any measurable selection $\widehat\theta_N\in\arg\min_{\theta\in\Theta} g_N(\theta)^\top \widehat A_N g_N(\theta)$ satisfies $\widehat\theta_N-\theta_T^\star\to_p 0$. If Assumptions~\ref{ass:dependence}, \ref{ass:local-correct-specification}, and \ref{ass:gmm-interior} also hold, then
\begin{equation*}
\sqrt{N}\,(\widehat\theta_N-\theta_T^\star)\Rightarrow \mathcal{N}(0,\Sigma),
\end{equation*}
where
\[
\Sigma:=(G^\top A G)^{-1}G^\top A\,\Omega_R\,A G\,(G^\top A G)^{-1}.
\]
For any continuously differentiable $h:\Theta\to\mathbb{R}^q$, let $\mathcal J_T:=\nabla_\theta h(\theta_T^\star)$. If $\mathcal J_T\to\mathcal J$, then $\sqrt{N}\,\big(h(\widehat\theta_N)-h(\theta_T^\star)\big)\Rightarrow \mathcal{N}(0,\mathcal J\Sigma\mathcal J^\top)$.
\end{theorem}

When $\Omega_R$ is nonsingular, the inverse-design-covariance weighting is $A^\star=\Omega_R^{-1}$ and the corresponding asymptotic covariance is $(G^\top\Omega_R^{-1}G)^{-1}$. This theorem is a finite-history limit theory for the design-based estimand $\theta_T^\star$. It does not reinterpret $\theta_T^\star$ as a population parameter. As with all the theory in this paper, the randomness in the limiting distribution is the assignment-shock randomness conditional on $\mathcal{E}_T$; the mean path itself is held fixed throughout the theorem.

The local correct-specification restriction has a narrow role: it keeps the main influence function in the familiar GMM form based on the centered moment innovation $e_{T,t}$. In an overidentified system with first-order pseudo-true mean $\bar m_T^\star:=\bar m_T(\theta_T^\star)$, the relevant moment is augmented by Jacobian noise. With $J_{T,t}:=\nabla_\theta g_t(W_t,\theta_T^\star)$ and $G_T:=T^{-1}\sum_t\mathbb E_T[J_{T,t}]$, the first-order term is
\[
\zeta_{T,t}:=G_T^\top A e_{T,t}+\big(J_{T,t}-G_T\big)^\top A\bar m_T^\star .
\]
Appendix~\ref{app:misspecified-gmm} records the corresponding expansion for the fixed-weight criterion with $\widehat A_N\equiv A$. Thus the clean HAC decomposition for the raw moment is a locally correctly specified result, which becomes more complicated under local misspecification.

For the running LP example, the theorem is obtained under the familiar requirements that the assignment shock and predetermined controls generate a finite-memory or short-memory moment process, that the relevant products have finite moments, and that the LP Gram matrix remains nonsingular. For the stable VAR example, it is obtained when the innovation process is short-memory, the autoregressive filters are absolutely summable, and $\theta_T^\star$ is locally identified. For HAC, these first-order conditions are not enough: the lag-product arrays, meaning arrays formed from products of moments at different lags, must also obey the covariance-envelope restrictions used in Proposition~\ref{prop:primitive-verification}.

The following corollary records one sufficient route for the running LP example. Fix a finite horizon set $\mathcal H$ and finite-dimensional regressors $\psi_t=(1,x_t,c_t^\top)^\top$. For each $h\in\mathcal H$, write $u_{T,t,h}^\star:=y_{t+h}-\psi_t^\top\theta_{T,h}^\star$ and $g^{\mathrm{LP}}_{T,t,h}:=\psi_tu_{T,t,h}^\star$. Suppose Assumption~\ref{ass:design-environment} holds. Conditional on $\mathcal E_T$, impose three sufficient restrictions. First, for some fixed integer $q<\infty$, the vector collecting $\psi_t$, $(u_{T,t,h}^\star)_{h\in\mathcal H}$, and the products entering the LP moment function and Jacobian is measurable with respect to the design shocks in the finite window $t-q,\ldots,t+\max\mathcal H+q$. Second, after enlarging $q$ if necessary, the underlying design-shock array is $q$-dependent, and the primitive products entering the moment, Jacobian, and moment-product arrays have uniformly bounded $4+\delta$ moments for some $\delta>0$. Third, the fixed-lag Ces\`aro limits of the primitive second moments exist, $T^{-1}\sum_{t=1}^T\mathbb E_T[\psi_t\psi_t^\top]\to Q_\psi\succ0$, and the sample-period LP normal equations are locally correctly specified.

\begin{corollary}\label{cor:lp-primitive}
Under the LP finite-window, moment, Ces\`aro-limit, nonsingularity, and local correct specification conditions in the preceding paragraph, the LP moment satisfies Assumptions~\ref{ass:smoothness}, \ref{ass:dependence}, \ref{ass:uniform-lln}, and~\ref{ass:local-correct-specification}; the same finite-window product-array restrictions verify the stochastic part of Assumption~\ref{ass:hac-regularity} for the kernel and bandwidth under consideration. If Assumptions~\ref{ass:gmm}, \ref{ass:estimand-separation}, and~\ref{ass:gmm-interior} also hold for the stacked LP criterion, then Theorem~\ref{thm:AN} applies to the stacked LP estimator. If the centered LP mean path also satisfies Assumption~\ref{ass:mean-path} and the kernel and bandwidth satisfy Assumption~\ref{ass:hac-kernel}, then Proposition~\ref{lem:fixed-env-decomp} and Theorem~\ref{thm:hac} apply to the same LP moment.
\end{corollary}

The next corollary gives the analogous sufficient route for the stable VAR normal equations used as a second running example. Fix $p$, $n$, and $m$, and consider the VAR moments $g^{\mathrm{VAR}}_t(\phi)=\operatorname{vec}(X_tu_t(\phi)^\top)$ with $X_t^\top=(1,Y_{t-1}^\top,\ldots,Y_{t-p}^\top,W_t^\top)$. Suppose Assumption~\ref{ass:design-environment} holds and, conditional on $\mathcal E_T$, the VAR moment function, its Jacobian, and the moment-product arrays satisfy the fully observed sample sufficient conditions in Assumption~\ref{ass:primitive-app} of Appendix~\ref{app:primitive-verification} after replacing the generic moment by $g_t^{\mathrm{VAR}}$. Suppose also that $T^{-1}\sum_{t=1}^T\mathbb E_T[X_tX_t^\top]\to Q_X\succ0$ and that the sample-period VAR normal equations are locally correctly specified.

\begin{corollary}\label{cor:var-primitive}
Under the VAR moment-product, nonsingularity, and local correct specification conditions in the preceding paragraph, the VAR moment satisfies Assumptions~\ref{ass:smoothness}, \ref{ass:dependence}, \ref{ass:uniform-lln}, and~\ref{ass:local-correct-specification}; the product-array restrictions verify Assumption~\ref{ass:hac-regularity} for the kernel and bandwidth covered by Assumption~\ref{ass:primitive-app}. If Assumptions~\ref{ass:gmm}, \ref{ass:estimand-separation}, and~\ref{ass:gmm-interior} also hold for the VAR criterion, then Theorem~\ref{thm:AN} applies to the VAR coefficient estimator. If the centered mean path satisfies Assumption~\ref{ass:mean-path} and the kernel and bandwidth satisfy Assumption~\ref{ass:hac-kernel}, then Proposition~\ref{lem:fixed-env-decomp} and Theorem~\ref{thm:hac} apply to the same VAR moment.
\end{corollary}

The two primitive corollaries are stated for fixed-dimensional LPs and VARs with a fixed horizon or lag order. The sufficient conditions in Appendix~\ref{app:primitive-verification} are stronger than necessary, especially the moment-product and covariance-envelope conditions used for HAC, but they make the product-array content of the high-level assumptions explicit. Cases with many horizons, high-dimensional controls, time-varying parameter dimension, data-selected adjustment sets, or bandwidth rules outside the stated lag-window conditions require separate empirical-process and HAC arguments.

\subsection{Mean-path decomposition}

The next assumption is imposed on the centered mean component of the moment to which the HAC calculation is applied. It requires that this fixed path have a stable short-run, or HAC, long-run covariance.

\begin{assumption}\label{ass:mean-path}
For each fixed $\ell\ge0$, the limit
\[
\Gamma_{g,\mu}(\ell):=\lim_{T\to\infty}T^{-1}\sum_{t=\ell+1}^T
\tilde \mu_{T,t}\tilde \mu_{T,t-\ell}^\top
\]
exists, with $\Gamma_{g,\mu}(-\ell):=\Gamma_{g,\mu}(\ell)^\top$ for $\ell>0$, and $\sum_{\ell\in\mathbb{Z}}\|\Gamma_{g,\mu}(\ell)\|<\infty$.
\end{assumption}

Assumption~\ref{ass:mean-path} should not be read as ruling out deterministic changes in the historical environment. Under the design-based interpretation, the realized paths of states, regimes, calendar variables, sampling windows, and potential dynamic effects are fixed once the conditioning environment is fixed. In that sense, the mean path $(\tilde\mu_{T,t})_{t\le T}$ is itself a deterministic object. The assumption answers a narrower question: after the components treated as fixed by the design have been accounted for in the moment, is the remaining sample-centered mean path the kind of object to which a growing-bandwidth HAC long-run-variance calculation can be applied?

The restriction is therefore a restriction on unresolved low-frequency variation in the moment, not a requirement that impulse responses or shock variances be time invariant. It allows predictable movement in dynamic causal effects to contribute to the mean-path term as long as the centered mean path has stable fixed-lag Ces\`aro autocovariances and those limiting autocovariances are summable in lag. This is the time-series analogue of allowing treatment-effect heterogeneity in the finite-population cross-sectional setting: the heterogeneity may be present in the fixed history, but the variance calculation still needs a regular limiting object. For example, variation driven by a realized short-memory state path can satisfy the assumption even though, conditional on the realized history, that path is fixed.

What the assumption does exclude from the raw HAC calculation is a persistent low-frequency component that has not otherwise been modeled or partialled out. A permanent deterministic break, a monotone trend, a fixed seasonal pattern, or a low-frequency drift of fixed magnitude can have fixed-lag Ces\`aro autocovariances that remain nonnegligible over arbitrarily many lags. A growing-bandwidth HAC estimator would then be accumulating a deterministic low-frequency component rather than converging to a finite long-run covariance. That failure does not contradict the motivation of the paper; it says that such components should not be forced into the same short-memory HAC approximation.

The condition is imposed on the moment after any residualization that is part of the design. If the fixed mean path can be decomposed as
\[
\tilde\mu_{T,t}=D_{T,t}\lambda_T+\nu_{T,t},
\]
where $D_{T,t}$ is a fixed low-dimensional deterministic block, such as trends, calendar terms, sampling-window indicators, or regime indicators, then raw HAC applied to the unreduced moment need not have a finite long-run limit when $D_{T,t}\lambda_T$ is persistent. Including $D_{T,t}$ in the residualization step removes this deterministic component before the lag-window limit is taken. The relevant mean-path condition is then the corresponding condition on the residual path $\nu_{T,t}$, or on the residualized moment more generally, not on the unreduced low-frequency component. This is why the implementation first residualizes deterministic terms associated with the fixed design and only then uses stochastic predetermined covariates for projection adjustment.

In the special case of time-invariant impulse responses and time-invariant shock variances, $\tilde\mu_{T,t}\equiv0$, so Assumption~\ref{ass:mean-path} is automatic. More generally, the assumption permits deterministic time variation that is either explicitly controlled by the fixed design or leaves behind a residual centered mean path with finite HAC long-run covariance.

Write
\begin{equation}\label{eq:Omegam}
\Omega_\mu:=\sum_{\ell\in\mathbb{Z}}\Gamma_{g,\mu}(\ell).
\end{equation}
The sum in \eqref{eq:Omegam} is a HAC long-run covariance: fixed-lag Ces\`aro limits are taken first, and the resulting lag sequence is then summed. It is not the finite-$T$ sum of all sample autocovariances of the centered path, which would vanish because $\sum_t\tilde\mu_{T,t}=0$. Thus $\Omega_\mu$ measures the short-memory long-run covariance of the residual predictable mean path that remains in the moment being fed to HAC.

The following elementary decomposition is the core object used by the HAC results. For each fixed $\ell\ge0$, define
\begin{align}
\Gamma_{g,\mathrm{hac}}(\ell)
&:=\lim_{T\to\infty}\frac{1}{T}\sum_{t=\ell+1}^T
\mathbb{E}_T\!\left[\Big(g_t(W_t,\theta_T^\star)-\bar m_T(\theta_T^\star)\Big)
\Big(g_{t-\ell}(W_{t-\ell},\theta_T^\star)-\bar m_T(\theta_T^\star)\Big)^\top\right], \label{eq:Gamma-ghac}\\
\Gamma_{g,\mathrm{hac}}(-\ell)&:=\Gamma_{g,\mathrm{hac}}(\ell)^\top,\qquad \ell>0,\nonumber
\end{align}
and let $\Omega_R^+:=\sum_{\ell\in\mathbb{Z}}\Gamma_{g,\mathrm{hac}}(\ell)$ whenever the sum exists.

\begin{proposition}\label{lem:fixed-env-decomp}
Under Assumptions~\ref{ass:design-environment}, \ref{ass:dependence}, and~\ref{ass:mean-path}, the fixed-lag limits in \eqref{eq:Gamma-ghac} exist, $\sum_{\ell\in\mathbb Z}\|\Gamma_{g,\mathrm{hac}}(\ell)\|<\infty$, and
\begin{equation}\label{eq:OmegaR}
\qquad
\Omega_R=\Omega_R^+-\Omega_\mu,
\end{equation}
where $\Omega_\mu\succeq0$.
\end{proposition}

This decomposition is one of the main results and highlights the main differences between the usual and design-based variances in the time series setting. Proposition~\ref{lem:fixed-env-decomp} compares two covariance calculations applied to the same moment. The design CLT uses the centered innovation $e_{T,t}$, while HAC applied to the centered observed moment also contains the fixed path $\tilde\mu_{T,t}$. This is the formal reason that the HAC limit contains the design variance plus a nonnegative mean-path component rather than an unrestricted collection of covariance terms. If one instead lets the conditioning object that defines the mean path be random under the design distribution, the same clean decomposition requires additional long-run orthogonality between the innovation and the predictable path. The main theory avoids that ambiguity by conditioning on $\mathcal{E}_T$.

In \eqref{eq:OmegaR}, the key term in the asymptotic variance is the design-based long-run covariance $\Omega_R$. Proposition~\ref{lem:fixed-env-decomp} shows that $\Omega_R$ equals the conventional HAC limit minus the nonnegative mean-path term $\Omega_\mu$. The formula reduces to the usual HAC long-run variance when $\tilde \mu_t\equiv 0$, and in the i.i.d.\ no-overlap setting of \citet{KakehiMatsushitaOtsu2026FinitePopulationGMM} it collapses to $\Omega_R= \Gamma_{g,\mathrm{hac}}(0) - \Gamma_{g,\mu}(0)$. In general, $\Omega_R$ equals the HAC-type component minus the mean-path term $\Omega_\mu$, which is nonnegative but unidentified. This is why conventional standard errors are conservative and why projection or regression adjustment can sharpen the bound under the conditions below. Appendix~\ref{app:incomplete-sampling} records the corresponding formulas under incomplete observation. The applied reading is simple: HAC counts predictable changes in the mean moment as time-series uncertainty, while the design variance does not when those changes are part of the fixed historical environment.

For LPs and VARs, Remark~\ref{rem:LP_VAR_mean_drift} gives a direct interpretation of $\Omega_\mu$. Under the fixed-state convention, the displayed conditional expectations are the fixed date-specific means $\mu_{T,t}(\theta_T^\star)$. Under the dynamic-state convention, the same formulas describe predictable components before averaging over the moving pre-shock state; the fixed mean path in Assumption~\ref{ass:mean-path} is then the corresponding $\mathbb E_T$ expectation. In the LP running example, these predictable components capture changes in the conditional impulse response $\tau_{t,h}$ or in the shock variance $\Var(x_t\mid\I_t)$ over $t$. Thus $\Omega_R$ captures the short-memory part of moment variation, while $\Omega_\mu=\Omega_R^+-\Omega_R$ captures variation in these predictable components.

\begin{remark}\label{rem:OmegaMu}
Write $\Omega_\mu:=\sum_{\ell\in\mathbb{Z}}\Gamma_{g,\mu}(\ell)=\mathrm{LRV}(\tilde\mu_t)$, where $\tilde\mu_t:=\mu_{T,t}(\theta_T^\star)-\bar m_T(\theta_T^\star)$ and the displayed conditional expectations below give sufficient expressions for $\mu_{T,t}(\theta_T^\star)$. Then, under the additional LP condition stated in Remark~\ref{rem:LP_VAR_mean_drift}(i), if $d_{t,h}^{LP}:=\Var(x_t\mid\I_t)(\tau_{t,h}-\beta_h^\star)$ and $\bar d_{T,h}^{LP}:=T^{-1}\sum_{t=1}^T d_{t,h}^{LP}$, then $\Omega_{\mu}^{LP}=\mathrm{LRV}(d_{t,h}^{LP}-\bar d_{T,h}^{LP})\,e_xe_x^\top$, where $e_x$ selects the $x$-row of the LP moment. For the VAR case, if
\[
d_t^{VAR}:=
\begin{bmatrix}
\operatorname{vec}\!\big(Z_t^{\mathrm{pred}}\Delta_t^\top\big)\\[0.2em]
\operatorname{vec}\!\big(\Sigma_{W,t}(B(Z_t)-B)^\top\big)
\end{bmatrix},
\qquad
\bar d_T^{VAR}:=\frac1T\sum_{t=1}^T d_t^{VAR},
\]
then $\Omega^{VAR}_\mu=\mathrm{LRV}(d_t^{VAR}-\bar d_T^{VAR})$. In both cases $\Omega_\mu\succeq 0$ (see Lemma~\ref{lem:sumGE-psd}).
\end{remark}

The appendix proof of Remark~\ref{rem:LP_VAR_mean_drift} derives the displayed conditional mean paths. Combining those paths with the definition of $\Omega_\mu$ gives the formulas in Remark~\ref{rem:OmegaMu}. For LPs, under the additional mean condition, the extra term is a rank-one block proportional to the long-run covariance of the centered scalar drift in the $x_t$-row. For VARs, the drift has two blocks: time variation in the deterministic and autoregressive coefficients through $Z_t^{\mathrm{pred}}\Delta_t^\top$ and state dependence in the impact matrix through $\Sigma_{W,t}(B(Z_t)-B)^\top$. By Lemma~\ref{lem:sumGE-psd}, positive semidefiniteness follows, so $\Omega_\mu$ measures time variation in predictable dynamic effects, the time-series analogue of heterogeneity in the cross-sectional design-based literature.

Assumptions~\ref{ass:dependence} and~\ref{ass:mean-path} separate the innovation part of the moment from the predictable mean-path drift induced by time-varying dynamic causal effects. For LPs with finite-lag controls and for stable VARs with short-memory innovations, these assumptions reduce to standard finite-moment and weak-dependence conditions, while allowing rich state dependence in impulse responses so long as the centered drift has finite long-run variance. Appendix~\ref{app:incomplete-sampling} records the incomplete-observation extension. Appendix~\ref{app:lpvar} gives some basic results on the bias-variance tradeoff in the choice of LPs versus VARs in the design-based setting. The results are essentially the same as in the standard case.

\section{Calculating standard errors} \label{sec:calculating-SEs}

\subsection{The HAC variance}

Under the finite-history design interpretation, HAC standard errors play the same role as Eicker--Huber--White standard errors in the cross-sectional finite-population setting: under the conditions below, they estimate a conservative variance limit for scalar design inference. They are computed as
\begin{equation}\label{eq:HAC}
\widehat{\Omega}_R^{+}(L,K):=\widehat{\Gamma}_{\mathrm{hac}}(0)+\sum_{\ell=1}^{T-1}K\!\left(\frac{\ell}{L}\right)\left(\widehat{\Gamma}_{\mathrm{hac}}(\ell)+\widehat{\Gamma}_{\mathrm{hac}}(\ell)^\top\right),
\end{equation}
where $K$ is a bounded, symmetric kernel function and $L=L_T$ is a bandwidth parameter. The empirical lag covariances are, for $\ell\ge 0$,
\[
\widehat{\Gamma}_{\mathrm{hac}}(\ell)
:=\frac{1}{T}\sum_{t=\ell+1}^{T} \big(g_t(W_t,\widehat\theta_N)-g_N(\widehat\theta_N)\big)\big(g_{t-\ell}(W_{t-\ell},\widehat\theta_N)-g_N(\widehat\theta_N)\big)^{\!\top}.
\]

Because Assumption~\ref{ass:local-correct-specification} imposes $\sqrt{T}\,\|\bar m_T(\theta_T^\star)\|=o(1)$, the raw and centered fixed-lag limits coincide, and recentering by either $g_N(\widehat\theta_N)$ or the infeasible $\bar m_T(\theta_T^\star)$ changes the weighted HAC sum by at most $O_p(L_T\|\bar m_T(\theta_T^\star)\|^2)+o_p(1)=o_p(1)$ under Assumption~\ref{ass:hac-kernel}. The centered formula is the natural finite-sample analogue of the centered long-run variance limit $\Omega_R^{+}$; the conventional uncentered formula has the same limit under the stated local correct specification condition. Write $\widehat{\Omega}_R^{+}$ as shorthand for $\widehat{\Omega}_R^{+}(L_T,K)$. The paper places mild assumptions on the kernel in order to rely on standard HAC theory.

\begin{assumption}\label{ass:hac-kernel}
$K$ is bounded and symmetric, with $K(0)=1$, $K(x)=0$ for $|x|>1$, and continuity at $0$. The bandwidth parameter $L_T$ satisfies $L_T\to\infty$ and $L_T/\sqrt T\to 0$.
\end{assumption}

Assumption~\ref{ass:hac-kernel} is the kernel and bandwidth condition for feasible HAC estimation. The rate $L_T/\sqrt T\to0$ is slightly stronger than the minimal infeasible lag-window requirement $L_T/T\to0$ for the same centered path; it is used only to make the feasible plug-in replacement $g_t(W_t,\widehat\theta_N)$ negligible uniformly across the $O(L_T)$ lags in the HAC sum. The assumption is otherwise the usual lag-window condition for consistent long-run variance estimation with a diverging bandwidth. It is only an implementation assumption for HAC and the multiplier bootstrap; it is not part of the core design-based consistency/CLT result for the GMM estimator itself. Thus admissible bandwidths share the same first-order limit, although finite-sample conservativeness can depend materially on the bandwidth. Let $s_{T,t}:=g_t(W_t,\theta_T^\star)-\bar m_T(\theta_T^\star)$. For $\ell\ge0$, write $\widehat\Gamma_s(\ell):=T^{-1}\sum_{t=\ell+1}^{T}s_{T,t}s_{T,t-\ell}^{\top}$, with the transpose convention for negative lags.

\begin{assumption}\label{ass:hac-regularity}
For the kernel and bandwidth sequence used in the HAC estimator, the fixed-lag limits $\Gamma_{g,\mathrm{hac}}(\ell)$ in \eqref{eq:Gamma-ghac} exist and satisfy $\sum_{\ell\in\mathbb{Z}}\|\Gamma_{g,\mathrm{hac}}(\ell)\|<\infty$. For every fixed $\ell$, $\widehat\Gamma_s(\ell)\xrightarrow{p}\Gamma_{g,\mathrm{hac}}(\ell)$. For that same kernel and bandwidth sequence,
\[
\left\|\sum_{1\le |\ell|\le L_T}K(|\ell|/L_T)
\big(\widehat\Gamma_s(\ell)-\Gamma_{g,\mathrm{hac}}(\ell)\big)\right\|\xrightarrow{p}0.
\]
If $\widehat\Gamma_{\hat s}(\ell)$ denotes the same lag covariance with $\widehat s_t:=g_t(W_t,\widehat\theta_N)-g_N(\widehat\theta_N)$ in place of $s_{T,t}$, then for each fixed $\ell$, $\|\widehat\Gamma_{\hat s}(\ell)-\widehat\Gamma_s(\ell)\|\to_p0$, and
\[
\left\|\sum_{1\le |\ell|\le L_T}K(|\ell|/L_T)
\big[
\widehat\Gamma_{\hat s}(\ell)-\widehat\Gamma_s(\ell)
\big]\right\|\to_p0 .
\]
\end{assumption}

Assumption~\ref{ass:hac-regularity} is the lag-window law of large numbers needed for HAC consistency. The fixed-lag clause controls the zero and finitely many nonzero lags, while the growing-window clause controls the additional lags admitted as $L_T$ diverges. The assumption is weaker than stationarity of the raw moment and stronger than the innovation CLT alone because HAC uses products over a growing set of lags.

Appendix~\ref{app:primitive-verification} gives a formal sufficient route. It imposes a uniform covariance-sum bound on the centered moment-product arrays over the retained lags; together with $L_T/\sqrt T\to0$, this bound gives the growing-lag law by Chebyshev's inequality. The feasible plug-in step then follows from the $T^{-1/2}$ GMM rate and the same bandwidth restriction. The $4+\delta$ moment-envelope moments used there are sufficient for products and are not part of the abstract moment-level theorem.

The lag-window class covered by Assumption~\ref{ass:hac-kernel} includes Bartlett, Parzen, and compactly supported flat-top kernels. Quadratic-spectral kernels, prewhitening, and fixed-$b$ asymptotics require separate arguments. If the centered moment has an exact finite dependence window, a fixed flat-top kernel whose $K=1$ region covers that window recovers the same conservative variance limit; this special case is stated as the fixed-bandwidth alternative in the theorem. The theorem therefore gives consistency for $\Omega_R^+$, the conservative HAC variance limit, rather than for $\Omega_R$ itself.

\begin{theorem}\label{thm:hac}
Suppose the sample is fully observed. Under Assumptions~\ref{ass:design-environment}, \ref{ass:dependence}, \ref{ass:mean-path}, \ref{ass:local-correct-specification}, \ref{ass:smoothness}, \ref{ass:uniform-lln}, \ref{ass:gmm}, \ref{ass:estimand-separation}, and \ref{ass:hac-regularity}, if the kernel and bandwidth satisfy Assumption~\ref{ass:hac-kernel}, then the feasible centered HAC estimator satisfies $\widehat{\Omega}_R^{+}(L_T,K) \xrightarrow{p} \Omega_R^{+}$.
The same conclusion holds with fixed $L$ if the fixed-lag and plug-in clauses of Assumption~\ref{ass:hac-regularity} hold for that fixed lag set, there exists $q<\infty$ such that the finite-$T$ centered lag covariance of the moment is zero for all $|\ell|>q$, eventually in $T$, and the flat-top kernel satisfies $K(\ell/L)=1$ for every integer $|\ell|\le q$.
\end{theorem}

This theorem identifies the probability limit of the feasible HAC estimator as the conservative moment matrix $\Omega_R^+$, not as the design long-run variance $\Omega_R$ of the centered moment. The distinction matters exactly when $\Omega_\mu\neq 0$: HAC treats the fixed mean path as part of the long-run moment variation, whereas the design moment variance treats only deviations from that path as random.

\begin{corollary}\label{cor:scalar-conservative}
Let
\[
\Sigma^+:=(G^\top A G)^{-1}G^\top A\,\Omega_R^+\,A G\,(G^\top A G)^{-1}.
\]
Under the conditions of Theorems~\ref{thm:AN} and~\ref{thm:hac}, $\Sigma^+-\Sigma\succeq0$. Hence, for any fixed vector $a$, $a^\top\Sigma a\le a^\top\Sigma^+a$.
\end{corollary}

\noindent Furthermore, the following result holds.

\begin{corollary}\label{cor:scalar-coverage}
Suppose the conditions of Theorems~\ref{thm:AN} and~\ref{thm:hac} hold. Let $\tau_T:=a^\top h(\theta_T^\star)$ and $\widehat\tau:=a^\top h(\widehat\theta_N)$ for a continuously differentiable functional $h:\Theta\to\mathbb R^q$ and a fixed vector $a\in\mathbb R^q$. Suppose $\mathcal J_T:=\nabla_\theta h(\theta_T^\star)\to\mathcal J$, set $d:=\mathcal J^\top a$, define $v:=d^\top\Sigma d$ and $v^+:=d^\top\Sigma^+d$, and let $\widehat v^+\to_p v^+$. If $v^+>0$, then the Wald interval
\[
\widehat\tau \pm z_{1-\alpha/2}\sqrt{\widehat v^+/T}
\]
has asymptotic design coverage at least $1-\alpha$. If $v>0$, the limiting coverage is
\[
2\Phi\!\left(z_{1-\alpha/2}\sqrt{v^+/v}\right)-1\ge 1-\alpha .
\]
\end{corollary}

Theorem~\ref{thm:hac} and Corollaries~\ref{cor:scalar-conservative}--\ref{cor:scalar-coverage} show that HAC standard errors deliver asymptotically conservative scalar design inference for smooth functions of the finite-history estimand under the fully observed, locally correctly specified conditions. In a constant-coefficient VAR with martingale-difference innovations and no mean-path drift, the relevant score autocovariances may vanish beyond lag zero, so the HAC expression collapses to the usual VAR covariance formula. The HAC notation is used here as a common long-run-variance representation that also covers state dependence, predictable mean-path drift, and other macro GMM moments with nontrivial serial dependence.

As mentioned above, $\Omega_\mu$ is not identified from a single realized history without additional structure. Section~\ref{sec:implementation} shows that projecting the moment on predetermined covariates can remove predictable components and reduce the conservative HAC variance limit; its interpretation as a tighter conservative variance bound requires the stronger conditions discussed there and in the appendix.

\subsection{The multiplier bootstrap}
Macro bootstrap procedures often approximate impulse-response distributions by resampling moment-contribution-like objects \citep{GoncalvesKilian2004,KilianLuetkepohl2017}. The design-based analogue here is a dependent Gaussian multiplier applied to centered GMM moments, and it estimates the same conservative HAC covariance matrix as in Theorem~\ref{thm:hac}. Let the per-date GMM moments stack into $g_t(\theta)\in\mathbb{R}^k$, with sample mean $g_N(\theta):=T^{-1}\sum_{t=1}^T g_t(\theta)$ as above and Jacobian $\widehat G_N(\theta):=\nabla_\theta g_N(\theta)$.

The following two examples help to make the Jacobian in the bootstrap update explicit. For LPs at each horizon $h\in\mathcal H$, let $\psi_t:=(1,x_t,c_t^\top)^\top$ and $\theta_h:=(\alpha_h,\beta_h,\gamma_h^\top)^\top$. If $g^{\mathrm{LP}}_{t,h}(\theta_h)=\psi_t(y_{t+h}-\psi_t^\top\theta_h)$ and $g_t(\theta)=\big(g^{\mathrm{LP}}_{t,h}(\theta_h)\big)_{h\in\mathcal H}$, then $\widehat G_N(\theta)$ is block diagonal with $h$-blocks $-T^{-1}\sum_t\psi_t\psi_t^\top$. When the conditional Gram average converges, the corresponding limit is $G_{\mathrm{LP}}=-\lim_T T^{-1}\sum_t\mathbb E_T[\psi_t\psi_t^\top]$. Moreover, for a VAR$(p)$, let $X_t^\top=(1,y_{t-1}^\top,\ldots,y_{t-p}^\top,x_t^\top)$, $u_t(\varphi):=y_t-\Phi_0-\Phi_1y_{t-1}-\cdots-\Phi_py_{t-p}-Bx_t$, and $\varphi:=\operatorname{vec}(\Phi_0,\Phi_1,\ldots,\Phi_p,B)$. The moment is $g^{\mathrm{VAR}}_t(\varphi)=\operatorname{vec}\!\big(X_t u_t(\varphi)^\top\big)=(I_n\!\otimes X_t)u_t(\varphi)$, so $\widehat G_N(\varphi)=-T^{-1}\sum_t(I_n\!\otimes X_t)\mathcal R_t$, where $\mathcal R_t$ is the linear map with $u_t(\varphi)=y_t-\mathcal R_t\varphi$.

For the multiplier bootstrap, choose a kernel/bandwidth pair such that the Toeplitz matrix $\big(K(|t-s|/L)\big)_{1\le t,s\le T}$ is positive semidefinite. Draw a mean-zero Gaussian vector $\xi=(\xi_1,\ldots,\xi_T)$ with covariance $\mathrm{Cov}^*(\xi_t,\xi_s)=K(|t-s|/L)$. Form the bootstrap moment average $B_N^*:=T^{-1/2}\sum_{t=1}^T \xi_t\big(g_t(\widehat\theta_N)-g_N(\widehat\theta_N)\big)$. Write $\widehat G_N:=\widehat G_N(\widehat\theta_N)$. The bootstrap update is
\[
\sqrt T\!\left(\widehat\theta_N^*-\widehat\theta_N\right)
=
-\Big(\widehat G_N^\top \widehat A_N \widehat G_N\Big)^{-1}\widehat G_N^\top \widehat A_N\,B_N^*,
\]
whose conditional covariance is the plug-in GMM covariance induced by the HAC moment matrix with kernel $K$ and bandwidth $L$. For the Gaussian multiplier bootstrap, the finite-sample covariance matrix of the multipliers must be positive semidefinite. The paper therefore restricts attention to kernels and bandwidths for which the Toeplitz matrix $\big(K(|t-s|/L_T)\big)_{1\le t,s\le T}$ is positive semidefinite.

\begin{theorem}\label{thm:bootstrap}
Suppose the conditions of Theorem~\ref{thm:hac} hold, Assumption~\ref{ass:gmm-interior} holds, and either the kernel and bandwidth satisfy Assumption~\ref{ass:hac-kernel} or the fixed flat-top alternative in Theorem~\ref{thm:hac} applies. Suppose also, where $L_T$ denotes the chosen bandwidth, possibly constant under the fixed flat-top alternative, that the multiplier Toeplitz matrix $\big(K(|t-s|/L_T)\big)_{1\le t,s\le T}$ is positive semidefinite for each $T$. Conditionally on the data, $\sqrt{T}(\widehat\theta_N^\ast-\widehat\theta_N)\Rightarrow^\ast \mathcal{N}(0,\Sigma^+)$ in probability, where
\[
\Sigma^+\;:=\;(G^\top A G)^{-1}G^\top A\,\Omega_R^+\,A G\,(G^\top A G)^{-1}.
\]
\end{theorem}
Thus the bootstrap reproduces the conservative Gaussian law with covariance $\Sigma^+$; it reproduces the exact design law only in the no-mean-path case $\Omega_\mu=0$. Appendix~\ref{app:incomplete-sampling} records how the moment covariance changes under incomplete observation.

\section{Projection and regression adjustment}\label{sec:implementation}
This section studies a feasible way to reduce the conservative HAC variance limit $\Omega_R^+$. The basic result is a positive-semidefinite, or Loewner-order, reduction of the HAC limit. A stronger interpretation as a tighter conservative bound for $\Omega_R$ requires long-run orthogonality, meaning zero HAC long-run cross-covariance, between the centered innovation and the part of the date-specific mean path left after projecting on the adjustment covariates. Equality with the design variance obtains when the residualized centered mean path has zero long-run variance; a simple sufficient case is that the centered mean path lies exactly in the linear span of the adjustment covariates.

If the date-specific design mean $\mu_t=\mathbb{E}_T[g_t(W_t,\theta_T^\star)]$ were known, one could form the infeasible residual moment $r_t^\star:=g_t(W_t,\theta_T^\star)-\mu_t$ and apply HAC to $(r_t^\star)_{t\le T}$. The resulting long-run matrix would estimate $\Omega_R$ because $\mathbb{E}_T[r_t^\star]=0$ date by date. The difficulty is that $\mu_t$ is not identified from one realized history.

The feasible procedure uses predetermined adjustment variables to approximate the date-specific design mean. The econometric idea is to subtract variation that is predictable before the shock before applying HAC. This can reduce the reported HAC variance limit when those variables track the date-specific mean path, but it does not by itself prove that the adjusted matrix estimates the design variance. Here a variable is predetermined when it is dated before the shock whose moment is being adjusted. After centering or residualizing variables treated as fixed by the design, denote the retained adjustment vector by $z_t$. When such variables are themselves generated by past shocks, the basic positive-semidefinite reduction result applies to the stacked process $(z_t,s_t)$, while the stronger conservative variance interpretation requires the long-run orthogonality condition stated in Proposition~\ref{prop:HAC_RA_conservative}.

The population long-run projection objects are defined as follows. For vector processes $(a_t)$ and $(b_t)$ with absolutely summable two-sided cross-covariances, write $m_{a,T,t}:=\mathbb{E}_T[a_t]$ and $\tilde m_{a,T,t}:=m_{a,T,t}-T^{-1}\sum_{s=1}^T m_{a,T,s}$, with analogous definitions for $b_t$. For $\ell\ge0$, set $\Gamma_{ab,\mathrm{hac}}(\ell):=\lim_{T\to\infty}T^{-1}\sum_{t=\ell+1}^T \mathbb{E}_T[a_t b_{t-\ell}^\top]$ and $\Gamma_{ab,\mu}(\ell):=\lim_{T\to\infty}T^{-1}\sum_{t=\ell+1}^T \tilde m_{a,T,t}\tilde m_{b,T,t-\ell}^\top$. For negative lags use $\Gamma_{ab,\mathrm{hac}}(-\ell):=\Gamma_{ba,\mathrm{hac}}(\ell)^\top$ and $\Gamma_{ab,\mu}(-\ell):=\Gamma_{ba,\mu}(\ell)^\top$ for $\ell>0$. Define $S_{ab}^{+}:=\sum_{\ell\in\mathbb{Z}}\Gamma_{ab,\mathrm{hac}}(\ell)$ and $S_{ab}^{R}:=S_{ab}^{+}-\sum_{\ell\in\mathbb{Z}}\Gamma_{ab,\mu}(\ell)$. Because Assumption~\ref{ass:local-correct-specification} imposes $\bar m_T(\theta_T^\star)\to0$, the raw and centered fixed-lag limits coincide for the moment, so $S_{gg}^{+}=\Omega_R^{+}$ and $S_{gg}^{R}=\Omega_R$ when $a_t=b_t=g_t(W_t,\theta_T^\star)$.

In implementation, intercepts, deterministic trends, and regime dummies are first treated as variables fixed by the design and residualized out. The stochastic predetermined adjustment variables are centered or residualized with respect to those deterministic terms before forming long-run projection matrices. Start from a vector of raw $\I_t$-measurable adjustment variables, and let $z_t\in\mathbb{R}^{m}$ denote the resulting retained covariate vector after dropping directions that are linearly redundant in long-run second moment. Write $s_t:=g_t(W_t,\theta_T^\star)$ for the moment to be adjusted. Here $s_t$ denotes the moment being residualized, not an assignment shock. If $z_t$ is generated by past shocks, the basic positive-semidefinite reduction result applies to the stacked process $(z_t,s_t)$; the stronger interpretation of the adjusted matrix as a conservative refinement of $\Omega_R$ requires the long-run orthogonality condition stated in Proposition~\ref{prop:HAC_RA_conservative}.

For empirical analogues, define $\widehat\Gamma_{ab,\mathrm{hac}}(\ell):=T^{-1}\sum_{t=\ell+1}^{T} a_t b_{t-\ell}^{\!\top}$ for $\ell\ge0$, set $\widehat\Gamma_{ab,\mathrm{hac}}(-\ell):=\widehat\Gamma_{ba,\mathrm{hac}}(\ell)^\top$ for $\ell>0$, and let $\hat S_{ab}^{+}:=\sum_{|\ell|\le L} K(|\ell|/L)\,\widehat\Gamma_{ab,\mathrm{hac}}(\ell)$. Because $z_t$ has been centered or residualized on variables treated as fixed by the design, and because local correct specification makes the sample-period mean of the moment asymptotically negligible, these uncentered cross-HAC formulas have the same first-order limits as their centered analogues.

\begin{assumption}\label{ass:z_nondeg}
The stacked process $(z_t^\top,s_t^\top)^\top$ satisfies the same fixed-lag and lag-window HAC regularity as in Assumption~\ref{ass:hac-regularity} componentwise, and $S_{zz}^{+}:=\sum_{\ell}\Gamma_{zz,\mathrm{hac}}(\ell)\succ 0$.
\end{assumption}

Assumption~\ref{ass:z_nondeg} bundles the nondegeneracy of the retained covariate span with the stacked HAC regularity needed to estimate $S_{zz}^+$, $S_{zs}^+$, and $S_{sz}^+$ consistently. The positive-definiteness condition is therefore a normalization of the chosen covariate span, not a substantive restriction that every raw candidate covariate be noncollinear. This assumption is used only in Proposition~\ref{prop:HAC_RA} and the corresponding regression-adjusted covariance formulas.

The feasible projection coefficient is $\hat B:=(\hat S_{zz}^{+})^{-1}\hat S_{zs}^{+}$. The residualized moments are $\hat r_t:=g_t(W_t,\widehat\theta_N)-\hat B^\top z_t$. The reported adjusted matrix is $\widehat\Omega_{R}^{+}(\hat r):=\hat S_{\hat r\hat r}^{+}$. The corresponding GMM covariance matrix is
\[
\widehat\Sigma_{\mathrm{RA}}(A)
:=
(\widehat G_N^\top A \widehat G_N)^{-1}
\widehat G_N^\top A\,\widehat\Omega_{R}^{+}(\hat r)\,A \widehat G_N
(\widehat G_N^\top A \widehat G_N)^{-1}.
\]

The following result shows that the projection-adjusted standard errors converge to their limits under the design-based asymptotics. It also states the exact sense in which asymptotic variance is reduced.

\begin{proposition}\label{prop:HAC_RA}
Suppose that Assumptions~\ref{ass:design-environment}, \ref{ass:dependence}, \ref{ass:smoothness}, \ref{ass:uniform-lln}, \ref{ass:local-correct-specification}, \ref{ass:gmm}, \ref{ass:estimand-separation}, \ref{ass:mean-path},\ref{ass:hac-kernel}, \ref{ass:hac-regularity}, and~\ref{ass:z_nondeg} hold. If $L_T\to\infty$ and $L_T/\sqrt T\to0$, then $\widehat\Omega_{R}^{+}(\hat r)\to_p\Omega_R^+(r):=S_{rr}^+$, where $r_t=s_t-B^{\circ\top}z_t$ and $B^\circ:=(S_{zz}^+)^{-1}S_{zs}^+$. Moreover:
\begin{enumerate}[label=(\roman*)]
\item The adjusted HAC variance limit satisfies $\Omega_R^+(r)=S_{ss}^{+}-S_{sz}^{+}(S_{zz}^{+})^{-1}S_{zs}^{+}\preceq S_{ss}^{+}=\Omega_R^+$. Equality holds if and only if $S_{zs}^+=0$.
\item For any fixed weighting matrix $A\succ0$, the corresponding GMM asymptotic covariance satisfies $\Sigma_{\mathrm{RA}}(A):=(G^\top A G)^{-1}G^\top A\,\Omega_R^+(r)\,A G\,(G^\top A G)^{-1}\preceq\Sigma^+(A)$, where $\Sigma^+(A)$ is the same sandwich with $\Omega_R^+$ in place of $\Omega_R^+(r)$.
\item Whenever $\Omega_R^+(r)$ and $\Omega_R^+$ are positive definite, the corresponding inverse-variance-weight covariances satisfy $\Sigma_{\mathrm{RA}}^\star:=(G^\top \Omega_R^+(r)^{-1}G)^{-1}\preceq (G^\top (\Omega_R^+)^{-1}G)^{-1}=:\Sigma_{\mathrm{raw}}^\star$.
\end{enumerate}
\end{proposition}

The next proposition states the stronger conservative variance interpretation available when the residualized adjustment component is long-run orthogonal to the centered innovation component of the moment. In the fully observed setup, date-specific centering of $e_{T,t}$ already gives $S_{e\mu}^+=0$. The additional substantive restriction is the long-run orthogonality $S_{ez}^+=0$. This condition holds automatically for deterministic residualization variables fixed under the design, such as intercepts and deterministic trends, and is substantive for stochastic lagged macro variables generated by past shocks.

\begin{proposition}\label{prop:HAC_RA_conservative}
Suppose the conditions of Proposition~\ref{prop:HAC_RA} hold with a fully observed sample. Let $B^\circ:=(S_{zz}^{+})^{-1}S_{zs}^{+}$, $r_t:=s_t-B^{\circ\top}z_t$, and $d_t:=\tilde\mu_{T,t}-B^{\circ\top}z_t$. Suppose also that the long-run cross matrix $S_{ez}^{+}$ exists and equals zero, and that $S_{dd}^{+}$ exists. Then
\[
\Omega_R^+(r)=\Omega_R+S_{dd}^{+}
=\Omega_R+S_{\mu\mu}^{+}-S_{\mu z}^{+}(S_{zz}^{+})^{-1}S_{z\mu}^{+},
\]
and hence $\Omega_R\preceq\Omega_R^+(r)\preceq\Omega_R^+$. Equality with $\Omega_R$ holds if and only if $S_{dd}^{+}=0$; in particular, equality holds when there is a fixed matrix $B_0$ with $\tilde\mu_{T,t}=B_0^\top z_t$ along the sequence.
\end{proposition}

Proposition~\ref{prop:HAC_RA} concerns positive-semidefinite reduction of the conservative HAC variance limit $\Omega_R^+$. Proposition~\ref{prop:HAC_RA_conservative} gives the additional conditions under which the adjusted object is also conservative for $\Omega_R$. In the fully observed fixed-environment case, the cross term $S_{e\mu}^{+}$ vanishes by date-specific centering of $e_{T,t}$; the substantive long-run orthogonality condition in the proposition is $S_{ez}^{+}=0$.

The distinction matters because ``predetermined'' is a timing statement, not a guarantee of design validity. A lagged macro variable is observed before the current shock, but if it was generated by earlier shocks it moves under designs that redraw those earlier shocks. In that case Proposition~\ref{prop:HAC_RA} still gives a reduction of the HAC limit; Proposition~\ref{prop:HAC_RA_conservative} supplies the extra orthogonality condition needed to interpret the adjusted matrix as a tighter conservative bound for $\Omega_R$. The result is fixed-dimensional; high-dimensional or data-selected adjustment sets require incorporating the selection step into the moment or proving that it is asymptotically negligible. In finite samples either standard error can be larger, so the result does not justify adding covariates mechanically.

In applications the paper therefore reports baseline HAC and regression-adjusted HAC side by side, describes the timing and content of the covariate set, and interprets adjustment through whether the covariates are credible proxies for predictable date-specific drift. A convenient moment-level scalar summary is $\mathrm{Gap}:=[\mathrm{tr}(\widehat\Omega_R^{+})-\mathrm{tr}(\widehat\Omega_R^{+}(\hat r))]/\mathrm{tr}(\widehat\Omega_R^{+})$. For a selected parameter block, the same trace ratio can be computed after applying the GMM sandwich and the corresponding selection matrix. Large values of $\mathrm{Gap}$ indicate substantial predictable movement in date-specific mean moments captured by the chosen covariate set; they do not, without the stronger conditions, estimate $\Omega_\mu$ itself.

\section{Monte Carlo evidence}\label{sec:mc}
\label{sec:monte-carlo}

In each Monte Carlo setting below, the realized state and disturbance paths that define the conditioning environment are held fixed, and the assignment shocks are repeatedly re-randomized. This makes the sample-period estimand and its design variance directly observable by simulation. The data-generating process is:
\[
y_t = 0.85 y_{t-1} + \tau_t W_t + u_t, \qquad u_t = 0.5 u_{t-1} + \sqrt{1-0.5^2}\,\eta_t, \qquad c_t = 0.9 c_{t-1} + \sqrt{1-0.9^2}\,\varepsilon_t,
\]
with i.i.d.\ standard-normal $(\varepsilon_t,\eta_t)$, i.i.d.\ Rademacher $W_t$, and a burn-in of $500$ observations. Here, $W_t$ is the re-randomized shock, and $\tau_t$ is the date-specific effect of that shock. Thus time variation in $\tau_t$ is a source of predictable drift in the local-projection moment: even though the shock is randomized, the effect of the shock can be larger or smaller at dates with different realized states.

The main text reports two designs chosen to isolate when regression adjustment should help. Design A sets $\tau_t=1+c_t$. Here the shock effect is larger when the persistent state $c_t$ is high, and the same state $c_t$ is used as the predetermined adjustment variable. This is the aligned case: the adjustment variable is deliberately chosen to track the predictable component of the mean path, so regression-adjusted HAC should move the feasible variance estimate toward the oracle benchmark. Design B sets $\tau_t=1+c_t+0.75(c_t^2-1)$, but still adjusts only with $c_t$. The term $c_t^2-1$ is a population-centered nonlinear state component, so Design B keeps the same basic state-dependence while adding mean-path variation that a linear adjustment in $c_t$ cannot span. This design asks whether regression adjustment can still provide benefits when the adjustment only imperfectly spans the predictable mean-path variation.

At each horizon $h=0,\ldots,12$, the simulation estimates the local projection
\[
y_{t+h} = \alpha_h + \beta_h W_t + \gamma_h y_{t-1} + \delta_h c_t + \mathrm{error}_{t,h},
\]
where $\beta_h$ is the LP coefficient whose sample-period target is $\beta_{T,h}^\star$, $\alpha_h$ absorbs the average level of $y_{t+h}$, $\gamma_h$ controls for lagged outcome persistence, and $\delta_h$ controls for the observed state $c_t$. The number of usable observations at horizon $h$ is $T_h:=T-h$, and all feasible procedures use Bartlett HAC with bandwidth $L_h=\lfloor T_h^{1/3}\rfloor$. The results below use $T=240$. For each environment and horizon, the reference simulation estimates the sample-period estimand $\beta_{T,h}^\star$ together with the design variance $V_{R,h}$, the raw HAC limit $V_{H,h}^{+}$, the regression-adjusted HAC limit $V_{RA,h}^{+}$, and the infeasible oracle variance $V_{O,h}$ obtained after removing the true mean path. The simulation then evaluates $90\%$ normal intervals for $\beta_{T,h}^\star$ under repeated shock re-randomizations. Reported entries are averages across environments, and parenthesized values are standard errors across environments.  The baseline designs use $J=20$ fixed environments, $R=250$ shock re-randomizations per environment, and $M=5000$ reference draws per environment. The misspecified-design appendix keeps $\tau_t=1+c_t$ but replaces $c_t$ in the adjustment by $y_{t-1}$. The appendix also reports a no-drift diagnostic with $\tau_t\equiv 1$, a sample-size comparison, and a comparison of the same LP estimator with several conventional local-projection inference choices. The no-drift, misspecified, and local-projection comparison appendices use the same $J$, $R$, and $M$ counts; the sample-size appendix uses a lighter grid with $J=10$, $R=120$, and $M=2000$.

The variance ratios show that the mean-path term is quantitatively important mainly at short horizons in these designs. Table~\ref{tab:mc_variance_main} reports selected values at horizons $h=0,4,8,12$. In Design A, HAC is notably conservative at short horizons: at $h=0$, $(V_H^{+}/V_R,V_{RA}^{+}/V_R,V_O/V_R)=(6.026,1.633,0.999)$. In Design B the same short-horizon conservativeness is more pronounced, with $(9.939,4.821,0.998)$ at $h=0$. By $h=4$ the corresponding triples are $(1.136,1.028,1.011)$ in Design A and $(1.170,1.082,1.006)$ in Design B, and by $h=8$ the gap is considerably smaller. Figure~\ref{fig:mc_variance_paths} shows the full horizon path. In these AR designs the quantitative relevance of the mean-path term is therefore concentrated at short horizons.

\begin{table}[t!]
  \centering
  \begin{tabular}{llccc}
\toprule
Design & $h$ & $V_H^+/V_R$ & $V_{RA}^+/V_R$ & $V_O/V_R$ \\
\midrule
A & 0 & 6.026 (0.334) & 1.633 (0.131) & 0.999 (0.004) \\
 & 4 & 1.136 (0.014) & 1.028 (0.008) & 1.011 (0.008) \\
 & 8 & 1.028 (0.011) & 1.005 (0.010) & 1.001 (0.010) \\
 & 12 & 1.000 (0.007) & 0.994 (0.006) & 0.993 (0.006) \\
\midrule
B & 0 & 9.939 (1.206) & 4.821 (0.446) & 0.998 (0.004) \\
 & 4 & 1.170 (0.016) & 1.082 (0.011) & 1.006 (0.006) \\
 & 8 & 1.035 (0.008) & 1.017 (0.008) & 1.001 (0.007) \\
 & 12 & 1.009 (0.006) & 1.004 (0.006) & 1.000 (0.006) \\
\bottomrule
\end{tabular}

  \caption{Selected variance ratios in the main Monte Carlo designs. Entries are averages across fixed environments; parenthesized values are standard errors across environments.}
  \label{tab:mc_variance_main}
\end{table}

The interval diagnostics show that these variance-ratio differences translate into shorter intervals when adjustment is aligned with the drift. Table~\ref{tab:mc_interval_main} reports the corresponding coverage and length results. In Design A at $h=0$, HAC coverage is $0.999$ with mean length $0.498$, whereas RA delivers coverage $0.938$ and length $0.258$, close to the oracle values $0.893$ and $0.210$. In Design B at $h=0$, HAC gives $(0.998,0.621)$ for (coverage, length), RA gives $(0.994,0.439)$, and the oracle gives $(0.898,0.218)$. By $h=4$ in Design B the remaining gap is considerably smaller: coverage is $(0.870,0.858,0.866)$ and mean length is $(0.770,0.734,0.744)$ for HAC, RA, and the oracle. The residual deviation from nominal coverage at medium horizons is therefore shared by RA and the oracle and is consistent with finite-sample approximation error common to both procedures. Figure~\ref{fig:mc_interval_paths} shows that the main differences operate through short-horizon interval length rather than through instability at longer horizons.

\begin{table}[t!]
  \centering
  {\small
\renewcommand{\arraystretch}{1.05}
\setlength{\tabcolsep}{3pt}
\begin{tabular}{@{}llcccccc@{}}
\toprule
Design & $h$ & \shortstack{HAC\\cov.} & \shortstack{RA\\cov.} & \shortstack{Oracle\\cov.} & \shortstack{HAC\\len.} & \shortstack{RA\\len.} & \shortstack{Oracle\\len.} \\
\midrule
A & 0 & 0.999 (0.000) & 0.938 (0.009) & 0.893 (0.006) & 0.498 (0.014) & 0.258 (0.011) & 0.210 (0.002) \\
 & 4 & 0.875 (0.006) & 0.863 (0.006) & 0.874 (0.006) & 0.699 (0.016) & 0.664 (0.016) & 0.681 (0.016) \\
 & 8 & 0.866 (0.006) & 0.860 (0.006) & 0.873 (0.006) & 0.763 (0.020) & 0.749 (0.020) & 0.773 (0.020) \\
 & 12 & 0.866 (0.006) & 0.859 (0.006) & 0.875 (0.005) & 0.778 (0.021) & 0.767 (0.021) & 0.792 (0.022) \\
\midrule
B & 0 & 0.998 (0.001) & 0.994 (0.002) & 0.898 (0.005) & 0.621 (0.043) & 0.439 (0.022) & 0.218 (0.003) \\
 & 4 & 0.870 (0.004) & 0.858 (0.005) & 0.866 (0.005) & 0.770 (0.037) & 0.734 (0.033) & 0.744 (0.033) \\
 & 8 & 0.859 (0.004) & 0.853 (0.004) & 0.865 (0.004) & 0.828 (0.038) & 0.814 (0.037) & 0.838 (0.038) \\
 & 12 & 0.869 (0.005) & 0.864 (0.005) & 0.879 (0.005) & 0.844 (0.039) & 0.834 (0.038) & 0.859 (0.039) \\
\bottomrule
\end{tabular}
}%

  \caption{Selected $90\%$ interval diagnostics in the main Monte Carlo designs. Coverage is for $\beta_{T,h}^\star$. Entries are averages across fixed environments; parenthesized values are standard errors across environments.}
  \label{tab:mc_interval_main}
\end{table}

\begin{figure}[t!]
  \centering
  \includegraphics[width=0.98\linewidth]{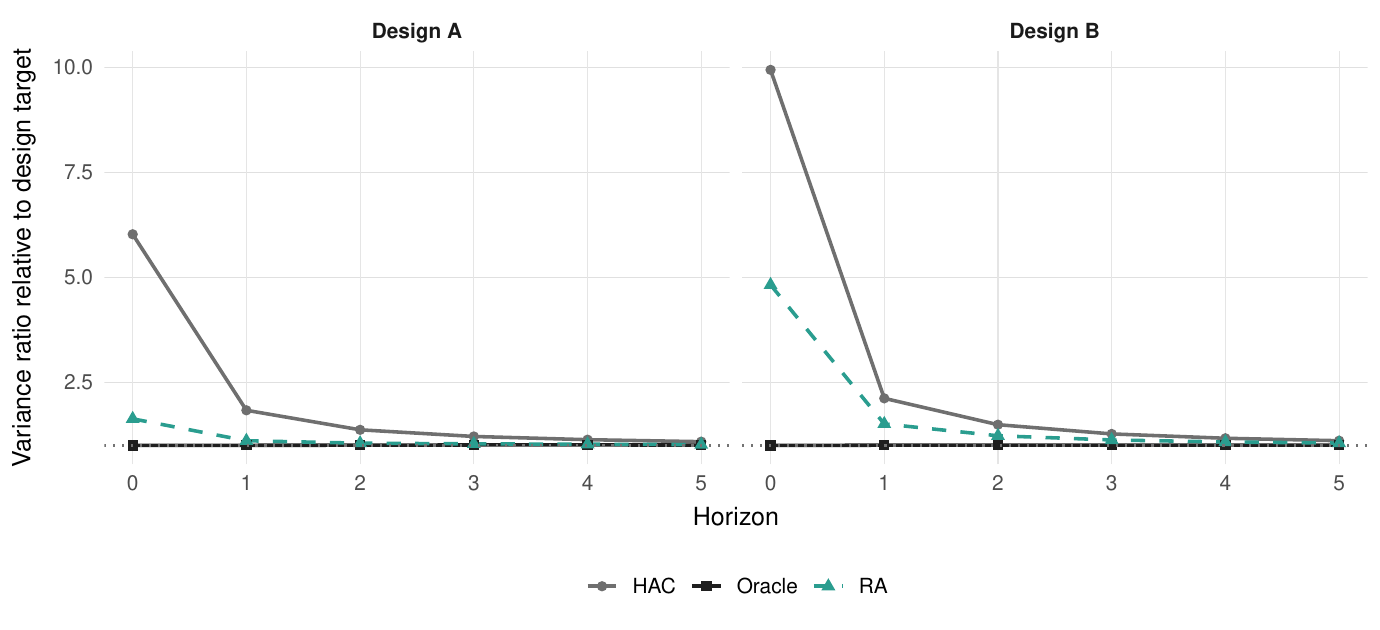}
  \caption{Variance-ratio paths for Designs A and B. The plotted objects are $V_H^{+}/V_R$, $V_{RA}^{+}/V_R$, and $V_O/V_R$ by horizon. The dashed line (on the $x$-axis) represents a ratio of one.}
  \label{fig:mc_variance_paths}
\end{figure}

\begin{figure}[t!]
  \centering
  \begin{subfigure}{0.92\linewidth}
    \centering
    \includegraphics[width=\linewidth,height=0.34\textheight,keepaspectratio]{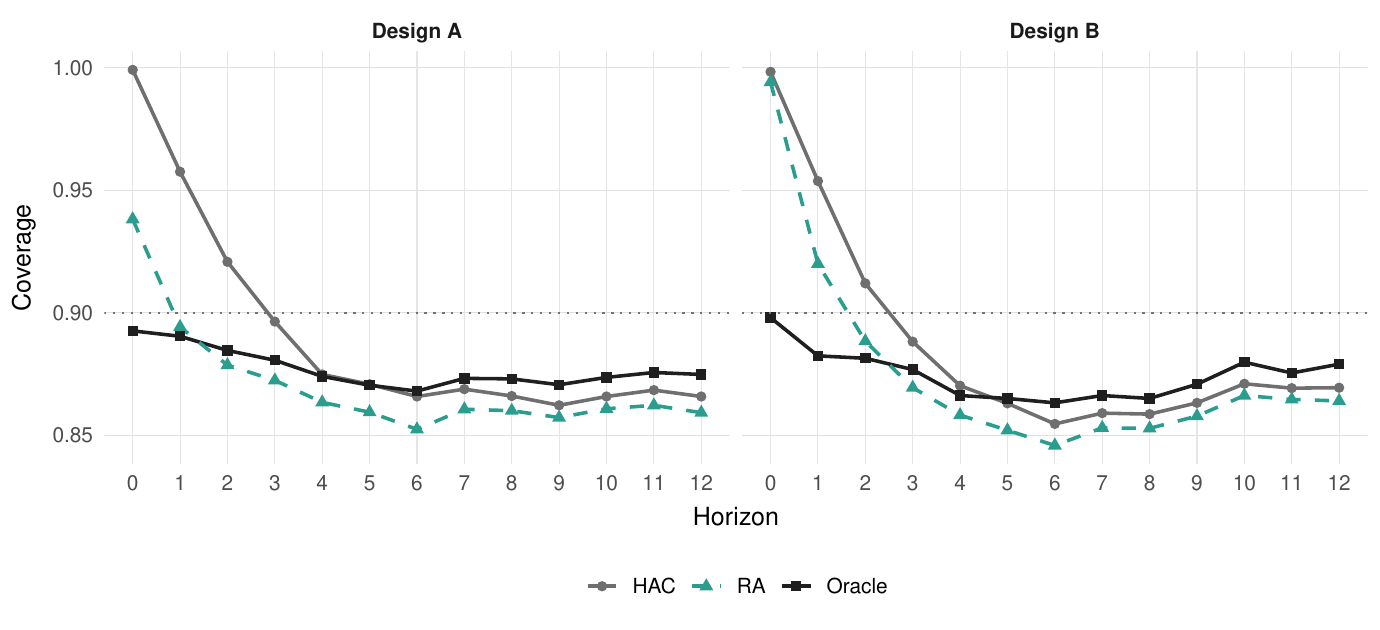}
    \caption{Coverage by horizon.}
  \end{subfigure}

  \vspace{0.6em}
  \begin{subfigure}{0.92\linewidth}
    \centering
    \includegraphics[width=\linewidth,height=0.34\textheight,keepaspectratio]{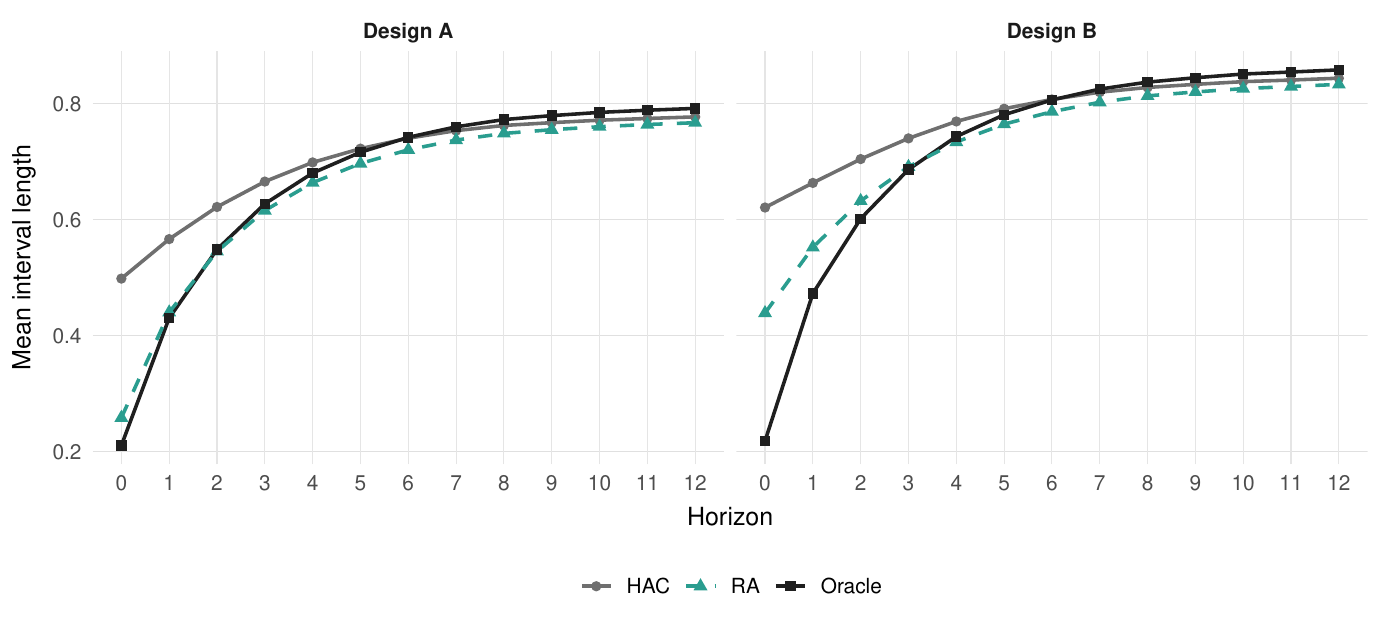}
    \caption{Mean interval length by horizon.}
  \end{subfigure}
  \caption{Interval diagnostics for Designs A and B. Coverage is for $\beta_{T,h}^\star$. The dashed line in the coverage panel marks nominal $0.90$.}
  \label{fig:mc_interval_paths}
\end{figure}

The simulations illustrate two conclusions under these designs. First, conventional HAC can be materially conservative for finite-history estimands, especially at short horizons. Second, predetermined projection can remove a substantial part of that conservativeness when the adjustment variable tracks the mean path, but the improvement is limited when it does not. These designs diagnose the variance decomposition and show how the variance limits move when the adjustment span is aligned or misaligned with the mean path. The appendix comparisons with misspecified adjustments and with common local-projection inference choices make the same distinction visible.

The deliberately adverse state-feedback design in Appendix~\ref{app:mc-stress} reinforces the same qualification. There the adjustment variable is affected by the shocks that are re-randomized, and the adjusted variance ratio falls to $0.033$ with coverage of $0.238$, so the failure is not a small finite-sample deterioration. The example emphasizes that timing alone is not enough for the conservative adjusted interpretation: adjustment variables must either be fixed by the design, or satisfy the long-run orthogonality conditions needed for Proposition~\ref{prop:HAC_RA_conservative}.

\section{Empirical application}\label{sec:empirical}

The empirical application asks whether the same distinction that appears in the simulations matters in a familiar macroeconomic setting. It studies monthly U.S. monetary-policy shocks using the updated Jaroci\'nski--Karadi shock series and public monthly macro data, with a simple diagnostic logic: in local projections, the shock is the object whose assignment variation identifies the response, but the moment used for inference can still vary predictably over time because the macroeconomic environment changes. If information available before the shock captures this predictable variation, then residualizing the HAC moment on that information should matter. The application asks how much of the HAC moment variation in standard monetary-shock regressions can be explained by economically interpretable pre-shock macro variables, and it treats the resulting standard-error reductions primarily as a diagnostic for predictable mean-moment variation unless the stronger conservative-adjustment conditions are maintained.

The monthly shock is the updated \textit{MP\_median} monetary-policy shock from \citet{JarocinskiKaradiUpdate2024}, and the macro block is built from public monthly U.S. series in the spirit of \citet{McCrackenNg2016}. The working sample is the pre-COVID overlap 1990:02--2019:08. The headline figures focus on the Moody's BAA--AAA corporate bond spread and log CPI; the main summary table also reports the 1-year Treasury rate, and the appendix adds log industrial production. For each horizon $h=0,\ldots,18$, the baseline linear local projection is
\[
y_{t+h}=\alpha_h+\beta_h x_t+\gamma_h^\top C_{t-1}+u_{t,h}.
\]
Here $y_{t+h}$ is the macro outcome at horizon $h$, $x_t$ is the monetary-policy shock, and $C_{t-1}$ is information dated before the shock. The coefficient $\beta_h$ is the horizon-$h$ response to the shock after controlling for lagged macro conditions. The application also estimates a state-dependent specification,
\[
y_{t+h}=\alpha_h+\beta_{L,h}(1-S_{t-1})x_t+\beta_{H,h}S_{t-1}x_t+\gamma_h^\top C_{t-1}+u_{t,h}.
\]
The indicator $S_{t-1}$ equals one when the unemployment rate is at least 6.5 percent. Thus $\beta_{H,h}$ is the response in slack periods, and $\beta_{L,h}$ is the response outside slack periods. This specification asks whether the response to monetary-policy shocks differs across macroeconomic states.

The baseline control vector $C_{t-1}$ contains six monthly lags of the dependent variable and of each of log industrial production, log CPI, the 1-year Treasury rate, the BAA--AAA spread, and the monetary-policy shock. These controls define the familiar local-projection specification. They are distinct from the regression-adjustment covariates used to study the HAC moment. The shock is scaled as in the published series; all reported responses and standard errors therefore use the units of that series. A one-unit response is a response to a one-unit movement in the published \textit{MP\_median} shock series, and the replication files record the raw-data scaling used in the analysis.

The regression-adjustment covariates are organized in three nested predetermined sets. The first, parsimonious set contains low-order time trends and, in state-dependent regressions, the state indicator. This set asks whether simple deterministic features of the sample explain part of the HAC moment. The second, macro-lag set adds the same lag block that appears in the local projection. This set asks whether the standard pre-shock controls already contain information about predictable moment variation. The third, rich macro set adds additional lags of equity prices, labor-market slack, the short rate, and commodity prices, implemented with the S\&P 500, the unemployment rate, the federal funds rate, and a broad commodity-price series. This set asks whether broader macro-financial information available before the shock further explains the moment variation. All regression-adjustment variables are dated before the monetary-policy shock whose moment is adjusted. That timing makes them predetermined for implementation, but because lagged macro variables may themselves reflect earlier shocks, the conservative variance interpretation rests on the additional orthogonality logic in Section~\ref{sec:implementation} rather than on timing alone. The covariate sets are fixed before inspecting horizon-by-horizon significance and are not selected separately by outcome or horizon. The point is therefore not to search for the most favorable adjustment variable, but to compare economically interpretable information sets of increasing richness.

The published high-frequency series is treated as the measured assignment shock used in the design, not as an estimated first-stage object. For the reported calculations, the conditioning environment fixes the sample window, the local-projection and residualization specifications, the observed lagged-information matrices used in $C_{t-1}$, the adjustment-covariate paths, and the measurement convention for $x_t$. The re-randomized object is the current measured monetary surprise entering the moment, with HAC providing a feasible short-memory long-run-variance approximation for the resulting moment sequence. Under this convention, the calculations do not add uncertainty from constructing the high-frequency shock series. A design that conditions on the realized shock path itself would leave no shock-assignment uncertainty to approximate, while a design that treats the high-frequency construction as random would require replacing the local-projection moment by the augmented moment vector in Appendix~\ref{app:two-step}. The empirical calculations below are therefore conditional on the measurement convention for the shock series and do not require monetary surprises to be literal randomized policy experiments.

The main pattern is that richer predetermined macro covariates produce larger standard-error reductions. Table~\ref{tab:money_main_summary} reports the corresponding empirical results. In the linear projections, the average HAC standard-error reduction for log CPI rises from 3.71\% under the parsimonious set to 16.12\% under the rich macro covariate set; for the BAA--AAA spread it rises from 1.57\% to 22.02\%; and for the 1-year Treasury rate it rises from 8.95\% to 32.27\%. In the slack-state CPI specification, the high-slack coefficient has the largest reduction, with a 19.96\% average decline in the standard error. These reductions are consistent with the Loewner-reduction result in Proposition~\ref{prop:HAC_RA}: larger sets of predetermined covariates reduce the reported HAC variance limit when they track predictable components of the mean path. Under the stronger conditions of Proposition~\ref{prop:HAC_RA_conservative}, the same reductions can be interpreted as tightening a conservative bound for the design variance; absent those conditions, the table is a diagnostic for movement in the date-specific mean path. It identifies neither $\Omega_\mu$ nor the exact design variance; it shows that the chosen covariate span explains part of the raw HAC moment variation.

Figure~\ref{fig:money_linear_main} reports the linear projections using the rich macro covariate set for the two main outcomes. The point estimates are essentially unchanged, while the reported bands based on the adjusted HAC limit are narrower. For the spread, at $h=5$ the estimated response is 0.015 with a HAC standard error of 0.011 and an RA standard error of 0.008. For log CPI, the largest differences occur at longer horizons; at $h=15$ and $h=18$, the adjusted-HAC pointwise interval excludes zero while the baseline HAC interval does not. These pointwise differences are descriptive evidence about how residualization changes the reported variance limit, and their interpretation as tighter valid design intervals requires the stronger conditions discussed above. Fixed-horizon simultaneous statements can use either Bonferroni or \v{S}id{\'a}k adjustments under the fixed-$J$ joint Gaussian limit in Appendix~\ref{app:simultaneous}; the appendix reports Bonferroni as the conservative reference.

Figure~\ref{fig:money_state_cpi_main} shows the state-dependent CPI responses under the slack split. The rich macro covariate set has limited effect at short horizons and a larger effect at medium and longer horizons, especially in the high-slack state. At $h=13$ the high-slack coefficient is $-0.191$ with HAC and RA standard errors 0.099 and 0.085, and at $h=18$ it is $-0.206$ with corresponding standard errors 0.122 and 0.100. In this application, regression adjustment leaves the response profile largely unchanged and narrows the reported bands where the predictable moment component appears quantitatively important, subject to the same conservative-interpretation qualification.

The appendix reports the remaining monetary results: Table~\ref{tab:app_money_linear_summary} adds log industrial production to the linear results, Table~\ref{tab:app_money_state_summary} reports the slack and zero-lower-bound (ZLB) specifications for both CPI and the spread, and Table~\ref{tab:app_money_varx_summary} reports a vector autoregression with exogenous controls (VARX) impulse-response analysis on the same data. In the VARX results, the rich macro covariate set reduces average standard errors by 25.71\% for cumulative inflation responses, 19.73\% for BAA--AAA spread responses, 15.77\% for 1-year Treasury-rate responses, and 14.86\% for cumulative industrial-production responses. The appendix also reports bandwidth and covariate-set sensitivity checks. The reductions remain positive across conventional HAC bandwidth choices and are driven primarily by larger sets of predetermined macro covariates rather than by demeaning or low-order trends alone. The appendix reports selected descriptive horizons at which pointwise significance changes when regression-adjusted HAC standard errors replace baseline HAC standard errors and shows that most do not survive fixed-horizon simultaneous adjustment. The empirical conclusion concerns predictable movement in date-specific mean moments and conservative variance limits rather than uniform rejection across response horizons.

\begin{table}[t!]
  \centering
  \small
  \setlength{\tabcolsep}{4.5pt}
  \renewcommand{\arraystretch}{1.08}
  \begin{minipage}{\linewidth}
    \centering
    \textbf{Panel A. Linear Local Projections}\\[0.35em]
    \begin{tabular}{llcccc}
      \toprule
      Outcome & $Z_t$ set & Mean \% red. & Median \% red. & Share $>0$ & HAC$\to$RA flips\\
      \midrule
      log CPI & Parsimonious & 3.71 & 3.27 & 1.00 & 0\\
       & Macro lags & 8.70 & 8.75 & 0.89 & 1\\
       & Rich macro & 16.12 & 16.80 & 0.95 & 4\\
      \addlinespace[0.25em]
      BAA--AAA spread & Parsimonious & 1.57 & 1.75 & 1.00 & 0\\
       & Macro lags & 13.51 & 12.82 & 1.00 & 0\\
       & Rich macro & 22.02 & 21.74 & 1.00 & 1\\
      \addlinespace[0.25em]
      1-year Treasury rate & Parsimonious & 8.95 & 8.71 & 1.00 & 0\\
       & Macro lags & 22.91 & 20.97 & 1.00 & 0\\
       & Rich macro & 32.27 & 31.58 & 1.00 & 1\\
      \bottomrule
    \end{tabular}
  \end{minipage}

  \vspace{0.55em}
  \begin{minipage}{\linewidth}
    \centering
    \textbf{Panel B. State-Dependent CPI Local Projection (Slack Split)}\\[0.35em]
    \begin{tabular}{llcccc}
      \toprule
      Coefficient & $Z_t$ set & Mean \% red. & Median \% red. & Share $>0$ & HAC$\to$RA flips\\
      \midrule
      Low slack & Parsimonious & 3.74 & 1.94 & 1.00 & 0\\
       & Macro lags & 6.02 & 8.30 & 0.74 & 0\\
       & Rich macro & 12.56 & 13.15 & 0.89 & 1\\
      \addlinespace[0.25em]
      High slack & Parsimonious & 6.04 & 5.21 & 1.00 & 4\\
       & Macro lags & 22.33 & 21.32 & 1.00 & 6\\
       & Rich macro & 19.96 & 17.89 & 1.00 & 5\\
      \bottomrule
    \end{tabular}
  \end{minipage}
  \caption{Empirical tightening in the monthly monetary application. Panel A reports linear local projections for the main outcomes. Panel B reports the slack-state CPI specification used in Figure~\ref{fig:money_state_cpi_main}. Entries are percentage reductions in the HAC standard error over horizons $h=0,\ldots,18$; \emph{HAC$\to$RA flips} counts horizons where the coefficient is insignificant under HAC but significant under regression-adjusted HAC.}
  \label{tab:money_main_summary}
\end{table}

\clearpage
\begin{figure}[p]
  \centering
  \captionsetup{font=footnotesize,skip=2pt}
  \includegraphics[width=0.88\linewidth]{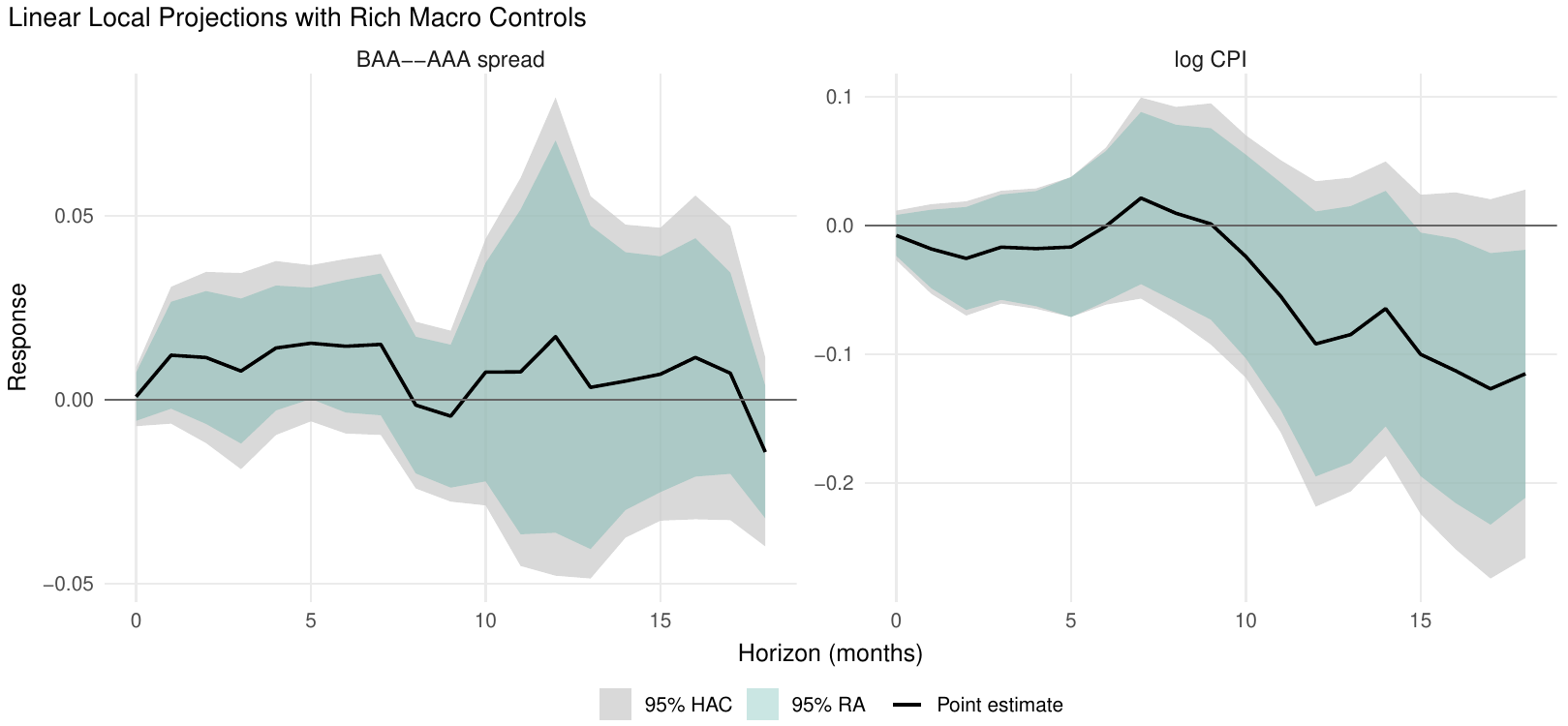}
  \caption{Linear local projections for the headline outcomes using the rich macro covariate set. Each panel reports the point estimate together with 95\% HAC bands and corresponding regression-adjusted HAC bands; the latter are interpreted with the qualifications in Section~\ref{sec:implementation}.}
  \label{fig:money_linear_main}
\end{figure}

\begin{figure}[p]
  \centering
  \captionsetup{font=footnotesize,skip=2pt}
  \includegraphics[width=0.88\linewidth]{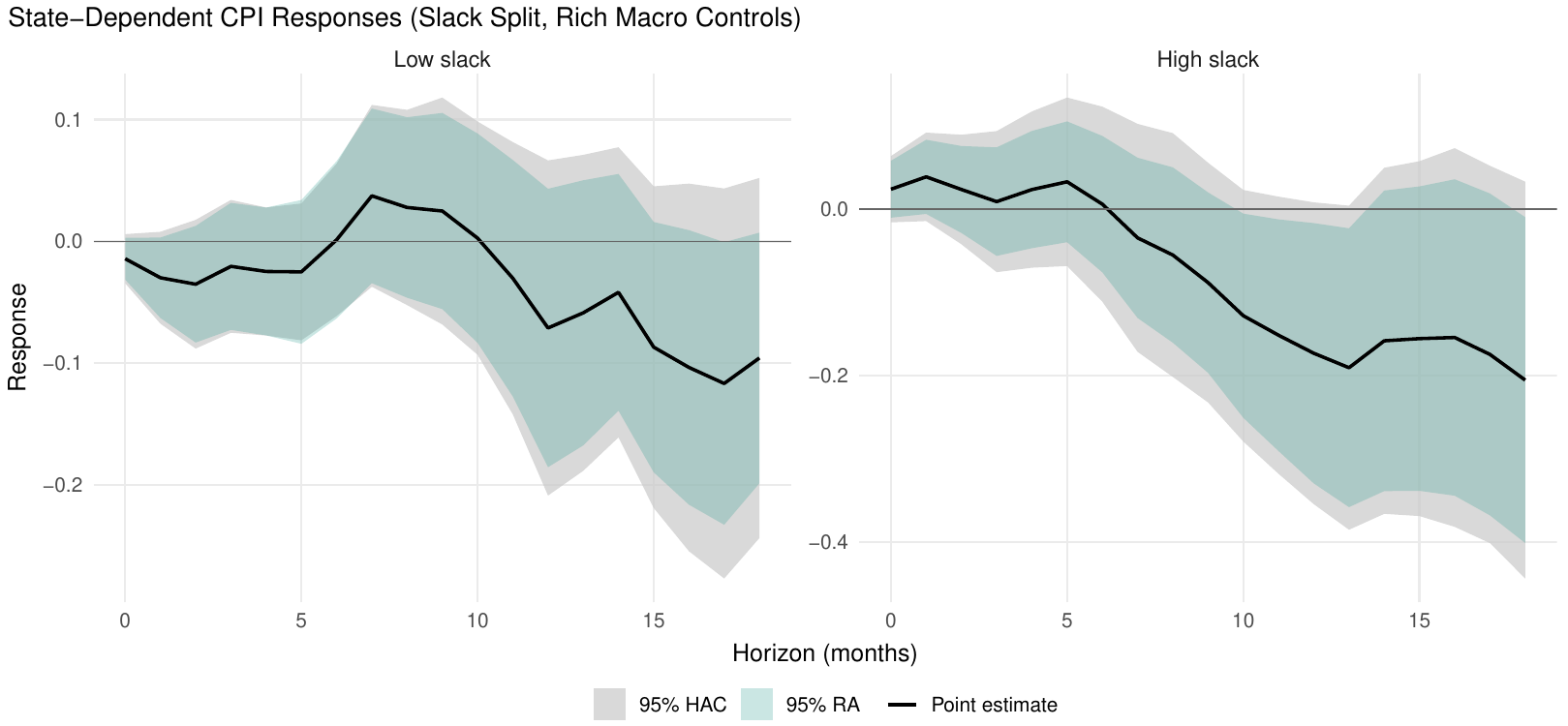}
  \caption{State-dependent CPI responses under the slack split with the rich macro covariate set. Each panel reports the point estimate together with 95\% HAC bands and corresponding regression-adjusted HAC bands; the latter are interpreted with the qualifications in Section~\ref{sec:implementation}.}
  \label{fig:money_state_cpi_main}
\end{figure}

\FloatBarrier
\section{Conclusion} \label{sec:conclusion}
This paper develops a design-based interpretation of time-series GMM inference. The central idea is that macroeconomic time series are often used to study what happened in one realized economy, not to average over an imagined population of economies. Conditional on the realized economic environment, the finite-history estimand is fixed and the relevant uncertainty comes only from the assignment variation in the shocks. While the design long-run variance for that uncertainty is $\Omega_R$, conventional HAC procedures generally estimate a different object, $\Omega_R^+=\Omega_R+\Omega_\mu$, where $\Omega_\mu$ is the long-run variance of the centered date-specific mean path. The decomposition shows exactly what HAC adds: besides innovation uncertainty, it also counts predictable drift in the moments across dates. Under the maintained conditions, this makes HAC conservative for scalar functions of the finite-history estimand rather than invalid.

The empirical takeaway from this paper is that baseline HAC intervals remain a useful conservative benchmark, but they need not be the end of the analysis. When predetermined variables explain predictable movement in the moment, regression-adjusted HAC can lower the conservative HAC limit, and under the long-run orthogonality conditions it can be interpreted as moving the bound closer to the design variance. Without those conditions, the reduction remains a diagnostic for predictable moment variation rather than a standalone validity claim. The simulations show why alignment matters: adjustment works well when the covariates track the mean path, though they naturally work less well when they only imperfectly span it. The monetary-shock application shows the same force in data, with rich macro information delivering material standard-error reductions and little movement in point estimates. More broadly, the paper separates uncertainty about shock assignment within the realized economy from uncertainty over alternative possible economies, and uses that distinction to clarify when familiar HAC tools are conservative and when additional structure can make them sharper. Future work could extend the framework to high-dimensional adjustment, fixed-$b$ inference, generated shocks, and overidentified GMM under local misspecification.

\clearpage
\bibliography{bibli}

\clearpage
\appendix

\clearpage
\section{Supplementary theoretical results}\label{app:supplementary-results}

This appendix collects theoretical results which supplement those in the main text. The main text specializes to the fully observed sample because that case is the transparent macro setting used for the main variance-limit results. Here, the appendix instead records a primitive verification route for the fully observed sample, the modifications under incomplete observation, the stronger projection-adjustment statement available under additional conditions for the fully observed sample, simultaneous inference over a fixed set of horizons, the LP--VAR comparison, the overidentified expansion under local misspecification, and two-step/generated-regressor constructions. Appendix~\ref{app:proofs} then gives the corresponding proofs and auxiliary probability lemmas before the simulation, empirical, and macro-estimator appendices.

\subsection{Notation}
The appendix uses the following notation repeatedly. Under incomplete observation, $R_{T,t}\in\{0,1\}$ is the inclusion indicator, $N:=\sum_{t=1}^T R_{T,t}$ is the realized sample size, $\rho_T:=T^{-1}\sum_{t=1}^T \mathbb{E}[R_{T,t}]$ is the sampling rate, and $\kappa_R(\ell)$ is the lag-$\ell$ pair-frequency limit. Whenever $R_{T,t}$ appears in an appendix result, the sample moment and sample Jacobian are $g_N(\theta):=N^{-1}\sum_{t=1}^T R_{T,t}g_t(W_t,\theta)$ and $\widehat G_N(\theta):=N^{-1}\sum_{t=1}^T R_{T,t}\nabla_\theta g_t(W_t,\theta)$ on the event $N>0$. As in the main text, $\mu_{T,t}(\theta_T^\star):=\mathbb{E}_T[g_t(W_t,\theta_T^\star)]$ is the date-specific design mean moment and $\tilde\mu_{T,t}:=\mu_{T,t}(\theta_T^\star)-\bar m_T(\theta_T^\star)$ is its centered version; conditional-on-$\mathcal{I}_{T,t}$ formulas are used only as sufficient expressions for this fixed path in examples. GMM expansions use $H:=(G^\top A G)^{-1}G^\top A$ for the GMM sandwich linear map, and the local VAR comparison uses $H_\phi:=(G_\phi^\top A_\phi G_\phi)^{-1}G_\phi^\top A_\phi$. When the sample is fully observed, these objects collapse to $R_{T,t}\equiv 1$, $N=T$, $\rho_T=1$, and $\kappa_R(\ell)=1$.

The next subsection gives one sufficient route for the moment-level and HAC assumptions used in the main text. It is stated for the fully observed sample because that is the case used in the main theorem statements. The incomplete-observation results below impose the corresponding conditions based on observed lag pairs directly.

\subsection{Primitive verification of HAC conditions}\label{app:primitive-verification}

The main text separates the econometric argument from the model-specific weak-dependence verification. The next condition gives one route for the fully observed sample that is strong enough for the canonical fixed-dimensional LP and VAR examples. It is sufficient rather than minimal because, in particular, the HAC part imposes moment and covariance restrictions on moment products, not only on the moment itself.

Let $J_{T,t}(\theta):=\nabla_\theta g_t(W_t,\theta)$ and $s_{T,t}:=g_t(W_t,\theta_T^\star)-\bar m_T(\theta_T^\star)$. For $\ell\ge0$, set $\bar\Gamma_{s,T}(\ell):=T^{-1}\sum_{t=\ell+1}^T\mathbb E_T[s_{T,t}s_{T,t-\ell}^\top]$, with the transpose convention for negative lags. The following assumption concerns the fully observed sample only, so $R_{T,t}\equiv1$ and $N=T$.

\begin{assumption}\label{ass:primitive-app}
For some $\delta>0$ and for the kernel and bandwidth sequence used in Assumption~\ref{ass:hac-kernel}, the following conditions hold.
\begin{enumerate}[label=(\roman*)]
\item The centered innovation array $(e_{T,t})$ is uniformly $\alpha$-mixing with coefficients $\alpha(h)$ satisfying $\sum_{h\ge1}\alpha(h)^{\delta/(2+\delta)}<\infty$, satisfies $\sup_{T,t}\mathbb E_T\|e_{T,t}\|^{2+\delta}<\infty$, and has fixed-lag covariance limits. For $\ell\in\mathbb Z$, the bound $\|T^{-1}\sum_{t=|\ell|+1}^T\mathbb E_T[e_{T,t}e_{T,t-|\ell|}^\top]\|\le b_e(\ell)$ holds eventually in $T$, where $\sum_{\ell\in\mathbb Z}b_e(\ell)<\infty$ and the transpose convention is used for negative lags.
\item For each scalar coordinate $f_{T,t}(\theta)$ of $g_t(W_t,\theta)-\mathbb E_T[g_t(W_t,\theta)]$ and of $J_{T,t}(\theta)-\mathbb E_T[J_{T,t}(\theta)]$,
\[
\sup_{\theta\in\Theta}\frac1T\sum_{h=-(T-1)}^{T-1}\sum_{t=|h|+1}^{T}
\left|\Cov_T(f_{T,t}(\theta),f_{T,t-|h|}(\theta))\right|=O(1),
\]
where $\Cov_T$ is conditional covariance under the fixed environment. The sample averages $T^{-1}\sum_t b_{g,T,t}$ and $T^{-1}\sum_t b_{g,1,T,t}$ are $O_p(1)$.
\item The fixed-lag limits $\Gamma_{g,\mathrm{hac}}(\ell)=\lim_T\bar\Gamma_{s,T}(\ell)$ exist and satisfy $\sum_{\ell\in\mathbb Z}\|\Gamma_{g,\mathrm{hac}}(\ell)\|<\infty$. For each coordinate pair $(i,j)$ and each $0\le\ell\le L_T$, define $Y^{ij}_{T,t}(\ell):=s_{T,t,i}s_{T,t-\ell,j}-\mathbb E_T[s_{T,t,i}s_{T,t-\ell,j}]$. Then
\[
\sup_{0\le\ell\le L_T}\sup_{i,j}\frac1T\sum_{h=-(T-1)}^{T-1}\sum_{t=|h|+1}^{T}
\left|\Cov_T(Y^{ij}_{T,t}(\ell),Y^{ij}_{T,t-|h|}(\ell))\right|=O(1),
\]
and the deterministic lag-window bias satisfies $\|\sum_{1\le |\ell|\le L_T}K(|\ell|/L_T)(\bar\Gamma_{s,T}(\ell)-\Gamma_{g,\mathrm{hac}}(\ell))\|\to0$.
\item With $\bar B_{T,t}:=B_{g,T,t}+T^{-1}\sum_{r=1}^T\mathbb E_T B_{g,T,r}$, the plug-in product-envelope averages satisfy
\[
\sup_{0\le\ell\le L_T}\frac1T\sum_{t=\ell+1}^{T}
\big(b_{g,T,t}\bar B_{T,t-\ell}+\bar B_{T,t}b_{g,T,t-\ell}\big)=O_p(1),
\]
with the same boundary convention as in the HAC sums.
\end{enumerate}
\end{assumption}

Assumption~\ref{ass:primitive-app} is appendix-only: the first condition gives a triangular-array route to the centered-moment CLT, the second condition is the covariance-sum requirement that makes the moment function and Jacobian uniform laws referee-checkable, and the third condition is the product-array requirement behind growing-bandwidth HAC consistency. The fourth condition is used only when replacing the infeasible moment at $\theta_T^\star$ by the feasible moment at $\widehat\theta_N$ over the $O(L_T)$ retained lags. All covariance-sum conditions are componentwise; because the moment and parameter dimensions are fixed, they imply the corresponding norm bounds used in the proofs.

The mixing condition here is a triangular-array condition conditional on $\mathcal E_T$. It permits time-varying coefficients and fixed deterministic paths; what must have short-memory behavior is the centered innovation array and, for HAC, the centered product array after the fixed mean path has been separated.

\begin{proposition}\label{prop:primitive-verification}
Suppose the sample is fully observed and fix a kernel and bandwidth satisfying Assumption~\ref{ass:hac-kernel}. If Assumptions~\ref{ass:design-environment}, \ref{ass:smoothness}, \ref{ass:local-correct-specification}, \ref{ass:gmm}, \ref{ass:estimand-separation}, \ref{ass:gmm-interior}, and~\ref{ass:primitive-app} hold, then Assumptions~\ref{ass:dependence}, \ref{ass:uniform-lln}, and~\ref{ass:hac-regularity} hold for that kernel and bandwidth. Consequently Theorem~\ref{thm:AN} applies; if Assumption~\ref{ass:mean-path} also holds, then Proposition~\ref{lem:fixed-env-decomp}, Theorem~\ref{thm:hac}, and Corollary~\ref{cor:scalar-conservative} apply for the same kernel and bandwidth.
\end{proposition}

\subsection{Incomplete observation }\label{app:incomplete-sampling}

This subsection records the changes needed when the econometrician observes only selected dates. The sample-period estimand remains $\theta_T^\star$; incomplete observation changes the sample average and the covariance constants, not the estimand.

Write $R_{T,t}\in\{0,1\}$ for inclusion, $N:=\sum_{t=1}^T R_{T,t}$, and $\rho_T:=T^{-1}\sum_{t=1}^T\mathbb E[R_{T,t}]\to\rho\in(0,1]$. For a fixed lag $\ell\in\mathbb Z$, set $D_T(\ell):=\sum_{t=|\ell|+1}^T R_{T,t}R_{T,t-|\ell|}$ and $\widehat\kappa_R(\ell):=D_T(\ell)/N$, with the convention $D_T(-\ell)=D_T(\ell)$. The pair-frequency limit is $\kappa_R(\ell):=\lim_{T\to\infty}(T\rho_T)^{-1}\sum_{t=|\ell|+1}^T\mathbb E[R_{T,t}R_{T,t-|\ell|}]$ whenever the limit exists; under a deterministic sampling window the expectation is the fixed product $R_{T,t}R_{T,t-|\ell|}$. Thus $\kappa_R(0)=1$ whenever $N/(T\rho_T)\to1$.

\begin{assumption}\label{ass:sampling-app}
The sample is generated by one of the following two sampling schemes.
\begin{enumerate}[label=(\roman*)]
\item In the random-sampling scheme, for each $T$, $(R_{T,t})_{t\in\mathbb Z}$ is strictly stationary, independent of the moment process indexed by the shock history conditional on $\mathcal E_T$, has mean $\rho_T$, and is $\alpha$-mixing with coefficients $\alpha_R(h)$ satisfying $\sum_{h\ge1}\alpha_R(h)^{\delta/(2+\delta)}<\infty$ for the same $\delta$ as in Assumption~\ref{ass:dependence}. Its autocovariances $\gamma_{R,T}(\ell):=\Cov(R_{T,t},R_{T,t-|\ell|})$ have fixed-lag limits $\gamma_R(\ell)$ and satisfy $|\gamma_{R,T}(\ell)|\le b_R(\ell)$ with $\sum_{\ell\in\mathbb Z}b_R(\ell)<\infty$. The realized sample size obeys $N/(T\rho_T)\to_p1$.
\item In the deterministic-window scheme, $R_{T,t}$ is nonrandom conditional on $\mathcal E_T$. Also $\rho_T=T^{-1}\sum_{t=1}^T R_{T,t}\to\rho$, and, for every fixed $\ell$, $D_T(\ell)/(T\rho_T)\to\kappa_R(\ell)$. For the moment function and Jacobian classes used in Lemmas~\ref{lem:ULLN} and~\ref{lem:ULLN-J},
\[
\sup_{\theta\in\Theta}\left\|(T\rho_T)^{-1}\sum_{t=1}^T(R_{T,t}-\rho_T)\mathbb E_T[g_t(W_t,\theta)]\right\|\to0
\]
and the same condition holds with $g_t(W_t,\theta)$ replaced by $\nabla_\theta g_t(W_t,\theta)$.
\end{enumerate}
For both schemes, each fixed-lag limit $\kappa_R(\ell)$ exists and $\sup_\ell\kappa_R(\ell)<\infty$. Expressions involving $g_N$ are evaluated on $N>0$; this event has probability approaching one under the assumption.
\end{assumption}

Assumption~\ref{ass:sampling-app} is appendix-only and separates the two sources of regularity needed under incomplete observation. Random sampling requires independence from the design shocks and weak dependence of the inclusion process, whereas deterministic sampling windows require explicit balance of the selected finite-history moments. The balance condition keeps the estimand equal to the sample-period estimand rather than replacing it by a selected-window estimand.

In the random-sampling scheme, the probability space is enlarged by an observation mechanism. Assumption~\ref{ass:design-environment} still describes the generation of the moment map conditional on $\mathcal E_T$: assignment shocks are the source of design randomness in the moment, while $R_{T,t}$ is additional sampling or observation randomness and is assumed independent of that map in the random-sampling scheme.

For HAC calculations based only on observed lag pairs, define, for $\ell\ge0$,
\[
a_{T,t}(\ell)
:=\mathbb E_T\!\left[
g_t(W_t,\theta_T^\star)g_{t-\ell}(W_{t-\ell},\theta_T^\star)^\top
\right].
\]
\begin{assumption}\label{ass:pair-average-app}
For each fixed $\ell\ge0$, the Ces\`aro limit $\Gamma_{\mathrm{hac}}(\ell):=\lim_{T\to\infty} T^{-1}\sum_{t=\ell+1}^T a_{T,t}(\ell)$ exists. If $\kappa_R(\ell)>0$, then
\[
D_T(\ell)^{-1}\sum_{t=\ell+1}^T R_{T,t}R_{T,t-\ell}a_{T,t}(\ell)\xrightarrow{p}\Gamma_{\mathrm{hac}}(\ell).
\]
For negative lags set $\Gamma_{\mathrm{hac}}(-\ell):=\Gamma_{\mathrm{hac}}(\ell)^\top$. If $\kappa_R(\ell)=0$, no observed-pair average is required at that lag; its weighted HAC contribution is defined to be zero.
\end{assumption}

Assumption~\ref{ass:pair-average-app} says that the observed lag pairs preserve the fixed-history second-moment limit of the full series whenever the lag has positive asymptotic pair frequency. It is automatic under a fully observed sample, is a genuine observed-pair restriction under deterministic sampling windows, and follows under random sampling from the same independence and weak-dependence route used for Lemma~\ref{lem:kappa-rho-HAC} when the lag-product array has a summable covariance envelope.

Under random sampling in Assumption~\ref{ass:sampling-app}, let $\Omega_R^{\mathrm{obs}}$ denote the covariance matrix of the incomplete-observation normalized sample moment. Its variance decomposition is
\begin{equation}\label{eq:OmegaR_sampling}
\Omega_R^{\mathrm{obs}}
=
\Omega_R^{+,\mathrm{obs}}-\rho\Omega_\mu,
\qquad
\Omega_R^{+,\mathrm{obs}}
:=
\Gamma_{g,\mathrm{hac}}(0)+\sum_{\ell\neq0}\kappa_R(\ell)\Gamma_{g,\mathrm{hac}}(\ell),
\qquad
\Omega_\mu:=\sum_{\ell\in\mathbb Z}\Gamma_{g,\mu}(\ell).
\end{equation}
Here $\Gamma_{g,\mathrm{hac}}(\ell)$ is the fixed-lag long-run covariance of the raw centered moment defined in \eqref{eq:Gamma-ghac}, and $\Gamma_{g,\mu}(\ell)$ is the corresponding lag limit of the centered mean path. The factor $\rho$ in front of $\Omega_\mu$ is specific to the independent random-sampling calculation in Lemma~\ref{lem:CLT-mom}. In a deterministic window, the centered innovation covariance is instead $\sum_{\ell\in\mathbb Z}\kappa_R(\ell)\Gamma_e(\ell)$ and a nonzero mean shift may remain. When $R_{T,t}\equiv1$, $\rho=1$ and $\kappa_R(\ell)=1$, so \eqref{eq:OmegaR_sampling} collapses to \eqref{eq:OmegaR}--\eqref{eq:Omegam}, with $\Omega_R^{\mathrm{obs}}=\Omega_R$ and $\Omega_R^{+,\mathrm{obs}}=\Omega_R^+$.

The deterministic-window GMM expansion needs the selected sample to identify and linearize the same sample-period estimand. Let $M_{N,T}(\theta):=N^{-1}\sum_{t=1}^T R_{T,t}\mathbb E_T[g_t(W_t,\theta)]$ and $G_{N,T}(\theta):=N^{-1}\sum_{t=1}^T R_{T,t}\mathbb E_T[\nabla_\theta g_t(W_t,\theta)]$.
\begin{assumption}\label{ass:det-selected-gmm}
For the deterministic-window scheme, the selected sample satisfies \(\sup_{\theta\in\Theta}\|g_N(\theta)-M_{N,T}(\theta)\|\to_p0\), \(\sup_{\theta\in\Theta}\|M_{N,T}(\theta)-\bar m_T(\theta)\|\to0\), \(\sup_{\theta\in\Theta}\|\widehat G_N(\theta)-G_{N,T}(\theta)\|\to_p0\), and \(\|G_{N,T}(\theta_T^\star)-G\|\to0\).
The estimator considered below is a measurable global minimizer of $g_N(\theta)^\top\widehat A_Ng_N(\theta)$ over $\Theta$.
\end{assumption}

Assumption~\ref{ass:det-selected-gmm} states the selected-window GMM regularity used below. The first two convergence conditions give uniform convergence of the selected sample criterion to the full-span population criterion. The latter two conditions give the Jacobian convergence needed for the GMM linearization. The global-minimizer clause is used for consistency; Assumption~\ref{ass:gmm-interior} supplies the interior local first-order condition. Together, these conditions keep the estimand equal to $\theta_T^\star$.

For deterministic sampling windows, let $m_T:=(T\rho_T)^{-1/2}\sum_{t=1}^T(R_{T,t}-\rho_T)\mu_{T,t}(\theta_T^\star)$, which measures the mean-path imbalance created by observing only the deterministic sampling window. Let $\Omega_B$ denote the covariance matrix in the selected-window centered-moment CLT, and write $H:=(G^\top A G)^{-1}G^\top A$.
\begin{theorem}\label{thm:det-GMM}
Suppose Assumptions~\ref{ass:design-environment}, \ref{ass:smoothness}, \ref{ass:gmm}, \ref{ass:estimand-separation}, \ref{ass:local-correct-specification}, \ref{ass:gmm-interior}, \ref{ass:sampling-app}, and~\ref{ass:det-selected-gmm} hold under the deterministic-window scheme. Suppose also that $m_T\to m$ and \(\sqrt N\big(g_N(\theta_T^\star)-\bar m_T(\theta_T^\star)\big)\Rightarrow\mathcal N(m,\Omega_B)\). Then \(\sqrt N(\widehat\theta_N-\theta_T^\star)\Rightarrow\mathcal N(-Hm,H\Omega_BH^\top)\).
The centered Gaussian limit obtains when $m=0$.
\end{theorem}

\begin{corollary}\label{cor:det-centered}
Under the assumptions of Theorem~\ref{thm:det-GMM}, if $m_T\to0$, then $\sqrt{N}(\widehat\theta_N-\theta_T^\star)\Rightarrow\mathcal N(0,H\Omega_BH^\top)$. The fully observed sample $R_{T,t}\equiv1$ is included because then $m_T\equiv0$.
\end{corollary}

\begin{proposition}\label{prop:window-centering}
Suppose the observed sample is a single contiguous window, so $I_T=\{s_T+1,\dots,s_T+N_T\}$, $R_{T,t}=\1\{t\in I_T\}$, and $N_T/T\to\rho\in(0,1]$. Write $\tilde\mu_{T,t}:=\mu_{T,t}(\theta_T^\star)-\bar m_T(\theta_T^\star)$. If
\[
\sup_{0\le u<v\le1}\left\|T^{-1/2}\sum_{t=\lfloor uT\rfloor+1}^{\lfloor vT\rfloor}\tilde\mu_{T,t}\right\|\to0,
\]
then $m_T\to0$.
\end{proposition}

\subsection{More on Projection adjustment}\label{app:ra-lower-bound}

The main text establishes the raw HAC positive-semidefinite reduction result and states the fully observed sample orthogonality case in which the adjusted variance limit remains conservative for the original design variance. This appendix records the conservative-variance logic in notation useful for fixed-dimensional predetermined adjustment variables and then gives a counterexample showing why incomplete observation requires additional restrictions. Let $z_t\in\mathbb R^m$ be the retained adjustment vector, let $g_t:=g_t(W_t,\theta_T^\star)$, and, for each fixed conformable matrix $B\in\mathbb R^{m\times k}$, write $r_t(B):=g_t-B^\top z_t$ and $d_t(B):=\tilde\mu_{T,t}-B^\top z_t$. All long-run matrices in this subsection are the growing-bandwidth HAC limits defined in Section~\ref{sec:implementation}; in particular, $S_{ab}^+=(S_{ba}^+)^\top$ when both sides exist.

\begin{assumption}\label{ass:RA-orth}
The sample is fully observed, $R_{T,t}\equiv1$, and Assumption~\ref{ass:local-correct-specification} holds. With $e_{T,t}:=g_t-\mathbb E_T[g_t]$, the long-run matrices $S_{ab}^+$ involving $a,b\in\{e,\tilde\mu,z\}$ exist through absolutely summable fixed-lag limits, $S_{zz}^+\succ0$, and $S_{ez}^+=0$. For every fixed conformable matrix $B$, $S_{d(B)d(B)}^+$ exists.
\end{assumption}

Assumption~\ref{ass:RA-orth} is the appendix version of the orthogonality condition for the fully observed sample in Proposition~\ref{prop:HAC_RA_conservative}. The cross term $S_{e\mu}^+$ is not assumed: it follows from date-specific centering of $e_{T,t}$ and the fixed status of the centered mean path. The substantive restriction is $S_{ez}^+=0$; together with Proposition~\ref{lem:fixed-env-decomp}, when the sample is fully observed, $\Omega_R=\Omega_R^+-S_{\mu\mu}^+$ under this notation. The condition is not that every predetermined adjustment variable is fixed under the design. It is that, after the design convention has fixed which objects move under shock re-randomizations, the residualized adjustment component is long-run orthogonal to the innovation component in the sense stated above.

Define the HAC-minimizing projection coefficient $B^*:=(S_{zz}^+)^{-1}S_{zg}^+$ and the adjusted variance limit $\Omega_R^+(r(B)):=S_{r(B)r(B)}^+$.
\begin{theorem}\label{thm:RA-fullobs}
Under Assumption~\ref{ass:RA-orth}, every fixed conformable matrix $B$ satisfies \(\Omega_R^+(r(B))=\Omega_R+S_{d(B)d(B)}^+\succeq\Omega_R\).
For $B^*=(S_{zz}^+)^{-1}S_{zg}^+$,
\[
\Omega_R\preceq\Omega_R^+(r(B^*))\preceq\Omega_R^+,
\qquad
\Omega_R^+(r(B^*))=\Omega_R+S_{\mu\mu}^+-S_{\mu z}^+(S_{zz}^+)^{-1}S_{z\mu}^+ .
\]
Equality with $\Omega_R$ holds if and only if $S_{d(B^*)d(B^*)}^+=0$.
\end{theorem}

\begin{corollary}\label{cor:RA-oracle}
Under Assumption~\ref{ass:RA-orth}, if there exists $B_0$ such that $S_{d(B_0)d(B_0)}^+=0$, then $\Omega_R^+(r(B^*))=\Omega_R$. In particular, the conclusion holds if $\tilde\mu_{T,t}=B_0^\top z_t$ for every $t$ along the sequence.
\end{corollary}

\begin{proposition}\label{prop:RA-counterexample}
Consider the random-sampling scheme in Assumption~\ref{ass:sampling-app} with i.i.d.\ Bernoulli$(\rho)$ inclusion indicators, $\rho\in(0,1)$, independent of the moment map. Let the scalar moment be fixed under the assignment design, $g_t=\mu_{T,t}$, where $T^{-1/2}\sum_{t=1}^T\mu_{T,t}\to0$, $T^{-1}\sum_{t=1}^T\mu_{T,t}^2\to\sigma_\mu^2>0$, and all nonzero empirical lag products have vanishing Ces\`aro limits. Let $z_t=\mu_{T,t}$ and take $B=1$. Then $r_t(B)=0$, so the adjusted observed-pair HAC limit is zero, while the exact incomplete-observation design covariance is $\Omega_R^{\mathrm{obs}}=(1-\rho)\sigma_\mu^2>0$. Thus the analogue of the fully observed lower bound is false under incomplete sampling without additional restrictions.
\end{proposition}

\subsection{Simultaneous inference}\label{app:simultaneous}

The main-text procedures are pointwise in horizon. For a fixed set of horizons, Bonferroni and \v{S}id{\'a}k adjustments both give simultaneous guarantees under the joint Gaussian condition below, but pointwise intervals should not themselves be read as simultaneous bands. Let $z_q$ denote the $q$th standard-normal quantile. Define $\mathcal C_{j,T}^{\mathrm{Bonf}}:=[\widehat\psi_j \pm z_{1-\alpha/(2J)}\,\widehat s_j/\sqrt T]$ and $\mathcal C_{j,T}^{\mathrm{Sid}}:=[\widehat\psi_j \pm z_{\{1+(1-\alpha)^{1/J}\}/2}\,\widehat s_j/\sqrt T]$ for $j=1,\ldots,J$.

\begin{theorem}\label{thm:bonf}
Let $\psi_T=(\psi_{T,1},\dots,\psi_{T,J})^\top$ be a fixed-$J$ vector of functions of the estimand, and fix $0<\alpha<1$. Suppose that $\widehat s_j>0$ with probability approaching one for each $j$, and that $T_T:=(T_{1,T},\ldots,T_{J,T})^\top\Rightarrow Z$, where $T_{j,T}:=\sqrt T(\widehat\psi_j-\psi_{T,j})/\widehat s_j$ and $Z$ is centered Gaussian with $\Var(Z_j)=\tau_j^2\le1$. Then both $\liminf_{T\to\infty}\Pr(\psi_{T,j}\in\mathcal C_{j,T}^{\mathrm{Bonf}}\ \text{for all }j=1,\dots,J)\ge1-\alpha$ and $\liminf_{T\to\infty}\Pr(\psi_{T,j}\in\mathcal C_{j,T}^{\mathrm{Sid}}\ \text{for all }j=1,\dots,J)\ge1-\alpha$.
\end{theorem}

\subsection{Local LP--VAR comparison}\label{app:lpvar}

The LP--VAR comparison is a local first-order statement for fixed-dimensional procedures and a fixed finite horizon set. Fix $\mathcal H=\{h_1,\dots,h_J\}$. For each $h\in\mathcal H$, consider the just-identified OLS local-projection moment $g_{t,h}^{LP}(\theta_h):=\psi_t(y_{t+h}-\psi_t^\top\theta_h)$, where $\psi_t=(1,x_t,c_t^\top)^\top$ and $\theta_h=(\alpha_h,\beta_h,\gamma_h^\top)^\top$. Let $s_x$ select the $x_t$ coefficient, so $\beta_h=s_x^\top\theta_h$. Write $\bar m_{h,T}(\theta):=T^{-1}\sum_{t=1}^T\mathbb E_T[g_{t,h}^{LP}(\theta)]$ and $Q_{h,T}:=T^{-1}\sum_{t=1}^T\mathbb E_T[\psi_t\psi_t^\top]$.

\begin{theorem}\label{thm:LP-local}
For each $h\in\mathcal H$, let $\theta_{h,T}^\star$ be the unique solution to $\bar m_{h,T}(\theta)=0$ for all large $T$. Suppose $Q_{h,T}\to Q_h\succ0$, $T^{-1}\sum_t\mathbb E_T[x_t]\to0$, and $T^{-1}\sum_t\mathbb E_T[x_t c_t^\top]\to0$. Suppose also that $\bar m_{h,T}(\theta_{0,h})=T^{-1/2}d_{h,T}+o(T^{-1/2})$ with $d_{h,T}\to d_h=(d_{0,h},0,d_{c,h}^\top)^\top$, and that the design-based CLT around $\theta_{h,T}^\star$ holds. Then, for each fixed $h\in\mathcal H$, \(\sqrt T(\widehat\beta_h^{LP}-\beta_{0,h})\Rightarrow\mathcal{N}(0,v_{LP,h})\), where $v_{LP,h}$ is the design-based asymptotic variance of $\sqrt T(\widehat\beta_h^{LP}-\beta_{h,T}^\star)$.
\end{theorem}

For the VAR comparison, let $\phi\in\mathbb{R}^{p_\phi}$ denote the VAR parameter, let $g_t^{VAR}(\phi)$ denote the corresponding GMM moment, and let $r(\phi)\in\mathbb{R}^J$ collect the impulse responses at the horizon set $\mathcal H$. Let $A_\phi\succ0$ be fixed. Define $\bar m_T^{VAR}(\phi):=T^{-1}\sum_{t=1}^T\mathbb E_T[g_t^{VAR}(\phi)]$, $Q_T^{VAR}(\phi):=\bar m_T^{VAR}(\phi)^\top A_\phi\bar m_T^{VAR}(\phi)$, $G_{\phi,T}:=T^{-1}\sum_t\mathbb E_T[\partial_\phi g_t^{VAR}(\phi_0)]$, and $H_{\phi,T}:=(G_{\phi,T}^\top A_\phi G_{\phi,T})^{-1}G_{\phi,T}^\top A_\phi$ when the inverse exists. Suppose $G_{\phi,T}\to G_\phi$, put $H_\phi:=(G_\phi^\top A_\phi G_\phi)^{-1}G_\phi^\top A_\phi$, let $R:=\nabla_\phi r(\phi_0)$, and let $\Sigma_\phi$ denote the design-based asymptotic covariance of $\sqrt T(\widehat\phi-\phi_T^\star)$.

\begin{theorem}\label{thm:VAR-local}
Suppose $G_\phi$ has full column rank, $\bar m_T^{VAR}(\phi_0)=T^{-1/2}d_{VAR,T}+o(T^{-1/2})$ with $d_{VAR,T}\to d_{VAR}$, and $r(\cdot)$ is continuously differentiable at $\phi_0$. Suppose also that there is a sequence $\phi_T^\star\to\phi_0$ of interior local minimizers of $Q_T^{VAR}$ such that, with $D_{\phi,T}^\star:=\nabla_\phi\bar m_T^{VAR}(\phi_T^\star)$, $(D_{\phi,T}^\star)^\top A_\phi\bar m_T^{VAR}(\phi_T^\star)=0$ and $D_{\phi,T}^\star=G_{\phi,T}+o(1)$. Finally, suppose that, for every sequence $\phi_T\to\phi_0$, \(\bar m_T^{VAR}(\phi_T)=\bar m_T^{VAR}(\phi_0)+G_{\phi,T}(\phi_T-\phi_0)+a_T(\phi_T)\), with \(\|a_T(\phi_T)\|=o(\|\phi_T-\phi_0\|)+o(T^{-1/2})\). If $\sqrt T(\widehat\phi-\phi_T^\star)\Rightarrow\mathcal N(0,\Sigma_\phi)$, then \(\sqrt T\big(r(\widehat\phi)-r(\phi_0)\big)\Rightarrow \mathcal{N}(b_{VAR},R\Sigma_\phi R^\top)\), where \(b_{VAR}:=-RH_\phi d_{VAR}\).
\end{theorem}

\begin{proposition}\label{prop:tangent}
Let $\Sigma_{LP}$ denote the positive-definite design-based asymptotic covariance of the stacked LP estimator for the IRF vector $\beta\in\mathbb{R}^J$. Let $C$ be full row rank with $\ker(C)=\operatorname{Im}(R)$. Suppose the VAR IRF estimator is asymptotically equivalent, under local correct specification, to the restricted Gaussian GMM estimator of $\beta$ subject to $C(\beta-\beta_0)=0$. Then $C\Sigma_{LP}C^\top$ is nonsingular and the VAR IRF covariance satisfies \(\Sigma_{VAR}=\Sigma_{LP}-\Sigma_{LP}C^\top(C\Sigma_{LP}C^\top)^{-1}C\Sigma_{LP}\preceq \Sigma_{LP}\).
\end{proposition}

\begin{proposition}\label{prop:LP-VAR-drift-design}
Under the conditions of Theorem~\ref{thm:LP-local}, $\sqrt T(\beta_{h,T}^\star-\beta_{0,h})\to0$ for each fixed $h\in\mathcal H$. Under the conditions of Theorem~\ref{thm:VAR-local}, $\sqrt T(r(\phi_T^\star)-r(\phi_0))\to b_{VAR}=-RH_\phi d_{VAR}$. If the tangent-space condition of Proposition~\ref{prop:tangent} also holds, then the VAR covariance for the impulse-response vector satisfies $\Sigma_{VAR}\preceq\Sigma_{LP}$. No covariance ordering is implied without that tangent-space condition.
\end{proposition}

The bias term $b_{VAR}$ is a local mean-drift term. If the VAR moment is correctly specified date by date so that $d_{VAR}=0$, then this design-drift term vanishes and the impulse-response limit is centered at the reference value $r(\phi_0)$.

\subsection{Overidentified GMM under local misspecification}\label{app:misspecified-gmm}

The main theorem uses Assumption~\ref{ass:local-correct-specification} to keep the first-order covariance in the familiar GMM sandwich form based on the centered moment innovations. This subsection records the expansion for the fixed-weight criterion when the pseudo-true mean is not negligible at the $T^{-1/2}$ scale.

If the weight matrix is estimated and the pseudo-true mean does not vanish, the influence function of the weight estimator enters the same first-order condition through the derivative of $G_T(\theta)^\top A\bar m_T(\theta)$ with respect to $A$. The statement below sets $\widehat A_N\equiv A$ to isolate the design issue.

Let $J_{T,t}(\theta):=\nabla_\theta g_t(W_t,\theta)$, $G_T(\theta):=T^{-1}\sum_t\mathbb E_T[J_{T,t}(\theta)]$, $\widehat G_N(\theta):=T^{-1}\sum_{t=1}^T J_{T,t}(\theta)$, and $\bar m_T^\star:=\bar m_T(\theta_T^\star)$. Suppose $\theta_T^\star$ is an interior minimizer of $Q_T(\theta)=\bar m_T(\theta)^\top A\bar m_T(\theta)$, so $G_T(\theta_T^\star)^\top A\bar m_T^\star=0$. Define
\[
\mathcal D_T:=\nabla_\theta\big(G_T(\theta)^\top A\bar m_T(\theta)\big)\big|_{\theta=\theta_T^\star},
\qquad
\zeta_{T,t}:=G_T(\theta_T^\star)^\top A e_{T,t}+\big(J_{T,t}(\theta_T^\star)-G_T(\theta_T^\star)\big)^\top A\bar m_T^\star .
\]

\begin{proposition}\label{prop:misspecified-gmm}
Consider the fixed-weight estimator obtained by setting $\widehat A_N\equiv A$ in the sample criterion. Suppose Assumptions~\ref{ass:design-environment}, \ref{ass:smoothness}, \ref{ass:uniform-lln}, \ref{ass:gmm}, \ref{ass:estimand-separation}, and~\ref{ass:gmm-interior} hold for this fixed $A$, but do not impose Assumption~\ref{ass:local-correct-specification}. Suppose $\mathcal D_T\to\mathcal D$ with $\mathcal D$ nonsingular. Suppose also that, on the high-probability event where the fixed-weight sample first-order condition holds, the local moment satisfies
\[
\widehat G_N(\widehat\theta_N)^\top A g_N(\widehat\theta_N)
-\widehat G_N(\theta_T^\star)^\top A g_N(\theta_T^\star)
=\mathcal D_T(\widehat\theta_N-\theta_T^\star)+o_p(T^{-1/2}),
\]
and
\[
T^{-1/2}\sum_t e_{T,t}=O_p(1),
\qquad
T^{-1/2}\sum_t\operatorname{vec}\!\left(J_{T,t}(\theta_T^\star)-G_T(\theta_T^\star)\right)=O_p(1).
\]
Suppose finally that $T^{-1/2}\sum_{t=1}^T\zeta_{T,t}\Rightarrow\mathcal N(0,\Omega_\zeta)$. Then
\[
\sqrt T(\widehat\theta_N-\theta_T^\star)
=-\mathcal D_T^{-1}T^{-1/2}\sum_{t=1}^T\zeta_{T,t}+o_p(1)
\Rightarrow \mathcal N(0,\mathcal D^{-1}\Omega_\zeta\mathcal D^{-\top}).
\]
\end{proposition}

Proposition~\ref{prop:misspecified-gmm} explains the role of Assumption~\ref{ass:local-correct-specification} in the main text. If $\sqrt T\|\bar m_T^\star\|=o(1)$, then the Jacobian-estimation contribution satisfies $T^{-1/2}\sum_{t=1}^T(J_{T,t}-G_T)^\top A\bar m_T^\star=o_p(1)$, and the derivative term in $\mathcal D_T$ involving $\bar m_T^\star$ is negligible under the same smoothness conditions. The expansion then reduces to the locally correctly specified GMM sandwich based on $G_T^\top A e_{T,t}$.

When $\bar m_T^\star$ is first-order nonzero, fixed-weight inference must instead use the augmented misspecification moment $\zeta_{T,t}$. Estimated weights or first-stage components require their corresponding influence terms to be stacked in the same first-order condition.

\subsection{Two-step generated regressors}\label{app:two-step}

This subsection treats generated shocks and proxy-style first stages as part of the moment. If the first-stage construction is included in the design distribution, the variance limit is formed from the augmented moment vector below, which stacks the first-stage influence contribution with the second-stage moment. Ignoring the first stage is therefore a conditioning convention, not an implication of the design-based argument.

Let $\pi\in\mathbb{R}^{p_\pi}$ be a first-stage parameter and $\theta\in\mathbb{R}^{p_\theta}$ a second-stage parameter. Let $g_t(\theta,\pi)\in\mathbb R^k$ be smooth in both arguments, and define $\bar m_T(\theta,\pi):=T^{-1}\sum_t\mathbb E_T[g_t(\theta,\pi)]$. For an estimand pair $(\theta_T^\star,\pi_T^\star)$, write $G_{\theta,T}:=T^{-1}\sum_t\mathbb E_T[\partial_\theta g_t(\theta_T^\star,\pi_T^\star)]$ and $G_{\pi,T}:=T^{-1}\sum_t\mathbb E_T[\partial_\pi g_t(\theta_T^\star,\pi_T^\star)]$, and let $H_{\theta,T}:=(G_{\theta,T}^\top A G_{\theta,T})^{-1}G_{\theta,T}^\top A$ when the inverse exists. Let $IF_{\pi,T,t}$ denote a first-stage influence function, and define the augmented second-stage moment as $\widetilde g_{T,t}:=g_t(\theta_T^\star,\pi_T^\star)+G_{\pi,T}IF_{\pi,T,t}$.

\begin{theorem}\label{thm:twostep}
Suppose $\sqrt T\|\bar m_T(\theta_T^\star,\pi_T^\star)\|=o(1)$, $G_{\theta,T}\to G_\theta$ with $G_\theta$ full column rank, $G_{\pi,T}\to G_\pi$, and $\widehat A_N\to_p A\succ0$. Suppose the first stage satisfies $\sqrt T(\widehat\pi-\pi_T^\star)=T^{-1/2}\sum_{t=1}^T IF_{\pi,T,t}+o_p(1)$ and, on an event with probability approaching one, the second-stage first-order condition has the local expansion
\[
0=G_{\theta,T}^\top A\Big[g_N(\theta_T^\star,\pi_T^\star)+G_{\theta,T}(\widehat\theta-\theta_T^\star)+G_{\pi,T}(\widehat\pi-\pi_T^\star)\Big]+o_p(T^{-1/2}).
\]
Then \(\sqrt T(\widehat\theta-\theta_T^\star)=-H_{\theta,T}T^{-1/2}\sum_{t=1}^T\widetilde g_{T,t}+o_p(1)\).
If, in addition, the augmented contribution is locally correctly specified, \(\sqrt T\|T^{-1}\sum_{t=1}^T\mathbb E_T[\widetilde g_{T,t}]\|=o(1)\), its centered version satisfies the corresponding CLT, the augmented contribution satisfies the corresponding mean-path, HAC product-array, and multiplier conditions, and the parameter-dependent second-stage moment satisfies the smoothness and uniform-law conditions needed for the displayed local expansion, then the main GMM, HAC, and bootstrap conclusions apply to the augmented moment representation after replacing the original limiting Jacobian by $G_\theta$.
\end{theorem}

\begin{corollary}\label{cor:stacked}
Let $h_t(\eta):=(a_t(\pi)^\top,g_t(\theta,\pi)^\top)^\top$ with $\eta=(\pi^\top,\theta^\top)^\top$. Suppose the recursive stacked system is just identified and that, after replacing $g_t$ by $h_t$ and $\theta$ by $\eta$, the stacked moment satisfies Assumptions~\ref{ass:design-environment}, \ref{ass:dependence}, \ref{ass:smoothness}, \ref{ass:uniform-lln}, \ref{ass:local-correct-specification}, \ref{ass:gmm}, \ref{ass:estimand-separation}, and~\ref{ass:gmm-interior}. Suppose further that its block Jacobian is
\[
G_{h,T}=\begin{bmatrix}
G_{a,T} & 0\\
G_{\pi,T} & G_{\theta,T}
\end{bmatrix},
\]
where $G_{a,T}$ and $G_{\theta,T}$ are nonsingular for all large $T$. Then the joint estimator satisfies the design-based GMM theorem with the stacked Jacobian and stacked long-run covariance. Its second-stage block is equivalently represented by Theorem~\ref{thm:twostep} with $IF_{\pi,T,t}=-G_{a,T}^{-1}a_t(\pi_T^\star)$.
\end{corollary}

\clearpage
\section{Proofs}\label{app:proofs}

This section collects the proofs in dependency order. The first group records deterministic variance and covariance facts; the next groups give the uniform-law and CLT tools, the main GMM and HAC arguments, the incomplete-observation extensions, and the specialized appendix results for regression adjustment, LP--VAR comparisons, misspecification, and generated regressors.

\subsection{Preliminary identities}

The first step is to record two elementary bounds used in several proofs. First, let $(U_{T,t})$ be a zero-mean triangular array such that
\[
T^{-1}\sum_{h=-(T-1)}^{T-1}\sum_{t=|h|+1}^T
\left\|\Cov(U_{T,t},U_{T,t-|h|})\right\|=O(1).
\]
Collecting the double covariance sum by lag gives
\[
\Var\Big(T^{-1}\sum_{t=1}^TU_{T,t}\Big)
=T^{-2}\sum_{h=-(T-1)}^{T-1}\sum_{t=|h|+1}^T\Cov(U_{T,t},U_{T,t-|h|}),
\]
and the displayed bound implies $\|\Var(T^{-1}\sum_{t=1}^TU_{T,t})\|=O(T^{-1})$.

A second repeatedly used calculation controls plug-in lag products. Let $\hat{s}_t$, $s_t$, and $\Delta_T$ be time series, and let $D_T:=\sum_t Z_t$ for nonnegative normalization weights. Suppose $\| g_t(W_t, \hat{\theta}_N) - s_t \|\le b_{T,t}\Delta_T$, $\|s_t\|\vee\|g_t(W_t, \hat{\theta}_N)\|\le B_{T,t}$, the relevant second moments of $b_{T,t}$ and $B_{T,t}$ are uniformly bounded, and $\Delta_T\to_p 0$. If the normalized averages $D_T^{-1}\sum_t Z_t b_{T,t}B_{T,t-|\ell|}$ and $D_T^{-1}\sum_t Z_t B_{T,t}b_{T,t-|\ell|}$ are $O_p(1)$ uniformly over the retained finite lag set, then
\[
D_T^{-1}\sum_t Z_t\,\big\|g_t(W_t, \hat{\theta}_N)g_{t-|\ell|}(W_{t-|\ell|}, \hat{\theta}_N)^\top - s_t s_{t-|\ell|}^\top\big\|
=O_p(1)\Delta_T=o_p(1).
\]
Indeed, with $\hat g_t:=g_t(W_t,\hat{\theta}_N)$, $\hat g_t\hat g_{t-|\ell|}^\top-s_ts_{t-|\ell|}^\top=(\hat g_t-s_t)s_{t-|\ell|}^\top+s_t(\hat g_{t-|\ell|}-s_{t-|\ell|})^\top+(\hat g_t-s_t)(\hat g_{t-|\ell|}-s_{t-|\ell|})^\top$. Hence the displayed average is bounded by $\Delta_TD_T^{-1}\sum_tZ_tb_{T,t}B_{T,t-|\ell|}+\Delta_TD_T^{-1}\sum_tZ_tB_{T,t}b_{T,t-|\ell|}+\Delta_T^2D_T^{-1}\sum_tZ_tb_{T,t}b_{T,t-|\ell|}=O_p(1)\Delta_T+O_p(1)\Delta_T^2$, where the last average is $O_p(1)$ by Cauchy--Schwarz and the stated second-moment bounds.

The next deterministic lemma is used whenever a growing-bandwidth HAC covariance of a fixed centered path is interpreted as a nonnegative matrix.

\begin{lemma}
\label{lem:sumGE-psd}
For each fixed $\ell\ge 0$, let $\Gamma_{\mu,T}(\ell):=T^{-1}\sum_{t=\ell+1}^T \tilde\mu_t\tilde\mu_{t-\ell}^\top$ and $\Gamma_{\mu,T}(-\ell):=\Gamma_{\mu,T}(\ell)^\top$ for $\ell>0$, and suppose the fixed-lag limits $\Gamma_{g,\mu}(\ell)=\lim_T\Gamma_{\mu,T}(\ell)$ exist and are absolutely summable. Then $\sum_{\ell\in\mathbb{Z}}\Gamma_{g,\mu}(\ell)\succeq0$.
\end{lemma}

\begin{proof}[Proof of Lemma \ref{lem:sumGE-psd}]
Fix $u\in\mathbb R^k$ and write $x_{T,t}:=u^\top\tilde\mu_t$ for $1\le t\le T$, extending the sequence by $x_{T,t}=0$ outside the sample. Define the finite-sample scalar autocovariance $c_T(\ell):=T^{-1}\sum_{t\in\mathbb Z}x_{T,t}x_{T,t-\ell}$. The extension by zero makes the sum finite and gives $c_T(\ell)=0$ when $|\ell|\ge T$. For every fixed $\ell$, the fixed-lag matrix limit implies $c_T(\ell)\to\gamma(\ell):=u^\top\Gamma_{g,\mu}(\ell)u$, with the transpose convention for negative lags.

Let $M\ge1$ be fixed. Expanding the square and collecting terms by the lag $\ell=j-r$ gives the finite-$T$ Fej\'er identity
\[
\frac1T\sum_{t\in\mathbb Z}\left(M^{-1/2}\sum_{j=0}^{M-1}x_{T,t-j}\right)^2
=\sum_{|\ell|<M}\left(1-\frac{|\ell|}{M}\right)c_T(\ell).
\]
The left side is a finite sum of squares and is therefore nonnegative. Passing to the limit in $T$ at the fixed value of $M$ gives \(\sum_{|\ell|<M}\left(1-\frac{|\ell|}{M}\right)\gamma(\ell)\ge0\).
Absolute summability of the matrix lag sequence implies $\sum_\ell|\gamma(\ell)|\le \|u\|^2\sum_\ell\|\Gamma_{g,\mu}(\ell)\|<\infty$. Hence the Fej\'er-weighted sums converge as $M\to\infty$ to $\sum_{\ell\in\mathbb Z}\gamma(\ell)$ by dominated convergence for series. It follows that \(u^\top\left(\sum_{\ell\in\mathbb Z}\Gamma_{g,\mu}(\ell)\right)u=\sum_{\ell\in\mathbb Z}\gamma(\ell)\ge0\).
Since the preceding argument holds for every $u\in\mathbb R^k$, the long-run matrix $\sum_{\ell\in\mathbb Z}\Gamma_{g,\mu}(\ell)$ is positive semidefinite.
\end{proof}

\begin{proof}[Proof of Proposition \ref{lem:fixed-env-decomp}]
Fix the conditioning environment and suppress the argument $\theta_T^\star$. Write $\bar m_T:=\bar m_T(\theta_T^\star)$, $\mu_{T,t}:=\mathbb E_T[g_t(W_t,\theta_T^\star)]$, $e_{T,t}:=g_t(W_t,\theta_T^\star)-\mu_{T,t}$, and $\tilde\mu_{T,t}:=\mu_{T,t}-\bar m_T$. Then $\mathbb E_T[e_{T,t}]=0$ for each date and \(g_t(W_t,\theta_T^\star)-\bar m_T=e_{T,t}+\tilde\mu_{T,t}\).
For a fixed nonnegative lag $\ell$, define the finite-$T$ averages
\[
\Gamma_{e,T}(\ell):=\frac1T\sum_{t=\ell+1}^T\mathbb E_T[e_{T,t}e_{T,t-\ell}^\top],
\qquad
\Gamma_{\mu,T}(\ell):=\frac1T\sum_{t=\ell+1}^T\tilde\mu_{T,t}\tilde\mu_{T,t-\ell}^\top.
\]
Expanding the centered observed product gives, for each retained pair $(t,t-\ell)$,
\[
\begin{aligned}
&\mathbb E_T[(g_t-\bar m_T)(g_{t-\ell}-\bar m_T)^\top] \\
&\quad=\mathbb E_T[e_{T,t}e_{T,t-\ell}^\top]
+\mathbb E_T[e_{T,t}]\tilde\mu_{T,t-\ell}^\top
+\tilde\mu_{T,t}\mathbb E_T[e_{T,t-\ell}]^\top
+\tilde\mu_{T,t}\tilde\mu_{T,t-\ell}^\top \\
&\quad=\mathbb E_T[e_{T,t}e_{T,t-\ell}^\top]+\tilde\mu_{T,t}\tilde\mu_{T,t-\ell}^\top.
\end{aligned}
\]
The two cross terms vanish date by date because the mean path is fixed under the conditional law and each $e_{T,t}$ is conditionally centered. Averaging over $t=\ell+1,\ldots,T$ therefore gives the exact finite-$T$ identity
\[
\frac1T\sum_{t=\ell+1}^T\mathbb E_T[(g_t-\bar m_T)(g_{t-\ell}-\bar m_T)^\top]
=\Gamma_{e,T}(\ell)+\Gamma_{\mu,T}(\ell).
\]
Assumption~\ref{ass:dependence} gives $\Gamma_{e,T}(\ell)\to\Gamma_e(\ell)$, and Assumption~\ref{ass:mean-path} gives $\Gamma_{\mu,T}(\ell)\to\Gamma_{g,\mu}(\ell)$. Thus the fixed-lag limit in \eqref{eq:Gamma-ghac} exists and equals $\Gamma_{g,\mathrm{hac}}(\ell)=\Gamma_e(\ell)+\Gamma_{g,\mu}(\ell)$ for every $\ell\ge0$. The definition for negative lags is by transposition, so the same identity holds for every $\ell\in\mathbb Z$.

The lag arrays $(\Gamma_e(\ell))_{\ell\in\mathbb Z}$ and $(\Gamma_{g,\mu}(\ell))_{\ell\in\mathbb Z}$ are absolutely summable by Assumptions~\ref{ass:dependence} and~\ref{ass:mean-path}. Hence $(\Gamma_{g,\mathrm{hac}}(\ell))_{\ell\in\mathbb Z}$ is absolutely summable, and term-by-term summation of the fixed-lag identity yields
\[
\Omega_R^+=\sum_{\ell\in\mathbb Z}\Gamma_{g,\mathrm{hac}}(\ell)
=\sum_{\ell\in\mathbb Z}\Gamma_e(\ell)+\sum_{\ell\in\mathbb Z}\Gamma_{g,\mu}(\ell)
=\Omega_R+\Omega_\mu.
\]
Finally, Lemma~\ref{lem:sumGE-psd} applies to the fixed centered path $(\tilde\mu_{T,t})_{t\le T}$ and gives $\Omega_\mu\succeq0$.
\end{proof}

\begin{proof}[Proof of Lemma \ref{lem:neyman-ts}]
Fix $h\ge0$ and condition on the fixed array $\{(Y^{(1)}_{t+h},Y^{(0)}_{t+h}):1\le t\le T_h\}$. Write $n:=T_h$, $a_t:=Y^{(1)}_{t+h}-\bar Y_h^{(1)}$, $b_t:=Y^{(0)}_{t+h}-\bar Y_h^{(0)}$, and $\hat p_n:=n^{-1}\sum_{t=1}^n W_t$. With an intercept and a binary regressor, the LP slope is the difference between the treated and untreated sample means. Since $\sum_t a_t=\sum_t b_t=0$,
\[
\widehat\tau_h-\bar\tau_h
=\frac{A_n}{\hat p_n}-\frac{B_n}{1-\hat p_n},
\qquad
A_n:=\frac1n\sum_{t=1}^n W_ta_t,
\quad
B_n:=\frac1n\sum_{t=1}^n(1-W_t)b_t .
\]
Adding and subtracting $A_n/p-B_n/(1-p)$ gives \(\widehat\tau_h-\bar\tau_h=n^{-1}\sum_{t=1}^n\phi_{t,h}+R_n\), where \(R_n=A_n\frac{p-\hat p_n}{p\hat p_n}-B_n\frac{\hat p_n-p}{(1-p)(1-\hat p_n)}\). Let $\mathcal E_n:=\{\hat p_n\in[p/2,(1+p)/2]\}$. On $\mathcal E_n$, the denominators in $R_n$ are bounded away from zero and $|R_n|\le C(|A_n|+|B_n|)|\hat p_n-p|$. Since $\sum_ta_t=\sum_tb_t=0$, $A_n=n^{-1}\sum_t(W_t-p)a_t$ and $B_n=-n^{-1}\sum_t(W_t-p)b_t$. The fixed-array fourth-moment inequality for independent centered Bernoulli sums gives $\mathbb E[A_n^4]\le Cn^{-4}((\sum_ta_t^2)^2+\sum_ta_t^4)=O(n^{-2})$, and the same calculation gives $\mathbb E[B_n^4]=O(n^{-2})$. Also, $\mathbb E[(\hat p_n-p)^4]=O(n^{-2})$. Therefore Cauchy--Schwarz yields $n\mathbb E[R_n^2\1_{\mathcal E_n}]\le Cn(\mathbb E[A_n^4]+\mathbb E[B_n^4])^{1/2}\mathbb E[(\hat p_n-p)^4]^{1/2}=O(n^{-1})$. Hoeffding's inequality gives $\Pr(\mathcal E_n^c)\le 2e^{-cn}$. With the harmless convention that an empty-cell mean is set to zero, all cell-mean deviations on $\mathcal E_n^c$ are bounded by $C(1+\max_t|a_t|+\max_t|b_t|)$. The fixed-array fourth-moment bound implies $\max_t|a_t|\vee\max_t|b_t|\le Cn^{1/4}$, so $n\mathbb E[R_n^2\1_{\mathcal E_n^c}]\le Cn(1+n^{1/4})^2e^{-cn}=o(1)$. Thus $r_{T,h}:=\sqrt n R_n$ satisfies $\mathbb E[r_{T,h}^2]=o(1)$, proving the linear expansion above.

For each date, \(\mathbb E[\phi_{t,h}]=Y^{(1)}_{t+h}-\bar Y_h^{(1)}-Y^{(0)}_{t+h}+\bar Y_h^{(0)}=\tau_{t,h}-\bar\tau_h\), and hence $n^{-1}\sum_t\mathbb E[\phi_{t,h}]=0$. Conditional on the fixed array, $\phi_{t,h}$ is a function only of $W_t$. The Bernoulli assignments are independent across dates, so $\Cov(\phi_{t,h},\phi_{s,h})=0$ for $t\ne s$ and \(\Var\!\left(n^{-1/2}\sum_{t=1}^n\phi_{t,h}\right)=n^{-1}\sum_{t=1}^n\Var(\phi_{t,h})\). The $L^2$ remainder contributes $o(n^{-1})$ to $\Var(\widehat\tau_h)$ by Cauchy--Schwarz. The variance calculation uses $\mathbb E[\phi_{t,h}^2]=a_t^2/p+b_t^2/(1-p)$ and $\mathbb E[\phi_{t,h}]=a_t-b_t=\tau_{t,h}-\bar\tau_h$, so
\[
\Var(\phi_{t,h})
=p\left(\frac{Y^{(1)}_{t+h}-\bar Y_h^{(1)}}{p}\right)^2
+(1-p)\left(\frac{Y^{(0)}_{t+h}-\bar Y_h^{(0)}}{1-p}\right)^2
-(\tau_{t,h}-\bar\tau_h)^2 .
\]
Averaging over $t=1,\dots,n$ yields $\Gamma_h(0)=S^2_{1,h}/p+S^2_{0,h}/(1-p)-S^2_{\tau,h}$, with the finite-history variances defined in the statement. Substituting into the preceding variance identity gives the result.
\end{proof}

\begin{proof}[Proof of Remark \ref{rem:LP_VAR_mean_drift}]
(i) Using $y_{t+h}=y_{t+h}(0)+\tau_{t,h}x_t$ and the definition of $\theta_h^\star$, the conditional mean moment decomposes as $\mathbb{E}[\psi_t(y_{t+h}-\psi_t^\top\theta_h^\star)\mid\I_t]=\mathbb{E}[\psi_t(y_{t+h}(0)-\alpha_h^\star-\gamma_h^{\star\top}c_t)\mid\I_t]+\mathbb{E}[\psi_t(\tau_{t,h}-\beta_h^\star)x_t\mid\I_t]$.
Let $e_x$ select the shock coordinate, so $e_x^\top\psi_t=x_t$. By design exogeneity, $e_x^\top\mathbb{E}[\psi_t(y_{t+h}(0)-\alpha_h^\star-\gamma_h^{\star\top}c_t)\mid\I_t]=\mathbb{E}[x_t(y_{t+h}(0)-\alpha_h^\star-\gamma_h^{\star\top}c_t)\mid\I_t]=0$. The $\I_t$-measurability of $\tau_{t,h}$ and the conditional centering of $x_t$ imply $e_x^\top\mathbb{E}[\psi_t(\tau_{t,h}-\beta_h^\star)x_t\mid\I_t]=\mathbb{E}[x_t^2\mid\I_t](\tau_{t,h}-\beta_h^\star)=\mathrm{Var}(x_t\mid\I_t)(\tau_{t,h}-\beta_h^\star)$, which proves the displayed formula for the shock row. If, in addition, $\mathbb{E}[y_{t+h}(0)\mid\I_t]=\alpha_h^\star+\gamma_h^{\star\top}c_t$, then the first term vanishes entirely, giving the displayed vector with zero intercept and control rows.

(ii) At $\phi^\star$, write $u_t(\phi^\star)=\Delta_t+(B(Z_t)-B)W_t+\varepsilon_t$, where $\Delta_t:=(\Phi_0(Z_t)-\Phi_0)+\sum_{i=1}^p(\Phi_i(Z_t)-\Phi_i)Y_{t-i}$.
Conditioning on $\I_t$, $Z^{\mathrm{pred}}_t$ and $\Delta_t$ are measurable, and
$\mathbb{E}[W_t\mid\I_t]=0$, $\mathbb{E}[\varepsilon_t\mid\I_t]=0$, $\mathbb{E}[W_t W_t^\top\mid\I_t]=\Sigma_{W,t}$.
Therefore $\mathbb{E}[Z^{\mathrm{pred}}_t u_t(\phi^\star)^\top\mid\I_t]=Z^{\mathrm{pred}}_t\,\Delta_t^\top$ and $\mathbb{E}[W_t u_t(\phi^\star)^\top\mid\I_t]=\Sigma_{W,t}\,(B(Z_t)-B)^\top$,
which, stacked and vectorized, gives the claim. In particular, the conditional mean moment vanishes when $\Delta_t\equiv 0$
and $\Sigma_{W,t}(B(Z_t)-B)^\top\equiv 0$.
\end{proof}

\subsection{Uniform-LLNs}

The next two lemmas record one sufficient route to the moment-function and Jacobian uniform laws used in the incomplete-observation appendix. The route is stronger than Assumption~\ref{ass:uniform-lln}, because it imposes covariance-envelope regularity directly on the generic-$\theta$ arrays. Throughout Lemmas~\ref{lem:ULLN} and~\ref{lem:ULLN-J}, use the following appendix-only strengthening. For each fixed $\theta$ and each scalar coordinate $f_{T,t}(\theta)$ of $g_t(W_t,\theta)-\mathbb E_T[g_t(W_t,\theta)]$, there is a summable sequence $b_g^{(\theta)}(h)$ such that
\[
\frac1T\sum_{t=|h|+1}^T
\left|\Cov_T\big(f_{T,t}(\theta),f_{T,t-|h|}(\theta)\big)\right|
\le b_g^{(\theta)}(h)
\]
for all $T$ and all integer lags $h$, with the transpose or coordinate convention used when $h<0$. The same condition holds for each scalar coordinate of $\nabla_\theta g_t(W_t,\theta)-\mathbb E_T[\nabla_\theta g_t(W_t,\theta)]$, with summable bound $b_J^{(\theta)}(h)$. The sample and expectation averages of the Lipschitz envelopes $b_{g,T,t}$ and $b_{g,1,T,t}$ are $O_p(1)$ and $O(1)$, respectively, and the corresponding moment envelopes have the bounded moments stated in Assumption~\ref{ass:smoothness}. For later use, define
\[
\Delta_T(\theta):=(T\rho_T)^{-1}\sum_{t=1}^T(R_{T,t}-\rho_T)g_t(W_t,\theta),\quad
S_T(\theta):=T^{-1}\sum_{t=1}^T\big[g_t(W_t,\theta)-\mathbb{E}_T[g_t(W_t,\theta)]\big],
\]
and $r_T(\theta):=(N^{-1}-(T\rho_T)^{-1})\sum_{t=1}^T R_{T,t}g_t(W_t,\theta)$.
\begin{lemma}\label{lem:ULLN}
Under Assumptions~\ref{ass:smoothness} and~\ref{ass:sampling-app} and the covariance-envelope strengthening just stated, the decomposition $g_N(\theta)-\bar m_T(\theta)=\Delta_T(\theta)+S_T(\theta)+r_T(\theta)$ holds for every $\theta\in\Theta$. Moreover, $\sup_{\theta\in\Theta}\|\Delta_T(\theta)\|=o_p(1)$, $\sup_{\theta\in\Theta}\|S_T(\theta)\|=o_p(1)$, and $\sup_{\theta\in\Theta}\|r_T(\theta)\|=O_p(N^{-1/2})=o_p(1)$. Consequently, $\sup_{\theta\in\Theta}\|g_N(\theta)-\bar m_T(\theta)\|=o_p(1)$.
\end{lemma}

\begin{proof}[Proof of Lemma \ref{lem:ULLN}]
Fix the environment and write $\mu_t(\theta):=\mathbb E_T[g_t(W_t,\theta)]$. The decomposition is the following identity, since $\bar m_T(\theta)=T^{-1}\sum_t\mu_t(\theta)$:
\[
\begin{aligned}
g_N(\theta)-\bar m_T(\theta)
&=\left(N^{-1}-(T\rho_T)^{-1}\right)\sum_{t=1}^TR_{T,t}g_t(W_t,\theta) \\
&\quad +(T\rho_T)^{-1}\sum_{t=1}^T(R_{T,t}-\rho_T)g_t(W_t,\theta)
+T^{-1}\sum_{t=1}^T\big[g_t(W_t,\theta)-\mu_t(\theta)\big].
\end{aligned}
\]
The three terms on the right are $r_T(\theta)$, $\Delta_T(\theta)$, and $S_T(\theta)$.

First consider $\Delta_T$. For $\theta,\theta'\in\Theta$,
\[
\|\Delta_T(\theta)-\Delta_T(\theta')\|
\le \frac{1}{T\rho_T}\sum_{t=1}^T |R_{T,t}-\rho_T|\,b_{g,T,t}\|\theta-\theta'\|
\le L_{\Delta,T}\|\theta-\theta'\|,
\]
where $L_{\Delta,T}:=\rho_T^{-1}T^{-1}\sum_t b_{g,T,t}=O_p(1)$ because $\rho_T\to\rho>0$. Thus $\Delta_T$ is stochastically equicontinuous. Fix a point $\theta_0\in\Theta$ and a scalar coordinate $r$. Under the random-sampling scheme, independence of the inclusion process and the moment process gives $\mathbb E_T[(R_{T,t}-\rho_T)g_{t,r}(\theta_0)]=0$ and, for each lag $h$,
\[
\Cov_T\big((R_{T,t}-\rho_T)g_{t,r}(\theta_0),(R_{T,t-|h|}-\rho_T)g_{t-|h|,r}(\theta_0)\big)
=\gamma_{R,T}(h)\mathbb E_T[g_{t,r}(\theta_0)g_{t-|h|,r}(\theta_0)].
\]
The expectation on the right is uniformly bounded by Cauchy--Schwarz and the envelope in Assumption~\ref{ass:smoothness}, while $\sum_h|\gamma_{R,T}(h)|$ is uniformly bounded by Assumption~\ref{ass:sampling-app}. Collecting covariances by lag gives $\Var_T(\Delta_{T,r}(\theta_0))=(T\rho_T)^{-2}\sum_{h=-(T-1)}^{T-1}\sum_{t=|h|+1}^T\gamma_{R,T}(h)\mathbb E_T[g_{t,r}(\theta_0)g_{t-|h|,r}(\theta_0)]\le CT^{-1}\rho_T^{-2}\sum_h|\gamma_{R,T}(h)|=O(T^{-1})$.

Under the deterministic-window scheme, put $a_{T,t}:=R_{T,t}/\rho_T-1$. Then $\Delta_T(\theta_0)=T^{-1}\sum_ta_{T,t}g_t(W_t,\theta_0)$ and $\sup_t|a_{T,t}|$ is bounded for all large $T$. The deterministic mean component $T^{-1}\sum_ta_{T,t}\mu_t(\theta)$ is $o(1)$ uniformly in $\theta$ by Assumption~\ref{ass:sampling-app}. The centered component has coordinate variance bounded by
\[
\Var_T\!\left(T^{-1}\sum_{t=1}^Ta_{T,t}\big(g_{t,r}(\theta_0)-\mu_{t,r}(\theta_0)\big)\right)
\le \frac{C}{T}\sum_{h\in\mathbb Z}b_g^{(\theta_0)}(h)=O(T^{-1}),
\]
so Chebyshev's inequality gives pointwise convergence. Hence, under either sampling scheme, $\Delta_T(\theta_0)=o_p(1)$ for each fixed $\theta_0$, after including the deterministic mean term in the deterministic-window case.

Let $(\theta_j)_{j=1}^{J(\eta)}$ be a finite $\eta$-net of the compact set $\Theta$. The preceding pointwise result, a union bound over the finite set of coordinates and net points, and the fixed dimension of the moment vector give $\max_j\|\Delta_T(\theta_j)\|=o_p(1)$. Given $\varepsilon,c>0$, choose $M<\infty$ such that $\Pr(L_{\Delta,T}\le M)\ge1-c$ for all large $T$, and then choose $\eta\le\varepsilon/(2M)$. On $\{L_{\Delta,T}\le M\}$, every $\theta$ is within $\eta$ of some $\theta_j$, so $\|\Delta_T(\theta)\|\le\max_j\|\Delta_T(\theta_j)\|+\varepsilon/2$. Letting $T\to\infty$ and then $c\downarrow0$ proves $\sup_{\theta\in\Theta}\|\Delta_T(\theta)\|\to_p0$.

The term $S_T$ is handled similarly but without sampling weights. Assumption~\ref{ass:smoothness} gives
\[
\|S_T(\theta)-S_T(\theta')\|
\le \left(T^{-1}\sum_{t=1}^Tb_{g,T,t}+T^{-1}\sum_{t=1}^T\mathbb E_Tb_{g,T,t}\right)\|\theta-\theta'\|,
\]
and the multiplier is $O_p(1)$. At a fixed $\theta_0$ and coordinate $r$, the covariance-envelope strengthening gives
\[
\Var_T\!\left(T^{-1}\sum_{t=1}^T\big(g_{t,r}(\theta_0)-\mu_{t,r}(\theta_0)\big)\right)
\le \frac{C}{T}\sum_{h\in\mathbb Z}b_g^{(\theta_0)}(h)=O(T^{-1}).
\]
Chebyshev's inequality gives convergence on every finite net, and the same compactness and stochastic-equicontinuity argument gives $\sup_{\theta\in\Theta}\|S_T(\theta)\|\to_p0$.

To control the denominator ratio, note first that under the deterministic-window scheme, $N=T\rho_T$ exactly, so $r_T(\theta)=0$. Under the random-sampling scheme, write $\sqrt N r_T(\theta)=A_TB_T(\theta)$, where $A_T:=(T\rho_T-N)/(T\rho_T)$ and $B_T(\theta):=N^{-1/2}\sum_tR_{T,t}g_t(W_t,\theta)$. The inclusion covariance envelope gives $\Var(N)=\sum_{h=-(T-1)}^{T-1}\sum_{t=|h|+1}^T\Cov(R_{T,t},R_{T,t-|h|})\le CT\sum_h|\gamma_{R,T}(h)|=O(T)$, so $A_T=O_p(T^{-1/2})$. On the event $\mathcal A_T:=\{|N/(T\rho_T)-1|\le1/2\}$, whose probability tends to one,
\[
\sup_{\theta\in\Theta}\|B_T(\theta)\|
\le \left(\frac{2T}{\rho_T}\right)^{1/2}T^{-1}\sum_{t=1}^TB_{g,T,t}=O_p(\sqrt T).
\]
Thus $\sup_\theta\|\sqrt N r_T(\theta)\|=O_p(1)$ and $\sup_\theta\|r_T(\theta)\|=O_p(N^{-1/2})$. Combining the three bounds proves the lemma.
\end{proof}

\begin{lemma}\label{lem:ULLN-J}
Under Assumptions~\ref{ass:smoothness} and~\ref{ass:sampling-app} and the covariance-envelope strengthening stated before Lemma~\ref{lem:ULLN}, for both sampling schemes,
\[
\sup_{\theta\in\Theta}\left\|
\frac{1}{N}\sum_{t=1}^T R_{T,t}\,\nabla_\theta g_t(W_t,\theta)
- \frac{1}{T}\sum_{t=1}^T \mathbb{E}_T[\nabla_\theta g_t(W_t,\theta)]
\right\|\xrightarrow{p}0.
\]
\end{lemma}

\begin{proof}[Proof of Lemma \ref{lem:ULLN-J}]
Write $h_t(\theta):=\nabla_\theta g_t(W_t,\theta)$ and $\nu_t(\theta):=\mathbb E_T[h_t(\theta)]$. The same algebra as in Lemma~\ref{lem:ULLN} gives
\[
\frac1N\sum_{t=1}^TR_{T,t}h_t(\theta)-\frac1T\sum_{t=1}^T\nu_t(\theta)
=\Delta_T^J(\theta)+S_T^J(\theta)+r_T^J(\theta),
\]
where $\Delta_T^J(\theta):=(T\rho_T)^{-1}\sum_t(R_{T,t}-\rho_T)h_t(\theta)$, $S_T^J(\theta):=T^{-1}\sum_t(h_t(\theta)-\nu_t(\theta))$, and $r_T^J(\theta):=(N^{-1}-(T\rho_T)^{-1})\sum_tR_{T,t}h_t(\theta)$. The Jacobian has fixed dimension, so it is enough to prove convergence coordinate by coordinate and then pass to any equivalent matrix norm.

Assumption~\ref{ass:smoothness} gives $\|h_t(\theta)-h_t(\theta')\|\le b_{g,1,T,t}\|\theta-\theta'\|$ and $\|\nu_t(\theta)-\nu_t(\theta')\|\le\mathbb E_Tb_{g,1,T,t}\|\theta-\theta'\|$. Hence $\|\Delta_T^J(\theta)-\Delta_T^J(\theta')\|\le \rho_T^{-1}T^{-1}\sum_tb_{g,1,T,t}\|\theta-\theta'\|$ and $\|S_T^J(\theta)-S_T^J(\theta')\|\le [T^{-1}\sum_tb_{g,1,T,t}+T^{-1}\sum_t\mathbb E_Tb_{g,1,T,t}]\|\theta-\theta'\|$. The two multipliers are $O_p(1)$, so both processes are stochastically equicontinuous.

Fix $\theta_0$ and a scalar coordinate $r$ of $h_t(\theta_0)$, denoted $h_{t,r}$, and write $\nu_{t,r}:=\mathbb E_T[h_{t,r}]$. Under the random-sampling scheme, independence gives $\mathbb E_T[(R_{T,t}-\rho_T)h_{t,r}]=0$ and, for each lag $q$, $\Cov_T((R_{T,t}-\rho_T)h_{t,r},(R_{T,t-|q|}-\rho_T)h_{t-|q|,r})=\gamma_{R,T}(q)\mathbb E_T[h_{t,r}h_{t-|q|,r}]$. Cauchy--Schwarz and the Jacobian envelope bound the last expectation uniformly, so collecting covariances by lag gives $\Var_T(\Delta_{T,r}^J(\theta_0))\le CT^{-1}\rho_T^{-2}\sum_q|\gamma_{R,T}(q)|=O(T^{-1})$. Under the deterministic-window scheme, with $a_{T,t}:=R_{T,t}/\rho_T-1$, the decomposition $\Delta_{T,r}^J(\theta_0)=T^{-1}\sum_ta_{T,t}\nu_{t,r}+T^{-1}\sum_ta_{T,t}[h_{t,r}-\nu_{t,r}]$ separates the deterministic balance term from the centered term. The first term is $o(1)$ by the Jacobian balance condition in Assumption~\ref{ass:sampling-app}, and the second has variance bounded by $CT^{-1}\sum_q b_J^{(\theta_0)}(q)=O(T^{-1})$. The same covariance envelope gives $\Var_T(S_{T,r}^J(\theta_0))\le CT^{-1}\sum_q b_J^{(\theta_0)}(q)=O(T^{-1})$. Chebyshev's inequality gives pointwise convergence at $\theta_0$.

A finite-net argument now gives the uniform result. If $(\theta_j)$ is an $\eta$-net and $L_T^J$ denotes the relevant Lipschitz multiplier, then on $\{L_T^J\le M\}$, $\sup_\theta\|\Delta_T^J(\theta)\|\le\max_j\|\Delta_T^J(\theta_j)\|+M\eta$, and the same bound holds for $S_T^J$. Taking a union bound over the finite net and fixed coordinates, then choosing $\eta$ small and letting the complement probability of $\{L_T^J\le M\}$ vanish, yields $\sup_\theta\|\Delta_T^J(\theta)\|=o_p(1)$ and $\sup_\theta\|S_T^J(\theta)\|=o_p(1)$.

The ratio term is zero under deterministic sampling. Under random sampling, work on
\[
\mathcal A_T=\{|N/(T\rho_T)-1|\le1/2\}.
\]
On this event,
\[
\sup_{\theta\in\Theta}\left\|N^{-1/2}\sum_{t=1}^TR_{T,t}h_t(\theta)\right\|
\le \left(\frac{2T}{\rho_T}\right)^{1/2}T^{-1}\sum_{t=1}^TB_{g,1,T,t}=O_p(\sqrt T),
\]
and $A_T=(T\rho_T-N)/(T\rho_T)=O_p(T^{-1/2})$. Hence $\sup_\theta\|r_T^J(\theta)\|=O_p(N^{-1/2})=o_p(1)$. Combining the uniform bounds for the three terms proves the Jacobian law.
\end{proof}

\subsection{Main Results}

\begin{proof}[Proof of Theorem \ref{thm:AN}, part (i)]
The proof uses the fully observed normalization $N=T$. For each realized sample, continuity of $g_t(W_t,\theta)$ in $\theta$ and compactness of $\Theta$ imply that the sample criterion $J_N(\theta):=g_N(\theta)^\top\widehat A_Ng_N(\theta)$ is continuous on a compact set. Hence a measurable argmin selection may be taken. Write $\Delta_T(\theta):=g_N(\theta)-\bar m_T(\theta)$. For every $\theta\in\Theta$,
\[
\begin{aligned}
J_N(\theta)-Q_T(\theta)
&=\bar m_T(\theta)^\top(\widehat A_N-A)\bar m_T(\theta)
  +2\bar m_T(\theta)^\top\widehat A_N\Delta_T(\theta) \\
&\quad +\Delta_T(\theta)^\top\widehat A_N\Delta_T(\theta).
\end{aligned}
\]
Assumption~\ref{ass:smoothness} gives $M_{m,T}:=\sup_{\theta\in\Theta}\|\bar m_T(\theta)\|\le T^{-1}\sum_t\mathbb E_T[B_{g,T,t}]=O(1)$ uniformly in $T$. Assumption~\ref{ass:uniform-lln} gives $M_{\Delta,T}:=\sup_{\theta\in\Theta}\|\Delta_T(\theta)\|\to_p0$, and Assumption~\ref{ass:gmm} gives $\widehat A_N\to_pA$, hence $\|\widehat A_N\|=O_p(1)$. The preceding display yields $\sup_{\theta\in\Theta}|J_N(\theta)-Q_T(\theta)|\le M_{m,T}^2\|\widehat A_N-A\|+2M_{m,T}\|\widehat A_N\|M_{\Delta,T}+\|\widehat A_N\|M_{\Delta,T}^2$, and the right side is $o_p(1)$.

Fix $\varepsilon>0$. By Assumption~\ref{ass:estimand-separation}, there are $\eta_\varepsilon>0$ and $T_\varepsilon<\infty$ such that, for all $T\ge T_\varepsilon$, \(\inf_{\|\theta-\theta_T^\star\|\ge\varepsilon}\big[Q_T(\theta)-Q_T(\theta_T^\star)\big]\ge2\eta_\varepsilon\).
On the event $\sup_{\vartheta\in\Theta}|J_N(\vartheta)-Q_T(\vartheta)|<\eta_\varepsilon$, every $\theta$ outside the $\varepsilon$-ball around $\theta_T^\star$ satisfies
\[
J_N(\theta)-J_N(\theta_T^\star)
\ge Q_T(\theta)-Q_T(\theta_T^\star)
-2\sup_{\vartheta\in\Theta}|J_N(\vartheta)-Q_T(\vartheta)|>0.
\]
Thus no global minimizer can lie outside that ball on that event. Consequently,
\[
\Pr(\|\widehat\theta_N-\theta_T^\star\|\ge\varepsilon)
\le
\Pr\left(\sup_{\theta\in\Theta}|J_N(\theta)-Q_T(\theta)|\ge\eta_\varepsilon\right)\to0,
\]
which proves consistency for the moving sample-period estimand.
\end{proof}

\begin{proof}[Proof of Theorem \ref{thm:AN}, part (ii)]
Throughout this part $N=T$, and part (i) gives $\widehat\theta_N-\theta_T^\star\to_p0$. Assumption~\ref{ass:gmm-interior} gives a fixed radius $\eta>0$ such that the closed ball around $\theta_T^\star$ is contained in $\Theta$ for all large $T$, and, with probability approaching one, $\widehat\theta_N$ is an interior local minimizer in that ball. Intersecting this event with $\{\|\widehat\theta_N-\theta_T^\star\|<\eta\}$ gives a high-probability event on which the line segment between $\theta_T^\star$ and $\widehat\theta_N$ lies in $\Theta$ and the sample first-order condition is available.

Write $\widehat G_N(\theta):=T^{-1}\sum_{t=1}^T\nabla_\theta g_t(W_t,\theta)$, $\widetilde G_N:=\widehat G_N(\widehat\theta_N)$, and \(\bar G_N:=\int_0^1\widehat G_N(\theta_T^\star+r(\widehat\theta_N-\theta_T^\star))\,dr\).
On the high-probability interior event, the derivative of $J_N(\theta)$ at $\widehat\theta_N$ is zero, so $\widetilde G_N^\top\widehat A_Ng_N(\widehat\theta_N)=0$. The mean-value expansion of the sample moment along the same segment gives $g_N(\widehat\theta_N)=g_N(\theta_T^\star)+\bar G_N(\widehat\theta_N-\theta_T^\star)$. Hence
\begin{equation}\label{eq:AN-foc-step3}
0=\widetilde G_N^\top\widehat A_N
\big[g_N(\theta_T^\star)+\bar G_N(\widehat\theta_N-\theta_T^\star)\big].
\end{equation}

The Jacobian law and consistency imply $\widetilde G_N\to_pG$ and $\bar G_N\to_pG$. To see this, let $\theta_N^\dagger$ be any random sequence with $\theta_N^\dagger-\theta_T^\star\to_p0$. Then
\[
\begin{aligned}
\|\widehat G_N(\theta_N^\dagger)-G_T\|
&\le
\sup_{\theta\in\Theta}\left\|\widehat G_N(\theta)-\frac1T\sum_{t=1}^T\mathbb E_T[\nabla_\theta g_t(W_t,\theta)]\right\| \\
&\quad +\frac1T\sum_{t=1}^T\mathbb E_T[b_{g,1,T,t}]\,\|\theta_N^\dagger-\theta_T^\star\|,
\end{aligned}
\]
where $G_T:=T^{-1}\sum_t\mathbb E_T[\nabla_\theta g_t(W_t,\theta_T^\star)]$. The first term is $o_p(1)$ by Assumption~\ref{ass:uniform-lln}, and the second is $o_p(1)$ by Assumption~\ref{ass:smoothness}. Assumption~\ref{ass:smoothness} also gives $G_T\to G$. Applying the same bound to the random points on the line segment gives $\|\bar G_N-G\|\le\int_0^1\|\widehat G_N(\theta_T^\star+r(\widehat\theta_N-\theta_T^\star))-G\|\,dr=o_p(1)$. Since $\widehat A_N\to_pA$, the matrix $\widetilde G_N^\top\widehat A_N\bar G_N$ converges in probability to $G^\top A G$. Because $A\succ0$ and $G$ has full column rank, $G^\top A G$ is nonsingular, so $\widetilde G_N^\top\widehat A_N\bar G_N$ is nonsingular with probability approaching one.

It remains to reduce the score in \eqref{eq:AN-foc-step3} to the centered innovation average. The population criterion is $Q_T(\theta)=\bar m_T(\theta)^\top A\bar m_T(\theta)$, and the deterministic ball in Assumption~\ref{ass:gmm-interior} makes $\theta_T^\star$ an interior minimizer of $Q_T$ for all large $T$. Differentiability therefore gives the population first-order condition \(G_T^\top A\bar m_T(\theta_T^\star)=0\).
Adding and subtracting this identity yields
\[
\begin{aligned}
\widetilde G_N^\top\widehat A_Ng_N(\theta_T^\star)
&=G_T^\top A\big(g_N(\theta_T^\star)-\bar m_T(\theta_T^\star)\big) \\
&\quad +\big(\widetilde G_N^\top\widehat A_N-G_T^\top A\big)\big(g_N(\theta_T^\star)-\bar m_T(\theta_T^\star)\big) \\
&\quad +\big(\widetilde G_N^\top\widehat A_N-G_T^\top A\big)\bar m_T(\theta_T^\star).
\end{aligned}
\]
Under full observation, $g_N(\theta_T^\star)-\bar m_T(\theta_T^\star)=T^{-1}\sum_te_{T,t}$, and Assumption~\ref{ass:dependence}(ii) gives this term as $O_p(T^{-1/2})$. Call the second and third lines in the display $R_{1T}$ and $R_{2T}$. Then $\|R_{1T}\|\le\|\widetilde G_N^\top\widehat A_N-G_T^\top A\|\,\|T^{-1}\sum_te_{T,t}\|=o_p(T^{-1/2})$ and $\|R_{2T}\|\le\|\widetilde G_N^\top\widehat A_N-G_T^\top A\|\,\|\bar m_T(\theta_T^\star)\|=o_p(T^{-1/2})$, using Assumption~\ref{ass:local-correct-specification} for the last factor. Replacing $G_T$ by $G$ in the leading line changes the expression by $\|(G_T-G)^\top A T^{-1}\sum_te_{T,t}\|=o_p(T^{-1/2})$.

Solving \eqref{eq:AN-foc-step3} on the high-probability nonsingularity event gives the linear representation
\[
\sqrt T(\widehat\theta_N-\theta_T^\star)
=-(G^\top A G)^{-1}G^\top A\sqrt T\big(g_N(\theta_T^\star)-\bar m_T(\theta_T^\star)\big)+o_p(1).
\]
Assumption~\ref{ass:dependence}(ii) gives the Gaussian limit of the centered moment with covariance $\Omega_R$, and Slutsky's theorem gives the asserted covariance matrix for $\widehat\theta_N$.

For the delta method, let $\mathcal J_T:=\nabla_\theta h(\theta_T^\star)$. Since $\widehat\theta_N-\theta_T^\star=O_p(T^{-1/2})$ and $h$ is continuously differentiable, the mean-value expansion gives $h(\widehat\theta_N)-h(\theta_T^\star)=\mathcal J_T(\widehat\theta_N-\theta_T^\star)+r_T$, where $\|r_T\|=o_p(T^{-1/2})$. The maintained convergence $\mathcal J_T\to\mathcal J$ and Slutsky's theorem give the displayed delta-method limit.
\end{proof}

\begin{lemma}\label{lem:alpha-sum}
Let $(R_t)$ and $(X_t)$ be two independent two-sided processes with mixing coefficients $\alpha_R(h)$ and $\alpha_X(h)$. If $Z_t:=(R_t,X_t)$, then $\alpha_Z(h)\le \alpha_R(h)+\alpha_X(h)$ for $h\ge 0$.
\end{lemma}

\begin{proof}[Proof of Lemma \ref{lem:alpha-sum}]
Let $\mathcal R_0:=\sigma(R_t:t\le0)$, $\mathcal R_h:=\sigma(R_t:t\ge h)$, $\mathcal X_0:=\sigma(X_t:t\le0)$, and $\mathcal X_h:=\sigma(X_t:t\ge h)$. The past and future $Z$-fields are $\mathcal R_0\vee\mathcal X_0$ and $\mathcal R_h\vee\mathcal X_h$. First take rectangles $A=A_R\cap A_X$ and $B=B_R\cap B_X$, with $A_R\in\mathcal R_0$, $B_R\in\mathcal R_h$, $A_X\in\mathcal X_0$, and $B_X\in\mathcal X_h$. Independence of the full $R$-field and the full $X$-field gives
\[
\begin{aligned}
&\mathbb P(A\cap B)-\mathbb P(A)\mathbb P(B) \\
&\quad=\mathbb P(A_R\cap B_R)\mathbb P(A_X\cap B_X)
      -\mathbb P(A_R)\mathbb P(B_R)\mathbb P(A_X)\mathbb P(B_X).
\end{aligned}
\]
Adding and subtracting $\mathbb P(A_R)\mathbb P(B_R)\mathbb P(A_X\cap B_X)$ gives
\[
\begin{aligned}
&\left|\mathbb P(A\cap B)-\mathbb P(A)\mathbb P(B)\right| \\
&\quad\le
\left|\mathbb P(A_R\cap B_R)-\mathbb P(A_R)\mathbb P(B_R)\right|\mathbb P(A_X\cap B_X) \\
&\qquad +\mathbb P(A_R)\mathbb P(B_R)
\left|\mathbb P(A_X\cap B_X)-\mathbb P(A_X)\mathbb P(B_X)\right| \\
&\quad\le \alpha_R(h)+\alpha_X(h).
\end{aligned}
\]
For fixed rectangular $B$, the class of $A\in\mathcal R_0\vee\mathcal X_0$ satisfying this bound is a monotone class containing the rectangle generator; hence it is the whole past field. Repeating the same argument with $A$ fixed extends the bound to all $B\in\mathcal R_h\vee\mathcal X_h$. Taking suprema over $A$ and $B$ proves $\alpha_Z(h)\le\alpha_R(h)+\alpha_X(h)$.
\end{proof}

\begin{lemma}\label{lem:mixing-CLT}
Let $(Z_{T,t})_{1\le t\le n_T}$ be centered triangular arrays. Suppose:
\begin{enumerate}[label=(\roman*)]
\item there exists a sequence $\alpha(h)$ with $\sum_{h\ge 1}\alpha(h)^{\delta/(2+\delta)}<\infty$ for some $\delta>0$ such that, for every $T$, the array $(Z_{T,t})_t$ is $\alpha$-mixing with coefficients bounded by $\alpha(h)$;
\item $\sup_{T,t}\mathbb{E}|Z_{T,t}|^{2+\delta}<\infty$;
\item for every $\varepsilon>0$, $\sum_{t=1}^{n_T}\mathbb{E}\left[Z_{T,t}^2\,\1\{|Z_{T,t}|>\varepsilon\}\right]\to 0$;
\item $\Var\left(\sum_{t=1}^{n_T}Z_{T,t}\right)\to \sigma^2\in[0,\infty)$.
\end{enumerate}
Then $\sum_{t=1}^{n_T} Z_{T,t}\Rightarrow \mathcal{N}(0,\sigma^2)$.
\end{lemma}

\begin{proof}[Proof of Lemma \ref{lem:mixing-CLT}]
This is the triangular-array central limit theorem stated in \citet[][Thm.~2.1]{Peligrad1996}; see also \citet{Utev1990} and \citet[][Thm.~19.1]{Billingsley1999}. Conditions (i) and (ii) give the uniform mixing and moment bounds required by that theorem, condition (iii) is its Lindeberg condition, and condition (iv) fixes the limiting variance. Applying the cited theorem gives the scalar normal limit, including the degenerate case $\sigma^2=0$.
\end{proof}

\begin{proof}[Proof of Proposition \ref{prop:primitive-verification}]
Full observation is maintained throughout, so $N=T$ and no inclusion weights appear. The proof verifies the stochastic high-level conditions not already imposed in the proposition.

First consider the centered-moment CLT in Assumption~\ref{ass:dependence}. Assumption~\ref{ass:primitive-app}(i) gives the fixed-lag covariance limits, the summable covariance envelope, and the uniform $2+\delta$ moment bound for $(e_{T,t})$. It remains to prove the vector central limit theorem. Fix a unit vector $u$ and set $Z_{T,t}:=T^{-1/2}u^\top e_{T,t}$. The scalar array is centered and has the same strong-mixing coefficients as $(u^\top e_{T,t})$. For every $\varepsilon>0$,
\[
\sum_{t=1}^T\mathbb E_T[Z_{T,t}^2\1\{|Z_{T,t}|>\varepsilon\}]
\le \varepsilon^{-\delta}\sum_{t=1}^T\mathbb E_T|Z_{T,t}|^{2+\delta}
\le C\varepsilon^{-\delta}T^{-\delta/2}\to0,
\]
so the Lindeberg condition in Lemma~\ref{lem:mixing-CLT} holds. Collecting covariances by lag gives
\[
\Var_T\!\left(\sum_{t=1}^TZ_{T,t}\right)
=\sum_{\ell=-(T-1)}^{T-1}\left(1-\frac{|\ell|}{T}\right)u^\top\Gamma_{e,T}(\ell)u,
\]
where $\Gamma_{e,T}(\ell):=T^{-1}\sum_{t=|\ell|+1}^T\mathbb E_T[e_{T,t}e_{T,t-|\ell|}^\top]$ with the transpose convention for negative lags. The fixed-lag limits and the summable envelope imply convergence of this variance to $u^\top\Omega_Ru$ by dominated convergence for series: first truncate the lag sum at a fixed $M$, pass to the fixed-lag limits, and then let $M\to\infty$ using the envelope. Lemma~\ref{lem:mixing-CLT} gives $T^{-1/2}\sum_tu^\top e_{T,t}\Rightarrow\mathcal N(0,u^\top\Omega_Ru)$. Since $u$ was arbitrary, Cram\'er--Wold gives $T^{-1/2}\sum_te_{T,t}\Rightarrow\mathcal N(0,\Omega_R)$.

Next verify the uniform laws in Assumption~\ref{ass:uniform-lln}. Let $f_{T,t}(\theta)$ be a scalar coordinate of either $g_t(W_t,\theta)-\mathbb E_T[g_t(W_t,\theta)]$ or $J_{T,t}(\theta)-\mathbb E_T[J_{T,t}(\theta)]$. Assumption~\ref{ass:primitive-app}(ii) implies, uniformly in $\theta$,
\[
\Var_T\!\left(T^{-1}\sum_{t=1}^T f_{T,t}(\theta)\right)
\le T^{-1}\sup_{\vartheta\in\Theta}\frac1T\sum_{h=-(T-1)}^{T-1}\sum_{t=|h|+1}^T
\left|\Cov_T(f_{T,t}(\vartheta),f_{T,t-|h|}(\vartheta))\right|=O(T^{-1}).
\]
Chebyshev's inequality gives pointwise convergence for each fixed $\theta$, and a union bound gives convergence on any finite net. For the moment process, Assumption~\ref{ass:smoothness} gives
\[
\begin{aligned}
&\left\|T^{-1}\sum_t\big(g_t(W_t,\theta)-g_t(W_t,\theta')
-\mathbb E_T[g_t(W_t,\theta)-g_t(W_t,\theta')]\big)\right\| \\
&\quad\le \left(T^{-1}\sum_tb_{g,T,t}+T^{-1}\sum_t\mathbb E_T b_{g,T,t}\right)\|\theta-\theta'\|.
\end{aligned}
\]
The multiplier on the right is $O_p(1)$. The same argument with the Jacobian Lipschitz envelope $b_{g,1,T,t}$ gives stochastic equicontinuity of the centered Jacobian process. Compactness of $\Theta$ now upgrades pointwise convergence to uniform convergence. Explicitly, if $U_T(\theta)$ is either centered process and $(\theta_j)$ is an $\eta$-net, then on the event that its Lipschitz multiplier $L_T$ is at most $M$, $\sup_{\theta\in\Theta}\|U_T(\theta)\|\le\max_j\|U_T(\theta_j)\|+M\eta$. Choosing $M$ with high probability, then $\eta$ smaller than the target tolerance divided by $M$, and then passing to the finite-net limit gives the uniform law. The moment, Jacobian, and parameter dimensions are fixed, so coordinatewise uniform convergence implies the norm convergence in Assumption~\ref{ass:uniform-lln}.

It remains to verify Assumption~\ref{ass:hac-regularity}. For a fixed lag $\ell$ and coordinate pair $(i,j)$, Assumption~\ref{ass:primitive-app}(iii) gives the deterministic product mean $\bar\Gamma_{s,T}(\ell)\to\Gamma_{g,\mathrm{hac}}(\ell)$ and the variance bound $\Var_T(T^{-1}\sum_{t=\ell+1}^TY_{T,t}^{ij}(\ell))=O(T^{-1})$. Therefore $\widehat\Gamma_s(\ell)\to_p\Gamma_{g,\mathrm{hac}}(\ell)$ for every fixed $\ell$. For the growing-window stochastic term, write $A_{T,\ell}:=\widehat\Gamma_s(\ell)-\bar\Gamma_{s,T}(\ell)$ and extend the definition to negative lags by transposition. The componentwise covariance-sum bound gives $\mathbb E_T\|A_{T,\ell}\|^2\le C/T$ uniformly over $0\le\ell\le L_T$. Since $K$ is bounded and dimensions are fixed,
\[
\begin{aligned}
\mathbb E_T\left\|\sum_{1\le |\ell|\le L_T}K(|\ell|/L_T)A_{T,\ell}\right\|^2
&\le C L_T\sum_{1\le |\ell|\le L_T}\mathbb E_T\|A_{T,\ell}\|^2 \\
&=O(L_T^2/T)=o(1).
\end{aligned}
\]
The deterministic lag-window bias imposed in Assumption~\ref{ass:primitive-app}(iii) then gives the growing-window clause for the infeasible centered contribution $s_{T,t}$.

The feasible plug-in clause uses the first-order rate from Theorem~\ref{thm:AN}, which is available here because the preceding paragraphs have verified Assumptions~\ref{ass:dependence} and~\ref{ass:uniform-lln}, and Theorem~\ref{thm:AN} does not require Assumption~\ref{ass:hac-regularity}. The theorem and Assumption~\ref{ass:local-correct-specification} give $\|\widehat\theta_N-\theta_T^\star\|=O_p(T^{-1/2})$, $\|g_N(\theta_T^\star)-\bar m_T(\theta_T^\star)\|=O_p(T^{-1/2})$, and $\|g_N(\widehat\theta_N)-\bar m_T(\theta_T^\star)\|=O_p(T^{-1/2})$. Let $\widehat s_t:=g_t(W_t,\widehat\theta_N)-g_N(\widehat\theta_N)$, $d_t:=\widehat s_t-s_{T,t}$, $\Delta_T:=\|\widehat\theta_N-\theta_T^\star\|$, and $M_T:=\|g_N(\widehat\theta_N)-\bar m_T(\theta_T^\star)\|$. Assumption~\ref{ass:smoothness} gives \(\|d_t\|\le b_{g,T,t}\Delta_T+M_T\) and \(\Delta_T+M_T=O_p(T^{-1/2})\).
Also $\|s_{T,t}\|\le\bar B_{T,t}$ and $\|\widehat s_t\|\le\bar B_{T,t}+\|d_t\|$. Adding and subtracting $\widehat s_ts_{T,t-\ell}^\top$ gives, uniformly over $0\le\ell\le L_T$,
\[
\begin{aligned}
&\frac1T\sum_{t=\ell+1}^T\|\widehat s_t\widehat s_{t-\ell}^\top-s_{T,t}s_{T,t-\ell}^\top\| \\
&\quad\le \frac1T\sum_{t=\ell+1}^T\left(\|d_t\|\bar B_{T,t-\ell}+\bar B_{T,t}\|d_{t-\ell}\|+\|d_t\|\|d_{t-\ell}\|\right).
\end{aligned}
\]
Let $U_{T,\ell}$ denote the right side of the preceding bound. The inequality $\|d_t\|\le b_{g,T,t}\Delta_T+M_T$ gives
\[
U_{T,\ell}\le(\Delta_T+M_T)A_{1,T,\ell}+ (\Delta_T+M_T)^2A_{2,T,\ell},
\]
where $A_{1,T,\ell}:=T^{-1}\sum_{t=\ell+1}^T(b_{g,T,t}\bar B_{T,t-\ell}+\bar B_{T,t}b_{g,T,t-\ell}+\bar B_{T,t}+\bar B_{T,t-\ell})$ and $A_{2,T,\ell}:=T^{-1}\sum_{t=\ell+1}^T(1+b_{g,T,t})(1+b_{g,T,t-\ell})$. Assumption~\ref{ass:primitive-app}(iv), $T^{-1}\sum_t\bar B_{T,t}=O_p(1)$, and Cauchy--Schwarz with the envelope moment bounds give $\sup_{\ell\le L_T}A_{1,T,\ell}=O_p(1)$ and $\sup_{\ell\le L_T}A_{2,T,\ell}=O_p(1)$. Therefore $\sup_{0\le\ell\le L_T}U_{T,\ell}=O_p(T^{-1/2})$. Since $K$ is bounded, with $C_K:=\sup_x|K(x)|$, the lag-window replacement is bounded by $C_K\sum_{|\ell|\le L_T}U_{T,|\ell|}=O_p(L_T/\sqrt T)=o_p(1)$. This verifies the feasible part of Assumption~\ref{ass:hac-regularity}. The final conclusions follow by applying Theorem~\ref{thm:AN}, Proposition~\ref{lem:fixed-env-decomp}, Theorem~\ref{thm:hac}, and Corollary~\ref{cor:scalar-conservative} under the additional assumptions stated in the proposition.
\end{proof}

\begin{proof}[Proof of Corollary~\ref{cor:lp-primitive}]
Fix the finite horizon set $\mathcal H$ and let $q<\infty$ be large enough that each single-date stacked LP moment and Jacobian is a function of primitive variables in a window of radius at most $q$ around the relevant date, after including the largest horizon in $\mathcal H$. For a fixed lag $\ell$, define $\mathcal W_t(\ell):=[t-q,t+q]\cup[t-\ell-q,t-\ell+q]$. Writing $P_t(\ell)$ for any scalar coordinate of $s_{T,t}s_{T,t-\ell}^\top$, the variable $P_t(\ell)$ is measurable with respect to primitive variables in $\mathcal W_t(\ell)$, and $P_{t-h}(\ell)$ is measurable with respect to $\mathcal W_{t-h}(\ell)$. Hence $\Cov(P_t(\ell),P_{t-h}(\ell))=0$ whenever $\mathcal W_t(\ell)\cap\mathcal W_{t-h}(\ell)=\emptyset$. The intersection can be nonempty only if $|h|\le2q$, $|h-\ell|\le2q$, or $|h+\ell|\le2q$. Thus the number of relevant $h$ values is bounded by $12q+3$, uniformly in $T$ and $\ell$, and the maintained moment bound gives $T^{-1}\sum_h\sum_t|\Cov(P_t(\ell),P_{t-h}(\ell))|\le C(12q+3)$ uniformly over the retained lags. The same finite-window argument applies to centered moment and Jacobian coordinates. Thus the covariance envelopes in Assumption~\ref{ass:primitive-app} have the finite-support form required for the stochastic innovation and product-array parts.

The uniform $4+\delta$ bound on the primitive products implies the $2+\delta$ moment bound for $e_{T,t}$, the second-moment envelope bound for the Jacobian, and the product-array moment bounds used in Assumptions~\ref{ass:smoothness} and~\ref{ass:primitive-app}. The fixed-lag Ces\`aro limits of the primitive second moments give the fixed-lag covariance limits for the moment and for the lag-product arrays. When the HAC conclusion is invoked, Assumption~\ref{ass:mean-path} supplies absolute summability for the centered mean path. For the deterministic lag-window bias in Assumption~\ref{ass:primitive-app}(iii), fix $M$ and control the finite set $|\ell|\le M$ by fixed-lag convergence; the remaining tail is bounded, after taking limits, by the innovation finite-support bound plus $\sum_{|\ell|>M}\|\Gamma_{g,\mu}(\ell)\|$, which is arbitrarily small for large $M$.

The LP Jacobian is block diagonal across horizons, with $h$th block $-T^{-1}\sum_t\mathbb E_T[\psi_t\psi_t^\top]$ in the population limit. The nonsingular Gram limit $Q_\psi\succ0$ therefore gives the full-column-rank limit for the stacked Jacobian. In the just-identified normal-equation case, the sample-period normal equations give $\bar m_T(\theta_T^\star)=0$ exactly; the corollary also allows the stated local-correct-specification condition, which is the weaker requirement needed by Theorem~\ref{thm:AN}. These arguments verify the moment-level stochastic, smoothness, product-array, and local-correct-specification conditions. The GMM weighting, moving-estimand separation, and interior/local-solution requirements are the additional assumptions imposed in the corollary rather than consequences of finite-window dependence alone. Under those additional assumptions, Theorem~\ref{thm:AN} applies to the stacked LP estimator. If Assumption~\ref{ass:mean-path} and the kernel and bandwidth conditions also hold, Proposition~\ref{prop:primitive-verification} verifies Assumption~\ref{ass:hac-regularity}; Proposition~\ref{lem:fixed-env-decomp} and Theorem~\ref{thm:hac} then give the stated HAC conclusion.
\end{proof}

\begin{proof}[Proof of Corollary~\ref{cor:var-primitive}]
For fixed $p$, $n$, and $m$, the VAR moment $g_t^{\mathrm{VAR}}(\phi)=\operatorname{vec}(X_tu_t(\phi)^\top)$ is affine in $\phi$ after conditioning on $X_t$, and its Jacobian is a finite-dimensional array of products of components of $X_t$. The paragraph preceding the corollary imposes Assumption~\ref{ass:primitive-app} on the VAR moment, Jacobian, and lag-product arrays. Proposition~\ref{prop:primitive-verification} therefore supplies the centered-moment CLT, the uniform moment and Jacobian laws, and the lag-window HAC law once the remaining GMM conditions in the corollary are imposed.

To identify the Jacobian limit, write the residual as $u_t(\phi)=y_t-\mathcal R_t\phi$, where $\mathcal R_t$ is the fixed linear map that converts the vectorized VAR coefficient into the fitted value. Since $d u_t(\phi)=-\mathcal R_t\,d\phi$, the differential satisfies $d\operatorname{vec}(X_tu_t(\phi)^\top)=\operatorname{vec}(X_t\,d u_t(\phi)^\top)=-(I_n\otimes X_t)\mathcal R_t\,d\phi$. Under the maintained coefficient ordering, $\mathcal R_t\,d\phi$ is the differential of the fitted value and, writing $d\phi=\operatorname{vec}(d\Phi)$, $(I_n\otimes X_t)\mathcal R_t\,d\phi=\operatorname{vec}(X_tX_t^\top d\Phi)=(I_n\otimes X_tX_t^\top)d\phi$, up to the fixed vec-permutation convention. Averaging and taking expectations gives $\bar G_T=-T^{-1}\sum_t\mathbb E_T[(I_n\otimes X_t)\mathcal R_t]=-I_n\otimes T^{-1}\sum_t\mathbb E_T[X_tX_t^\top]$. The limit $Q_X\succ0$ implies that the limiting Jacobian has full column rank. The sample-period VAR normal equations give local correct specification in the just-identified normal-equation case; otherwise the local-correct-specification condition imposed in the corollary is exactly the high-level condition needed to remove a first-order mean component. The primitive product restrictions verify the moment-level stochastic, smoothness, and product-array requirements, while the GMM weighting, moving-estimand separation, and interior/local-solution requirements remain the separate assumptions stated in the corollary. Under those additional assumptions, Theorem~\ref{thm:AN} gives the design-based asymptotic distribution of the VAR coefficient estimator. If the centered VAR mean path also satisfies Assumption~\ref{ass:mean-path}, Proposition~\ref{lem:fixed-env-decomp} identifies the conservative HAC variance limit and Theorem~\ref{thm:hac} gives feasible HAC consistency for that limit.
\end{proof}

The proofs below verify the four conditions of Lemma~\ref{lem:mixing-CLT} for the incomplete-observation CLTs. These appendix-only sampling lemmas use a primitive mixing route in addition to the high-level conditions in the main text: for each unit vector $u$, the scalar arrays $u^\top g_t(W_t,\theta_T^\star)$ and $u^\top e_{T,t}$ have $\alpha$-mixing coefficients dominated by a common envelope $\alpha(h)$ satisfying $\sum_h\alpha(h)^{\delta/(2+\delta)}<\infty$. This route is used only for incomplete observation; the fully observed main theorem uses the innovation CLT in Assumption~\ref{ass:dependence}(ii) directly.

\begin{lemma}\label{lem:schemeA}
Let $Y_{T,t}:=(R_{T,t}-\rho_T)g_t(W_t,\theta_T^\star)\in\mathbb R^d$. Under Assumptions~\ref{ass:dependence}, \ref{ass:smoothness}, \ref{ass:mean-path}, \ref{ass:local-correct-specification}, and~\ref{ass:sampling-app}, and under the random-sampling scheme and the appendix-only mixing route just stated,
\[
T^{-1/2}\sum_{t=1}^T Y_{T,t}\Rightarrow\mathcal N(0,\Sigma_A),
\qquad
\Sigma_A:=\sum_{\ell\in\mathbb Z}\gamma_R(\ell)\Gamma_{g,\mathrm{hac}}(\ell).
\]
The series defining $\Sigma_A$ is absolutely convergent.
\end{lemma}

\begin{proof}[Proof of Lemma \ref{lem:schemeA}]
Fix $u\in\mathbb R^d$ with $\|u\|=1$, write $g_{T,t}:=g_t(W_t,\theta_T^\star)$, and set $Z_{T,t}:=T^{-1/2}u^\top(R_{T,t}-\rho_T)g_{T,t}$. The variables are centered because the inclusion process is independent of the moment process and $\mathbb E[R_{T,t}-\rho_T]=0$. It remains to verify Lemma~\ref{lem:mixing-CLT} for $\sum_tZ_{T,t}$.

Let $Z'_t:=(R_{T,t},g_{T,t})$. The random-sampling scheme makes the inclusion process independent of the design-shock process, and Lemma~\ref{lem:alpha-sum} gives $\alpha_{Z'}(h)\le\alpha_R(h)+\alpha(h)$. With $\eta:=\delta/(2+\delta)\in(0,1)$, $(a+b)^\eta\le a^\eta+b^\eta$, so the scalar array $(Z_{T,t})$ has mixing coefficients with summable $\eta$ powers. For the moment and Lindeberg conditions, put $p:=2+\delta$. Since $|R_{T,t}-\rho_T|\le1$ and Assumption~\ref{ass:smoothness} gives a uniformly bounded $p$th moment for $g_{T,t}$, \(\sup_{T,t}\mathbb E_T|Z_{T,t}|^p\le C\).
For every $\varepsilon>0$,
\[
\sum_{t=1}^T\mathbb E_T\big[Z_{T,t}^2\1\{|Z_{T,t}|>\varepsilon\}\big]
\le \varepsilon^{-\delta}\sum_{t=1}^T\mathbb E_T|Z_{T,t}|^{2+\delta}
\le C\varepsilon^{-\delta}T^{-\delta/2}\to0.
\]

It remains to identify the variance. For $\ell\ge0$, define $M_{g,T}(\ell):=T^{-1}\sum_{t=\ell+1}^T\mathbb E_T[g_{T,t}g_{T,t-\ell}^\top]$ and set $M_{g,T}(-\ell):=M_{g,T}(\ell)^\top$. Independence of $R$ and the moment process gives, for each fixed lag, $\Cov_T(Y_{T,t},Y_{T,t-|\ell|})=\gamma_{R,T}(\ell)\mathbb E_T[g_{T,t}g_{T,t-|\ell|}^\top]$, with the transpose convention for negative lags. Hence
\[
\Var_T\!\left(\sum_{t=1}^TZ_{T,t}\right)
=\sum_{\ell=-(T-1)}^{T-1}\gamma_{R,T}(\ell)u^\top M_{g,T}(\ell)u.
\]
For fixed $\ell$, $M_{g,T}(\ell)$ has the same limit as the centered HAC product. Let $\bar m_T:=\bar m_T(\theta_T^\star)$ and $\bar M_{g,T}(\ell):=T^{-1}\sum_{t=\ell+1}^T\mathbb E_T[(g_{T,t}-\bar m_T)(g_{T,t-\ell}-\bar m_T)^\top]$. If $A_{T,\ell}:=T^{-1}\sum_{t=\ell+1}^T\mathbb E_Tg_{T,t}$ and $B_{T,\ell}:=T^{-1}\sum_{t=\ell+1}^T\mathbb E_Tg_{T,t-\ell}$, then $M_{g,T}(\ell)-\bar M_{g,T}(\ell)=A_{T,\ell}\bar m_T^\top+\bar m_TB_{T,\ell}^\top-\bar m_T\bar m_T^\top$. The fixed-lag averages $A_{T,\ell}$ and $B_{T,\ell}$ are $O(1)$ by the envelope, and Assumption~\ref{ass:local-correct-specification} gives $\|\bar m_T\|=o(T^{-1/2})$, so the difference is $o(1)$. Therefore $M_{g,T}(\ell)\to\Gamma_{g,\mathrm{hac}}(\ell)$.

The dominated-convergence passage over lags is justified as follows. Assumption~\ref{ass:sampling-app} gives $|\gamma_{R,T}(\ell)|\le b_R(\ell)$ with $\sum_\ell b_R(\ell)<\infty$, while Cauchy--Schwarz and the moment envelope give $\sup_{T,\ell}\|M_{g,T}(\ell)\|<\infty$. Truncate the variance sum at $|\ell|\le L$, pass to fixed-lag limits, and then let $L\to\infty$. The variance converges to $u^\top\Sigma_Au$. Lemma~\ref{lem:mixing-CLT} gives the scalar normal limit, and Cram\'er--Wold gives the vector limit. Absolute convergence of $\Sigma_A$ follows from the same summable bound.
\end{proof}

The next lemma is a generic weighted CLT used for deterministic sampling windows and for the full-sample centered innovation term. Let $S_T:=\sum_{t=1}^T a_{T,t}X_{T,t}$, where $X_{T,t}\in\mathbb R^d$ is centered, has uniformly bounded $2+\delta$ moments, and satisfies the appendix-only mixing envelope stated before Lemma~\ref{lem:schemeA}. For deterministic weights, assume $\max_{t\le T}|a_{T,t}|\to0$ and $\sum_{t=1}^T a_{T,t}^2\to1$. For each $\ell\in\mathbb Z$, define \(C_T(\ell):=\sum_{t=|\ell|+1}^T a_{T,t}a_{T,t-|\ell|}\Cov_T(X_{T,t},X_{T,t-|\ell|})\), with $C_T(-\ell):=C_T(\ell)^\top$ for $\ell>0$. Suppose that, for each fixed $\ell$, $C_T(\ell)\to\kappa(\ell)\Gamma(\ell)$, and that \(\lim_{L\to\infty}\limsup_{T\to\infty}\sum_{|\ell|>L}\|C_T(\ell)\|=0\).
These two covariance conditions imply absolute convergence of $\sum_{\ell\in\mathbb Z}\kappa(\ell)\Gamma(\ell)$. If the weights are random, the same conditions are imposed conditionally on the weight sigma field: the maximum-weight, sum-of-squares, fixed-lag weighted-covariance, and tail conditions hold in probability, with $\sup_T\mathbb E\sum_ta_{T,t}^2<\infty$ and $\mathbb E\sum_t|a_{T,t}|^{2+\delta}\to0$. The deterministic-window application uses deterministic weights $R_{T,t}/\sqrt{T\rho_T}$; its fixed-lag weighted-covariance condition is the selected-window covariance analogue of the pair-frequency condition.

\begin{lemma}\label{lem:schemeB}
Suppose the weighted triangular array introduced in the preceding paragraph satisfies the stated uniform $(2+\delta)$ moment bound, appendix-only mixing envelope, deterministic or random weight conditions, fixed-lag weighted-covariance convergence, and covariance-tail condition. Then \(S_T\Rightarrow\mathcal N(0,\Sigma_B)\), where \(\Sigma_B:=\sum_{\ell\in\mathbb Z}\kappa(\ell)\Gamma(\ell)\).
\end{lemma}

\begin{proof}[Proof of Lemma \ref{lem:schemeB}]
Fix $u\in\mathbb R^d$ with $\|u\|=1$, set $Y_{T,t}:=u^\top X_{T,t}$ and $Z_{T,t}:=a_{T,t}Y_{T,t}$. First take the weights to be deterministic. Multiplication by deterministic weights does not enlarge sigma fields, so the scalar array $(Z_{T,t})$ inherits the common mixing envelope. With $p:=2+\delta$, the uniform moment bound on $X_{T,t}$ and boundedness of $\sum_ta_{T,t}^2$ imply $\sup_{T,t}\mathbb E_T|Z_{T,t}|^p<\infty$. The Lindeberg condition follows from
\[
\sum_{t=1}^T\mathbb E_T\big[Z_{T,t}^2\1\{|Z_{T,t}|>\varepsilon\}\big]
\le \varepsilon^{-\delta}\sum_{t=1}^T |a_{T,t}|^{2+\delta}\mathbb E_T|Y_{T,t}|^{2+\delta}
\le C\varepsilon^{-\delta}(\max_t|a_{T,t}|^\delta)\sum_{t=1}^Ta_{T,t}^2=o(1).
\]
The variance equals \(\Var_T\!\left(\sum_{t=1}^TZ_{T,t}\right)=\sum_{\ell=-(T-1)}^{T-1}u^\top C_T(\ell)u\).
For fixed $L$, the fixed-lag covariance condition gives convergence of the partial sum over $|\ell|\le L$ to $\sum_{|\ell|\le L}\kappa(\ell)u^\top\Gamma(\ell)u$. The weighted tail condition makes the remaining lags uniformly negligible, and also makes the limiting tail negligible. Thus the variance converges to $\sigma^2(u):=\sum_{\ell\in\mathbb Z}\kappa(\ell)u^\top\Gamma(\ell)u$. Lemma~\ref{lem:mixing-CLT} gives $\sum_tZ_{T,t}\Rightarrow\mathcal N(0,\sigma^2(u))$, including the degenerate case. Cram\'er--Wold gives the vector conclusion.

For random weights satisfying the conditional version of the restrictions, condition on the weight sigma field. Along every subsequence there is a further subsequence on which the weight and covariance conditions hold almost surely. The deterministic conditional argument then gives the scalar CLT along that further subsequence. For bounded Lipschitz $f$, $\mathbb E f(S_T)-\mathbb E f(Z_{\Sigma_B})=\mathbb E[\mathbb E(f(S_T)-f(Z_{\Sigma_B})\mid\mathcal A_T^w)]$, where $\mathcal A_T^w$ is the weight sigma field and $Z_{\Sigma_B}\sim\mathcal N(0,\Sigma_B)$. The inner term tends to zero along the almost-sure subsequence, and boundedness of $f$ gives dominated convergence. The subsequence principle gives the unconditional conclusion.
\end{proof}

Let $X_{T,t}:=g_t(W_t,\theta_T^\star)-\mu_t(\theta_T^\star)$. For each fixed $\ell\ge0$, set
\[
\Gamma_e(\ell):=\lim_T T^{-1}\sum_{t=\ell+1}^T\Cov_T(X_{T,t},X_{T,t-\ell}),
\qquad
\Gamma_e(-\ell):=\Gamma_e(\ell)^\top.
\]
Under Assumption~\ref{ass:mean-path}, $\Gamma_e(\ell)=\Gamma_{g,\mathrm{hac}}(\ell)-\Gamma_{g,\mu}(\ell)$ for every fixed $\ell$.
\begin{lemma}\label{lem:CLT-mom}
Let Assumptions~\ref{ass:dependence}, \ref{ass:mean-path}, and~\ref{ass:local-correct-specification} hold for $g_t(W_t,\theta_T^\star)$. For incomplete observation, also impose Assumptions~\ref{ass:smoothness} and~\ref{ass:sampling-app} and the appendix-only mixing, weight, weighted-covariance, and covariance-tail conditions used in Lemmas~\ref{lem:schemeA}--\ref{lem:schemeB}. Under the random-sampling scheme, \(\sqrt N\big(g_N(\theta_T^\star)-\bar m_T(\theta_T^\star)\big)\Rightarrow\mathcal N(0,\Omega_R^{\mathrm{obs}})\), where $\Omega_R^{\mathrm{obs}}$ is defined in \eqref{eq:OmegaR_sampling}. Under the deterministic-window scheme, if $m_T\to m$, then
\[
\sqrt N\big(g_N(\theta_T^\star)-\bar m_T(\theta_T^\star)\big)
\Rightarrow\mathcal N\left(m,\sum_{\ell\in\mathbb Z}\kappa_R(\ell)\Gamma_e(\ell)\right).
\]
If $R_{T,t}=1$ for all $t$ and $T$, both conclusions reduce to the fully observed sample limit with covariance $\sum_{\ell\in\mathbb Z}\Gamma_e(\ell)$.
\end{lemma}

\begin{proof}[Proof of Lemma \ref{lem:CLT-mom}]
Suppress the argument $\theta_T^\star$ and write $g_t:=g_t(W_t,\theta_T^\star)$, $\mu_t:=\mu_{T,t}(\theta_T^\star)$, $\bar m_T:=T^{-1}\sum_t\mu_t$, and $X_t:=g_t-\mu_t$. If the sample is fully observed, then $N=T$ and $g_N(\theta_T^\star)-\bar m_T=T^{-1}\sum_tX_t$, so Assumption~\ref{ass:dependence}(ii) gives the fully observed limit. The remaining argument treats incomplete observation.

Lemma~\ref{lem:ULLN} gives the identity
\[
g_N(\theta_T^\star)-\bar m_T
=(T\rho_T)^{-1}\sum_{t=1}^T(R_{T,t}-\rho_T)g_t+T^{-1}\sum_{t=1}^TX_t+r_T(\theta_T^\star).
\]
The denominator term is negligible at the $\sqrt N$ scale. Under deterministic-window sampling, $N=T\rho_T$ exactly and $r_T(\theta_T^\star)=0$. Under random sampling, write $\sqrt N r_T(\theta_T^\star)=A_TB_T$, where $A_T:=(T\rho_T-N)/(T\rho_T)$ and $B_T:=N^{-1/2}\sum_tR_{T,t}g_t$. The inclusion covariance envelope gives $\Var(N)=O(T)$, so $A_T=O_p(T^{-1/2})$. To bound $B_T$, decompose
\[
N^{-1/2}\sum_tR_{T,t}g_t
=N^{-1/2}\sum_t(R_{T,t}-\rho_T)g_t
+\frac{\rho_T\sqrt T}{\sqrt N}T^{-1/2}\sum_tX_t
+\frac{\rho_TT}{\sqrt N}\bar m_T.
\]
The first term is $O_p(1)$ by Lemma~\ref{lem:schemeA}, the second is $O_p(1)$ by Lemma~\ref{lem:schemeB} with weights $T^{-1/2}$, and the third is $o_p(1)$ by $N/(T\rho_T)\to_p1$ and local correct specification. Hence $B_T=O_p(1)$ and $\sqrt N r_T(\theta_T^\star)=o_p(1)$.

Multiplying by $\sqrt N$ and using $N/(T\rho_T)\to_p1$ gives \(\sqrt N\big(g_N(\theta_T^\star)-\bar m_T\big)=Z_{1,T}+Z_{2,T}+o_p(1)\), where \(Z_{1,T}:=\sqrt{\rho_T}\,T^{-1/2}\sum_{t=1}^TX_t\) and \(Z_{2,T}:=(T\rho_T)^{-1/2}\sum_{t=1}^T(R_{T,t}-\rho_T)g_t\).
Under random sampling, Lemma~\ref{lem:schemeB} with weights $a_{T,t}=T^{-1/2}$ gives $T^{-1/2}\sum_tX_t\Rightarrow\mathcal N(0,\sum_\ell\Gamma_e(\ell))$, so $Z_{1,T}$ has limiting covariance $\rho\sum_\ell\Gamma_e(\ell)$. Lemma~\ref{lem:schemeA}, divided by $\sqrt{\rho_T}$, gives $Z_{2,T}\Rightarrow\mathcal N(0,\rho^{-1}\sum_\ell\gamma_R(\ell)\Gamma_{g,\mathrm{hac}}(\ell))$.

For joint convergence, fix conformable vectors $a$ and $b$ and apply Lemma~\ref{lem:mixing-CLT} to
\[
T^{-1/2}\sum_{t=1}^T
\left[
\sqrt{\rho_T}\,a^\top X_t
\,+\,\rho_T^{-1/2}b^\top(R_{T,t}-\rho_T)g_t
\right].
\]
Lemma~\ref{lem:alpha-sum} supplies the mixing bound for the joint array. The moment and Lindeberg conditions follow from the same envelope inequalities used in Lemmas~\ref{lem:schemeA} and~\ref{lem:schemeB}. The finite-$T$ cross-covariance between the two scalar components is zero: for all dates $s,t$, $\mathbb E[(R_{T,s}-\rho_T)a^\top X_t g_s^\top b]=\mathbb E[R_{T,s}-\rho_T]\mathbb E_T[a^\top X_t g_s^\top b]=0$ by independence. Therefore the limiting variance of every scalar combination has no cross term, and the limiting covariance of $Z_{1,T}+Z_{2,T}$ is
\[
\Sigma_A=\rho\sum_{\ell\in\mathbb Z}\Gamma_e(\ell)+\rho^{-1}\sum_{\ell\in\mathbb Z}\gamma_R(\ell)\Gamma_{g,\mathrm{hac}}(\ell).
\]
Using $\Gamma_e(\ell)=\Gamma_{g,\mathrm{hac}}(\ell)-\Gamma_{g,\mu}(\ell)$, the lag-zero coefficient satisfies $\rho\Gamma_e(0)+\rho^{-1}\gamma_R(0)\Gamma_{g,\mathrm{hac}}(0)=\Gamma_{g,\mathrm{hac}}(0)-\rho\Gamma_{g,\mu}(0)$ because $\gamma_R(0)=\rho(1-\rho)$. For $\ell\ne0$, the coefficient satisfies $\rho\Gamma_e(\ell)+\rho^{-1}\gamma_R(\ell)\Gamma_{g,\mathrm{hac}}(\ell)=\kappa_R(\ell)\Gamma_{g,\mathrm{hac}}(\ell)-\rho\Gamma_{g,\mu}(\ell)$ because $\kappa_R(\ell)=\rho+\gamma_R(\ell)/\rho$. Summing these identities over lags gives $\Gamma_{g,\mathrm{hac}}(0)+\sum_{\ell\ne0}\kappa_R(\ell)\Gamma_{g,\mathrm{hac}}(\ell)-\rho\sum_{\ell\in\mathbb Z}\Gamma_{g,\mu}(\ell)$, which is $\Omega_R^{\mathrm{obs}}$ in \eqref{eq:OmegaR_sampling}. Slutsky's theorem proves the random-sampling conclusion.

Under deterministic-window sampling, there is no denominator remainder, and the decomposition is exact after multiplication by $\sqrt N=\sqrt{T\rho_T}$:
\[
\sqrt N\big(g_N(\theta_T^\star)-\bar m_T\big)
=(T\rho_T)^{-1/2}\sum_{t=1}^TR_{T,t}X_t
+(T\rho_T)^{-1/2}\sum_{t=1}^T(R_{T,t}-\rho_T)\mu_t.
\]
The second term is $m_T$, while Lemma~\ref{lem:schemeB}, applied with deterministic weights $a_{T,t}=R_{T,t}/\sqrt{T\rho_T}$ and the selected-window weighted-covariance and tail conditions maintained in this lemma, gives convergence of the first term to a centered normal law with covariance $\sum_{\ell\in\mathbb Z}\kappa_R(\ell)\Gamma_e(\ell)$. Since $m_T\to m$, Slutsky's theorem gives the deterministic-window conclusion. If $R_{T,t}=1$ for all $t$, then $\rho_T=1$, $m_T=0$, and $\kappa_R(\ell)=1$, so the formula reduces to the fully observed covariance $\sum_\ell\Gamma_e(\ell)$.
\end{proof}

The next two lemmas separate the sampling mechanics from the lag-product law of large numbers based on observed pairs. The first records observed-pair frequency convergence, and the second proves the fixed-lag HAC limit over observed lag pairs and makes explicit the plug-in envelope condition used when $\widehat\theta_N$ replaces $\theta_T^\star$.

\subsection{Incomplete-observation proofs}

\begin{proof}[Proof of Theorem \ref{thm:det-GMM}]
With all probabilities conditional on the conditioning environment, first consider consistency. Assumption~\ref{ass:det-selected-gmm} gives $\sup_{\theta\in\Theta}\|g_N(\theta)-\bar m_T(\theta)\|\to_p0$, because $g_N(\theta)-\bar m_T(\theta)=[g_N(\theta)-M_{N,T}(\theta)]+[M_{N,T}(\theta)-\bar m_T(\theta)]$ uniformly in $\theta$. The envelope part of Assumption~\ref{ass:smoothness} implies $M_{m,T}:=\sup_{\theta\in\Theta}\|\bar m_T(\theta)\|=O(1)$, and Assumption~\ref{ass:gmm} gives $\widehat A_N\to_p A$. With $M_{\Delta,T}:=\sup_{\theta\in\Theta}\|g_N(\theta)-\bar m_T(\theta)\|$, the same expansion as in Theorem~\ref{thm:AN}(i) gives $\sup_{\theta\in\Theta}|g_N(\theta)^\top\widehat A_Ng_N(\theta)-Q_T(\theta)|\le M_{m,T}^2\|\widehat A_N-A\|+2M_{m,T}\|\widehat A_N\|M_{\Delta,T}+\|\widehat A_N\|M_{\Delta,T}^2=o_p(1)$. Since Assumption~\ref{ass:det-selected-gmm} takes $\widehat\theta_N$ to be a measurable global minimizer, the moving-estimand argmin argument used in Theorem~\ref{thm:AN}(i), together with Assumption~\ref{ass:estimand-separation}, gives $\widehat\theta_N-\theta_T^\star\to_p0$.

Now linearize on the high-probability event where the estimator lies in the interior neighborhood from Assumption~\ref{ass:gmm-interior}. Let $\widetilde G_N:=\widehat G_N(\widehat\theta_N)$ and \(\bar G_N:=\int_0^1\widehat G_N(\theta_T^\star+r(\widehat\theta_N-\theta_T^\star))\,dr\). If $\theta_N^\dagger-\theta_T^\star\to_p0$, then $\|\widehat G_N(\theta_N^\dagger)-G\|\le\sup_{\theta\in\Theta}\|\widehat G_N(\theta)-G_{N,T}(\theta)\|+\|G_{N,T}(\theta_N^\dagger)-G_{N,T}(\theta_T^\star)\|+\|G_{N,T}(\theta_T^\star)-G\|$. The middle term is bounded by $N^{-1}\sum_tR_{T,t}\mathbb E_Tb_{g,1,T,t}\|\theta_N^\dagger-\theta_T^\star\|=o_p(1)$, while the first and third terms are $o_p(1)$ by Assumption~\ref{ass:det-selected-gmm}. Applying this bound to $\theta_N^\dagger=\widehat\theta_N$ gives $\widetilde G_N\to_pG$. Applying it to points on the line segment gives $\|\bar G_N-G\|\le\int_0^1\|\widehat G_N(\theta_T^\star+r(\widehat\theta_N-\theta_T^\star))-G\|\,dr=o_p(1)$. The mean-value expansion gives $g_N(\widehat\theta_N)=g_N(\theta_T^\star)+\bar G_N(\widehat\theta_N-\theta_T^\star)$. The sample first-order condition is therefore \(0=\widetilde G_N^\top\widehat A_Ng_N(\theta_T^\star)+\widetilde G_N^\top\widehat A_N\bar G_N(\widehat\theta_N-\theta_T^\star)\).

The deterministic ball in Assumption~\ref{ass:gmm-interior} makes $\theta_T^\star$ an interior minimizer of $Q_T(\theta)=\bar m_T(\theta)^\top A\bar m_T(\theta)$ for all large $T$. Thus differentiability gives the population first-order condition $G_T^\top A\bar m_T(\theta_T^\star)=0$, where $G_T:=T^{-1}\sum_t\mathbb E_T[\nabla_\theta g_t(W_t,\theta_T^\star)]$. Adding and subtracting this identity yields
\[
\widetilde G_N^\top\widehat A_Ng_N(\theta_T^\star)
=G_T^\top A\big(g_N(\theta_T^\star)-\bar m_T(\theta_T^\star)\big)+R_{1T}+R_{2T},
\]
where $R_{1T}:=(\widetilde G_N^\top\widehat A_N-G_T^\top A)(g_N(\theta_T^\star)-\bar m_T(\theta_T^\star))$ and $R_{2T}:=(\widetilde G_N^\top\widehat A_N-G_T^\top A)\bar m_T(\theta_T^\star)$. The centered-moment CLT implies $g_N(\theta_T^\star)-\bar m_T(\theta_T^\star)=O_p(N^{-1/2})$. The Jacobian and weight convergence imply $\widetilde G_N^\top\widehat A_N-G_T^\top A=o_p(1)$, so $R_{1T}=o_p(N^{-1/2})$. Since $N/T\to\rho>0$ and Assumption~\ref{ass:local-correct-specification} gives $\sqrt T\|\bar m_T(\theta_T^\star)\|=o(1)$, $R_{2T}=o_p(N^{-1/2})$.

The matrix $\widetilde G_N^\top\widehat A_N\bar G_N$ converges in probability to $G^\top A G$, which is nonsingular because $A\succ0$ and $G$ has full column rank. Solving the first-order condition gives
\[
\sqrt N(\widehat\theta_N-\theta_T^\star)
=-(G^\top A G)^{-1}G^\top A\sqrt N\big(g_N(\theta_T^\star)-\bar m_T(\theta_T^\star)\big)+o_p(1).
\]
The centered-moment CLT and Slutsky's theorem give the asserted limit with mean $-Hm$ and covariance $H\Omega_BH^\top$.
\end{proof}

\begin{proof}[Proof of Corollary \ref{cor:det-centered}]
Theorem~\ref{thm:det-GMM} gives the limit $\mathcal N(-Hm,H\Omega_BH^\top)$. If $m_T\to0$, then $m=0$. If $R_{T,t}\equiv1$, then $\rho_T=1$ and $R_{T,t}-\rho_T=0$ for every $t$, so $m_T=0$ exactly.
\end{proof}

\begin{proof}[Proof of Proposition \ref{prop:window-centering}]
For a deterministic contiguous window, $\rho_T=N_T/T$ and $T\rho_T=N_T$. Since $\sum_{t=1}^T(R_{T,t}-\rho_T)=0$, replacing $\mu_{T,t}(\theta_T^\star)$ by $\tilde\mu_{T,t}$ leaves $m_T$ unchanged: \(m_T=(T\rho_T)^{-1/2}\sum_{t=1}^T(R_{T,t}-\rho_T)\tilde\mu_{T,t}\).
The centered path satisfies $\sum_{t=1}^T\tilde\mu_{T,t}=0$, and therefore
\[
\sum_{t=1}^T(R_{T,t}-\rho_T)\tilde\mu_{T,t}
=\sum_{t\in I_T}\tilde\mu_{T,t}-\rho_T\sum_{t=1}^T\tilde\mu_{T,t}
=\sum_{t\in I_T}\tilde\mu_{T,t}.
\]
The set $I_T$ is one of the contiguous intervals covered by the partial-sum bound, so $\|\sum_{t\in I_T}\tilde\mu_{T,t}\|=o(\sqrt T)$. Hence $\|m_T\|\le\rho_T^{-1/2}o(1)\to0$ because $\rho_T\to\rho>0$.
\end{proof}

\begin{lemma}\label{lem:kappa-rho-HAC}
Under Assumption~\ref{ass:sampling-app}, for each fixed $\ell\in\mathbb Z$, $\widehat\kappa_R(\ell)\xrightarrow{p}\kappa_R(\ell)$ and $N/T\xrightarrow{p}\rho$.
\end{lemma}

\begin{proof}[Proof of Lemma \ref{lem:kappa-rho-HAC}]
It is enough to treat $\ell\ge0$. For $\ell=0$, $D_T(0)=N$, so $\widehat\kappa_R(0)=1=\kappa_R(0)$, and $N/T\to_p\rho$ follows from $N/(T\rho_T)\to_p1$ and $\rho_T\to\rho$. Fix $\ell>0$. Since $\widehat\kappa_R(\ell)=(T^{-1}D_T(\ell))/(N/T)$, it remains to show $T^{-1}D_T(\ell)\to_p\rho\kappa_R(\ell)$.

Under random sampling, set $H_{T,t}(\ell):=R_{T,t}R_{T,t-\ell}$. This is a bounded fixed-lag transformation of the inclusion process. If $|h|>2\ell$, the sigma fields generated by $H_{T,t}(\ell)$ and $H_{T,t-h}(\ell)$ are separated by at least $|h|-2\ell$ dates. The covariance inequality for bounded strongly mixing variables gives $|\Cov(H_{T,t}(\ell),H_{T,t-h}(\ell))|\le C\alpha_R((|h|-2\ell)_+)$. Because $\alpha_R(h)\le\alpha_R(h)^{\delta/(2+\delta)}$ for all large $h$ and the latter sequence is summable, the lag covariance bound is summable for fixed $\ell$. Hence
\[
\Var\!\left(T^{-1}\sum_{t=\ell+1}^T\big(H_{T,t}(\ell)-\mathbb EH_{T,t}(\ell)\big)\right)
\le C T^{-1}\sum_{h\in\mathbb Z}\alpha_R((|h|-2\ell)_+)=O(T^{-1}).
\]
Chebyshev's inequality gives $T^{-1}\sum_t(H_{T,t}(\ell)-\mathbb EH_{T,t}(\ell))=o_p(1)$. Strict stationarity of the inclusion process gives $\mathbb EH_{T,t}(\ell)=\mathbb E[R_{T,t}R_{T,t-\ell}]=\rho_T^2+\gamma_{R,T}(\ell)$, so $T^{-1}\sum_{t=\ell+1}^T\mathbb EH_{T,t}(\ell)=(1-\ell/T)(\rho_T^2+\gamma_{R,T}(\ell))\to\rho^2+\gamma_R(\ell)=\rho\kappa_R(\ell)$. Under deterministic-window sampling, $R_{T,t}$ is fixed and $T^{-1}D_T(\ell)\to\rho\kappa_R(\ell)$ is exactly the pair-frequency condition in Assumption~\ref{ass:sampling-app}. Division by $N/T\to_p\rho>0$ gives $\widehat\kappa_R(\ell)\to_p\kappa_R(\ell)$.
\end{proof}

\begin{lemma}\label{lem:pair-LLN-hac}
Fix $\ell\ge0$ with $\kappa_R(\ell)>0$, set $\widehat s_t:=g_t(W_t,\widehat\theta_N)-g_N(\widehat\theta_N)$, and define
\[
\widetilde\Gamma_{\mathrm{hac}}(\ell):=D_T(\ell)^{-1}\sum_{t=\ell+1}^TR_{T,t}R_{T,t-\ell}\widehat s_t\widehat s_{t-\ell}^\top
\]
on the event $D_T(\ell)>0$, with any fixed value on the complement. Suppose Assumptions~\ref{ass:design-environment}, \ref{ass:sampling-app}, \ref{ass:pair-average-app}, and~\ref{ass:smoothness} hold, $\widehat\theta_N\to_p\theta_T^\star$, $g_N(\widehat\theta_N)=o_p(1)$, and the following fixed-lag product conditions hold. With
\[
Y_{T,t}(\ell):=g_t(W_t,\theta_T^\star)g_{t-\ell}(W_{t-\ell},\theta_T^\star)^\top-
\mathbb E_T[g_t(W_t,\theta_T^\star)g_{t-\ell}(W_{t-\ell},\theta_T^\star)^\top],
\]
assume
\[
\sup_T T^{-1}\sum_{h=-(T-1)}^{T-1}\sum_{t=|h|+1}^{T}
\left\|\Cov_T(\operatorname{vec}(Y_{T,t}(\ell)),\operatorname{vec}(Y_{T,t-|h|}(\ell)))\right\|<\infty,
\]
and
\[
T^{-1}\sum_{t=\ell+1}^T b_g(W_t)B_g(W_{t-\ell})=O_p(1),
\qquad
T^{-1}\sum_{t=\ell+1}^T B_g(W_t)b_g(W_{t-\ell})=O_p(1).
\]
Then $\widetilde\Gamma_{\mathrm{hac}}(\ell)\xrightarrow{p}\Gamma_{\mathrm{hac}}(\ell)$, with $\widetilde\Gamma_{\mathrm{hac}}(-\ell):=\widetilde\Gamma_{\mathrm{hac}}(\ell)^\top$ for $\ell>0$. If $\kappa_R(\ell)=0$, the lag is assigned zero weight in the HAC sum over observed lag pairs.
\end{lemma}

\begin{proof}[Proof of Lemma \ref{lem:pair-LLN-hac}]
All sums in this proof run over $t=\ell+1,\ldots,T$. Since $\kappa_R(\ell)>0$, Lemma~\ref{lem:kappa-rho-HAC} gives $D_T(\ell)/(T\rho_T)\to_p\kappa_R(\ell)$; hence $D_T(\ell)>0$ with probability approaching one and $T/D_T(\ell)=O_p(1)$. Write $Z_t:=R_{T,t}R_{T,t-\ell}$, \(a_t:=\mathbb E_T[g_t(W_t,\theta_T^\star)g_{t-\ell}(W_{t-\ell},\theta_T^\star)^\top]\), and \(Y_t:=Y_{T,t}(\ell)\), and abbreviate $g_t(\theta):=g_t(W_t,\theta)$. First remove selected-sample centering by defining \(\widetilde\Gamma_0(\ell):=D_T(\ell)^{-1}\sum Z_tg_t(\widehat\theta_N)g_{t-\ell}(\widehat\theta_N)^\top\).
On $D_T(\ell)>0$,
\[
\begin{aligned}
\widetilde\Gamma_0(\ell)-\Gamma_{\mathrm{hac}}(\ell)
&=D_T(\ell)^{-1}\sum Z_tY_t
 +\left(D_T(\ell)^{-1}\sum Z_ta_t-\Gamma_{\mathrm{hac}}(\ell)\right) \\
&\quad +D_T(\ell)^{-1}\sum Z_t\big[g_t(\widehat\theta_N)g_{t-\ell}(\widehat\theta_N)^\top-g_t(\theta_T^\star)g_{t-\ell}(\theta_T^\star)^\top\big].
\end{aligned}
\]
For the first term, write it as $(T/D_T(\ell))\bar I_T$ with $\bar I_T:=T^{-1}\sum Z_tY_t$. Under random sampling, the observed-pair indicators are independent of the moment-product array, and $\mathbb E_TY_t=0$. Thus $\Cov_T(\operatorname{vec}(Z_tY_t),\operatorname{vec}(Z_sY_s))=\mathbb E[Z_tZ_s]\Cov_T(\operatorname{vec}(Y_t),\operatorname{vec}(Y_s))$, and $0\le\mathbb E[Z_tZ_s]\le1$; under deterministic-window sampling the same bound holds with fixed factors $Z_tZ_s\le1$. Therefore $\Var_T(\bar I_T)\le T^{-2}\sum_h\sum_t\|\Cov_T(\operatorname{vec}(Y_t),\operatorname{vec}(Y_{t-h}))\|=O(T^{-1})$ by the product covariance condition. Chebyshev's inequality gives $\bar I_T=O_p(T^{-1/2})$, and $T/D_T(\ell)=O_p(1)$ makes the first term $o_p(1)$. The second term is $o_p(1)$ by Assumption~\ref{ass:pair-average-app}.

For the plug-in term, let $\Delta_T:=\|\widehat\theta_N-\theta_T^\star\|$. Assumption~\ref{ass:smoothness} gives
\[
\begin{aligned}
&\|g_t(\widehat\theta_N)g_{t-\ell}(\widehat\theta_N)^\top-g_t(\theta_T^\star)g_{t-\ell}(\theta_T^\star)^\top\| \\
&\quad\le \Delta_T\big[b_g(W_t)B_g(W_{t-\ell})+B_g(W_t)b_g(W_{t-\ell})
+\Delta_T b_g(W_t)b_g(W_{t-\ell})\big].
\end{aligned}
\]
The first two observed-pair averages are $O_p(1)$. For example,
\[
D_T(\ell)^{-1}\sum Z_tb_g(W_t)B_g(W_{t-\ell})
=\frac{T}{D_T(\ell)}\,T^{-1}\sum Z_tb_g(W_t)B_g(W_{t-\ell})
=O_p(1),
\]
and the same calculation applies with $B_g(W_t)b_g(W_{t-\ell})$. For the last average, Cauchy--Schwarz gives
\[
D_T(\ell)^{-1}\sum Z_tb_g(W_t)b_g(W_{t-\ell})
\le
\frac{T}{D_T(\ell)}
\left[T^{-1}\sum b_g(W_t)^2\right]^{1/2}
\left[T^{-1}\sum b_g(W_{t-\ell})^2\right]^{1/2}
=O_p(1),
\]
using the second-moment envelope in Assumption~\ref{ass:smoothness}. Since $\Delta_T=o_p(1)$, the plug-in term is $o_p(1)$. Hence $\widetilde\Gamma_0(\ell)\to_p\Gamma_{\mathrm{hac}}(\ell)$.

It remains to replace uncentered observed-pair products by products centered at $g_N(\widehat\theta_N)$. Let $\bar g_N:=g_N(\widehat\theta_N)$. Since $\widehat s_t=g_t(\widehat\theta_N)-\bar g_N$,
\[
\begin{aligned}
\widetilde\Gamma_{\mathrm{hac}}(\ell)-\widetilde\Gamma_0(\ell)
&=-\left(D_T(\ell)^{-1}\sum Z_tg_t(\widehat\theta_N)\right)\bar g_N^\top \\
&\quad-\bar g_N\left(D_T(\ell)^{-1}\sum Z_tg_{t-\ell}(\widehat\theta_N)\right)^\top
+\bar g_N\bar g_N^\top.
\end{aligned}
\]
The two observed-pair averages are $O_p(1)$. For example,
\[
\left\|D_T(\ell)^{-1}\sum Z_tg_t(\widehat\theta_N)\right\|
\le
\frac{T}{D_T(\ell)}\,T^{-1}\sum_t
\left[B_g(W_t)+b_g(W_t)\Delta_T\right]
=O_p(1)
\]
by consistency and Assumption~\ref{ass:smoothness}, and the lagged average is identical. Since $\bar g_N=o_p(1)$ by assumption, observed-pair centering changes the limit by $o_p(1)$. Negative lags follow by transposition, and lags with $\kappa_R(\ell)=0$ have zero asymptotic observed-pair weight by definition.
\end{proof}

\begin{proof}[Proof of Theorem \ref{thm:hac}]
The proof is for the fully observed sample, so $N=T$. Let $\widehat\Gamma_{\hat s}(\ell)$ denote the empirical lag covariance in \eqref{eq:HAC} built from $\widehat s_t:=g_t(W_t,\widehat\theta_N)-g_N(\widehat\theta_N)$, and let $\widehat\Gamma_s(\ell)$ denote the infeasible analogue built from $s_{T,t}:=g_t(W_t,\theta_T^\star)-\bar m_T(\theta_T^\star)$. Define
\[
\widetilde\Omega_T^+(L_T):=\sum_{|\ell|\le L_T}K(|\ell|/L_T)\widehat\Gamma_s(\ell),
\qquad
\Omega_K^+(L_T):=\sum_{|\ell|\le L_T}K(|\ell|/L_T)\Gamma_{g,\mathrm{hac}}(\ell),
\]
with the convention $K(|\ell|/L_T)=0$ for $|\ell|>L_T$ and $K(0)=1$. Then
\[
\|\widehat\Omega_R^+-\Omega_R^+\|
\le
\|\widehat\Omega_R^+-\widetilde\Omega_T^+(L_T)\|
+\|\widetilde\Omega_T^+(L_T)-\Omega_K^+(L_T)\|
+\|\Omega_K^+(L_T)-\Omega_R^+\|.
\]
The first term is $o_p(1)$ by the feasible plug-in clause of Assumption~\ref{ass:hac-regularity}. This is the only step in the theorem that uses the replacement of $\theta_T^\star$ by $\widehat\theta_N$ and the finite-sample centering by $g_N(\widehat\theta_N)$.

For the second term, separate the zero lag from the nonzero lags:
\[
\widetilde\Omega_T^+(L_T)-\Omega_K^+(L_T)
=\big(\widehat\Gamma_s(0)-\Gamma_{g,\mathrm{hac}}(0)\big)
+\sum_{1\le|\ell|\le L_T}K(|\ell|/L_T)\big(\widehat\Gamma_s(\ell)-\Gamma_{g,\mathrm{hac}}(\ell)\big).
\]
The fixed-lag clause of Assumption~\ref{ass:hac-regularity} controls the zero lag, and the growing-window clause controls the weighted nonzero-lag sum. Hence the second term is $o_p(1)$.

It remains to control deterministic truncation and kernel bias. Let $C_K:=\sup_x|K(x)|<\infty$. For any fixed integer $M$ with $M<L_T$ eventually,
\[
\begin{aligned}
\|\Omega_K^+(L_T)-\Omega_R^+\|
&\le
\sum_{|\ell|\le M}|K(|\ell|/L_T)-1|\,\|\Gamma_{g,\mathrm{hac}}(\ell)\| \\
&\quad +(1+C_K)\sum_{|\ell|>M}\|\Gamma_{g,\mathrm{hac}}(\ell)\|.
\end{aligned}
\]
For fixed $M$, continuity of $K$ at zero makes the first term vanish as $T\to\infty$. Absolute summability of $(\Gamma_{g,\mathrm{hac}}(\ell))$ makes the second term arbitrarily small by taking $M$ large. Therefore $\Omega_K^+(L_T)\to\Omega_R^+$ and the diverging-bandwidth HAC estimator converges in probability to $\Omega_R^+$.

In the fixed flat-top alternative, use the same decomposition with $L_T=L$ fixed. The population truncation term is
\[
\Omega_K^+(L)-\Omega_R^+
=\sum_{|\ell|\le L}(K(|\ell|/L)-1)\Gamma_{g,\mathrm{hac}}(\ell)-\sum_{|\ell|>L}\Gamma_{g,\mathrm{hac}}(\ell),
\]
and it equals zero because the finite-$T$ centered lag covariance is eventually zero beyond $q$ and the flat-top kernel equals one on every integer lag $|\ell|\le q$. The fixed-lag and plug-in clauses of Assumption~\ref{ass:hac-regularity} then give $\widehat\Omega_R^+\to_p\Omega_R^+$. Proposition~\ref{lem:fixed-env-decomp} finally gives $\Omega_R^+=\Omega_R+\Omega_\mu$ with $\Omega_\mu\succeq0$, so the probability limit is a positive-semidefinite upper bound on the design moment covariance.
\end{proof}

\begin{proof}[Proof of Corollary \ref{cor:scalar-conservative}]
Proposition~\ref{lem:fixed-env-decomp} gives $\Omega_R^+-\Omega_R=\Omega_\mu$, and Lemma~\ref{lem:sumGE-psd} gives $\Omega_\mu\succeq0$. Let $H:=(G^\top A G)^{-1}G^\top A$. The GMM sandwich map is a congruence transformation in the moment covariance matrix, so \(\Sigma^+-\Sigma=H(\Omega_R^+-\Omega_R)H^\top=H\Omega_\mu H^\top\succeq0\).
Premultiplying and postmultiplying by any fixed vector $a$ gives $a^\top\Sigma a\le a^\top\Sigma^+a$.
\end{proof}

\begin{proof}[Proof of Corollary \ref{cor:scalar-coverage}]
The delta-method conclusion in Theorem~\ref{thm:AN} gives $\sqrt T(\widehat\tau-\tau_T)\Rightarrow\mathcal N(0,v)$. Corollary~\ref{cor:scalar-conservative} gives $0\le v\le v^+$. Since $v^+>0$ and $\widehat v^+\to_pv^+$, the reported standard error is positive with probability approaching one.

If $v>0$, Slutsky's theorem gives \(\frac{\sqrt T(\widehat\tau-\tau_T)}{\sqrt{\widehat v^+}}\Rightarrow \mathcal N(0,v/v^+)\).
Therefore the limiting coverage probability of the two-sided interval is
\[
\Pr\left(|Z|\le z_{1-\alpha/2}\sqrt{v^+/v}\right)
=2\Phi\!\left(z_{1-\alpha/2}\sqrt{v^+/v}\right)-1,
\qquad Z\sim\mathcal N(0,1).
\]
Because $v/v^+\le1$, the displayed probability is at least $1-\alpha$. If $v=0$, then $\sqrt T(\widehat\tau-\tau_T)=o_p(1)$, while $\sqrt{\widehat v^+}\to_p\sqrt{v^+}>0$. The standardized statistic converges to zero in probability, and the coverage probability converges to one.
\end{proof}

The following lemma records the finite-sample positivity of the HAC quadratic form used by the multiplier construction.

\begin{lemma}
\label{lem:finite-PSD}
Let $N:=\sum_{t=1}^TR_{T,t}>0$, let $w_0=1$, let $w_{|\ell|}=K(|\ell|/L)$ for $1\le |\ell|\le L$, and let $w_{|\ell|}=0$ otherwise. Terms with zero observed-pair count are omitted from the observed-pair representation of $\widehat\Omega_R^{\mathrm{HAC}}(L)$. If the Toeplitz matrix $(w_{|t-s|})_{1\le t,s\le T}$ is nonnegative definite, then
\[
\widehat\Omega_R^{\mathrm{HAC}}(L):=\widetilde\Gamma_{\mathrm{hac}}(0)+\sum_{1\le|\ell|\le L}K(|\ell|/L)\widehat\kappa_R(\ell)\widetilde\Gamma_{\mathrm{hac}}(\ell)
\]
is positive semidefinite almost surely.
\end{lemma}

\begin{proof}[Proof of Lemma \ref{lem:finite-PSD}]
For $u\in\mathbb{R}^k$, set $\hat s_t:=g_t\!\big(W_t,\hat\theta_N\big)-g_N(\hat\theta_N)$ and $\phi_t:=u^\top \hat s_t$. For $|\ell|\ge1$ with $D_T(\ell)>0$,
\[
\widehat\kappa_R(\ell)\widetilde\Gamma_{\mathrm{hac}}(\ell)
=\frac{D_T(\ell)}{N}\frac1{D_T(\ell)}\sum_{t=|\ell|+1}^TR_{T,t}R_{T,t-|\ell|}\hat s_t\hat s_{t-|\ell|}^\top,
\]
and the same product is zero by convention if $D_T(\ell)=0$. For $\ell=0$, $D_T(0)=N$ and the term is $N^{-1}\sum_tR_{T,t}\hat s_t\hat s_t^\top$. Hence
\[
u^\top \widehat\Omega_R^{\mathrm{HAC}}(L)\,u
=N^{-1}\sum_{t=1}^T\sum_{s=1}^T R_{T,t}R_{T,s}\,w_{|t-s|}\,\phi_t\,\phi_s.
\]
The matrix with entries $R_{T,t}R_{T,s}\,w_{|t-s|}$ equals
\[
\operatorname{diag}(R_{T,1},\ldots,R_{T,T})
(w_{|t-s|})_{t,s}
\operatorname{diag}(R_{T,1},\ldots,R_{T,T}),
\]
which is positive semidefinite because the Toeplitz matrix is positive semidefinite. Therefore the displayed quadratic form is nonnegative for every $u$, proving the claim.
\end{proof}

\begin{proof}[Proof of Theorem \ref{thm:bootstrap}]
All starred probability statements are conditional on the observed sample. Write $\hat s_t:=g_t(W_t,\widehat\theta_N)-g_N(\widehat\theta_N)$ and $\widehat B_N^\ast:=T^{-1/2}\sum_{t=1}^T\xi_t\hat s_t$. The multiplier vector is Gaussian with $\mathbb E^\ast[\xi_t]=0$ and $\mathbb E^\ast[\xi_t\xi_s]=K(|t-s|/L_T)$. The Toeplitz positive-semidefiniteness condition ensures that this conditional Gaussian vector exists for every finite $T$. Conditional on the data, $\widehat B_N^\ast$ is Gaussian with mean zero and covariance
\[
\begin{aligned}
\Var^\ast(\widehat B_N^\ast)
&=T^{-1}\sum_{t=1}^T\sum_{s=1}^TK(|t-s|/L_T)\hat s_t\hat s_s^\top \\
&=\widehat\Omega_R^+(L_T,K),
\end{aligned}
\]
where the second equality is the centered lag-window identity for a fully observed sample.

Theorem~\ref{thm:hac} gives $\widehat\Omega_R^+(L_T,K)\to_p\Omega_R^+$ under the diverging-bandwidth conditions and under the stated fixed flat-top alternative. Let $Z$ be a standard normal vector of dimension $k$. If $C_n\to C$ in matrix norm and $C_n,C\succeq0$, then $C_n^{1/2}\to C^{1/2}$ in matrix norm because the positive-semidefinite square-root map is continuous in finite dimension. Conditional on the data, $\widehat B_N^\ast$ has the same law as $\widehat\Omega_R^+(L_T,K)^{1/2}Z$. Therefore, for every bounded 1-Lipschitz function $f$, $|\mathbb E^\ast f(\widehat B_N^\ast)-\mathbb E f(Z_\Omega)|=|\mathbb E f(\widehat\Omega_R^+(L_T,K)^{1/2}Z)-\mathbb E f((\Omega_R^+)^{1/2}Z)|\le\|\widehat\Omega_R^+(L_T,K)^{1/2}-(\Omega_R^+)^{1/2}\|\mathbb E\|Z\|\to_p0$, with $Z_\Omega\sim\mathcal N(0,\Omega_R^+)$. Equivalently, $\widehat B_N^\ast\Rightarrow^\ast\mathcal N(0,\Omega_R^+)$ in probability.

The bootstrap parameter is defined by the one-step GMM update
\[
\sqrt T(\widehat\theta_N^\ast-\widehat\theta_N)
=-(\widehat G_N^\top\widehat A_N\widehat G_N)^{-1}\widehat G_N^\top\widehat A_N\widehat B_N^\ast.
\]
The consistency proof in Theorem~\ref{thm:AN}, the uniform Jacobian law, and Assumption~\ref{ass:smoothness} give $\widehat G_N\to_pG$, and Assumption~\ref{ass:gmm} gives $\widehat A_N\to_pA$. Since $G^\top A G$ is nonsingular, \((\widehat G_N^\top\widehat A_N\widehat G_N)^{-1}\widehat G_N^\top\widehat A_N\to_p(G^\top A G)^{-1}G^\top A\).
Conditional Slutsky's theorem yields
\[
\sqrt T(\widehat\theta_N^\ast-\widehat\theta_N)\Rightarrow^\ast
\mathcal N\!\left(0,(G^\top A G)^{-1}G^\top A\,\Omega_R^+\,A G\,(G^\top A G)^{-1}\right)
\]
in probability, which is the asserted covariance $\Sigma^+$.
\end{proof}

The next lemma records the fully observed sample reduction to the conventional HAC formulas, connecting the appendix notation to the main-text estimator and confirming that the observed-pair notation collapses to the conventional centered lag-window construction when every date is observed.

\begin{lemma}\label{lem:collapse-HAC}
Suppose $R_{T,t}\equiv 1$ for all $t$ and write $\hat s_t:=g_t(W_t,\hat\theta_N)-g_N(\hat\theta_N)$. Then:

\begin{enumerate}\setlength\itemsep{2pt}
\item[(a)] For any kernel $K$ and bandwidth $1\le L<T$ as in Assumption~\ref{ass:hac-kernel}, the feasible centered HAC estimator
\[
\widehat\Omega_{R}^{\mathrm{HAC}}(L)
=\widetilde\Gamma_{\mathrm{hac}}(0)+\sum_{1\le|\ell|\le L}K(|\ell|/L)\,\widehat\kappa_R(\ell)\,\widetilde\Gamma_{\mathrm{hac}}(\ell)
\]
reduces exactly to the classical centered HAC estimator built from $(\hat s_t)_{t=1}^T$:
\begin{align*}
\widehat\Omega^{\mathrm{hac}}(L)
&:=\frac{1}{T}\sum_{t=1}^T \hat s_t\hat s_t^\top
   \\
   &+\sum_{\ell=1}^{L}K(\ell/L)\,\frac{1}{T}\sum_{t=\ell+1}^T
   \Big(\hat s_t\hat s_{t-\ell}^\top+\hat s_{t-\ell}\hat s_t^\top\Big).
\end{align*}

\item[(b)] The bootstrap moment and update become $\widehat B_T^\ast=T^{-1/2}\sum_{t=1}^T \xi_t\,\hat s_t$ and $\sqrt{T}(\widehat\theta_T^\ast-\widehat\theta_T)=-(\widehat G_T^\top \widehat A_T \widehat G_T)^{-1}\widehat G_T^\top \widehat A_T\,\widehat B_T^\ast$, and $\Var^\ast(\widehat B_T^\ast)=\widehat\Omega^{\mathrm{hac}}(L)$ exactly. If, in addition, the conditions giving the linear representation in Theorem~\ref{thm:AN} hold and $L=L_T$ satisfies Assumption~\ref{ass:hac-kernel}, replacing $\hat s_t$ by $g_t(W_t,\widehat\theta_N)$ in the HAC matrix or multiplier moment changes the limiting covariance by $o_p(1)$.

\end{enumerate}
\end{lemma}

\begin{proof}[Proof of Lemma \ref{lem:collapse-HAC}]
Since $R_{T,t}\equiv1$, $N=T$. For $0\le\ell<T$, the observed-pair count is $D_T(\ell)=T-\ell$, and therefore $\widehat\kappa_R(\ell)=(T-\ell)/T$. For each positive retained lag $\ell$,
\[
\widetilde\Gamma_{\mathrm{hac}}(\ell)=(T-\ell)^{-1}\sum_{t=\ell+1}^T\hat s_t\hat s_{t-\ell}^\top,
\qquad
\widetilde\Gamma_{\mathrm{hac}}(-\ell)=\widetilde\Gamma_{\mathrm{hac}}(\ell)^\top.
\]
Hence
\[
K(\ell/L)\widehat\kappa_R(\ell)\widetilde\Gamma_{\mathrm{hac}}(\ell)
=K(\ell/L)T^{-1}\sum_{t=\ell+1}^T\hat s_t\hat s_{t-\ell}^\top,
\]
and the negative-lag term is the transpose of the same expression. For $\ell=0$, $\widetilde\Gamma_{\mathrm{hac}}(0)=T^{-1}\sum_t\hat s_t\hat s_t^\top$. Summing the zero lag and both orientations of each positive lag gives the displayed classical centered HAC formula. This proves part (a).

For part (b), the same fully observed substitution gives $\widehat B_T^\ast=T^{-1/2}\sum_t\xi_t\hat s_t$ and the stated one-step update. Conditional on the data,
\[
\Var^\ast(\widehat B_T^\ast)
=T^{-1}\sum_{t=1}^T\sum_{s=1}^TK(|t-s|/L)\hat s_t\hat s_s^\top
=\widehat\Omega^{\mathrm{hac}}(L),
\]
which proves the exact covariance claim.

It remains to justify the uncentered variant under the diverging-bandwidth conditions. Let $\bar g_T:=g_N(\widehat\theta_N)$ and $u_t:=g_t(W_t,\widehat\theta_N)=\hat s_t+\bar g_T$. The first-order expansion in Theorem~\ref{thm:AN} and local correct specification imply $\|\bar g_T\|=O_p(T^{-1/2})$, and Assumption~\ref{ass:smoothness} together with consistency gives $T^{-1}\sum_t\|\hat s_t\|=O_p(1)$. It is enough to bound the zero and positive lags, since negative lags are transposes. With $\bar K:=\sup_x|K(x)|$,
\[
\begin{aligned}
&\left\|T^{-1}\sum_{\ell=0}^{L_T}K(\ell/L_T)
\sum_{t=\ell+1}^T\big(u_tu_{t-\ell}^\top-\hat s_t\hat s_{t-\ell}^\top\big)\right\| \\
&\quad\le C\sum_{\ell=0}^{L_T}\left(\|\bar g_T\|\,T^{-1}\sum_{t=\ell+1}^T(\|\hat s_t\|+\|\hat s_{t-\ell}\|)+\|\bar g_T\|^2\right) \\
&\quad=O_p(L_T/\sqrt T+L_T/T)=o_p(1).
\end{aligned}
\]
The same calculation applies to the conditional covariance of the multiplier moment, because conditional covariance is obtained by the corresponding kernel-weighted double sum. Therefore the centered and uncentered fully observed sample dependent-multiplier constructions calibrated to the HAC long-run variance have the same limiting covariance under $L_T/\sqrt T\to0$.
\end{proof}

\subsection{Regression-adjustment and simultaneous-inference proofs}

\begin{proof}[Proof of Proposition \ref{prop:HAC_RA}]
Write \(\hat s_t:=g_t(W_t,\widehat\theta_N)\) and \(s_t:=g_t(W_t,\theta_T^\star)\). The retained adjustment vector has dimension \(m\), the moment has dimension \(k\), and the projection coefficient \(B\) is \(m\times k\). Thus \(S_{zz}^+\) is \(m\times m\), \(S_{zs}^+\) is \(m\times k\), and \(B^\circ=(S_{zz}^+)^{-1}S_{zs}^+\) is conformable with the residual \(s_t-B^{\circ\top}z_t\).

Assumption~\ref{ass:z_nondeg} gives the fixed-lag limits and the lag-window laws of large numbers for each block of the stacked process \((z_t^\top,s_t^\top)^\top\). Let $\Delta_T:=\|\widehat\theta_N-\theta_T^\star\|$. Assumption~\ref{ass:smoothness} gives $\|\hat s_t-s_t\|\le b_{g,T,t}\Delta_T$. For a $zs$ block, $\|z_t(\hat s_{t-\ell}-s_{t-\ell})^\top\|\le\|z_t\|b_{g,T,t-\ell}\Delta_T$. For an $ss$ block,
\[
\|\hat s_t\hat s_{t-\ell}^\top-s_ts_{t-\ell}^\top\|
\le \|\hat s_t-s_t\|\|s_{t-\ell}\|+
\|s_t\|\|\hat s_{t-\ell}-s_{t-\ell}\|+
\|\hat s_t-s_t\|\|\hat s_{t-\ell}-s_{t-\ell}\|.
\]
Theorem~\ref{thm:AN} gives $\Delta_T=O_p(T^{-1/2})$. The same envelope averages used in the HAC plug-in argument make the lag-window replacements negligible. For example, the $zs$ replacement is bounded by $C_K\Delta_T\sum_{|\ell|\le L_T}T^{-1}\sum_t\|z_t\|b_{g,T,t-|\ell|}=O_p(L_T\Delta_T)=o_p(1)$, and the displayed $ss$ inequality gives the same order plus an $O_p(L_T\Delta_T^2)$ term. Hence the limits are unchanged when $s_t$ is replaced by $\hat s_t$, so \(\hat S_{zz}^{+}\to_p S_{zz}^{+}\), \(\hat S_{zs}^{+}\to_p S_{zs}^{+}\), \(\hat S_{sz}^{+}\to_p S_{sz}^{+}\), and \(\hat S_{ss}^{+}\to_p S_{ss}^{+}\). Since \(S_{zz}^{+}\succ0\), \(\hat S_{zz}^{+}\) is nonsingular with probability approaching one, and continuity of inversion yields \(\hat B=(\hat S_{zz}^{+})^{-1}\hat S_{zs}^{+}\to_p B^\circ\).

For each finite \(T\), the lag-window map \((a,b)\mapsto \hat S_{ab}^+\) is bilinear. Therefore, on the event on which \(\hat S_{zz}^+\) is nonsingular,
\[
\widehat\Omega_R^{+}(\hat r)
=\hat S_{ss}^{+}-\hat S_{sz}^{+}\hat B-\hat B^\top \hat S_{zs}^{+}+\hat B^\top \hat S_{zz}^{+}\hat B.
\]
The preceding block convergence and the continuous mapping theorem give
\[
\widehat\Omega_R^{+}(\hat r)\to_p
S_{ss}^{+}-S_{sz}^{+}B^\circ-B^{\circ\top}S_{zs}^{+}+B^{\circ\top}S_{zz}^{+}B^\circ
=:S_{rr}^{+}(B^\circ).
\]
This proves the consistency assertion.

It remains to prove the order statements. For an arbitrary conformable matrix \(B\), define \(r_t(B):=s_t-B^\top z_t\). The two-sided convention for the long-run matrices gives \(S_{sz}^{+}=S_{zs}^{+\top}\), and bilinearity gives \(S_{r(B)r(B)}^{+}=S_{ss}^{+}-S_{sz}^{+}B-B^\top S_{zs}^{+}+B^\top S_{zz}^{+}B\).
Let \(\Delta:=B-B^\circ\). Since \(S_{zz}^{+}B^\circ=S_{zs}^{+}\) and \(B^{\circ\top}S_{zz}^{+}=S_{sz}^{+}\), expanding the preceding identity at \(B=B^\circ+\Delta\) gives the exact identity \(S_{r(B)r(B)}^{+}=S_{r(B^\circ)r(B^\circ)}^{+}+\Delta^\top S_{zz}^{+}\Delta\).
The last term is positive semidefinite because \(S_{zz}^{+}\succ0\). Taking \(B=0\) gives
\[
\Omega_R^{+}(r)=S_{r(B^\circ)r(B^\circ)}^{+}
=S_{ss}^{+}-S_{sz}^{+}(S_{zz}^{+})^{-1}S_{zs}^{+}\preceq S_{ss}^{+}=\Omega_R^{+}.
\]
The Loewner gap is \(S_{sz}^{+}(S_{zz}^{+})^{-1}S_{zs}^{+}\). If this matrix is zero, then for each column \(x_j\) of \(S_{zs}^{+}\) it holds that \(x_j^\top(S_{zz}^{+})^{-1}x_j=0\); positive definiteness of \((S_{zz}^{+})^{-1}\) implies \(x_j=0\). Conversely, if $S_{zs}^{+}=0$, then that gap is zero. Thus equality holds if and only if \(S_{zs}^{+}=0\), proving part (i).

For part (ii), fix \(A\succ0\) and write \(C(A):=A G(G^\top A G)^{-1}\). Assumption~\ref{ass:smoothness} gives full column rank of \(G\), so \(G^\top A G\) is nonsingular. For any positive semidefinite \(\Omega\), the fixed-weight GMM covariance can be written as \(C(A)^\top\Omega C(A)\). Congruence preserves the Loewner order: if \(M\succeq0\), then \(C(A)^\top M C(A)\succeq0\). Applying this to \(M=\Omega_R^{+}-\Omega_R^{+}(r)\) yields \(\Sigma_{\mathrm{RA}}(A)\preceq\Sigma^{+}(A)\).

For part (iii), positive definiteness of both variance limits permits inversion. The order \(\Omega_R^{+}(r)\preceq\Omega_R^{+}\) and the operator monotonicity of the inverse on the positive-definite cone imply \(\Omega_R^{+}(r)^{-1}\succeq(\Omega_R^{+})^{-1}\). Premultiplication and postmultiplication by \(G\) preserve the Loewner order, so \(G^\top\Omega_R^{+}(r)^{-1}G\succeq G^\top(\Omega_R^{+})^{-1}G\). Both matrices are positive definite because \(G\) has full column rank, and inverting them reverses the Loewner order. This proves the inverse-variance-weight covariance inequality.
\end{proof}

\begin{proof}[Proof of Proposition \ref{prop:HAC_RA_conservative}]
Write \(\mu_{T,t}:=\mathbb E_T[s_t]\), \(\bar\mu_T:=T^{-1}\sum_{t=1}^T\mu_{T,t}\), \(\tilde\mu_{T,t}:=\mu_{T,t}-\bar\mu_T\), and \(e_{T,t}:=s_t-\mu_{T,t}\). Then \(\mathbb E_T[e_{T,t}]=0\) date by date and \(s_t=e_{T,t}+\tilde\mu_{T,t}+\bar\mu_T\). Assumption~\ref{ass:local-correct-specification} gives \(\sqrt T\|\bar\mu_T\|=o(1)\). Put $a_{T,t}:=e_{T,t}+\tilde\mu_{T,t}$. For each fixed lag, define
\[
A_{T,\ell}:=T^{-1}\sum_t\mathbb E_Ta_{T,t},
\qquad
B_{T,\ell}:=T^{-1}\sum_t\mathbb E_Ta_{T,t-\ell}.
\]
Then
\[
T^{-1}\sum_t\mathbb E_T[(a_{T,t}+\bar\mu_T)(a_{T,t-\ell}+\bar\mu_T)^\top-a_{T,t}a_{T,t-\ell}^\top]
=A_{T,\ell}\bar\mu_T^\top+\bar\mu_TB_{T,\ell}^\top+\bar\mu_T\bar\mu_T^\top.
\]
The second-moment envelope gives
\[
\sup_{\ell\le L_T}(\|A_{T,\ell}\|\vee\|B_{T,\ell}\|)
\le T^{-1}\sum_t\mathbb E_T\|a_{T,t}\|=O(1).
\]
Hence the norm of the fixed-lag difference is $O(\|\bar\mu_T\|+\|\bar\mu_T\|^2)=o(1)$, and its growing-bandwidth contribution is $O_p(L_T\|\bar\mu_T\|+L_T\|\bar\mu_T\|^2)=o_p(1)$ because $L_T/\sqrt T\to0$. Thus the long-run calculation may be made with \(e_{T,t}+\tilde\mu_{T,t}\).

For each fixed nonnegative lag \(\ell\),
\[
T^{-1}\sum_{t=\ell+1}^T\mathbb E_T[e_{T,t}\tilde\mu_{T,t-\ell}^\top]
=T^{-1}\sum_{t=\ell+1}^T\mathbb E_T[e_{T,t}]\tilde\mu_{T,t-\ell}^\top=0.
\]
The negative-lag terms are transposes of the corresponding positive-lag terms. Hence \(S_{e\mu}^{+}=S_{\mu e}^{+}=0\). The maintained orthogonality condition \(S_{ez}^{+}=0\) also gives \(S_{ze}^{+}=0\). With \(d_t:=\tilde\mu_{T,t}-B^{\circ\top}z_t\), \(S_{ed}^{+}=S_{e\mu}^{+}-S_{ez}^{+}B^\circ=0\) and \(S_{de}^{+}=0\).
Bilinearity of the long-run operator gives \(S_{rr}^{+}=S_{ee}^{+}+S_{dd}^{+}\).
In the fully observed sample, \(S_{ee}^{+}=\Omega_R\) by the definition of the design covariance in Assumption~\ref{ass:dependence}. To verify the sign of the residual term, fix \(u\in\mathbb R^k\), put \(q_{T,t}:=u^\top d_t\) for \(1\le t\le T\), and extend the scalar sequence by zero outside the sample. With \(c_T(\ell):=T^{-1}\sum_{t\in\mathbb Z}\mathbb E_T[q_{T,t}q_{T,t-\ell}]\), the Fej\'er identity gives, for each fixed \(M\ge1\),
\[
T^{-1}\mathbb E_T\left[\sum_{t\in\mathbb Z}\left(M^{-1/2}\sum_{j=0}^{M-1}q_{T,t-j}\right)^2\right]
=\sum_{|\ell|<M}\left(1-\frac{|\ell|}{M}\right)c_T(\ell)\ge0.
\]
For this fixed $M$, fixed-lag convergence gives $\sum_{|\ell|<M}(1-|\ell|/M)c(\ell)\ge0$, where $c(\ell):=\lim_Tc_T(\ell)$. Absolute summability of the long-run matrices involving $d_t$ then gives $u^\top S_{dd}^{+}u=\lim_{M\to\infty}\sum_{|\ell|<M}(1-|\ell|/M)c(\ell)\ge0$. Hence \(S_{dd}^{+}\succeq0\), and \(\Omega_R^+(r)=\Omega_R+S_{dd}^{+}\succeq\Omega_R\).

It remains to compute \(S_{dd}^{+}\). Since \(S_{ze}^{+}=0\), \(S_{zs}^{+}=S_{z\mu}^{+}\) and \(B^\circ=(S_{zz}^{+})^{-1}S_{z\mu}^{+}\).
Expanding \(d_t=\tilde\mu_{T,t}-B^{\circ\top}z_t\) gives
\[
\begin{aligned}
S_{dd}^{+}
&=S_{\mu\mu}^{+}-S_{\mu z}^{+}B^\circ-B^{\circ\top}S_{z\mu}^{+}+B^{\circ\top}S_{zz}^{+}B^\circ \\
&=S_{\mu\mu}^{+}-S_{\mu z}^{+}(S_{zz}^{+})^{-1}S_{z\mu}^{+}.
\end{aligned}
\]
The upper bound \(\Omega_R^+(r)\preceq\Omega_R^+\) is part (i) of Proposition~\ref{prop:HAC_RA}. Since \(S_{dd}^{+}\succeq0\), equality with \(\Omega_R\) holds if and only if \(S_{dd}^{+}=0\). If there is a fixed matrix \(B_0\) with \(\tilde\mu_{T,t}=B_0^\top z_t\) along the sequence, then \(S_{z\mu}^{+}=S_{zz}^{+}B_0\), so \(B^\circ=B_0\) and the residual matrix \(S_{dd}^{+}\) is zero.
\end{proof}

\clearpage

\begin{proof}[Proof of Theorem \ref{thm:RA-fullobs}]
Fix a conformable matrix \(B\). By the common-mean calculation in Proposition~\ref{prop:HAC_RA_conservative}, replacing $g_t$ by $e_{T,t}+\tilde\mu_{T,t}$ changes every fixed-lag and lag-window long-run calculation by $o(1)$. With \(d_t(B):=\tilde\mu_{T,t}-B^\top z_t\), this representation gives \(r_t(B)=e_{T,t}+d_t(B)\).

The date-specific-centering calculation in the proof of Proposition~\ref{prop:HAC_RA_conservative} gives \(S_{e\mu}^+=0\). Assumption~\ref{ass:RA-orth} gives \(S_{ez}^+=0\). Therefore, for the fixed matrix \(B\), \(S_{e d(B)}^+=S_{e\mu}^+-S_{ez}^+B=0\) and \(S_{d(B)e}^+=0\). Bilinearity gives \(S_{r(B)r(B)}^+=S_{ee}^+ + S_{d(B)d(B)}^+\). In the fully observed sample, \(S_{ee}^+=\Omega_R\). For each fixed vector $u$, the Fej\'er identity applied to $q_{T,t}=u^\top d_t(B)$ gives $u^\top S_{d(B)d(B)}^+u\ge0$, exactly as in Proposition~\ref{prop:HAC_RA_conservative}. Hence \(S_{d(B)d(B)}^+\succeq0\) and \(\Omega_R^+(r(B))=\Omega_R+S_{d(B)d(B)}^+\succeq\Omega_R\).

Now consider the population HAC projection. The long-run quadratic form satisfies \(S_{r(B)r(B)}^+=S_{gg}^+-S_{gz}^+B-B^\top S_{zg}^+ + B^\top S_{zz}^+B\).
Let \(B^*:=(S_{zz}^+)^{-1}S_{zg}^+\) and \(\Delta:=B-B^*\). Expanding at $B=B^*+\Delta$, the cross terms are $-S_{gz}^+\Delta-\Delta^\top S_{zg}^+ + B^{*\top}S_{zz}^+\Delta+\Delta^\top S_{zz}^+B^*=0$ because $S_{zz}^+B^*=S_{zg}^+$ and $B^{*\top}S_{zz}^+=S_{gz}^+$. The remainder is $\Delta^\top S_{zz}^+\Delta\succeq0$, so \(S_{r(B)r(B)}^+=S_{r(B^*)r(B^*)}^+ + \Delta^\top S_{zz}^+\Delta\).
Thus \(B^*\) minimizes the adjusted long-run matrix in Loewner order. Taking \(B=0\) yields \(\Omega_R^+(r(B^*))\preceq S_{gg}^+=\Omega_R^+\). Combining this upper bound with the lower bound already proved gives \(\Omega_R\preceq\Omega_R^+(r(B^*))\preceq\Omega_R^+\).
Finally, \(S_{ez}^+=0\) implies \(S_{ze}^+=0\), and date-specific centering implies \(S_{e\mu}^+=S_{\mu e}^+=0\). Therefore \(S_{zg}^+=S_{z\mu}^+\) and \(S_{gz}^+=S_{\mu z}^+\). Substituting \(B^*\) into \(S_{d(B^*)d(B^*)}^+\) gives \(S_{d(B^*)d(B^*)}^+=S_{\mu\mu}^+-S_{\mu z}^+(S_{zz}^+)^{-1}S_{z\mu}^+\).
Together with the identity \(\Omega_R^+(r(B^*))=\Omega_R+S_{d(B^*)d(B^*)}^+\), this is the stated formula. Since the difference between \(\Omega_R^+(r(B^*))\) and \(\Omega_R\) is the positive semidefinite matrix \(S_{d(B^*)d(B^*)}^+\), equality with \(\Omega_R\) holds if and only if \(S_{d(B^*)d(B^*)}^+=0\).
\end{proof}

\begin{proof}[Proof of Corollary \ref{cor:RA-oracle}]
If \(S_{d(B_0)d(B_0)}^+=0\), Theorem~\ref{thm:RA-fullobs} gives \(\Omega_R^+(r(B_0))=\Omega_R\). The completion-of-squares identity in the same theorem gives \(\Omega_R^+(r(B^*))\preceq\Omega_R^+(r(B_0))=\Omega_R\), because \(B^*\) is the population HAC projection coefficient. The lower bound in Theorem~\ref{thm:RA-fullobs} gives \(\Omega_R\preceq\Omega_R^+(r(B^*))\). Combining the two inequalities yields \(\Omega_R^+(r(B^*))=\Omega_R\). If \(\tilde\mu_{T,t}=B_0^\top z_t\) for every \(t\) along the sequence, then \(d_t(B_0)=0\) for every \(t\), so \(S_{d(B_0)d(B_0)}^+=0\).
\end{proof}

\begin{proof}[Proof of Proposition \ref{prop:RA-counterexample}]
The adjustment is exact on the retained scalar moment: \(r_t(B)=g_t-z_t=0\) for every date. Hence every sample lag product of the adjusted contribution is zero, and the adjusted HAC limit is \(\Omega_R^+(r(B))=0\).

For the unadjusted contribution, \(g_t=\mu_{T,t}\) is fixed under the assignment design, so there is no centered assignment innovation. The raw fixed-lag HAC limits are
\[
\Gamma_{g,\mathrm{hac}}(0)=\lim_{T\to\infty}\frac1T\sum_{t=1}^T\mu_{T,t}^2=\sigma_\mu^2,
\qquad
\Gamma_{g,\mathrm{hac}}(\ell)=0\quad(\ell\ne0),
\]
where the second equality is the vanishing of nonzero empirical lag products. Thus, under the random-sampling weights in \eqref{eq:OmegaR_sampling}, \(\Omega_R^{+,\mathrm{obs}}=\sigma_\mu^2\). Let \(\bar\mu_T:=T^{-1}\sum_{t=1}^T\mu_{T,t}\). The strengthened centering condition gives \(\bar\mu_T=o(T^{-1/2})\). For each fixed \(\ell\), the pointwise expansion $(\mu_{T,t}-\bar\mu_T)(\mu_{T,t-\ell}-\bar\mu_T)-\mu_{T,t}\mu_{T,t-\ell}=-\bar\mu_T\mu_{T,t}-\bar\mu_T\mu_{T,t-\ell}+\bar\mu_T^2$ gives
\[
\frac1T\sum_{t=\ell+1}^T(\mu_{T,t}-\bar\mu_T)(\mu_{T,t-\ell}-\bar\mu_T)
-\frac1T\sum_{t=\ell+1}^T\mu_{T,t}\mu_{T,t-\ell}=o(1),
\]
because the difference is bounded by constants times \(|\bar\mu_T|T^{-1}\sum_t|\mu_{T,t}|+\bar\mu_T^2\), and \(T^{-1}\sum_t|\mu_{T,t}|=O(1)\) by Cauchy--Schwarz. The centered mean-path long-run variance is therefore \(\Omega_\mu=\sigma_\mu^2\). Under i.i.d.\ Bernoulli sampling with rate \(\rho\), equation~\eqref{eq:OmegaR_sampling} gives \(\Omega_R^{\mathrm{obs}}=\Omega_R^{+,\mathrm{obs}}-\rho\Omega_\mu=(1-\rho)\sigma_\mu^2>0\).
Thus the adjusted observed-pair HAC limit is zero while the exact incomplete-observation design covariance is strictly positive. The fully observed lower bound therefore does not extend to incomplete random sampling without additional restrictions.
\end{proof}

\begin{proof}[Proof of Theorem \ref{thm:bonf}]
Let $c_B:=z_{1-\alpha/(2J)}$ and $c_S:=z_{\{1+(1-\alpha)^{1/J}\}/2}$. Since $\Pr(\widehat s_j>0\ \text{for all }j)=1-o(1)$, the Bonferroni event is equivalent up to $o(1)$ probability to $\{|T_{j,T}|\le c_B$ for all $j\}$, and the \v{S}id{\'a}k event is equivalent up to $o(1)$ probability to $\{|T_{j,T}|\le c_S$ for all $j\}$. The limiting Gaussian law assigns probability zero to the boundary of each of these rectangles. This is also true when the limiting covariance matrix is singular, because each boundary face is contained in an event of the form $\{Z_j=\pm c\}$ with $c>0$, while $Z_j$ is either nondegenerate normal or identically zero. Hence weak convergence on these continuity sets gives convergence of the corresponding joint coverage probabilities.

For Bonferroni, fix $j$. If $\tau_j=0$, then $\Pr(|Z_j|>c_B)=0$; if $0<\tau_j\le1$, write $Z_j=\tau_j Z_0$ with $Z_0\sim\mathcal N(0,1)$, so $\Pr(|Z_j|>c_B)\le\Pr(|Z_0|>c_B)=\alpha/J$. The union bound gives $\Pr(\max_{j\le J}|Z_j|>c_B)\le\alpha$, and therefore $\Pr(|Z_j|\le c_B\ \text{for all }j)\ge1-\alpha$. The preceding convergence gives the asserted Bonferroni coverage bound.

For \v{S}id{\'a}k, the same marginal argument gives $\Pr(|Z_j|\le c_S)\ge\Pr(|Z_0|\le c_S)=(1-\alpha)^{1/J}$ for every $j$. \v{S}id{\'a}k's Gaussian inequality for centered normal vectors gives $\Pr(|Z_j|\le c_S\ \text{for all }j)\ge\prod_{j=1}^J\Pr(|Z_j|\le c_S)$ when the covariance matrix is nonsingular. If it is singular, apply the nonsingular case to $Z+\varepsilon^{1/2}U$, where $U\sim\mathcal N(0,I_J)$ is independent of $Z$, and let $\varepsilon\downarrow0$ using the boundary-continuity argument above. Thus $\Pr(|Z_j|\le c_S\ \text{for all }j)\ge\prod_{j=1}^J\Pr(|Z_j|\le c_S)\ge1-\alpha$. The preceding convergence gives the asserted \v{S}id{\'a}k coverage bound.
\end{proof}

\subsection{LP--VAR, misspecification, and generated-regressor proofs}

\begin{proof}[Proof of Theorem \ref{thm:LP-local}]
Fix $h\in\mathcal H$. Since $Q_{h,T}\to Q_h\succ0$, $Q_{h,T}$ is nonsingular for all large $T$. The population LP moment can be written as $\bar m_{h,T}(\theta)=T^{-1}\sum_t\mathbb E_T[\psi_t y_{t+h}]-Q_{h,T}\theta$. Hence the normal equation defining $\theta_{h,T}^\star$ is exactly \(0=\bar m_{h,T}(\theta_{0,h})-Q_{h,T}(\theta_{h,T}^\star-\theta_{0,h})\),
and therefore \(\theta_{h,T}^\star-\theta_{0,h}=Q_{h,T}^{-1}\bar m_{h,T}(\theta_{0,h})=T^{-1/2}Q_{h,T}^{-1}d_{h,T}+o(T^{-1/2})\).
It remains to identify the shock coordinate of the limiting drift. Permute the coordinates so that the shock coordinate appears first and the nuisance block $(1,c_t^\top)^\top$ appears second. The restrictions $T^{-1}\sum_t\mathbb E_T[x_t]\to0$ and $T^{-1}\sum_t\mathbb E_T[x_t c_t^\top]\to0$ imply that the permuted limit of $Q_{h,T}$ is block diagonal between these two blocks. Write $Q_h=\begin{psmallmatrix}q_x&0\\0&Q_c\end{psmallmatrix}$ and $d_h=(0,d_c^\top)^\top$. Since $Q_h\succ0$, $q_x>0$ and $Q_c\succ0$, and $Q_h^{-1}d_h=(0,(Q_c^{-1}d_c)^\top)^\top$. Hence the shock coordinate of $Q_h^{-1}d_h$ is zero. Thus \(\sqrt T(\beta_{h,T}^\star-\beta_{0,h})=s_x^\top Q_{h,T}^{-1}d_{h,T}+o(1)\to s_x^\top Q_h^{-1}d_h=0\).
The design-based CLT gives $\sqrt T(\widehat\beta_h^{LP}-\beta_{h,T}^\star)\Rightarrow\mathcal N(0,v_{LP,h})$. Adding the preceding deterministic $o(1)$ shift yields the stated centered limit around $\beta_{0,h}$.
\end{proof}

\begin{proof}[Proof of Theorem \ref{thm:VAR-local}]
Write $m_T(\phi):=\bar m_T^{VAR}(\phi)$, $G_T:=G_{\phi,T}$, $A:=A_\phi$, and $\delta_T:=\phi_T^\star-\phi_0$. Since $G_T\to G_\phi$, $G_\phi$ has full column rank, and $A\succ0$, the matrix $M_T:=G_T^\top A G_T$ is nonsingular for all large $T$ and $M_T^{-1}$ is bounded. The local expansion at $\phi_T^\star$ gives \(m_T(\phi_T^\star)=T^{-1/2}d_{VAR,T}+G_T\delta_T+a_T(\phi_T^\star)+b_T\), where $\|a_T(\phi_T^\star)\|=o(\|\delta_T\|)+o(T^{-1/2})$ and $\|b_T\|=o(T^{-1/2})$. The interior first-order condition is $(D_{\phi,T}^\star)^\top A m_T(\phi_T^\star)=0$, and $D_{\phi,T}^\star=G_T+E_T$ with $\|E_T\|=o(1)$. Substitution gives the central system
\[
\begin{aligned}
0&=M_T\delta_T+G_T^\top A T^{-1/2}d_{VAR,T}+G_T^\top A a_T(\phi_T^\star)+c_T, \\
c_T&:=G_T^\top A b_T+E_T^\top A\big[T^{-1/2}d_{VAR,T}+G_T\delta_T+a_T(\phi_T^\star)+b_T\big].
\end{aligned}
\]
Since $d_{VAR,T}=O(1)$, $G_T=O(1)$, $\|E_T\|=o(1)$, and $\|a_T(\phi_T^\star)\|=o(\|\delta_T\|)+o(T^{-1/2})$, this definition gives $\|c_T\|=o(T^{-1/2})+o(\|\delta_T\|)$. Boundedness of $M_T^{-1}$ then implies $\|\delta_T\|\le CT^{-1/2}+\eta_T\|\delta_T\|+o(T^{-1/2})$ with $\eta_T\to0$. For all large $T$, absorbing the middle term gives $\delta_T=O(T^{-1/2})$. Consequently $a_T(\phi_T^\star)=o(T^{-1/2})$, and the first-order condition reduces to \(M_T\sqrt T(\phi_T^\star-\phi_0)=-G_T^\top A d_{VAR,T}+o(1)\). Thus $\sqrt T(\phi_T^\star-\phi_0)=-H_{\phi,T}d_{VAR,T}+o(1)\to-H_\phi d_{VAR}$, where $H_{\phi,T}:=(G_T^\top A G_T)^{-1}G_T^\top A$.

The design-based CLT gives $\widehat\phi-\phi_T^\star=O_p(T^{-1/2})$, and the preceding identity gives $\phi_T^\star-\phi_0=O(T^{-1/2})$; hence $\widehat\phi\to_p\phi_0$. Differentiability of $r$ at $\phi_0$ yields \(\sqrt T\big(r(\widehat\phi)-r(\phi_0)\big)=R\sqrt T(\widehat\phi-\phi_T^\star)+R\sqrt T(\phi_T^\star-\phi_0)+o_p(1)\).
The first term converges to $\mathcal N(0,R\Sigma_\phi R^\top)$, and the second term converges to $-RH_\phi d_{VAR}$. Slutsky's theorem gives the stated limit.
\end{proof}

\begin{proof}[Proof of Proposition \ref{prop:tangent}]
Since $\Sigma_{LP}\succ0$ and $C$ has full row rank, $C\Sigma_{LP}C^\top\succ0$. Let $X\sim\mathcal N(0,\Sigma_{LP})$. The joint vector $(X,CX)$ is Gaussian with covariance blocks $\Var(X)=\Sigma_{LP}$, $\Cov(X,CX)=\Sigma_{LP}C^\top$, and $\Var(CX)=C\Sigma_{LP}C^\top$. The Gaussian conditioning formula gives \(\Var(X\mid CX=0)=\Sigma_{LP}-\Sigma_{LP}C^\top(C\Sigma_{LP}C^\top)^{-1}C\Sigma_{LP}\).
The subtracted matrix is positive semidefinite because $\Sigma_{LP}C^\top(C\Sigma_{LP}C^\top)^{-1}C\Sigma_{LP}=M^\top M$, where $M:=(C\Sigma_{LP}C^\top)^{-1/2}C\Sigma_{LP}$. Hence the conditional covariance is weakly smaller than $\Sigma_{LP}$ in Loewner order. The tangent-space equivalence identifies the VAR influence vector with this restricted Gaussian influence vector on $\ker(C)=\operatorname{Im}(R)$. Therefore the conditional covariance above is the VAR IRF covariance under the stated local restriction, and the Loewner inequality follows.
\end{proof}

\begin{proof}[Proof of Proposition \ref{prop:LP-VAR-drift-design}]
The first sentence follows from Theorem~\ref{thm:LP-local}, which shows $\sqrt T(\beta_{h,T}^\star-\beta_{0,h})\to0$ for each fixed $h\in\mathcal H$. The second sentence follows from Theorem~\ref{thm:VAR-local}, which gives the local pseudo-true shift $\sqrt T(\phi_T^\star-\phi_0)\to-H_\phi d_{VAR}$ and therefore the impulse-response drift $b_{VAR}=-RH_\phi d_{VAR}$. Proposition~\ref{prop:tangent} supplies the separate restriction under which the VAR covariance is the covariance of the Gaussian LP limit after imposing the linear tangent restrictions. Without that tangent-space restriction, the preceding theorems compare local centering and local drift, but they do not order the covariance matrices.
\end{proof}

\begin{proof}[Proof of Proposition \ref{prop:misspecified-gmm}]
Write $G_T:=G_T(\theta_T^\star)$, $J_{T,t}:=J_{T,t}(\theta_T^\star)$, $\Delta_{e,T}:=T^{-1}\sum_t e_{T,t}$, and $\Delta_{J,T}:=T^{-1}\sum_t(J_{T,t}-G_T)$. The consistency argument in Theorem~\ref{thm:AN} uses uniform criterion convergence and moving-estimand separation, not Assumption~\ref{ass:local-correct-specification}; it therefore gives $\widehat\theta_N-\theta_T^\star\to_p0$ for the fixed-weight criterion. On the high-probability event where the sample first-order condition holds, the local expansion gives \(0=\widehat G_N(\theta_T^\star)^\top A g_N(\theta_T^\star)+\mathcal D_T(\widehat\theta_N-\theta_T^\star)+o_p(T^{-1/2})\).
The population first-order condition at the interior pseudo-true value is $G_T^\top A\bar m_T^\star=0$. At $\theta_T^\star$, the sample moment and sample Jacobian decompose as $g_N(\theta_T^\star)=\bar m_T^\star+\Delta_{e,T}$ and $\widehat G_N(\theta_T^\star)=G_T+\Delta_{J,T}$. Substituting these decompositions and using the population first-order condition yields
\[
\begin{aligned}
\widehat G_N(\theta_T^\star)^\top A g_N(\theta_T^\star)
&=G_T^\top A\Delta_{e,T}+\Delta_{J,T}^\top A\bar m_T^\star
   +\Delta_{J,T}^\top A\Delta_{e,T}  \\
&=T^{-1}\sum_{t=1}^T\zeta_{T,t}+\Delta_{J,T}^\top A\Delta_{e,T} .
\end{aligned}
\]
The two root-$T$ boundedness assumptions give $\Delta_{J,T}=O_p(T^{-1/2})$ and $\Delta_{e,T}=O_p(T^{-1/2})$, so $\Delta_{J,T}^\top A\Delta_{e,T}=O_p(T^{-1})=o_p(T^{-1/2})$. Hence \(\widehat G_N(\theta_T^\star)^\top A g_N(\theta_T^\star)=T^{-1}\sum_{t=1}^T\zeta_{T,t}+o_p(T^{-1/2})\).
Multiplying the first-order condition by $\sqrt T$ gives \(\sqrt T(\widehat\theta_N-\theta_T^\star)=-\mathcal D_T^{-1}T^{-1/2}\sum_{t=1}^T\zeta_{T,t}+o_p(1)\),
because $\mathcal D_T\to\mathcal D$ and $\mathcal D$ is nonsingular. Since $\mathcal D_T^{-1}\to\mathcal D^{-1}$ and $T^{-1/2}\sum_t\zeta_{T,t}\Rightarrow\mathcal N(0,\Omega_\zeta)$, Slutsky's theorem gives the stated Gaussian limit.
\end{proof}

\begin{proof}[Proof of Theorem \ref{thm:twostep}]
Let $g_N(\theta,\pi):=T^{-1}\sum_tg_t(\theta,\pi)$. The first-stage influence representation implies $\widehat\pi-\pi_T^\star=O_p(T^{-1/2})$. Because $G_{\theta,T}\to G_\theta$, $G_\theta$ has full column rank, and $A\succ0$, the matrix $G_{\theta,T}^\top A G_{\theta,T}$ is nonsingular for all large $T$. The local expansion of the second-stage first-order condition gives
\[
0=G_{\theta,T}^\top A\Big[g_N(\theta_T^\star,\pi_T^\star)+G_{\theta,T}(\widehat\theta-\theta_T^\star)+G_{\pi,T}(\widehat\pi-\pi_T^\star)\Big]+o_p(T^{-1/2}).
\]
Solving this linear system yields
\[
\sqrt T(\widehat\theta-\theta_T^\star)
=-H_{\theta,T}\Big[\sqrt T g_N(\theta_T^\star,\pi_T^\star)+G_{\pi,T}\sqrt T(\widehat\pi-\pi_T^\star)\Big]+o_p(1).
\]
Since $G_{\pi,T}\to G_\pi$, substituting the first-stage representation gives
\[
\sqrt T g_N(\theta_T^\star,\pi_T^\star)+G_{\pi,T}\sqrt T(\widehat\pi-\pi_T^\star)
=T^{-1/2}\sum_{t=1}^T\big(g_t(\theta_T^\star,\pi_T^\star)+G_{\pi,T}IF_{\pi,T,t}\big)+o_p(1).
\]
Define $\widetilde g_{T,t}:=g_t(\theta_T^\star,\pi_T^\star)+G_{\pi,T}IF_{\pi,T,t}$. The preceding displays give
\[
\sqrt T(\widehat\theta-\theta_T^\star)
=-H_{\theta,T}T^{-1/2}\sum_{t=1}^T\widetilde g_{T,t}+o_p(1),
\]
which is the augmented-moment representation. If $\widetilde e_{T,t}:=\widetilde g_{T,t}-T^{-1}\sum_s\mathbb E_T[\widetilde g_{T,s}]$, augmented local correct specification gives $T^{-1/2}\sum_t\widetilde g_{T,t}=T^{-1/2}\sum_t\widetilde e_{T,t}+o_p(1)$. The maintained augmented CLT gives the limit of $T^{-1/2}\sum_t\widetilde e_{T,t}$, and the augmented mean-path and product-array conditions define the corresponding $\Omega_R$ and $\Omega_R^+$. With these replacements, the first-order map is $H_{\theta}:=(G_\theta^\top A G_\theta)^{-1}G_\theta^\top A$, so the proofs of Theorems~\ref{thm:AN}, \ref{thm:hac}, and~\ref{thm:bootstrap} apply with $g_t$ replaced by $\widetilde g_{T,t}$ and limiting Jacobian $G_\theta$.
\end{proof}

\begin{proof}[Proof of Corollary \ref{cor:stacked}]
Applying Theorem~\ref{thm:AN} to the stacked just-identified system gives \(\sqrt T(\widehat\eta-\eta_T^\star)=-G_{h,T}^{-1}T^{-1/2}\sum_{t=1}^T h_t(\eta_T^\star)+o_p(1)\).
For the block lower-triangular Jacobian in the corollary,
\[
G_{h,T}^{-1}
=\begin{bmatrix}
G_{a,T}^{-1} & 0\\
-G_{\theta,T}^{-1}G_{\pi,T}G_{a,T}^{-1} & G_{\theta,T}^{-1}
\end{bmatrix}.
\]
The first-stage block is therefore \(\sqrt T(\widehat\pi-\pi_T^\star)=-G_{a,T}^{-1}T^{-1/2}\sum_{t=1}^T a_t(\pi_T^\star)+o_p(1)\), so $IF_{\pi,T,t}=-G_{a,T}^{-1}a_t(\pi_T^\star)$.
\[
\begin{aligned}
\sqrt T(\widehat\theta-\theta_T^\star)
&=G_{\theta,T}^{-1}G_{\pi,T}G_{a,T}^{-1}T^{-1/2}\sum_{t=1}^T a_t(\pi_T^\star)
  -G_{\theta,T}^{-1}T^{-1/2}\sum_{t=1}^T g_t(\theta_T^\star,\pi_T^\star)+o_p(1)\\
&=-G_{\theta,T}^{-1}T^{-1/2}\sum_{t=1}^T\big(g_t(\theta_T^\star,\pi_T^\star)+G_{\pi,T}IF_{\pi,T,t}\big)+o_p(1).
\end{aligned}
\]
This is the augmented-moment representation in Theorem~\ref{thm:twostep}, specialized to a just-identified second-stage block. The stacked covariance matrix is the long-run covariance of $h_t(\eta_T^\star)$; equivalently, the second-stage covariance can be written using the long-run covariance of the displayed augmented moment vector.
\end{proof}

\section{Additional Monte Carlo results}\label{app:mc-extra}

\subsection{No-drift diagnostic}\label{app:mc-nodrift}
The no-mean-path example sets $\tau_t \equiv 1$. The mean path is absent, so the design variance and the conservative variance limit should be nearly identical. Table~\ref{tab:mc_nodrift} reports this case. At $h=0$, both $V_H^{+}/V_R$ and $V_{RA}^{+}/V_R$ equal $1.000$, HAC and RA-HAC coverage are $0.885$ and $0.881$, and mean interval length is $0.202$ and $0.199$. The same near-equality persists at the longer selected horizons, so the remaining coverage deviations are consistent with finite-sample approximation error rather than a mean-path effect.

\begin{table}[!htbp]
  \centering
  \begin{tabular}{lcccccc}
\toprule
$h$ & $V_H^+/V_R$ & $V_{RA}^+/V_R$ & \shortstack{HAC\\cov.} & \shortstack{RA-HAC\\cov.} & \shortstack{HAC\\len.} & \shortstack{RA-HAC\\len.} \\
\midrule
0 & 1.000 (0.004) & 1.000 (0.004) & 0.885 (0.004) & 0.881 (0.005) & 0.202 (0.002) & 0.199 (0.002) \\
4 & 1.012 (0.006) & 1.012 (0.006) & 0.864 (0.005) & 0.860 (0.005) & 0.597 (0.008) & 0.590 (0.008) \\
8 & 0.999 (0.011) & 0.999 (0.011) & 0.855 (0.004) & 0.851 (0.004) & 0.678 (0.012) & 0.670 (0.012) \\
12 & 0.992 (0.008) & 0.992 (0.008) & 0.861 (0.004) & 0.855 (0.004) & 0.695 (0.014) & 0.687 (0.014) \\
\bottomrule
\end{tabular}

  \caption{No-drift diagnostic. RA-HAC denotes regression-adjusted HAC. Entries are averages across fixed environments; parenthesized values are standard errors across environments.}
  \label{tab:mc_nodrift}
\end{table}

\subsection{Misspecified adjustment}\label{app:mc-misspecified}
Misspecified Design C keeps $\tau_t = 1 + c_t$, but the feasible regression adjustment replaces $c_t$ by $y_{t-1}$. The change is substantive because the adjustment variable now varies across shock re-randomizations; the corresponding projection objects are therefore formed by averaging their sample analogues across the reference simulation rather than from a fixed covariate path. In this design, the adjustment variable is predetermined in timing but generated by the same dynamic system being re-randomized, so the fixed-covariate adjustment logic does not apply.

Table~\ref{tab:mc_designC_app} shows that the improvement is negligible. At $h=0$, $(V_H^{+}/V_R, V_{RA}^{+}/V_R, V_O/V_R)=(6.026,5.938,0.999)$, so the misspecified regression adjustment remains close to HAC rather than moving toward the oracle. The same conclusion appears in the interval diagnostics: at $h=0$, HAC and RA-HAC both have coverage essentially equal to one, while mean interval length falls only from $0.498$ to $0.490$, far above the oracle length $0.210$. Figure~\ref{fig:mc_designC_paths} shows that this lack of improvement persists across horizons. The adjustment does not reduce standard errors automatically; its value depends on whether the predetermined covariate captures the relevant predictable date-specific drift.

\begin{table}[p]
  \centering
  {\small
\renewcommand{\arraystretch}{1.05}
\setlength{\tabcolsep}{5pt}
\begin{tabular}{@{}lccc@{}}
\toprule
\multicolumn{4}{@{}l}{\emph{Panel A. Variance ratios}}\\
\midrule
$h$ & $V_H^+/V_R$ & $V_{RA}^+/V_R$ & $V_O/V_R$ \\
\midrule
0 & 6.026 (0.334) & 5.938 (0.324) & 0.999 (0.004) \\
4 & 1.136 (0.014) & 1.133 (0.014) & 1.011 (0.008) \\
8 & 1.028 (0.011) & 1.027 (0.010) & 1.001 (0.010) \\
12 & 1.000 (0.007) & 1.000 (0.007) & 0.993 (0.006) \\
\bottomrule
\end{tabular}

\vspace{0.8em}

\begin{tabular}{@{}lcccccc@{}}
\toprule
\multicolumn{7}{@{}l}{\emph{Panel B. Pointwise interval diagnostics}}\\
\midrule
$h$ & \shortstack{HAC\\cov.} & \shortstack{RA-HAC\\cov.} & \shortstack{Oracle\\cov.} & \shortstack{HAC\\len.} & \shortstack{RA-HAC\\len.} & \shortstack{Oracle\\len.} \\
\midrule
0 & 0.999 (0.000) & 0.999 (0.001) & 0.893 (0.006) & 0.498 (0.014) & 0.490 (0.014) & 0.210 (0.002) \\
4 & 0.875 (0.006) & 0.872 (0.006) & 0.874 (0.006) & 0.699 (0.016) & 0.690 (0.016) & 0.681 (0.016) \\
8 & 0.866 (0.006) & 0.861 (0.006) & 0.873 (0.006) & 0.763 (0.020) & 0.754 (0.020) & 0.773 (0.020) \\
12 & 0.866 (0.006) & 0.861 (0.006) & 0.875 (0.005) & 0.778 (0.021) & 0.769 (0.021) & 0.792 (0.022) \\
\bottomrule
\end{tabular}
}%

  \caption{Misspecified adjustment design (Design C). Panel A reports variance ratios, and Panel B reports pointwise coverage and mean interval length. RA-HAC denotes regression-adjusted HAC. Entries are averages across fixed environments; parenthesized values are standard errors across environments.}
  \label{tab:mc_designC_app}
\end{table}

\begin{figure}[!htbp]
  \centering
  \includegraphics[width=0.82\linewidth]{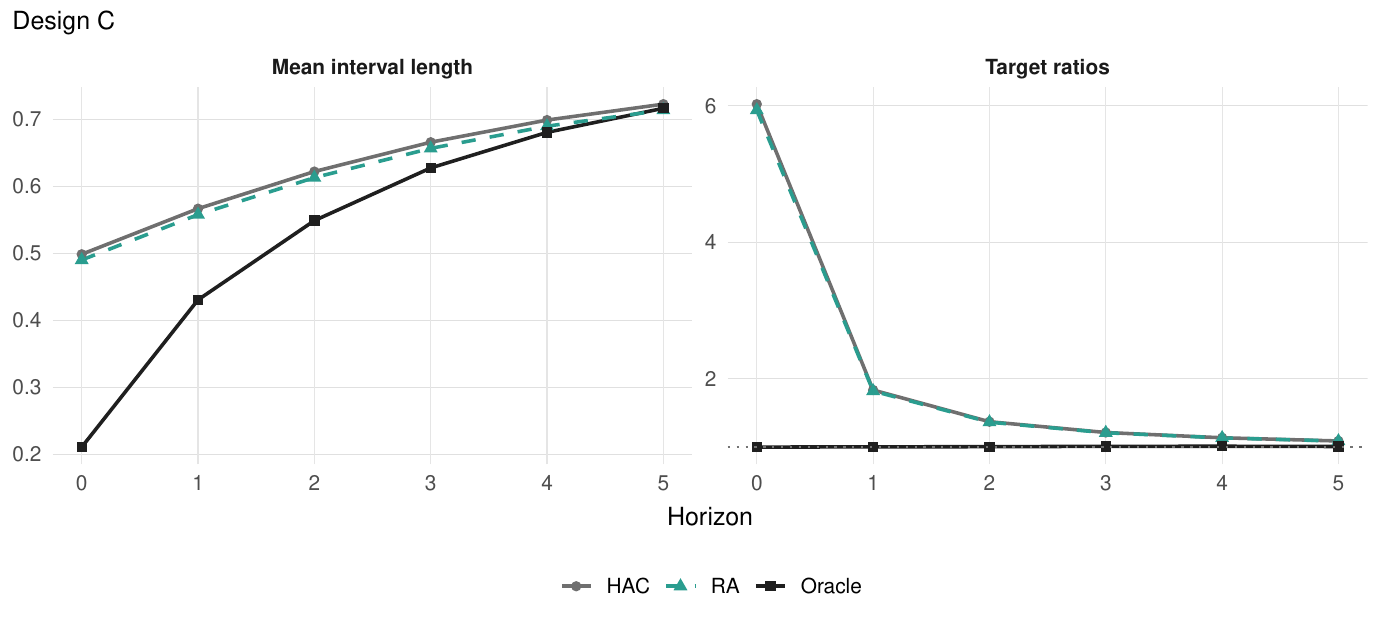}
  \caption{Misspecified adjustment design. The left panel reports variance ratios, and the right panel reports mean interval length. Because the predetermined adjustment is misspecified, the RA-HAC path remains close to HAC rather than moving toward the oracle.}
  \label{fig:mc_designC_paths}
\end{figure}

\subsection{Sample-size comparison}\label{app:mc-sample-size}
Table~\ref{tab:mc_sample_size} summarizes the same three designs at $T \in \{120,240,600\}$, averaging each metric over horizons $h=0,\ldots,12$ within environment and then across environments. Because the table averages over all horizons, it is best viewed as a summary of overall scale rather than as the main evidence on short-horizon conservativeness. The conservative-versus-adjusted pattern is present throughout. At $T=240$, Design A has average HAC and RA-HAC ratios $1.551$ and $1.066$, whereas Design C has $1.551$ and $1.530$; at $T=600$ the same comparison is $1.698$ versus $1.035$ in Design A and $1.698$ versus $1.695$ in Design C. Average interval length falls mechanically with $T$, but the distinction between aligned and misspecified adjustment remains visible in the horizon-averaged variance ratios.

\begin{table}[!htbp]
  \centering
  {\small
\renewcommand{\arraystretch}{1.05}
\setlength{\tabcolsep}{3pt}
\begin{tabular}{@{}llcccccc@{}}
\toprule
$T$ & Design & \shortstack{avg. HAC\\ratio} & \shortstack{avg. RA-HAC\\ratio} & \shortstack{avg. HAC\\cov.} & \shortstack{avg. RA-HAC\\cov.} & \shortstack{avg. HAC\\len.} & \shortstack{avg. RA-HAC\\len.} \\
\midrule
120 & A & 1.446 (0.089) & 1.046 (0.015) & 0.871 (0.006) & 0.852 (0.006) & 0.980 (0.037) & 0.919 (0.031) \\
 & B & 1.774 (0.160) & 1.383 (0.105) & 0.858 (0.010) & 0.846 (0.010) & 1.095 (0.083) & 1.042 (0.079) \\
 & C & 1.446 (0.089) & 1.441 (0.089) & 0.871 (0.006) & 0.865 (0.006) & 0.980 (0.037) & 0.964 (0.036) \\
\midrule
240 & A & 1.551 (0.057) & 1.066 (0.017) & 0.899 (0.004) & 0.880 (0.005) & 0.692 (0.021) & 0.641 (0.020) \\
 & B & 1.777 (0.117) & 1.393 (0.068) & 0.895 (0.005) & 0.886 (0.005) & 0.744 (0.042) & 0.706 (0.037) \\
 & C & 1.551 (0.057) & 1.530 (0.051) & 0.899 (0.004) & 0.895 (0.004) & 0.692 (0.021) & 0.682 (0.021) \\
\midrule
600 & A & 1.698 (0.044) & 1.035 (0.009) & 0.909 (0.005) & 0.883 (0.006) & 0.465 (0.009) & 0.423 (0.008) \\
 & B & 2.248 (0.164) & 1.604 (0.080) & 0.909 (0.006) & 0.898 (0.005) & 0.534 (0.024) & 0.500 (0.020) \\
 & C & 1.698 (0.044) & 1.695 (0.043) & 0.909 (0.005) & 0.906 (0.005) & 0.465 (0.009) & 0.462 (0.009) \\
\bottomrule
\end{tabular}
}%

  \caption{Sample-size comparison, horizon-averaged over $h=0,\ldots,12$. RA-HAC denotes regression-adjusted HAC. Entries are averages across fixed environments; parenthesized values are standard errors across environments.}
  \label{tab:mc_sample_size}
\end{table}

\subsection{Comparison with common LP inference}\label{app:mc-macro-practice}
This subsection keeps the LP estimator fixed and changes only the inference procedure. In addition to Bartlett HAC, regression adjustment, and the oracle, it reports a conventional HAC rule that sets the Bartlett bandwidth equal to $h+1$, heteroskedasticity-robust intervals without long-run correction, and a heteroskedasticity-robust Rademacher multiplier bootstrap with $499$ bootstrap draws in each replication. The comparison is intended to show how the design-based procedures compare with common empirical LP choices when the estimator itself is held fixed.

Tables~\ref{tab:mc_macro_variance_selected}--\ref{tab:mc_macro_interval_selected} compare the procedures at the selected horizons $h=0,4,8,12$. In aligned Design A at $h=0$, the variance ratios are $6.026$ for HAC, $1.633$ for regression-adjusted HAC, $2.765$ for the $h+1$ HAC rule, $1.935$ for heteroskedasticity-robust inference, and $0.999$ for the oracle. The corresponding mean interval lengths are $0.498$, $0.258$, $0.343$, $0.290$, and $0.210$, so regression-adjusted HAC is the closest feasible procedure to the oracle among the methods reported here.

In nonlinear Design B at $h=0$, regression-adjusted HAC still improves materially on baseline HAC, but it is no longer uniformly tighter than the conventional alternatives: the variance ratios are $9.939$, $4.821$, $4.489$, $2.904$, and $0.998$ for HAC, regression-adjusted HAC, the $h+1$ HAC rule, heteroskedasticity-robust inference, and the oracle. In misspecified Design C, regression-adjusted HAC remains close to HAC throughout. Figure~\ref{fig:mc_macro_selected} shows the corresponding horizon paths for Designs A and C. As in the main text, regression adjustment is most effective when the predetermined adjustment tracks the mean path, and its contribution is limited otherwise.

\begin{table}[p]
  \centering
  {\footnotesize
\renewcommand{\arraystretch}{1.05}
\setlength{\tabcolsep}{3.5pt}
\begin{tabular}{@{}llccccc@{}}
\toprule
Design & $h$ & HAC & RA-HAC & $h{+}1$ HAC & HC & Oracle \\
\midrule
A & 0 & 6.026 (0.334) & 1.633 (0.131) & 2.765 (0.101) & 1.935 (0.051) & 0.999 (0.004) \\
 & 4 & 1.136 (0.014) & 1.028 (0.008) & 1.124 (0.014) & 1.053 (0.016) & 1.011 (0.008) \\
 & 8 & 1.028 (0.011) & 1.005 (0.010) & 1.036 (0.009) & 1.003 (0.015) & 1.001 (0.010) \\
 & 12 & 1.000 (0.007) & 0.994 (0.006) & 1.007 (0.005) & 0.989 (0.008) & 0.993 (0.006) \\
\midrule
B & 0 & 9.939 (1.206) & 4.821 (0.446) & 4.489 (0.430) & 2.904 (0.229) & 0.998 (0.004) \\
 & 4 & 1.170 (0.016) & 1.082 (0.011) & 1.156 (0.015) & 1.056 (0.013) & 1.006 (0.006) \\
 & 8 & 1.035 (0.008) & 1.017 (0.008) & 1.042 (0.008) & 1.006 (0.011) & 1.001 (0.007) \\
 & 12 & 1.009 (0.006) & 1.004 (0.006) & 1.014 (0.005) & 0.998 (0.007) & 1.000 (0.006) \\
\midrule
C & 0 & 6.026 (0.334) & 5.938 (0.324) & 2.765 (0.101) & 1.935 (0.051) & 0.999 (0.004) \\
 & 4 & 1.136 (0.014) & 1.133 (0.014) & 1.124 (0.014) & 1.053 (0.016) & 1.011 (0.008) \\
 & 8 & 1.028 (0.011) & 1.027 (0.010) & 1.036 (0.009) & 1.003 (0.015) & 1.001 (0.010) \\
 & 12 & 1.000 (0.007) & 1.000 (0.007) & 1.007 (0.005) & 0.989 (0.008) & 0.993 (0.006) \\
\bottomrule
\end{tabular}
}%

  \caption{Common local-projection inference choices: selected-horizon variance ratios relative to $V_R$. RA-HAC denotes regression-adjusted HAC, and HC denotes heteroskedasticity-robust inference without long-run correction. Entries are averages across fixed environments; parenthesized values are standard errors across environments.}
  \label{tab:mc_macro_variance_selected}
\end{table}

\begin{table}[p]
  \centering
  {\footnotesize
\renewcommand{\arraystretch}{1.05}
\setlength{\tabcolsep}{3.5pt}
\begin{tabular}{@{}llcccccc@{}}
\toprule
\multicolumn{8}{@{}l}{\emph{Panel A. Coverage}}\\
\midrule
Design & $h$ & HAC & RA-HAC & $h{+}1$ HAC & HC & HC boot. & Oracle \\
\midrule
A & 0 & 0.999 (0.000) & 0.938 (0.009) & 0.991 (0.002) & 0.973 (0.003) & 0.973 (0.003) & 0.893 (0.006) \\
 & 4 & 0.875 (0.006) & 0.863 (0.006) & 0.877 (0.006) & 0.887 (0.006) & 0.884 (0.005) & 0.874 (0.006) \\
 & 8 & 0.866 (0.006) & 0.860 (0.006) & 0.858 (0.006) & 0.880 (0.006) & 0.878 (0.006) & 0.873 (0.006) \\
 & 12 & 0.866 (0.006) & 0.859 (0.006) & 0.851 (0.005) & 0.882 (0.005) & 0.882 (0.005) & 0.875 (0.005) \\
\midrule
B & 0 & 0.998 (0.001) & 0.994 (0.002) & 0.993 (0.003) & 0.984 (0.004) & 0.984 (0.004) & 0.898 (0.005) \\
 & 4 & 0.870 (0.004) & 0.858 (0.005) & 0.872 (0.004) & 0.882 (0.004) & 0.878 (0.005) & 0.866 (0.005) \\
 & 8 & 0.859 (0.004) & 0.853 (0.004) & 0.848 (0.004) & 0.874 (0.005) & 0.871 (0.005) & 0.865 (0.004) \\
 & 12 & 0.869 (0.005) & 0.864 (0.005) & 0.853 (0.006) & 0.885 (0.005) & 0.880 (0.005) & 0.879 (0.005) \\
\midrule
C & 0 & 0.999 (0.000) & 0.999 (0.001) & 0.991 (0.002) & 0.973 (0.003) & 0.973 (0.003) & 0.893 (0.006) \\
 & 4 & 0.875 (0.006) & 0.872 (0.006) & 0.877 (0.006) & 0.887 (0.006) & 0.884 (0.005) & 0.874 (0.006) \\
 & 8 & 0.866 (0.006) & 0.861 (0.006) & 0.858 (0.006) & 0.880 (0.006) & 0.878 (0.006) & 0.873 (0.006) \\
 & 12 & 0.866 (0.006) & 0.861 (0.006) & 0.851 (0.005) & 0.882 (0.005) & 0.882 (0.005) & 0.875 (0.005) \\
\bottomrule
\end{tabular}

\vspace{0.8em}

\begin{tabular}{@{}llcccccc@{}}
\toprule
\multicolumn{8}{@{}l}{\emph{Panel B. Mean interval length}}\\
\midrule
Design & $h$ & HAC & RA-HAC & $h{+}1$ HAC & HC & HC boot. & Oracle \\
\midrule
A & 0 & 0.498 (0.014) & 0.258 (0.011) & 0.343 (0.007) & 0.290 (0.005) & 0.288 (0.005) & 0.210 (0.002) \\
 & 4 & 0.699 (0.016) & 0.664 (0.016) & 0.699 (0.016) & 0.706 (0.017) & 0.702 (0.017) & 0.681 (0.016) \\
 & 8 & 0.763 (0.020) & 0.749 (0.020) & 0.751 (0.019) & 0.791 (0.021) & 0.786 (0.021) & 0.773 (0.020) \\
 & 12 & 0.778 (0.021) & 0.767 (0.021) & 0.750 (0.020) & 0.810 (0.022) & 0.805 (0.022) & 0.792 (0.022) \\
\midrule
B & 0 & 0.621 (0.043) & 0.439 (0.022) & 0.428 (0.023) & 0.351 (0.016) & 0.349 (0.016) & 0.218 (0.003) \\
 & 4 & 0.770 (0.037) & 0.734 (0.033) & 0.771 (0.037) & 0.774 (0.035) & 0.769 (0.035) & 0.744 (0.033) \\
 & 8 & 0.828 (0.038) & 0.814 (0.037) & 0.815 (0.037) & 0.861 (0.040) & 0.856 (0.040) & 0.838 (0.038) \\
 & 12 & 0.844 (0.039) & 0.834 (0.038) & 0.815 (0.037) & 0.881 (0.041) & 0.876 (0.040) & 0.859 (0.039) \\
\midrule
C & 0 & 0.498 (0.014) & 0.490 (0.014) & 0.343 (0.007) & 0.290 (0.005) & 0.288 (0.005) & 0.210 (0.002) \\
 & 4 & 0.699 (0.016) & 0.690 (0.016) & 0.699 (0.016) & 0.706 (0.017) & 0.702 (0.017) & 0.681 (0.016) \\
 & 8 & 0.763 (0.020) & 0.754 (0.020) & 0.751 (0.019) & 0.791 (0.021) & 0.786 (0.021) & 0.773 (0.020) \\
 & 12 & 0.778 (0.021) & 0.769 (0.021) & 0.750 (0.020) & 0.810 (0.022) & 0.805 (0.022) & 0.792 (0.022) \\
\bottomrule
\end{tabular}
}%

  \caption{Common local-projection inference choices: selected-horizon interval diagnostics. Panel A reports pointwise coverage, and Panel B reports mean interval length. RA-HAC denotes regression-adjusted HAC, HC denotes heteroskedasticity-robust inference without long-run correction, and HC boot.\ denotes the heteroskedasticity-robust Rademacher multiplier bootstrap. Entries are averages across fixed environments; parenthesized values are standard errors across environments.}
  \label{tab:mc_macro_interval_selected}
\end{table}

\begin{figure}[!htbp]
  \centering
  \includegraphics[width=0.84\linewidth]{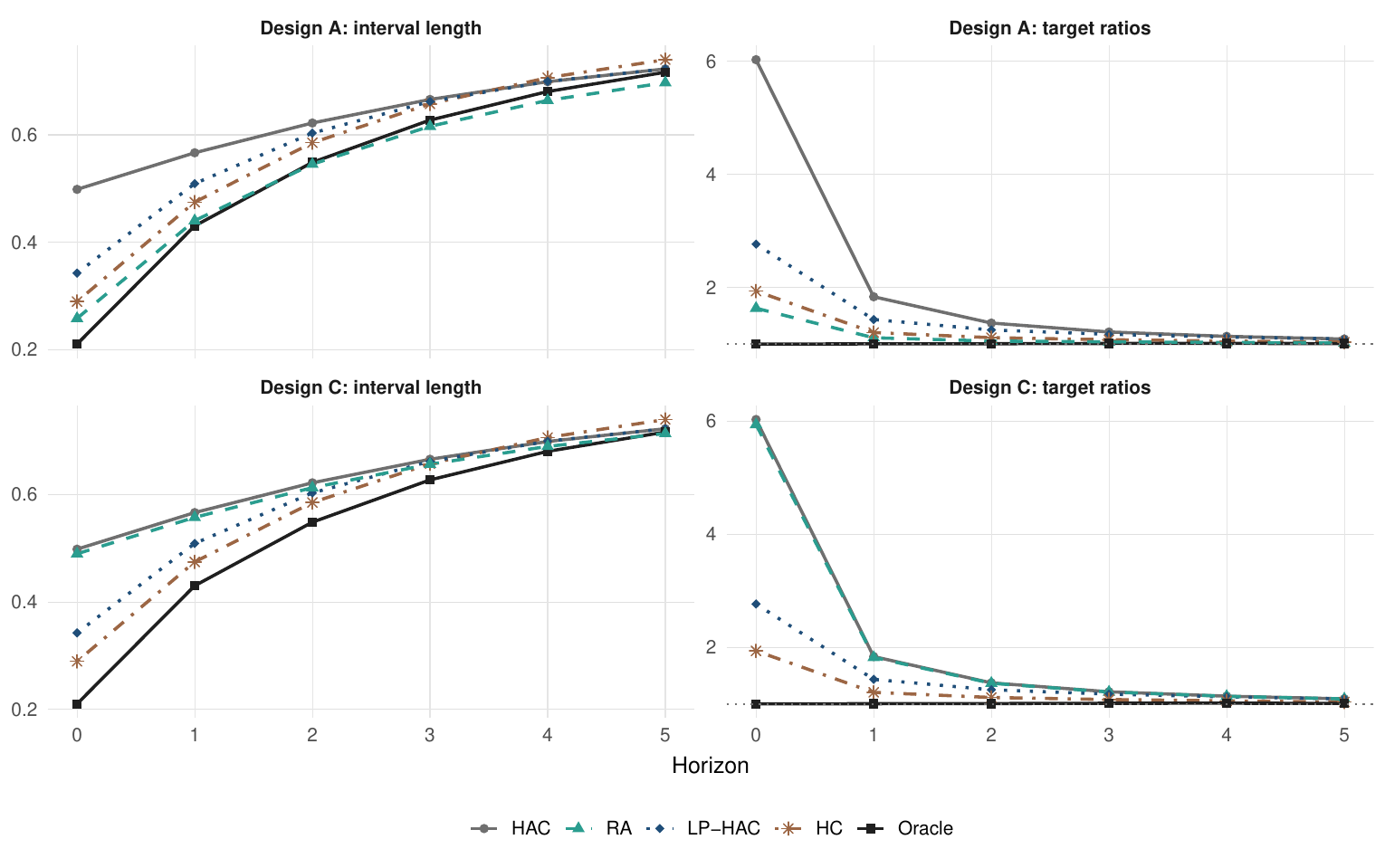}
  \caption{Common local-projection inference choices for Designs A and C. The top row reports variance ratios and the bottom row mean interval length by horizon. In the aligned design, the RA-HAC path is the closest feasible curve to the oracle at short horizons; in the misspecified design, it remains near HAC.}
  \label{fig:mc_macro_selected}
\end{figure}

\subsection{Bandwidth sensitivity}\label{app:mc-bandwidth}
This subsection keeps the data-generating designs fixed and varies only the HAC bandwidth rule. Table~\ref{tab:mc_extra_bandwidth} reports horizon-averaged variance ratios, coverage, and interval length for Designs A--C under three $T_h^{1/3}$ choices, the conventional $h+1$ rule, and the \citet{NeweyWest1994} rule. The oracle is stable across bandwidths, so the sensitivity is in the feasible HAC and regression-adjusted procedures rather than in the design variance itself. HAC remains conservative throughout, with the strongest conservativeness under the larger $T_h^{1/3}$ choice. Regression adjustment remains close to the design benchmark in Design A, yields smaller improvements in Design B, and is nearly indistinguishable from HAC in misspecified Design C. The $h+1$ rule delivers the shortest intervals, but it also lowers coverage relative to the better-balanced $T_h^{1/3}$ choices. These comparisons are consistent with the main-text interpretation: bandwidth choice is a second-order implementation issue asymptotically, but it can materially affect finite-sample conservativeness.

\begin{table}[p]
  \centering
  \small
  \setlength{\tabcolsep}{5pt}
  \begin{tabular}{@{}llccc@{}}
\toprule
Design & Bandwidth rule & \shortstack{Variance\\ratio} & Coverage & \shortstack{Mean\\length} \\
\midrule
A & Small $T_h^{1/3}$ & (1.452, 1.047) & (0.895, 0.875) & (0.708, 0.664) \\
 & Baseline $T_h^{1/3}$ & (1.587, 1.059) & (0.893, 0.869) & (0.711, 0.655) \\
 & Large $T_h^{1/3}$ & (1.755, 1.078) & (0.889, 0.862) & (0.713, 0.645) \\
 & $h+1$ & (1.255, 1.029) & (0.885, 0.863) & (0.683, 0.644) \\
 & Newey--West (1994) & (1.452, 1.047) & (0.895, 0.875) & (0.708, 0.664) \\
\midrule
B & Small $T_h^{1/3}$ & (1.700, 1.324) & (0.892, 0.882) & (0.761, 0.726) \\
 & Baseline $T_h^{1/3}$ & (1.877, 1.392) & (0.889, 0.877) & (0.764, 0.720) \\
 & Large $T_h^{1/3}$ & (2.081, 1.462) & (0.885, 0.869) & (0.766, 0.710) \\
 & $h+1$ & (1.398, 1.194) & (0.883, 0.869) & (0.733, 0.698) \\
 & Newey--West (1994) & (1.700, 1.324) & (0.892, 0.882) & (0.761, 0.726) \\
\midrule
C & Small $T_h^{1/3}$ & (1.421, 1.407) & (0.895, 0.892) & (0.694, 0.687) \\
 & Baseline $T_h^{1/3}$ & (1.548, 1.527) & (0.892, 0.888) & (0.697, 0.687) \\
 & Large $T_h^{1/3}$ & (1.707, 1.677) & (0.889, 0.883) & (0.699, 0.685) \\
 & $h+1$ & (1.238, 1.231) & (0.885, 0.879) & (0.670, 0.659) \\
 & Newey--West (1994) & (1.421, 1.407) & (0.895, 0.892) & (0.694, 0.687) \\
\bottomrule
\end{tabular}

  \caption{Bandwidth sensitivity. Entries average variance ratios, pointwise coverage, and mean interval length over horizons and fixed environments. In each parenthesized pair, the first entry is HAC and the second is RA-HAC.}
  \label{tab:mc_extra_bandwidth}
\end{table}

\subsection{Additional diagnostic designs}\label{app:mc-stress}
Table~\ref{tab:mc_extra_stress} collects two diagnostic designs that illustrate when the feasible procedures require additional structure. Panel A reports a state-feedback design in which the adjustment variable moves under shock re-randomization. There the adjustment variable is affected by re-randomized assignment shocks, HAC understates the design variance, and regression adjustment performs poorly: the adjusted variance ratio falls to $0.033$ and coverage to $0.238$, far below the oracle. This is therefore a substantive failure of regression adjustment rather than a minor quantitative deterioration.

Panel B reports the misspecified overidentified GMM design. There the main issue is variance misspecification rather than regression adjustment: the naive variance estimator is severely downward biased, while the misspecification-robust alternative is close to the design benchmark and restores coverage near the nominal level. Regression adjustment is a variance-reduction device when predetermined covariates track the predictable date-specific mean path. When adjustment variables are affected by re-randomized assignment shocks, the additional orthogonality conditions behind the adjusted lower bound fail. Overidentified settings require misspecification-robust variance formulas once the extra moments are not exactly correct.

\begin{table}[!htbp]
  \centering
  \small
  \setlength{\tabcolsep}{5pt}
  \begin{tabular}{@{}llcc@{}}
\toprule
Diagnostic design & Method & \shortstack{Variance\\ratio} & Coverage \\
\midrule
\multicolumn{4}{@{}l}{\emph{Panel A. State-feedback adjustment failure, $V_R=29.865$}} \\
State feedback & HAC & 0.631 & 0.826 \\
 & RA-HAC & 0.033 & 0.238 \\
 & Oracle & -- & 0.908 \\
\addlinespace[0.25em]
\multicolumn{4}{@{}l}{\emph{Panel B. Misspecified overidentified GMM, $V_R=5.032$}} \\
Misspecified GMM & Naive & 0.392 & 0.722 \\
 & Misspecification-robust & 0.971 & 0.905 \\
 & Oracle & -- & 0.920 \\
\bottomrule
\end{tabular}

  \caption{Additional diagnostic designs. Panel A reports a state-feedback design in which the adjustment variable moves under shock re-randomization. Panel B reports the misspecified overidentified GMM design. RA-HAC denotes regression-adjusted HAC.}
  \label{tab:mc_extra_stress}
\end{table}

\subsection{Simultaneous coverage}\label{app:mc-simultaneous-extra}
The main text studies pointwise inference by horizon, but pointwise intervals should not be read as simultaneous bands, so explicit simultaneous constructions are needed when uniform coverage across horizons is the object. Table~\ref{tab:mc_extra_simultaneous} reports the corresponding results. Average pointwise coverage remains close to the nominal $90\%$ level for HAC, regression adjustment, and the oracle. If one reinterprets those same pointwise intervals as uniform bands, simultaneous coverage falls sharply, to about one-half in both designs. Bonferroni and \v{S}id{\'a}k bands are both valid under the fixed-horizon Gaussian limit in Appendix~\ref{app:simultaneous}; the table reports Bonferroni as a conservative reference. The oracle max-$t$ benchmark is less conservative than Bonferroni but still has below-nominal coverage in finite samples.

\begin{table}[p]
  \centering
  \small
  \setlength{\tabcolsep}{5pt}
  \begin{tabular}{@{}llccccc@{}}
\toprule
Design & Method & \shortstack{Average\\pointwise} & \shortstack{Sim.,\\pointwise} & \shortstack{Sim.,\\Bonferroni} & \shortstack{Sim.,\\oracle max-$t$} & \shortstack{Oracle\\critical value} \\
\midrule
A & HAC & 0.892 & 0.531 & 0.913 & 0.874 & 2.465 \\
A & Oracle & 0.877 & 0.456 & 0.903 & 0.855 & 2.465 \\
A & RA-HAC & 0.870 & 0.446 & 0.888 & 0.837 & 2.465 \\
B & HAC & 0.890 & 0.529 & 0.909 & 0.873 & 2.483 \\
B & Oracle & 0.879 & 0.466 & 0.916 & 0.874 & 2.483 \\
B & RA-HAC & 0.878 & 0.489 & 0.896 & 0.857 & 2.483 \\
\bottomrule
\end{tabular}

  \caption{Simultaneous coverage over a fixed horizon set. The table compares average pointwise coverage with simultaneous coverage under pointwise critical values, Bonferroni bands, and the oracle max-$t$ benchmark. RA-HAC denotes regression-adjusted HAC.}
  \label{tab:mc_extra_simultaneous}
\end{table}

\clearpage

\FloatBarrier
\section{Additional empirical application diagnostics}\label{sec:app-empirical}
This appendix collects the additional monetary results. Figure~\ref{fig:app_money_linear_additional} reports the remaining linear local projections using the rich macro covariate set, while Figure~\ref{fig:app_money_state_additional} reports two additional state-dependent specifications: the slack specification for the corporate spread and the zero-lower-bound (ZLB) specification for CPI\@. Table~\ref{tab:app_money_linear_summary} and Table~\ref{tab:app_money_state_summary} provide the full summary of the reduction in standard errors across local-projection specifications, and Table~\ref{tab:app_money_varx_summary} reports the VARX impulse-response results discussed in the main text.

\begin{figure}[H]
  \centering
  \includegraphics[width=0.82\linewidth]{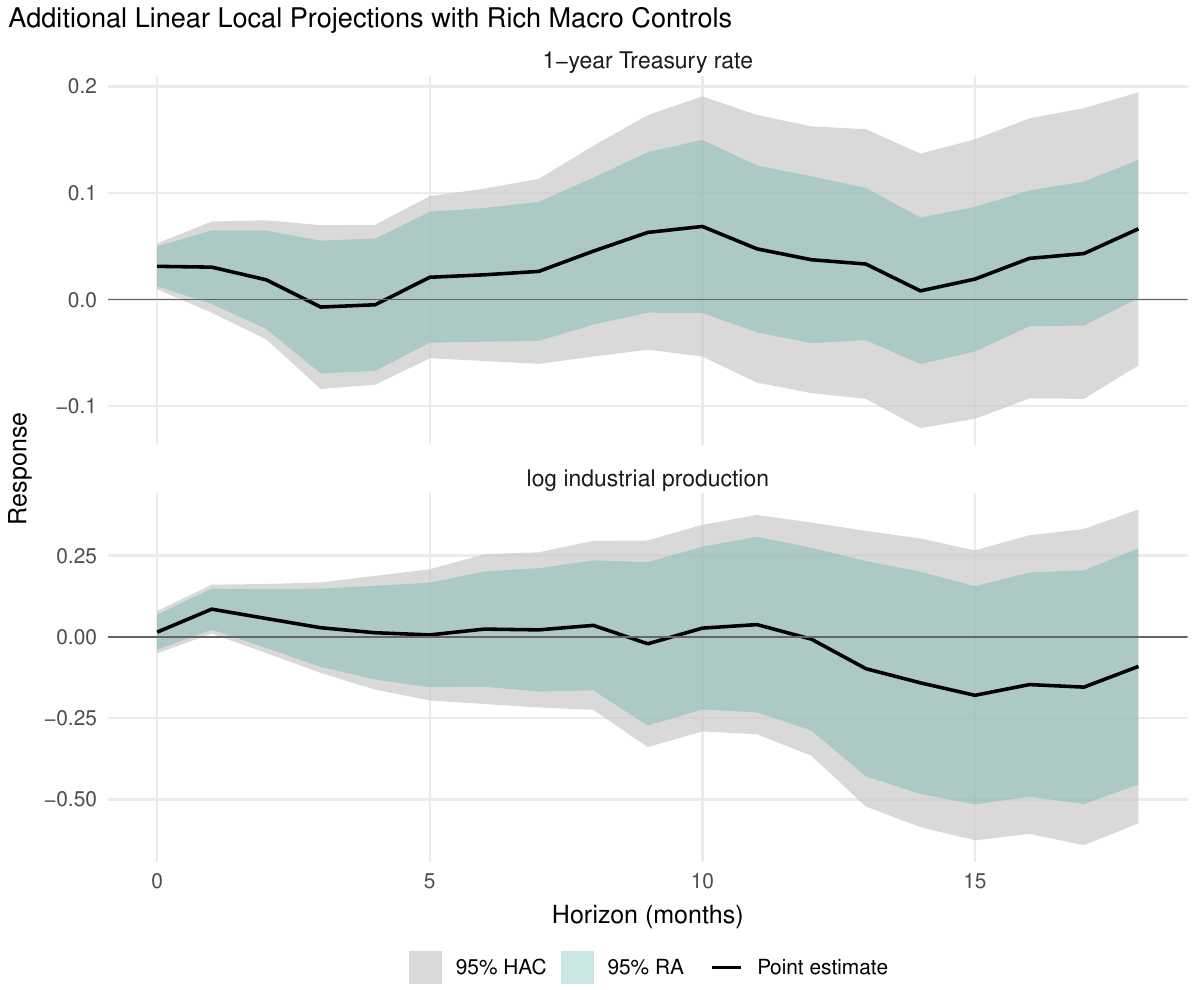}
  \caption{Additional linear local projections using the rich macro covariate set. The panels report the 1-year Treasury rate and log industrial production together with 95\% HAC bands and corresponding regression-adjusted HAC bands; the latter are interpreted with the qualifications in Section~\ref{sec:implementation}.}
  \label{fig:app_money_linear_additional}
\end{figure}

\begin{figure}[H]
  \centering
  \includegraphics[width=0.84\linewidth]{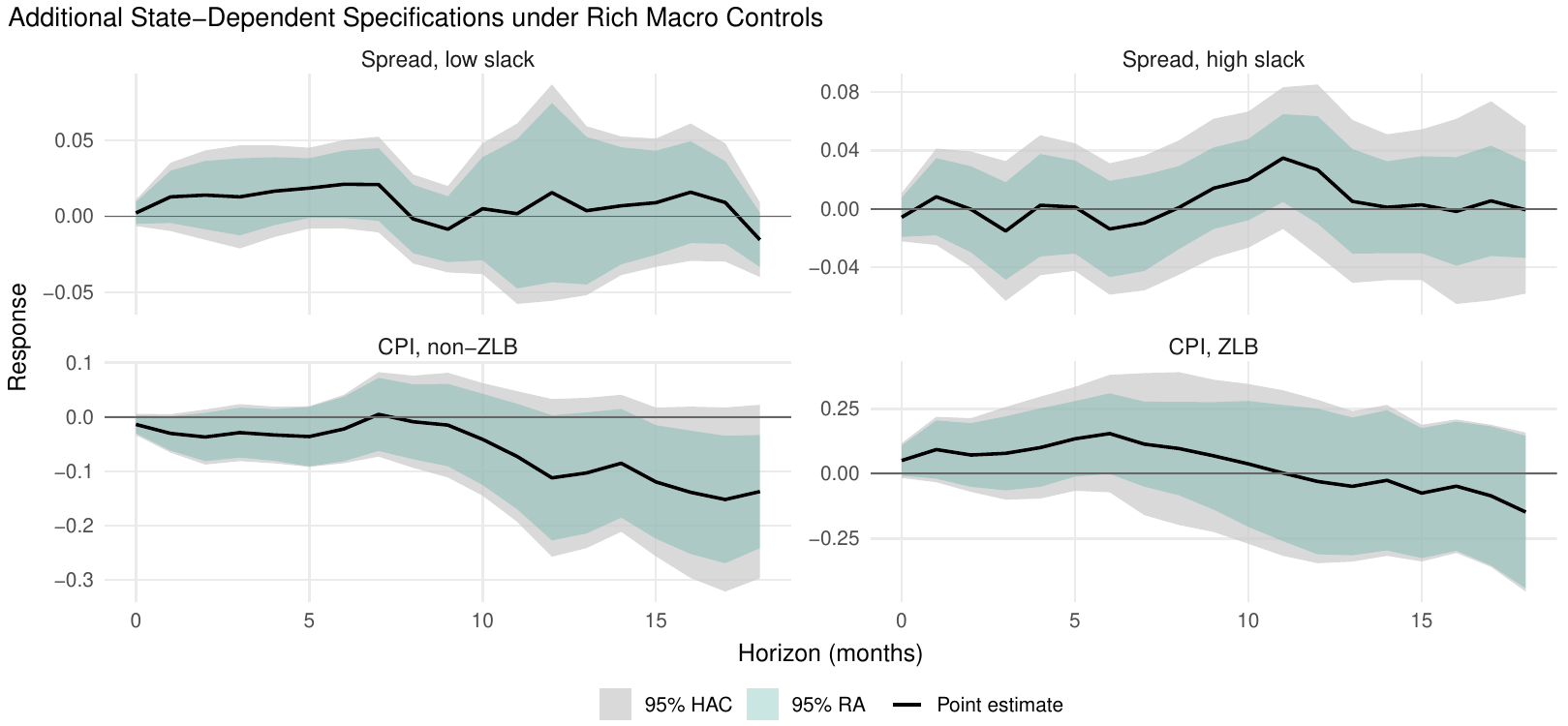}
  \caption{Additional state-dependent specifications using the rich macro covariate set. The top row reports the slack split for the BAA--AAA spread; the bottom row reports the zero-lower-bound (ZLB) split for CPI.}
  \label{fig:app_money_state_additional}
\end{figure}

\begin{table}[p]
  \centering
  \footnotesize
  \setlength{\tabcolsep}{4.3pt}
  \renewcommand{\arraystretch}{1.06}
  \begin{tabular}{llcccc}
    \toprule
    Outcome & $Z_t$ set & Mean reduction & Median reduction & Share positive & HAC$\to$RA flips\\
    \midrule
    log CPI & Parsimonious & 3.71 & 3.27 & 1.00 & 0\\
     & Macro lags & 8.70 & 8.75 & 0.89 & 1\\
     & Rich macro & 16.12 & 16.80 & 0.95 & 4\\
    \addlinespace[0.25em]
    BAA--AAA spread & Parsimonious & 1.57 & 1.75 & 1.00 & 0\\
     & Macro lags & 13.51 & 12.82 & 1.00 & 0\\
     & Rich macro & 22.02 & 21.74 & 1.00 & 1\\
    \addlinespace[0.25em]
    1-year Treasury rate & Parsimonious & 8.95 & 8.71 & 1.00 & 0\\
     & Macro lags & 22.91 & 20.97 & 1.00 & 0\\
     & Rich macro & 32.27 & 31.58 & 1.00 & 1\\
    \addlinespace[0.25em]
    log industrial production & Parsimonious & 3.75 & 1.58 & 1.00 & 0\\
     & Macro lags & 13.89 & 14.85 & 1.00 & 0\\
     & Rich macro & 20.81 & 21.19 & 1.00 & 0\\
    \bottomrule
  \end{tabular}
  \caption{Full linear-local-projection tightening summary for the monetary application. Entries are percentage reductions in the HAC standard error over horizons $h=0,\ldots,18$.}
  \label{tab:app_money_linear_summary}
\end{table}

\begin{table}[p]
  \centering
  \footnotesize
  \setlength{\tabcolsep}{3pt}
  \renewcommand{\arraystretch}{1.05}
  \begin{minipage}{0.96\linewidth}
    \centering
    \textbf{Panel A. CPI}\\[0.25em]
    \begin{tabular*}{\linewidth}{@{\extracolsep{\fill}}p{0.19\linewidth}p{0.21\linewidth}p{0.21\linewidth}ccc@{}}
      \toprule
      Specification & Coefficient & $Z_t$ Set & Mean & Median & Flips\\
      \midrule
      Slack split & Low / non-ZLB & Parsimonious & 3.74 & 1.94 & 0\\
       &  & Macro lags & 6.02 & 8.30 & 0\\
       &  & Rich macro & 12.56 & 13.15 & 1\\
       & High / ZLB & Parsimonious & 6.04 & 5.21 & 4\\
       &  & Macro lags & 22.33 & 21.32 & 6\\
       &  & Rich macro & 19.96 & 17.89 & 5\\
      \addlinespace[0.25em]
      ZLB split & Low / non-ZLB & Parsimonious & 4.32 & 3.03 & 0\\
       &  & Macro lags & 6.73 & 7.61 & 3\\
       &  & Rich macro & 17.79 & 19.01 & 4\\
       & High / ZLB & Parsimonious & 1.06 & 1.11 & 0\\
       &  & Macro lags & 15.48 & 13.09 & 0\\
       &  & Rich macro & 17.24 & 13.53 & 0\\
      \bottomrule
    \end{tabular*}
  \end{minipage}
  \vspace{0.55em}
  \begin{minipage}{0.96\linewidth}
    \centering
    \textbf{Panel B. BAA--AAA Spread}\\[0.25em]
    \begin{tabular*}{\linewidth}{@{\extracolsep{\fill}}p{0.19\linewidth}p{0.21\linewidth}p{0.21\linewidth}ccc@{}}
      \toprule
      Specification & Coefficient & $Z_t$ Set & Mean & Median & Flips\\
      \midrule
      Slack split & Low / non-ZLB & Parsimonious & 1.21 & 0.96 & 0\\
       &  & Macro lags & 14.00 & 15.15 & 0\\
       &  & Rich macro & 22.17 & 23.42 & 0\\
       & High / ZLB & Parsimonious & 3.51 & 0.53 & 0\\
       &  & Macro lags & 23.34 & 20.93 & 1\\
       &  & Rich macro & 33.47 & 35.81 & 1\\
      \addlinespace[0.25em]
      ZLB split & Low / non-ZLB & Parsimonious & 1.38 & 1.29 & 0\\
       &  & Macro lags & 12.77 & 14.98 & 2\\
       &  & Rich macro & 19.29 & 20.37 & 2\\
       & High / ZLB & Parsimonious & 0.72 & 0.43 & 0\\
       &  & Macro lags & 24.16 & 21.07 & 2\\
       &  & Rich macro & 34.05 & 33.02 & 5\\
      \bottomrule
    \end{tabular*}
  \end{minipage}
  \caption{State-dependent local-projection tightening summary for the monetary application. Entries are percentage standard-error reductions over horizons $h=0,\ldots,18$; \emph{Flips} counts HAC$\to$RA significance changes.}
  \label{tab:app_money_state_summary}
\end{table}

\begin{table}[p]
  \centering
  \footnotesize
  \setlength{\tabcolsep}{4.3pt}
  \renewcommand{\arraystretch}{1.06}
  \begin{tabular}{llcccc}
    \toprule
    Impulse response & $Z_t$ set & Mean reduction & Median reduction & Share positive & HAC$\to$RA flips\\
    \midrule
    Cumulative inflation IRF & Parsimonious & 7.00 & 7.97 & 1.00 & 0\\
     & Macro lags & 16.78 & 17.19 & 1.00 & 0\\
     & Rich macro & 25.71 & 25.43 & 1.00 & 0\\
    \addlinespace[0.25em]
    Spread IRF & Parsimonious & 0.88 & 0.83 & 0.95 & 0\\
     & Macro lags & 14.60 & 12.62 & 1.00 & 3\\
     & Rich macro & 19.73 & 20.62 & 1.00 & 2\\
    \addlinespace[0.25em]
    GS1 IRF & Parsimonious & 10.32 & 8.53 & 1.00 & 0\\
     & Macro lags & 12.92 & 10.73 & 1.00 & 0\\
     & Rich macro & 15.77 & 13.76 & 1.00 & 0\\
    \addlinespace[0.25em]
    Cumulative IP IRF & Parsimonious & 4.66 & 3.96 & 1.00 & 0\\
     & Macro lags & 3.69 & 3.18 & 1.00 & 0\\
     & Rich macro & 14.86 & 14.98 & 1.00 & 0\\
    \bottomrule
  \end{tabular}
  \caption{Beyond-LP proof of concept: VARX/IRF tightening on the same monetary dataset. \emph{Cumulative} responses cumulate monthly inflation and IP-growth responses to the policy shock.}
  \label{tab:app_money_varx_summary}
\end{table}

\subsection{Sensitivity checks}\label{app:money-robustness}
The bandwidth and covariate-set sensitivity exercises use the same monthly sample, 1990:02--2019:08, and the same monetary-shock series as the main application. The appendix reports linear and state-dependent specifications over horizons $h=0,\ldots,18$, varying the HAC bandwidth, the predetermined adjustment set, and the interpretation of significance when the object is a fixed horizon set rather than a pointwise statement.

The additional diagnostics sharpen three points: first, the magnitude of the reductions is not an artifact of the bandwidth choice. Table~\ref{tab:app_money_bandwidth_robustness} shows that, for the linear projections using the rich macro covariate set, average standard-error reductions remain substantial across the alternative HAC bandwidth rules considered here.

Second, the source of the standard-error reductions is economically informative. Table~\ref{tab:app_money_feature_ablation} shows that intercept-only residualization is essentially uninformative, deterministic trends alone produce only modest reductions, and the larger gains arise when the predetermined set is expanded to include macro and financial information available at date $t-1$.

Third, pointwise significance gains in the application should not be interpreted as uniform over horizons. Table~\ref{tab:app_money_simultaneous_sensitivity} shows that, once one applies the fixed-horizon Bonferroni adjustment from Appendix~\ref{app:simultaneous}, most horizons at which pointwise significance changes under regression adjustment no longer show a simultaneous rejection. \v{S}id{\'a}k gives the same asymptotic simultaneous guarantee under the fixed-horizon Gaussian limit, with a slightly smaller critical value. This is consistent with the paper's broader message: regression adjustment can materially narrow reported HAC bands, but whether those narrower bands are valid conservative design intervals depends on the inferential object and on the additional orthogonality conditions.

\begin{table}[t!]
  \centering
  \small
  \setlength{\tabcolsep}{4.5pt}
  \begin{tabular*}{\textwidth}{@{\extracolsep{\fill}}llccccc@{}}
    \toprule
    Outcome & Bandwidth & \shortstack{Mean\\reduction} & \shortstack{Median\\reduction} & \shortstack{Share\\positive} & \shortstack{Pointwise\\reversals} & \shortstack{\v{S}id\'{a}k\\reversals}\\
    \midrule
    log CPI & $h+1$ & 16.12 & 16.80 & 0.95 & 4 & 0\\
     & $0.75(h+1)$ & 14.32 & 14.67 & 1.00 & 4 & 0\\
     & $1.5(h+1)$ & 19.46 & 20.94 & 1.00 & 4 & 0\\
     & $n^{1/3}$ & 13.14 & 13.73 & 0.95 & 4 & 0\\
     & 12 & 18.97 & 19.68 & 1.00 & 4 & 0\\
    \addlinespace[0.25em]
    BAA--AAA spread & $h+1$ & 22.02 & 21.74 & 1.00 & 1 & 0\\
     & $0.75(h+1)$ & 20.76 & 19.50 & 1.00 & 0 & 0\\
     & $1.5(h+1)$ & 22.91 & 21.43 & 1.00 & 1 & 0\\
     & $n^{1/3}$ & 20.77 & 19.50 & 1.00 & 1 & 0\\
     & 12 & 22.99 & 24.05 & 1.00 & 1 & 0\\
    \addlinespace[0.25em]
    1-year Treasury rate & $h+1$ & 32.27 & 31.58 & 1.00 & 1 & 1\\
     & $0.75(h+1)$ & 26.46 & 24.80 & 1.00 & 0 & 0\\
     & $1.5(h+1)$ & 39.47 & 42.49 & 1.00 & 2 & 1\\
     & $n^{1/3}$ & 23.30 & 23.37 & 1.00 & 0 & 1\\
     & 12 & 34.43 & 35.45 & 1.00 & 0 & 1\\
    \bottomrule
  \end{tabular*}
  \caption{Bandwidth robustness in the monetary application. The table reports the rich-macro regression-adjustment results for the main linear local projections over horizons $h=0,\ldots,18$.}
  \label{tab:app_money_bandwidth_robustness}
\end{table}

\begin{table}[p]
  \centering
  \small
  \setlength{\tabcolsep}{4pt}
  \begin{tabular*}{\textwidth}{@{\extracolsep{\fill}}llccccc@{}}
    \toprule
    Outcome & Adjustment set & \shortstack{Mean\\reduction} & \shortstack{Median\\reduction} & \shortstack{Share\\positive} & \shortstack{Pointwise\\reversals} & \shortstack{\v{S}id\'{a}k\\reversals}\\
    \midrule
    log CPI & Intercept only & 0.00 & 0.00 & 0.11 & 0 & 0\\
     & Trends & 3.71 & 3.27 & 1.00 & 0 & 0\\
     & Macro lags, no trends & 9.63 & 9.44 & 0.95 & 1 & 0\\
     & Macro lags + trends & 8.70 & 8.75 & 0.89 & 1 & 0\\
     & Financial lags + trends & 11.49 & 9.19 & 0.95 & 4 & 0\\
     & Labor/commodity + trends & 7.27 & 6.83 & 1.00 & 0 & 0\\
     & Rich macro, no trends & 15.24 & 16.86 & 0.95 & 4 & 0\\
     & Rich macro & 16.12 & 16.80 & 0.95 & 4 & 0\\
    \addlinespace[0.25em]
    BAA--AAA spread & Intercept only & -0.00 & 0.00 & 0.05 & 0 & 0\\
     & Trends & 1.57 & 1.75 & 1.00 & 0 & 0\\
     & Macro lags, no trends & 10.52 & 9.82 & 1.00 & 0 & 0\\
     & Macro lags + trends & 13.51 & 12.82 & 1.00 & 0 & 0\\
     & Financial lags + trends & 7.75 & 8.11 & 1.00 & 0 & 0\\
     & Labor/commodity + trends & 7.55 & 7.43 & 1.00 & 0 & 0\\
     & Rich macro, no trends & 19.36 & 19.23 & 1.00 & 1 & 0\\
     & Rich macro & 22.02 & 21.74 & 1.00 & 1 & 0\\
    \bottomrule
  \end{tabular*}
  \caption{Feature-set ablations for regression-adjusted HAC in the monetary application. The baseline bandwidth is $h+1$. The intercept-only row removes only the sample mean of the influence-function path.}
  \label{tab:app_money_feature_ablation}
\end{table}

\begin{table}[t!]
  \centering
  \small
  \setlength{\tabcolsep}{4.5pt}
  \begin{tabular}{lcccc}
    \toprule
    Specification & HAC pointwise & RA pointwise & HAC \v{S}id\'{a}k & RA \v{S}id\'{a}k\\
    \midrule
    log CPI (linear) & 0 & 4 & 0 & 0\\
    1-year Treasury rate (linear) & 1 & 2 & 0 & 1\\
    BAA--AAA spread (linear) & 0 & 1 & 0 & 0\\
    CPI slack split: Low slack/non-ZLB & 0 & 1 & 0 & 0\\
    CPI slack split: High slack/ZLB & 0 & 5 & 0 & 0\\
    \bottomrule
  \end{tabular}
  \caption{Pointwise versus fixed-horizon \v{S}id\'{a}k significance counts in the monetary application. Counts are over horizons $h=0,\ldots,18$ using the rich-macro adjustment set and the baseline bandwidth $h+1$.}
  \label{tab:app_money_simultaneous_sensitivity}
\end{table}

The selected pointwise significance-change table reports horizons at which the rich macro covariate set changes pointwise significance relative to baseline HAC. It is included as a descriptive diagnostic; the simultaneous-sensitivity table above is the relevant reference when the inferential object is a fixed horizon set.

\begin{table}[t!]
  \centering
  \small
  \setlength{\tabcolsep}{4.7pt}
  \renewcommand{\arraystretch}{1.08}
  \begin{tabular}{lccccc}
    \toprule
    Specification & $h$ & Estimate & HAC s.e. & RA s.e. & HAC sig. / RA sig.\\
    \midrule
    log CPI (linear LP) & 15 & -0.100 & 0.063 & 0.048 & No / Yes\\
    log CPI (linear LP) & 18 & -0.115 & 0.073 & 0.049 & No / Yes\\
    BAA--AAA spread (linear LP) & 5 & 0.015 & 0.011 & 0.008 & No / Yes\\
    CPI, high slack & 13 & -0.191 & 0.099 & 0.085 & No / Yes\\
    CPI, low slack & 17 & -0.117 & 0.082 & 0.059 & No / Yes\\
    CPI, high slack & 18 & -0.206 & 0.122 & 0.100 & No / Yes\\
    \bottomrule
  \end{tabular}
  \caption{Selected inference-relevant horizons under the rich-macro adjustment. These rows summarize cases in which regression adjustment changes the practical interpretation of the response by turning an insignificant coefficient under HAC into a significant coefficient under regression-adjusted HAC.}
  \label{tab:money_selected_flips}
\end{table}

\FloatBarrier

\clearpage

\section{Additional macro estimators}\label{app:macro-estimators}
This section is a translation guide showing how several standard macro estimators instantiate the moment-vector notation and where the mean-path term enters.

Throughout, let $x_t$ be the randomized design shock, $c_t$ a vector of predetermined controls, and $z_t$ an external instrument measurable with respect to the relevant pre-shock information set and independent of potential outcomes beyond its correlation with the shock of interest. If the realized path of an instrument, control, or regime indicator is included in $\mathcal E_T$, it is fixed under the design; otherwise it moves with past assignment shocks or with the auxiliary observation rule used to construct it. Let $\psi_t=(1,x_t,c_t^\top)^\top$, and let the sample-period design estimand be $\theta_T^\star\in\arg\min_\theta Q_T(\theta)$ with $Q_T(\theta)=\bar m_T(\theta)^\top A\bar m_T(\theta)$. The design long-run variance of the centered moment is denoted $\Omega_R$, and the corresponding GMM asymptotic covariance is $(G^\top A G)^{-1}G^\top A\Omega_R A G(G^\top A G)^{-1}$.

Table~\ref{tab:macro_translation} summarizes the mapping for the main macro procedures considered in the paper. The moment-vector column records the per-period contribution to the moment. The third column combines the natural conditioning set and a choice of predetermined adjustment variables with the leading source of the mean drift $\mu_t(\theta_T^\star)$.

\begin{table}[p]
  \centering
  \footnotesize
  \renewcommand{\arraystretch}{1.08}
  \setlength{\tabcolsep}{4pt}
  \begin{tabularx}{\textwidth}{@{}p{0.19\textwidth}Y Y@{}}
    \toprule
    Estimator & Per-period moment vector & Conditioning set, predetermined covariates, and main source of mean drift \\
    \midrule
    \multicolumn{3}{@{}l}{\emph{Panel A. Reduced-form response estimators}}\\[-0.2em]
    \addlinespace[0.3em]
    Local projection (LP) &
    $g^{\mathrm{LP}}_{t,h}(\theta_h)=\psi_t\big(y_{t+h}-\psi_t^\top\theta_h\big)$, with $\psi_t=(1,x_t,c_t^\top)^\top$. &
    $\I_t=\sigma(\mathcal F_{t-1},Z_t)$, with $Z_t$ a short vector of lagged values of $(x,c,y)$ and low-order deterministic terms. The mean drift is driven by $\Var(x_t\mid\I_t)(\tau_{t,h}-\beta_h^\star)$, that is, by time variation in the impulse response and/or in the shock variance. \\
    \addlinespace[0.35em]
    LP-IV (external instrument) &
    $g^{\mathrm{LPIV}}_{t,h}(\theta_h)=z_t\big(y_{t+h}-\psi_t^\top\theta_h\big)$. &
    $\I_t=\sigma(\mathcal F_{t-1},Z_t)$, with $Z_t$ containing lagged values of $(z,x,c)$ and low-order trends. The drift is driven by $\Cov(z_t,x_t\mid\I_t)(\tau_{t,h}-\beta_{T,h}^\star)$, so both time-varying instrument relevance and treatment-effect heterogeneity matter. \\
    \addlinespace[0.35em]
    VAR($p$) with observed shocks &
    $g^{\mathrm{VAR}}_t(\phi)=\operatorname{vec}\!\big(Z_tu_t(\phi)^\top\big)$, where $u_t(\phi)=Y_t-\Phi_0-\sum_{j=1}^p A_jY_{t-j}-BW_t$. &
    The natural predetermined block is $Z_t^{\mathrm{pred}}=(1,Y_{t-1}^\top,\ldots,Y_{t-p}^\top)^\top$, so $\I_t=\sigma(\mathcal F_{t-1},Z_t^{\mathrm{pred}})$. The mean drift comes from state dependence in $(\Phi_0,A_j,B)$; in the notation of Remark~\ref{rem:LP_VAR_mean_drift}(ii), it is summarized by the blocks $Z_t^{\mathrm{pred}}\Delta_t^\top$ and $\Sigma_{W,t}(B(Z_t)-B)^\top$. \\
    \midrule
    \multicolumn{3}{@{}l}{\emph{Panel B. Structural and second-moment systems}}\\[-0.2em]
    \addlinespace[0.3em]
    Proxy SVAR / SVAR-IV &
    $g^{\mathrm{SVARIV}}_t(\theta)=\big(u_t z_t-\pi b_1,\ e_i^\top b_1-1\big)^\top$. &
    $\I_t=\sigma(\mathcal F_{t-1},Z_t)$, with $Z_t$ collecting lags of $(y_t,z_t)$ and regime indicators for instrument strength. The mean drift is $(\pi_t b_1(Z_t)-\pi^\star b_1^\star,0)^\top$, so time variation in the instrument signal or in the impact vector both matter. \\
    \addlinespace[0.35em]
    Heteroskedastic SVAR &
    $g^{\mathrm{HSVAR}}_t(\theta)=\bigoplus_{r=1}^R d_{r,t}\,\operatorname{vech}\big(u_tu_t^\top-B\Lambda_rB^\top\big)$. &
    $\I_t=\sigma(\mathcal F_{t-1},d_{1,t},\ldots,d_{R,t})$, and the predetermined covariates stack regime dummies with low-order time dummies. The drift is driven by state dependence in $B(Z_t)$ or in the regime-specific variances $\Lambda_{r,t}$. \\
    \addlinespace[0.35em]
    MD-GMM on\\ autocovariances &
    $g^{\mathrm{MD}}_t(\theta)=\bigoplus_{k\in\mathcal K}\operatorname{vec}\big(r_{k,t}-\Gamma_k(\theta)\big)$, with $r_{k,t}=y_ty_{t-k}^\top$. &
    $\I_t=\sigma(\mathcal F_{t-1},y_{t-1},\ldots,y_{t-K})$, and the predetermined covariates use these lags together with regime indicators when relevant. The mean drift is driven by time variation in the matched second moments, namely $\mathbb{E}[y_ty_{t-k}^\top\mid\I_t]-\Gamma_k(\theta_T^\star)$. \\
    \bottomrule
  \end{tabularx}

  \vspace{0.35em}
  \begin{minipage}{\textwidth}
  \footnotesize \emph{Notes.} The table summarizes how the examples instantiate the moment-vector and conditioning-set notation; it is not a formal result. The detailed derivations, including the explicit form of the mean-path term in each case, are given in the subsections below. In each row, $Z_t$ denotes the predetermined covariate vector dated before the shock and used for regression adjustment, and $\operatorname{vech}$ stacks the distinct entries of a symmetric matrix.
  \end{minipage}
  \caption{Common macro estimators in the design-based GMM setup.}
  \label{tab:macro_translation}
\end{table}

\paragraph{A. Local Projections with an External Instrument (LP-IV).}
Fix a horizon $h\ge 0$. Consider the linear projection
$y_{t+h}=\alpha_h+\beta_h x_t+\gamma_h^\top c_t+u_{t,h}$, instrumented with $z_t$ for $x_t$.
Let $\theta_h=(\alpha_h,\beta_h,\gamma_h^\top)^\top$. Define the (possibly overidentified) IV moments by $g^{\mathrm{LPIV}}_{t,h}(\theta_h):= z_t\,(y_{t+h}-\psi_t^\top\theta_h)\in\mathbb{R}^{d_z}$. The sample-period IV estimand $\theta_{T,h}^\star$ solves
$T^{-1}\sum_t\mathbb{E}[z_t(y_{t+h}-\psi_t^\top\theta_{T,h}^\star)]=0$ in the just-identified or locally correctly specified case. In the overidentified case, define $\theta_{T,h}^\star$ as the minimizer of $\bar m_{T,h}(\theta)^\top A\bar m_{T,h}(\theta)$, where $\bar m_{T,h}(\theta):=T^{-1}\sum_t\mathbb E[z_t(y_{t+h}-\psi_t^\top\theta)]$. Write the $h$-period causal effect as
$\tau_{t,h}:=Y_{t+h}^{(1)}-Y_{t+h}^{(0)}$ in the direct potential-outcome notation.

At $\theta_{T,h}^\star$, the instrument-direction drift is $\mu_{t,h}:=\mathbb{E}[g^{\mathrm{LPIV}}_{t,h}(\theta_{T,h}^\star)\mid\I_t]=\Cov(z_t,x_t\mid\I_t)(\tau_{t,h}-\beta_{T,h}^\star)$, provided the instrument is mean-zero and orthogonal to $c_t$ in design and the conditional mean is correctly absorbed by the control block. In that case, the mean path is constant (and hence zero) whenever
$\Cov(z_t,x_t\mid\I_t)\,(\tau_{t,h}-\beta_{T,h}^\star)$ is constant in $t$.
This is the IV analogue of the LP drift in Remark~\ref{rem:LP_VAR_mean_drift}(i). If the variables entering $y_{t+h}$, $x_t$, $c_t$, and $z_t$ are measurable with respect to a fixed finite shock window and the underlying design-shock array is finite-dependent, then the LP-IV moment is finite-dependent after enlarging the window by a constant depending on $h$. In that special case, a finite flat-top HAC with a kernel whose $K=1$ region covers the finite dependence window estimates the corresponding HAC variance limit $\Omega_R^+$. This limit equals the design long-run variance $\Omega_R$ of the centered moment only when the centered mean-path component has zero long-run variance, or after an adjustment step has removed that component under the fully observed long-run orthogonality conditions in Proposition~\ref{prop:HAC_RA_conservative}.
Asymptotic normality for $\hat\theta_h$ follows from Theorem~\ref{thm:AN} with
$G=\lim_T T^{-1}\sum_t \mathbb{E}_T[\nabla_\theta g^{\mathrm{LPIV}}_{t,h}]$
and weighting matrix $A$.

\paragraph{B. Proxy-SVAR / SVAR-IV.}
Let $u_t$ be the reduced-form innovations from a stable VAR($p$) for $y_t$, and suppose an external
instrument $z_t$ satisfies $u_t=B(Z_t)\varepsilon_t$ with $\varepsilon_t$ structural shocks and
$\mathbb{E}[z_t\varepsilon_{1t}]=\pi_t\neq 0$, $\mathbb{E}[z_t\varepsilon_{jt}]=0$ for $j\neq 1$.
The object is the impact column $b_1$ of $B$ associated with shock~1. Define the parameter $\theta=(b_1,\pi)$ with normalization $e_i^\top b_1=1$ for some chosen $i$. Consider moments $g^{\mathrm{SVARIV}}_{t}(\theta):=(u_t z_t - \pi\, b_1,\ e_i^\top b_1 - 1)^\top\in\mathbb{R}^{n+1}$. The first block enforces $\mathbb{E}[u_t z_t]=\pi\, b_1$ (the covariance restriction), and the second fixes scale. At $\theta_T^\star=(b_1^\star,\pi^\star)$, $\mu_t:=\mathbb{E}[g^{\mathrm{SVARIV}}_{t}(\theta_T^\star)\mid\I_t]=(\pi_t\,b_1(Z_t)-\pi^\star b_1^\star,0)^\top$. Hence the mean moment is zero whenever $\pi_t\,b_1(Z_t)\equiv \pi^\star b_1^\star$; time-invariance of both the instrument signal $\pi_t$ and the
impact vector $b_1(Z_t)$ is a convenient sufficient condition. State dependence in $B(Z_t)$ or time-variation
in $\pi_t$ can therefore generate a nonzero mean path even at the best-fitting $(b_1^\star,\pi^\star)$.

With stable reduced-form VAR dynamics the influence function is short-memory, so standard
design HAC with diverging bandwidth is conservative and adds the mean-path term. A finite
window aligned to exact finite dependence estimates the corresponding HAC variance limit,
which equals the design long-run variance only when the centered mean-path component has
zero long-run variance or has been removed by a valid adjustment.

\paragraph{C. Heteroskedasticity-identified SVAR (variance changes).}
Let $u_t=B\varepsilon_t$ with diagonal $\Var(\varepsilon_t\mid R_t^{\mathrm{reg}}=r)=\Lambda_r$ for regimes
$r=1,\dots,R$ indicated by $d_{r,t}=\1\{R_t^{\mathrm{reg}}=r\}$ (regime indicators are fixed if their paths are included in $\mathcal E_T$, and otherwise move with the assignment or state history).
Consider the parameter $\theta=(\operatorname{vec}B, (\operatorname{vech}\Lambda_r)_{r=1}^R)$ and the stacked second-moment conditions
$g^{\mathrm{HSVAR}}_{t}(\theta):=\bigoplus_{r=1}^R d_{r,t}\,\operatorname{vech}(u_t u_t^\top- B\Lambda_r B^\top)\in\mathbb{R}^{R\,n(n+1)/2}$, where $\operatorname{vech}$ stacks the distinct entries of a symmetric matrix. For each regime $r$, $\mathbb{E}[u_tu_t^\top\mid R_t^{\mathrm{reg}}=r]=B\Lambda_r B^\top$. At $\theta_T^\star=(B^\star,(\Lambda_r^\star)_{r=1}^R)$, $\mu_t=\bigoplus_{r=1}^R \mathbb{E}[d_{r,t}\,\operatorname{vech}(B(Z_t)\Lambda_{r,t}B(Z_t)^\top - B^\star \Lambda_r^\star B^{\star\top})]$. Thus $\mu_t\equiv 0$ whenever $B(Z_t)\Lambda_{r,t}B(Z_t)^\top\equiv B^\star \Lambda_r^\star B^{\star\top}$ within each regime; the stronger conditions
$B(Z_t)\equiv B^\star$ and $\Lambda_{r,t}\equiv \Lambda_r^\star$ are sufficient but not necessary.

If, in addition, $(u_t,R_t^{\mathrm{reg}})$ is serially independent, then the moment process is $0$-dependent in $t$ and
a flat-top HAC with $L=0$ estimates the corresponding HAC variance limit for these moments. This limit equals the design long-run variance of the centered moment only when the centered mean-path component has zero long-run variance, or after an adjustment step has removed that component under the fully observed long-run orthogonality conditions in Proposition~\ref{prop:HAC_RA_conservative}; otherwise one should use the short-memory HAC conditions from the main text.

\paragraph{D. Factor-augmented VARs (FAVAR).}
Let $X_t:=(y_t^\top,F_t^\top)^\top$ stack observables $y_t$ and latent factors $F_t$, and let $d_X:=\dim(X_t)$. Suppose $\widehat F_t$ is estimated in a first step (e.g., by large-$N$ PCA) and then fit a stable
VAR($p$) in the augmented state. Define $Q_t^\top=(1,X_{t-1}^\top,\dots,X_{t-p}^\top,x_t^\top)$ and $Q_t^{\mathrm{pred}}:=(1,X_{t-1}^\top,\dots,X_{t-p}^\top)^\top$. The second-stage moment is $g^{\mathrm{FAVAR}}_{t}(\varphi):=(I_{d_X}\otimes Q_t)\,u_t(\varphi)$, where $u_t(\varphi)=X_t-\Phi_0-\sum_{j=1}^pA_j X_{t-j}-B x_t$, so the example inherits Remark~\ref{rem:LP_VAR_mean_drift}(ii).

With state dependence $\Phi_0(Z_t),A_j(Z_t),B(Z_t)$, set $\Delta_t:=(\Phi_0(Z_t)-\Phi_0)+\sum_{j=1}^p(A_j(Z_t)-A_j)X_{t-j}$ and let $\Sigma_{x,t}:=\mathbb E[x_t x_t^\top\mid\I_t]$. Then
\[
\mathbb E[g^{\mathrm{FAVAR}}_t(\varphi_T^\star)\mid\I_t]
=
\operatorname{vec}\!\begin{pmatrix}
Q_t^{\mathrm{pred}}\Delta_t^\top\\[0.2em]
\Sigma_{x,t}(B(Z_t)-B)^\top
\end{pmatrix}.
\]
Thus $\mu_t\equiv 0$ whenever both stacked blocks vanish; in particular this holds if $\Phi_0$, $A_j$, and $B$ do not vary with $t$. If the factor-estimation error is negligible in the second-stage moment at the $T^{-1/2}$ scale and in the HAC product averages, substituting $\widehat F_t$ leaves the first-order mean-path and variance calculations unchanged. Otherwise the factor-estimation step should be included through the augmented-moment construction in Appendix~\ref{app:two-step}.

\paragraph{E. Minimum-distance GMM on second moments (autocovariance matching).}
Let $\mathcal{K}=\{0,1,\dots,K\}$ be a finite set of lags. For each $k\in\mathcal{K}$ define per-period
autocovariance contributions $r_{k,t}:=y_t y_{t-k}^\top$ (set $r_{k,t}=0$ if $t\le k$). A structural
model implies $\Gamma_k(\theta):=\mathbb{E}[y_t y_{t-k}^\top]$. The MD-GMM moments are $g^{\mathrm{MD}}_{t}(\theta):=\bigoplus_{k\in\mathcal{K}}\operatorname{vec}(r_{k,t}-\Gamma_k(\theta))$. At $\theta_T^\star$, $\mu_t=\bigoplus_{k\in\mathcal{K}}\operatorname{vec}(\mathbb{E}[y_t y_{t-k}^\top\mid\I_t]-\Gamma_k(\theta_T^\star))$, hence $\mu_t\equiv 0$ if and only if the second moments are time-invariant at
the lags used.

Because each block uses observations dated $t,\ldots,t-K$, if the underlying design-shock process is $q$-dependent, then the moment process is $(q+K)$-dependent; in that special case a finite flat-top HAC with a kernel whose $K=1$ region covers the resulting dependence window estimates the corresponding HAC variance limit. Otherwise the moment process is generally short-memory rather than finite-dependent, and the diverging-bandwidth HAC conditions in the main text are the appropriate route. As above, the HAC variance limit equals the design long-run variance only when the centered mean-path component has zero long-run variance, or after a valid adjustment has removed that component.

\clearpage

\end{document}